\documentclass[12pt, a4paper]{article}
\usepackage{lipsum}
\usepackage{amsmath}
\usepackage{amssymb}
\usepackage{amsthm}
\usepackage{mathtools}
\usepackage{fancyhdr}
\usepackage{graphicx}
\usepackage{verbatim}
\usepackage{latexsym}
\usepackage{mathrsfs}
\usepackage{mathdots}

\usepackage{tikz}
\usepackage{tikz-cd}

\usepackage{luatex85} 
\usepackage[all]{xy}

\usepackage{float}
\usepackage{authblk}
\usepackage{siunitx}
\usepackage{eqparbox}

\usepackage{todonotes}
\setuptodonotes{inline}

\usepackage[shortlabels]{enumitem}
\setlist[itemize]{leftmargin=*}

\usepackage{pdfpages}
\usepackage{xcolor}
\usepackage[calcwidth]{titlesec}
\usepackage{stmaryrd}
\usepackage{wallpaper}
\usepackage{setspace}
\usepackage[top = 2 cm, bottom = 3 cm, left = 2 cm, right = 2 cm]{geometry}
\usepackage{bbm}

\usepackage[UKenglish]{isodate}
\cleanlookdateon

\usepackage{anyfontsize}
\usepackage{pgf,pgfplots}
\usepackage{textcomp}
\usepackage[margin = 1cm]{caption}
\usepackage{stackengine}

\usepackage{subcaption}

\usepackage{wasysym}
\usepackage{esint} 

\usepackage[style=alphabetic,
  sorting = nyt,
  sortcites=true,
  giveninits=true,
  date=year,
  isbn=false,
  maxbibnames=99,
  maxalphanames=6,
  backend=biber]{biblatex}
\DeclareFieldFormat*{titlecase}{\MakeSentenceCase{#1}}

\DeclareSourcemap{\maps[datatype=bibtex]{\map{\step[fieldset=shorthand, null]}}}

\DeclareSourcemap{
  \maps[datatype=bibtex]{
    \map[overwrite]{
      \step[fieldsource=doi, final]
      \step[fieldset=url, null]
      \step[fieldset=eprint, null]
    }  
  }
}

\DeclareSourcemap{
  \maps[datatype=bibtex, overwrite]{
    \map{
      \step[fieldset=language, null]
      \step[fieldset=month, null]
      \step[fieldset=urldate, null]
    }
  }
}

\AtEveryBibitem{%
  \clearlist{language}%
  \clearlist{month}%
  \clearlist{urldate}%
  \clearfield{urldate}%
  \ifentrytype{inproceedings}
    {\clearlist{booktitle}%
     \clearlist{publisher}%
     \clearlist{location}}
}

\renewbibmacro{in:}{}

\DeclareFieldFormat[misc]{title}{\mkbibquote{#1}}
\DeclareFieldFormat[article,inproceedings]{volume}{\mkbibbold{#1}}
\DeclareFieldFormat[inproceedings]{series}{\mkbibitalic{#1}}

\usepackage[hidelinks,unicode,breaklinks=true]{hyperref}
\usepackage[capitalise,noabbrev,sort]{cleveref}

\usepackage{bookmark}

\usepackage{seqsplit}
\usepackage[notcite,notref,final]{showkeys}

\pgfplotsset{compat=1.16}
\usetikzlibrary{arrows}
\usetikzlibrary{arrows.meta}
\usetikzlibrary{patterns}
\usetikzlibrary{calc}
\usetikzlibrary{math} 

\usetikzlibrary{positioning,decorations.pathreplacing} 

\usepackage{chngcntr}
\counterwithin{figure}{section}
\counterwithin{equation}{section}

\let\C\relax

\newcommand{\R}{\mathbb{R}}
\newcommand{\C}{\mathbb{C}}
\newcommand{\Q}{\mathbb{Q}}
\newcommand{\N}{\mathbb{N}}
\newcommand{\Z}{\mathbb{Z}}

\renewcommand{\S}{\mathbb{S}}

\newcommand{\mcS}{\mathcal{S}}
\newcommand{\mcG}{\mathcal{G}}
\newcommand{\mcC}{\mathcal{C}}
\newcommand{\mcT}{\mathcal{T}}
\newcommand{\mcL}{\mathcal{L}}
\newcommand{\mcD}{\mathcal{D}}
\newcommand{\mcA}{\mathcal{A}}
\newcommand{\mcN}{\mathcal{N}}
\newcommand{\mcE}{\mathcal{E}}

\newcommand{\x}{\vec{x}}

\newtheoremstyle{theorems}
  {3pt}
  {3pt}
  {\itshape}
  {}
  {\bfseries}
  {.}
  { }
  {}

\newtheoremstyle{proofparts}
  {3pt}
  {0pt}
  {}
  {\parindent}
  {\scshape}
  {:}
  {\newline}
  {}

\newtheoremstyle{claims}
  {2pt}
  {2pt}
  {}
  {\parindent}
  {\bfseries}
  {.}
  { }
  {}

\makeatletter
\newcommand{\newreptheorem}[2]{\newtheorem*{rep@#1}{\rep@title}\newenvironment{rep#1}[1]{\def\rep@title{#2 \ref*{##1}}\begin{rep@#1}}{\end{rep@#1}}}
\makeatother

\theoremstyle{theorems}
\newtheorem{thm}{Theorem}[section]

\newtheorem{lemma}[thm]{Lemma}
\newtheorem*{lemma*}{Lemma}

\newtheorem{prop}[thm]{Proposition}

\newreptheorem{thm}{Theorem}
\newreptheorem{lemma}{Lemma}

\theoremstyle{definition}
\newtheorem{defn}[thm]{Definition}
\newtheorem{fakedef}[thm]{``Definition''}

\newtheorem{remark}[thm]{Remark}
\newtheorem{notation}[thm]{Notation}

\theoremstyle{proofparts}

\makeatletter
\@addtoreset{proofpart}{thm}
\makeatother

\theoremstyle{claims}

\newtheorem*{claim*}{Claim}

\crefname{thm}{theorem}{theorems}
\crefname{problem}{problem}{problems}
\crefname{lemma}{lemma}{lemmas}
\crefname{cor}{corollary}{corollaries}
\crefname{prop}{proposition}{propositions}
\crefname{conj}{conjecture}{conjectures}
\crefname{defn}{definition}{definitions}
\crefname{note}{note}{notes}
\crefname{ex}{example}{examples}
\crefname{remark}{remark}{remarks}
\crefname{notation}{notation}{notations}
\crefname{assumption}{assumption}{assumptions}
\crefname{claim}{claim}{claims}
\crefname{claim*}{claim}{claims}

\renewcommand{\subsectionmark}[1]{}
\makeatletter
\newcommand{\Biggg}{\bBigg@{3}}
\newcommand{\vast}{\bBigg@{4}}
\newcommand{\Vast}{\bBigg@{5}}
\makeatother

\newcommand{\norm}[1]{\left\Vert #1 \right\Vert}
\newcommand{\abs}[1]{\left\vert #1 \right\vert}

\DeclareMathOperator{\supp}{supp}

\DeclareMathOperator{\dist}{dist}

\DeclareMathOperator{\Vol}{Vol}
\DeclareMathOperator{\conv}{conv}

\definecolor{emphcolor}{rgb}{0,0,1}           

\newcommand{\ip}[2]{\left\langle #1 \middle\vert #2 \right\rangle}
\newcommand{\longip}[3]{\left\langle #1 \middle\vert #2 \middle\vert #3 \right\rangle}

\newcommand{\expect}[1]{\left\langle #1 \right\rangle}

\newcommand{\ud}{\,\textnormal{d}}

\let\epsilon\varepsilon
\let\varepsilon\epsilon

\newcommand{\eps}{\epsilon}

\title{Ground state energy of the dilute spin-polarized Fermi gas: Upper bound via cluster expansion}
\author{Asbjørn Bækgaard Lauritsen\thanks{\href{mailto:alaurits@ist.ac.at}{\nolinkurl{alaurits@ist.ac.at}}}\,\,}
\author{Robert Seiringer\thanks{\href{mailto:robert.seiringer@ist.ac.at}{\nolinkurl{robert.seiringer@ist.ac.at}}}}
\affil{IST Austria, Am Campus 1, 3400 Klosterneuburg, Austria}

\begin{document}
\maketitle
\begin{abstract}
We prove an upper bound on the ground state energy of the dilute spin-polarized  Fermi gas
capturing the leading correction to the kinetic energy resulting from repulsive interactions. 
One of the main ingredients in the proof is a rigorous implementation of the fermionic cluster expansion of 
Gaudin, Gillespie and Ripka (Nucl. Phys. A, 176.2 (1971), pp. 237–260).
\end{abstract}

\tableofcontents

\section{Introduction and main results}
We consider a Fermi gas of $N$ particles in a box $\Lambda = \Lambda_L = [-L/2, L/2]^d$
in $d$ dimensions, $d=1,2,3$. We will mostly focus on the case $d=3$. 
The particles interact via a two-body interaction $v$, which we assume to be positive, radial and of compact support.
In particular we allow for $v$ to have a hard core, i.e. $v(x)=\infty$ for $|x|\leq r$ for some $r>0$.
In natural units where $\hbar=1$ and the mass of the particles is $m=1/2$ the Hamiltonian of the system takes the form 
\[ 
	H_N = \sum_{i=1}^{N} -\Delta_{x_i} + \sum_{j < k} v(x_j - x_k).
\]
We are interested in spin-polarized fermions, meaning that all the spins are aligned. 
We may thus equivalently forget about the spin. 
This means that the Hamiltonian should be realized on the fermionic $N$-particle space of antisymmetric wavefunctions 
$L^2_a(\Lambda^{N}) = \bigwedge^N L^2(\Lambda)$. 
We consider the ground state energy density in the thermodynamic limit
\[
	e_d(\rho) 
		= \lim_{\substack{L\to \infty \\ N/L^d \to \rho}} \inf_{\substack{\Psi_N \in L^2_a(\Lambda^N) \\ \norm{\Psi_N}_{L^2}^2 = 1}} 
		\frac{\longip{\Psi_N}{H_N}{\Psi_N}}{L^d}.
\]
It is a result of Robinson \cite{Robinson.1971} that the thermodynamic limit exists, and that 
it is independent  of boundary conditions (say, Dirichlet, Neumann or periodic).

We study the dilute limit, where the inter-particle spacing is large compared to the length scale set by the interaction. 
For spin-polarized fermions, the relevant lengthscale  is the \emph{$p$-wave scattering length} $a$ 
which we define below.
Our main theorem is the upper bound 
\[
	e_{d=3}(\rho) 
		\leq \frac{3}{5}(6\pi^2)^{2/3} \rho^{5/3} 
    + \frac{12\pi}{5}(6\pi)^{2/3} a^3\rho^{8/3}
    \biggl[
		1 - \frac{9 }{35}(6\pi^2)^{2/3}a_0^2 \rho^{2/3} 
		+ o\left((a^3\rho)^{2/3}\right)
	\biggr]
\]
in the dilute limit $a^3\rho \ll 1$, where $a_0$ is another length related to the scattering length and effective range, also defined below.
The leading term $\frac{3}{5}(6\pi^2)^{2/3} \rho^{5/3}$ is the kinetic energy density of the free Fermi gas. 
The next term $\frac{12\pi}{5}(6\pi)^{2/3} a^3\rho^{8/3}$ naturally results from the two-body interactions using that the two-body density vanishes quadratically at incident points, leading to the cubic behavior in the scattering length. 
Finally, the correction term of order $a^3a_0^2 \rho^{10/3}$ is a consequence of the fourth-order behaviour of the two-particle density.

This formula is expected to be sharp \cite{Ding.Zhang.2019}. 
(How the length $a_0$ is related to the effective range appearing in \cite{Ding.Zhang.2019} is perhaps not immediate. We discuss this below.)
To order $a^3\rho^{8/3}$ the formula follows by truncating expansion formulas of Jastrow \cite{Jastrow.1955}, 
Iwamoto and Yamada \cite{Iwamoto.Yamada.1957}, Clark and Westhaus \cite{Clark.Westhaus.1968,Westhaus.Clark.1968} 
or Gaudin, Gillespie and Ripka \cite{Gaudin.Gillespie.ea.1971}.
Additionally the formula (to order $a^3\rho^{8/3}$) is claimed by Efimov and Amus'ya \cite{Efimov.Amusya.1965,Efimov.1966}, 
see also \cite{Wellenhofer.Drischler.ea.2020} and references therein.
Our result thus verifies this formula from the physics literature, at least as an upper bound.
An important ingredient in our proof is a rigorous implementation of the cluster expansion introduced by Gaudin, Gillespie and Ripka \cite{Gaudin.Gillespie.ea.1971}.

For the dilute Fermi gas one can also study the setting where different spins are present. 
This is studied in \cite{Lieb.Seiringer.ea.2005,Falconi.Giacomelli.ea.2021,Giacomelli.2023}, see also \cite{Giacomelli.2022}.
This system is realized by having the Hamiltonian $H_N$ act on a definite spin-sector 
$L^2_a(\Lambda^{N_\uparrow})\otimes L^2_a(\Lambda^{N_\downarrow})$,
where one fixes the number of spin-up and -down particles to be $N_\uparrow$ and $N_\downarrow = N - N_\uparrow$ respectively. 
The energy density satisfies (in $3$ dimensions)
\[
	e_{d=3}(\rho_\uparrow,\rho_\downarrow) = \frac{3}{5}(6\pi^2)^{2/3} \left(\rho_\uparrow^{5/3} + \rho_\downarrow^{5/3}\right)
		+ 8\pi a_{s} \rho_\downarrow \rho_\uparrow + o(a_{s}\rho^2),
\]
where $\rho_\sigma$ denotes the density of particles of spin $\sigma\in\{\uparrow,\downarrow\}$ 
and $\rho = \rho_\downarrow + \rho_\uparrow$.
Here $a_{s}$ is the $s$-wave scattering length of the interaction. 
The leading term is again the kinetic energy density of a free Fermi gas.
The next to leading order correction was first shown in \cite{Lieb.Seiringer.ea.2005} 
and later in \cite{Falconi.Giacomelli.ea.2021,Giacomelli.2023} using different methods.
The next correction is conjectured to be the Huang--Yang term \cite{Huang.Yang.1957} of order $a_{s}^2\rho^{7/3}$,
see \cite{Giacomelli.2022,Giacomelli.2023}.
Note that even the Huang--Yang term of order $a_{s}^2\rho^{7/3}$ is much larger than the leading correction 
in the spin-polarized case of order $a^3 \rho^{8/3}$.

For the dilute Fermi gas with spin, effectively only fermions of different spins interact (to leading order).
For fermions of different spins, the Pauli exclusion principle does not give any restriction, 
and the energy correction of the interaction is the same as for a dilute Bose gas (to leading order). 
For fermions of the same spin, the Pauli exclusion principle gives an inherent repulsion 
between the fermions. This gives the effect that the energy correction of the interaction is much smaller 
for fermions all of the same spin.

In addition to the dilute Fermi gas, much work has been done on  dilute Bose gases.
Here one realizes the Hamiltonian on 
the bosonic $N$-particle space of symmetric functions $L^2_s(\Lambda^N) = L^2(\Lambda)^{\otimes_{\textnormal{sym}}N}$ instead. 
One has the asymptotic formula (in $3$ dimensions)
\[ 
e_{d=3}(\rho) 
	= 4\pi a_{s} \rho^2\left(1 
	+ \frac{128}{15\sqrt{\pi}} (a_{s}^3\rho)^{1/2} 
	+ o\left( (a_{s}^3\rho)^{1/2}\right)\right).
\]
The leading term was shown by Dyson \cite{Dyson.1957} for an upper bound and Lieb and Yngvason \cite{Lieb.Yngvason.1998}
for the lower bound. The next correction, known as the Lee--Huang--Yang correction \cite{Lee.Huang.ea.1957},
was shown as an upper bound in \cite{Yau.Yin.2009,Basti.Cenatiempo.ea.2021} 
and as a lower bound in \cite{Fournais.Solovej.2020,Fournais.Solovej.2022}.
In some sense, the term $\frac{12\pi}{5}(6\pi)^{2/3}a^3\rho^{8/3}$ for the same-spin fermions is 
the fermionic analogue of the $4\pi a_{s} \rho^2$ for the bosons. 
It is the leading correction to the energy of the free Fermi/Bose gas.

Finally also some lower-dimensional problems have been studied. 
The $2$-dimensional dilute Fermi gas with different spins present is studied in \cite{Lieb.Seiringer.ea.2005},
where the leading correction to the kinetic energy is shown. 
Recently also the bosonic problem has been studied in \cite{Fournais.Girardot.ea.2022}.
We show that for the spin-polarized setting in $2$ dimensions we have the upper bound 
\[
	e_{d=2}(\rho) \leq 
      2\pi\rho^2 + 4\pi^2a^2\rho^{3} 
		[1 + o(1)]
    \qquad 
    \textnormal{as } a^2\rho \to 0.
\]
Additionally, the $1$-dimensional spin-polarized Fermi gas is studied in \cite{Agerskov.Reuvers.ea.2022}.
Agerskov, Reuvers and Solovej \cite{Agerskov.Reuvers.ea.2022} show that 
\[
	e_{d=1}(\rho) = \frac{\pi^2}{3}\rho^3 + \frac{2\pi^2}{3}a\rho^{4} \left[1 + o(1)\right]
  \qquad 
  \textnormal{as } a\rho \to 0.
\]
We give a new proof of this as an upper bound with an improved error term.

\subsection{Precise statement of results}
We now give the precise statement of our main theorems. 
We start with the $3$-dimensional setting.
First, we define the $p$-wave scattering length.
(See also \cites[Appendix A]{Lieb.Yngvason.2001}{Seiringer.Yngvason.2020}.)
\begin{defn}\label{defn.scattering.length}
The \emph{$p$-wave scattering length} $a$ of the interaction $v$ is defined by the  minimization problem
\[
12\pi a^3 = \inf \left\{\int_{\R^3} |x|^2\left( |\nabla f_0(x)|^2 + \frac{1}{2}v(x) |f_0(x)|^2\right) \ud x : f_0(x)\to 1 \textnormal{ as } |x|\to \infty \right\}.
\]
The minimizer $f_0$ is the \emph{($p$-wave) scattering function}.
(In case $v$ has a hard core, i.e. $v(x)=\infty$ for $|x|\leq r$ one has to interpret $v(x)\ud x$ as a measure. 
Necessarily then the minimizer has $f_0(x)=0$ for $|x|\leq r$.)
\end{defn}

\noindent 
We collect properties of the scattering function $f_0$ in \Cref{sec.scattering.function}.
We  define the length $a_0$ as follows.
\begin{defn}
The length $a_0$ is given by 
\[
	3a_0^2 = \frac{1}{12\pi a^3} \int_{\R^3} |x|^4 \left( |\nabla f_0(x)|^2 + \frac{1}{2}v(x) |f_0(x)|^2\right) \ud x,
\]
where $f_0$ is the scattering function of \Cref{defn.scattering.length}.
The normalization is chosen so that a hard core interaction of radius $R_0$ has $a_0=a=R_0$, see \Cref{rmk.hard.core.a.a0}.
(If $v$ has a hard core we interpret $v(x)\ud x$ as a measure as in \Cref{defn.scattering.length}.)
\end{defn}

\noindent
We can now state our main theorem. 
\begin{thm}\label{thm.main}
Suppose that $v\geq 0$ is radial and compactly supported.
Then, for sufficiently small $a^3\rho$, the ground-state energy density satisfies 
\[
\begin{aligned}
	e_{d=3}(\rho) 
		& \leq \frac{3}{5}(6\pi^2)^{2/3} \rho^{5/3} 
		\\ & \qquad 
		+ \frac{12\pi}{5}(6\pi)^{2/3} a^3\rho^{8/3}
		\biggl[
		1 - \frac{9 }{35}(6\pi^2)^{2/3}a_0^2 \rho^{2/3} 
		+ O\left( (a^3\rho)^{2/3 + 1/21} |\log(a^3\rho)|^6\right)
		\biggr].
\end{aligned}
\]
\end{thm}

\noindent
The essential steps in the proof are as follows.
\begin{enumerate}[(1)]

\item Show the absolute convergence of the formal cluster expansion formulas of \cite{Gaudin.Gillespie.ea.1971} for the reduced densities of a 
Jastrow-type trial state. The criterion for absolute convergence will not hold uniformly in the system size, and in order to allow for a larger particle number we need to 
 introduce the ``Fermi polyhedron'', described in \Cref{sec.polyhedron}, as an approximation to the Fermi ball. 
The formulas of \cite{Gaudin.Gillespie.ea.1971} 
are computed in \Cref{sec.calc.C_N,sec.calc.rho1,sec.calc.rho2,sec.calc.rho3} and stated in \Cref{thm.gaudin.expansion}. 
The absolute convergence is proven in \Cref{sec.abs.conv}.

\item Bound the energy of the Jastrow-type trial state. 
For this we shall in particular need bounds on ``derivative Lebesgue constants'' given in \Cref{lem.derivative.lebesgue.constant}
and proven in \Cref{sec.derivative.lebesgue.constant}.
The computation of the energy of such a Jastrow-type trial state is given in \Cref{sec.energy.in.box}.

\item Use a box method to glue together trial states in smaller boxes to obtain a bound in the thermodynamic limit. 
This is done in \Cref{sec.box.method}.
\end{enumerate}

\begin{remark}
The term of order $a^3a_0^2\rho^{10/3}$ is in fact the same as is claimed in \cite{Ding.Zhang.2019}. 
To see this we relate the \emph{effective range} $R_{\textnormal{eff}}$ to the length $a_0$.
In the physics literature the effective range is defined via the formula 
\begin{equation}
k^3 \cot \delta(k) = -\frac{3}{a^3} - \frac{1}{2R_{\textnormal{eff}}} k^2 + \textnormal{higher order in $k$}
\qquad k\to 0
\label{eqn.delta(k)}
\end{equation}
for the phase shift $\delta(k)$ of low energy $p$-wave scattering.
A formula for the effective range is found in \cite[Equation (56)]{Hammer.Lee.2010}. 
With this we find 
\begin{prop}\label{prop.effective.range}
The effective range is given by 
\begin{equation*}
R_{\textnormal{eff}}^{-1} = \frac{18}{5} a_0^2 a^{-3}.
\end{equation*}
\end{prop}
\noindent
Using this formula we recover the formula \cite[Equation (15)]{Ding.Zhang.2019} to order $\rho^{10/3}$. 
The formula of \cite{Ding.Zhang.2019} reads 
\begin{equation*}
\frac{e_{d=3}(\rho)}{\rho}
= k_F^2
  \left[\frac{3}{5} 
  + \frac{2}{5\pi}a^3k_F^3 
  - \frac{1}{35\pi} a^6 R_{\textnormal{eff}}^{-1} k_F^5 
  + \frac{2066 - 312\log 2}{10395\pi^2} a^6 k_F^6 
  + \textnormal{higher order}
  \right],
\end{equation*}
with $k_F = (6\pi^2\rho)^{1/3}$ the Fermi momentum.
We give the proof of \cref{prop.effective.range} in \Cref{sec.scattering.function} below.
\end{remark}

\begin{remark}[Numerical investigation]
The validity of the formula in \Cref{thm.main} is investigated numerically in \cite{Bertaina.Tarallo.ea.2023} using Quantum Monte Carlo simulations.
We plot their findings in \Cref{fig.plot} and compare them to the formula in \Cref{thm.main} 
and the claimed formula to order $\rho^{11/3}$ of \cite[Equation (15)]{Ding.Zhang.2019}.
\begin{figure}[htb]
\centering
\includegraphics{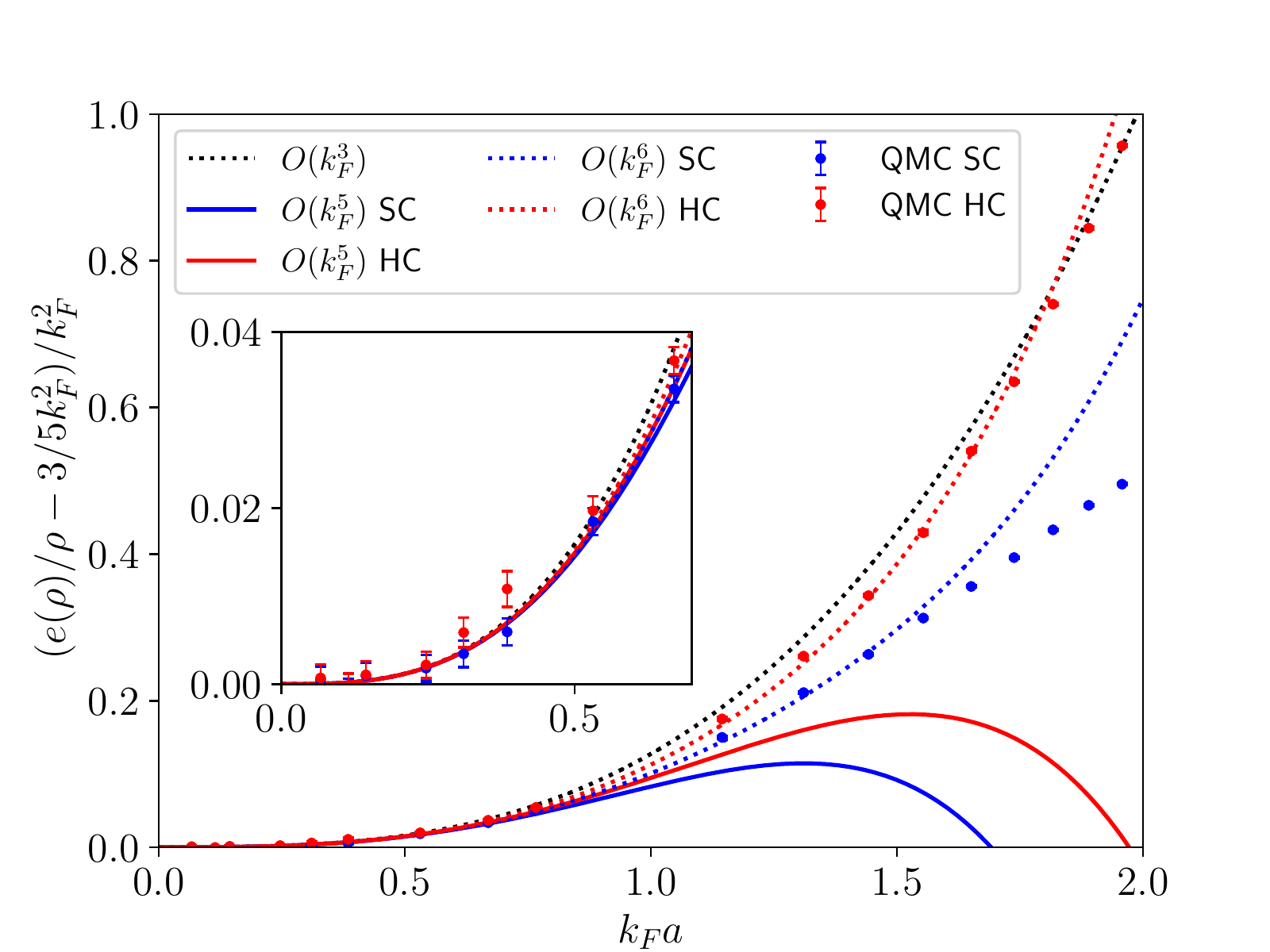}
\caption{Energies of dilute Fermi gasses. The curves and points labelled {\sf HC} are for a \emph{Hard Core} interaction of radius $a$. The curves and points labelled {\sf SC} 
are for a \emph{Soft Core} interaction of radius $2a$ and strength $V_0$ chosen so that it's scattering length is $a$.
(Meaning $v(x) = V_0 \chi_{|x|\leq 2a}$, $\chi$ being the characteristic function.)
The points (labelled {\sf QMC}) are \emph{Quantum Monte Carlo} simulations from \cite{Bertaina.Tarallo.ea.2023}.
The curves include the (conjectured) corrections up to the labelled order in $k_F = (6\pi^2\rho)^{1/3}$.}
\label{fig.plot}
\end{figure}
\end{remark}

\begin{remark}
One may weaken the assumptions on the interaction $v$ a bit at the cost of a longer proof. 
The compact support and that $v\geq 0$ are not strictly necessary. 
Essentially, we just need sufficiently good bounds on integrals of the scattering function $f_0$ as used in \Cref{sec.abs.conv,sec.energy.in.box} and 
that the ``stability condition'' of the tree-graph bound \cites[Proposition 6.1]{Poghosyan.Ueltschi.2009}{Ueltschi.2018}
used in \Cref{sec.abs.conv} is satisfied. 
\end{remark}

\begin{remark}
With the same method one should be able to improve the error bound slightly.
At best one could get the error to be $O_\eps( \rho^{5/3}(a^3\rho)^{2-\eps})$ for any $\eps > 0$ 
(i.e., the error-term in \Cref{thm.main}, $O((a^3\rho)^{2/3 + 1/21} |\log(a^3\rho)|^6)$ 
could be replaced by $O_\eps( (a^3\rho)^{1 - \eps})$). 
The bound of the error-term in \Cref{thm.main} arises from bounding the tail of the Gaudin-Gillespie-Ripka-expansion.
Exact calculation for small diagrams (meaning small number of involved particles) reveal that this bound is very crude. 
Using such exact calculations for more diagrams would improve the error-bound as stated.
This is somewhat similar to the recent work on the Bose gas \cite{Basti.Cenatiempo.ea.2022a}.
We shall discuss this further in \Cref{rmk.more.small.diagrams}.
\end{remark}

\noindent
We consider  the lower-dimensional problems next. 
We start with $2$ dimensions, where the scattering length is defined as follows.
\begin{defn}\label{defn.scattering.length.d=2}
The ($2$-dimensional) \emph{$p$-wave scattering length} $a$ of the interaction $v$ is defined by the  minimization problem
\[
	4\pi a^2 = \inf 
	\left\{\int_{\R^2} |x|^2 \left(|\nabla f_0(x)|^2 + \frac{1}{2}v(x) |f_0(x)|^2\right)\ud x 
			: f_0(x) \to 1 \textnormal{ as } |x|\to \infty \right\}
\]
The minimizer $f_0$ is the ($2$-dimensional) \emph{($p$-wave) scattering function}.
\end{defn}

\noindent
With this, we may state the $2$-dimensional analogue of \Cref{thm.main}.
\begin{thm}[Two dimensions]\label{thm.main.d=2}
Suppose that $v\geq 0$ is radial and compactly supported. 
Then, for sufficiently small $a^2\rho$, the ground-state energy density satisfies 
\[
\begin{aligned}
	e_{d=2}(\rho) 
	& \leq	
    2\pi\rho^2 + 4\pi^2a^2\rho^{3} 
		\biggl[
		1
		+ O\left(a^2\rho |\log(a^2\rho)|^2\right)
		\biggr].
\end{aligned}
\]
\end{thm}

\noindent
We sketch in \Cref{sec.two.dimensions} how to adapt the proof in the $3$-dimensional setting to $2$ dimensions.

Finally, we consider the $1$-dimensional problem.
The scattering length is defined as follows.
\begin{defn}\label{defn.scattering.length.d=1}
The ($1$-dimensional) \emph{$p$-wave scattering length} $a$ of the interaction $v$ is defined by the  minimization problem
\[
	2 a = \inf 
	\left\{\int_{\R} |x|^2 \left(|\partial f_0(x)|^2 + \frac{1}{2}v(x) |f_0(x)|^2\right)\ud x 
			: f_0(x) \to 1 \textnormal{ as } |x|\to \infty \right\}
\]
The minimizer $f_0$ is the ($1$-dimensional) \emph{($p$-wave) scattering function}.
\end{defn}

\noindent
We show in \Cref{prop.odd.wave.=.p.wave} that \Cref{defn.scattering.length.d=1} 
agrees with the (seemingly different) definition of the scattering length in \cite{Agerskov.Reuvers.ea.2022}.
With this, we may state the $1$-dimensional analogue of \Cref{thm.main}.
\begin{thm}[One dimension]\label{thm.main.d=1}
Suppose that $v\geq 0$ is even and compactly supported. 
Suppose moreover that $\int \left(\frac{1}{2}vf_0^2 + |\partial f_0|^2\right) \ud x < \infty$, 
where $f_0$ denotes the ($p$-wave) scattering function.
Then, for sufficiently small $a\rho$, the ground-state energy density satisfies 
\[
\begin{aligned}
	e_{d=1}(\rho) 
	& \leq	
		\frac{\pi^2}{3}\rho^3 + \frac{2\pi^2}{3}a\rho^{4} 
		\biggl[
		1
    + O\left((a\rho)^{9/13}\right)
		\biggr].
\end{aligned}
\]
\end{thm}

\noindent
We remark that Agerskov, Reuvers and Solovej \cite{Agerskov.Reuvers.ea.2022} 
recently showed (almost) the same result with a matching lower bound 
$e_{d=1}(\rho)\geq\frac{\pi^2}{3}\rho^3 + \frac{2\pi^2}{3}a\rho^{4} (1 + o(1))$.
Compared to their result we treat a slightly different class of potentials and obtain an improved error bound.
The conjectured next contribution is of order $a^2\rho^5$, see \cite{Agerskov.Reuvers.ea.2022}.


\begin{remark}[{On the assumptions on $v$}]
Any smooth interaction or an interaction with a hard core (meaning that $v(x) = +\infty$ for $|x|\leq a_0$ for some $a_0 > 0$) 
satisfies $\int \left(\frac{1}{2}vf_0^2 + |\partial f_0|^2\right) \ud x < \infty$, 
see \Cref{prop.h.c.implies.int.vf2,prop.smooth.implies.int.vf2}.
\end{remark}

\noindent
We sketch in \Cref{sec.one.dimension} how to adapt the proof in the $3$-dimensional setting to $1$ dimension.
This turns out to be more involved than adapting the argument to $2$ dimensions.

The paper is structured as follows. 
In \Cref{sec.preliminary.computations} we give some preliminary computations 
and in particular we introduce the ``Fermi polyhedron'', 
a polyhedral approximation to the Fermi ball.
In \Cref{sec.gaudin.expansion} we introduce the fermionic cluster expansion of Gaudin, Gillespie and Ripka 
\cite{Gaudin.Gillespie.ea.1971} and we find conditions on absolute convergence of the resulting formulas. 
In the subsequent  \Cref{sec.energy.in.box} we compute the energy of a Jastrow-type trial state 
and glue  many of them together using a box method to form trial states of arbitrary many particles.
Finally, in \Cref{sec.lower.dimensions} we sketch how to adapt the argument to the lower-dimensional settings.
In \Cref{sec.small.diagrams} we give computations of ``small diagrams'' needed for some bounds in \Cref{sec.energy.in.box,sec.one.dimension}
and in \Cref{sec.derivative.lebesgue.constant} we give the proof of \Cref{lem.derivative.lebesgue.constant}, 
an important lemma used in \Cref{sec.energy.in.box}.

\section{Preliminary computations}\label{sec.preliminary.computations}
We will construct a trial state using a box method, and bound the energy of such trial state.
To use such a box method we need to use Dirichlet boundary conditions in each smaller box.
In \Cref{lem.change.b.c.} we show that we may construct trial states with Dirichlet boundary condition 
out of trial states with periodic boundary conditions. We will thus 
use periodic boundary conditions in the box $\Lambda = [-L/2, L/2]^3$.
For periodic boundary conditions, the Hamiltonian is given by 
\[
	H_N = H_{N,L}^\textnormal{per} = \sum_{j=1}^N -\Delta_{j} + \sum_{i < j} v_{\textnormal{per}}(x_i - x_j),
\]
where $\Delta_{j}$ denotes the Laplacian on the $j$'th coordinate  
and $v_{\textnormal{per}}(x) = \sum_{n\in \Z^3} v(x + nL)$, the periodized interaction.
By a slight abuse of notation we write $v=v_{\textnormal{per}}$, since we will choose $L$ bigger than the range of $v$.

The trial state in each smaller box is given by the Jastrow-type \cite{Jastrow.1955} trial state 
(also known as a Bijl-Dingle-Jastrow-type trial state)
\begin{equation}\label{eqn.define.trial.state}
  \psi_N = \frac{1}{\sqrt{C_N}} \prod_{i < j} f(x_i-x_j) D_N(x_1,\ldots,x_N),
\end{equation}
where $f$ is a scaled and cut-off version of the scatting function $f_0$, 
$D_N$ is an appropriately chosen Slater determinant,
and $C_N$ is a normalization constant.
More precisely, 
\[ 
	f(x) = \begin{cases}
	\frac{1}{1-a^3/b^3}f_0(|x|) & 
  |x| \leq b,
	\\ 1 & 
  |x|\geq b	,
	\end{cases}
	\quad 
	D_N(x_1,\ldots,x_N) = \det\left[u_k(x_i)\right]_{\substack{1 \leq i \leq N \\ k \in P_F}},
  \quad
  u_k(x) = \frac{1}{L^{3/2}} e^{ikx},
\]
where $f_0$ is the $p$-wave scattering function, 
$\abs{\cdot}:= \min_{n\in \Z^3}\abs{\cdot - nL}_{\R^3}$ (with $\abs{\cdot}_{\R^3}$ denoting the norm on $\R^3$),
$b > R_0$, the range of $v$, is some cut-off to be chosen later, 
$P_F$ is a polyhedral approximation 
to the Fermi ball $B_F$ of radius $k_F$ 
described in \Cref{sec.polyhedron},
and the number of particles is $N = \#P_F$, the number of points in $P_F$.
We  choose $b$ to be larger than the range of $v$; 
in particular, then $f$ is continuous.
(Note that the metric on the torus is $d(x,y) = |x-y|$.
We will abuse notation slightly and denote by $\abs{\cdot}$ also the absolute value of some number or the norm on $\R^3$.)

Before going further with the proof we first fix some notation.
\begin{notation}
We introduce the following.
\begin{itemize}
\item
For any function $h$ and edge (of some graph) $e=(i,j)$ we will write $h_e = h_{ij} = h(x_i - x_j)$.

\item 
We denote by $C$ a generic positive constant whose value may change line by line.

\item 
For expressions $A, B$ we write $A \lesssim B$ if there exists some constant $C > 0$ such that $A \leq CB$.
If both $A\lesssim B$ and $B \lesssim A$ we write $A \sim B$.

\item
For a vector $x=(x^1,\ldots,x^d)\in\R^d$ we write $x^1,\ldots,x^d$ for its components.
\end{itemize}
We will fix the Fermi momentum $k_F$ and then choose $L,N$ large but finite depending on $k_F$.
The density of particles in the trial state $\psi_N$ is $\rho := N/L^3$. 
The limit of small density $a^3\rho \to 0$ will  be realized as $k_F a \to 0$.
\end{notation}

\noindent
To compute the energy of the trial state $\psi_N$ note that for (real-valued) functions $F,G$ we have 
\[
	\int \abs{\nabla (FG)}^2 = \int |\nabla F|^2 |G|^2 - \int |F|^2 G\Delta G.
\]
Using this on $F = \prod_{i<j}f_{ij}$ and $G = D_N$ we have 
\begin{equation}\label{eqn.compute.energy.first}
\begin{aligned}
\longip{\psi_N}{H_N}{\psi_N}
	& 	
	= E_0 + 2 \sum_{j < k} \longip{\psi_N}{ \abs{\frac{\nabla f(x_j - x_k)}{f(x_j - x_k)}}^2 + \frac{1}{2}v(x_j - x_k)}{\psi_N}
  \\ & \qquad 
		+ 6\sum_{i < j < k} \longip{\psi_N}{\frac{\nabla f_{ij} \nabla f_{jk}}{f_{ij} f_{jk}}}{\psi_N}
	\\ &
	= E_0 + \iint \rho^{(2)}_{\textnormal{Jas}}(x_1, x_2) \left(\abs{\frac{\nabla f(x_1 - x_2)}{f(x_1-x_2)}}^2 + \frac{1}{2}v(x_1 - x_2)\right) \ud x_1 \ud x_2
	\\ & \qquad 
	+ \iiint \rho_{\textnormal{Jas}}^{(3)}(x_1,x_2,x_3) \frac{\nabla f_{12} \nabla f_{23}}{f_{12} f_{23}} \ud x_1 \ud x_2 \ud x_3,
\end{aligned}
\end{equation}	
where $E_0 = \sum_{k \in P_F} |k|^2$ is the kinetic energy of $\frac{1}{\sqrt{N!}}D_N$ and 
$\rho^{(n)}_{\textnormal{Jas}}$ denotes the $n$-particle reduced density of the trial state $\psi_N$, given by 
\begin{equation}\label{eqn.define.rho(n)}
	\rho^{(n)}_{\textnormal{Jas}}(x_1,\ldots,x_n) = N(N-1)\cdots (N-n+1) \idotsint \abs{\psi_N(x_1,\ldots,x_N)}^2 \ud x_{n+1} \ldots \ud x_N,
	\quad
	n=1,\ldots,N.
\end{equation}
The division by $f$ is non-problematic even where $f=0$, since it cancels with the corresponding factors of $f$ in $\psi_N$.
We need to compute $\rho^{(2)}_{\textnormal{Jas}}$ and bound $\rho_{\textnormal{Jas}}^{(3)}$. 
Before we start on this endeavour we first recall some properties of the scattering function.

\subsection{The scattering function}\label{sec.scattering.function}
The scattering function $f_0$ is defined by the minimization problem in \Cref{defn.scattering.length},
see also \cites[Appendix A]{Lieb.Yngvason.2001}{Seiringer.Yngvason.2020}. 
In particular $f_0$ satisfies the corresponding Euler-Lagrange equation
\[
-4 x \cdot \nabla f_0 - 2 |x|^2 \Delta f_0 + |x|^2 v f_0 = 0.
\] 
The minimizer $f_0$ is radial and with a slight abuse of notation we sometimes write $f_0(|x|) = f_0(x)$.
In radial coordinates the Euler-Lagrange equations reads
\begin{equation}\label{eqn.f.scatt.radial}
-\partial_r^2 f_0 - \frac{4}{r}\partial_r f_0 + \frac{1}{2}vf_0 = 0,
\end{equation}
where $\partial_r$ denotes the derivative in the radial direction.
This is the same equation as for $s$-wave scattering in $5$ dimensions, see \cite[Appendix A]{Lieb.Yngvason.2001}.
Thus, properties of this carry over. In particular	
$f_0(x) = 1 - \frac{a^3}{|x|^3}$ for $x$ outside the support of $v$.
Moreover
\begin{lemma}[{\cite[Lemma A.1]{Lieb.Yngvason.2001}}]
\label{prop.bound.f0.hc}
The scattering function $f_0$ satisfies $\left[1 - \frac{a^3}{|x|^3}\right]_+ \leq f_0(x) \leq 1$ for all $x$ and 
$|\nabla f_0(x)|\leq \frac{3a^3}{|x|^4}$ for $|x| > a$.
\end{lemma}
\noindent
We give a short proof here for completeness.
\begin{proof}
From the radial Euler-Lagrange equation \eqref{eqn.f.scatt.radial} we have $\partial_r (r^4 \partial_r f_0) = v r^4 f_0/2 \geq 0$.
Denote by $f_{\textnormal{hc}} = \left[1 - \frac{a^3}{|x|^3}\right]_+$ the solution for a hard core potential of range $a$.
Then 
\[
	r^4 \partial_r f_{\textnormal{hc}} = \begin{cases}
	3a^3 & r > a \\ 0 & r < a
	\end{cases}
\]
In particular $\partial_r (r^4 \partial_r f_{\textnormal{hc}}) = 0$ for $r > a$. 
We thus see that $\partial_r f_0 \leq \partial_r f_{\textnormal{hc}} = 3a^3 r^{-4}$ 
and $f_0 \geq f_\textnormal{hc}$ for $r > a$ by integrating.
Trivially $f_0 \geq 0 = f_{\textnormal{hc}}$ for $r \leq a$.
\end{proof}

\begin{remark}\label{rmk.hard.core.a.a0}
A hard core interaction of range $R_0> 0$,
\[
	v_{\textnormal{hc}}(x) = \begin{cases}
	+\infty & |x|\leq R_0, \\ 0 & |x|> R_0,
	\end{cases}
\]
has $f_0(x) = f_{\textnormal{hc}}(x) = \left[1 - \frac{a^3}{|x|^3}\right]_+$ and thus $a_0 = a = R_0$.
\end{remark}

\noindent
Finally, we give the 
\begin{proof}[Proof of \Cref{prop.effective.range}]
Let $R_0$ denote the range of the interaction. 
Then the effective range is given by \cite[Equation (56)]{Hammer.Lee.2010} 
\begin{equation}
  -R_{\textnormal{eff}}^{-1} = -\frac{2}{R} - \frac{2R^2}{a^3} + \frac{2R^5}{5a^6} - 2\int_0^R u(r)^2 \ud r,
  \qquad 
  \textnormal{for }
  R \geq R_0
  \label{eqn.Reff.ini}
\end{equation}
where
$u$ solves \cite[Equation (22)]{Hammer.Lee.2010}
\begin{equation*}
-\partial_r^2 u + \frac{2}{r^2}u + \frac{1}{2}v u = 0
\end{equation*}
with $\partial_r$ denoting the radial derivative.
In particular then $f_0 = -\frac{a^3}{r^2} u$ satisfies the scattering equation, \Cref{eqn.f.scatt.radial}.
For $r \geq R_0$ we find using \cite[Equation (27)]{Hammer.Lee.2010} and \Cref{eqn.delta(k)}
\begin{equation*}
  u(r) = \lim_{k\to 0} \frac{\sin(kr + \delta(k)) - kr\cos(kr+\delta(k))}{r \sin\delta(k)}
    = \frac{-r^2}{a^3}\left[1 - \frac{a^3}{r^3}\right].
\end{equation*}
We conclude that $f_0 = -\frac{a^3}{r^2} u$ is indeed the scattering function. 
Thus \eqref{eqn.Reff.ini} reads
\begin{equation*}
  -R_{\textnormal{eff}}^{-1} 
  = -\frac{2}{R} - \frac{2R^2}{a^3} + \frac{2R^5}{5a^6}
    - \frac{2}{a^6} \int_0^R r^4 f_0^2 \ud r,
    \qquad 
  \textnormal{for }
  R \geq R_0
\end{equation*}
The remainder of the proof is a simple calculation using integration by parts and the scattering equation.
We omit the details.
This concludes the proof of \Cref{prop.effective.range}.
\end{proof}

\subsection{The ``Fermi polyhedron''}\label{sec.polyhedron}
We introduce a polyhedral approximation $P_F$ of the Fermi ball
$B_F = \{k\in \frac{2\pi}{L}\Z^3 : |k|\leq k_F\}$. 
The main properties we will need of the polyhedral approximation 
are given in \Cref{lem.KE.polyhedron,lem.lebesgue.constant.fermi.polyhedron,lem.derivative.lebesgue.constant}.
We discuss why we need such a polyhedral approximation in \Cref{rmk.why.Fermi.polyhedron}.
The problem is that  
\[
  \int_{[0,L]^3} \frac{1}{L^3}\abs{\sum_{k \in B(k_F)\cap \frac{2\pi}{L}\Z^3} e^{ikx}} \ud x
  = \frac{1}{(2\pi)^3} \int_{[0,2\pi]^3} \abs{\sum_{q \in B(cN^{1/3})\cap \Z^3} e^{iqu}} \ud u 
  \sim N^{1/3}
\]
for large $N$ (see \cite{Ganzburg.Liflyand.2019,Liflyand.2006} and references therein) is too big for our purposes. 
Note that this behaviour is a consequence of taking the absolute value. 
In fact we have that 
$\frac{1}{L^3}\int \sum_{k \in B(k_F)\cap \frac{2\pi}{L}\Z^3} e^{ikx} \ud x= 1$.

This type of quantity is referred to as the \emph{Lebesgue constant} \cite{Ganzburg.Liflyand.2019,Liflyand.2006} 
of some domain $\Omega$,
\[
\mcL(\Omega) := \frac{1}{(2\pi)^3}\int_{[0,2\pi]^3} \abs{\sum_{q \in \Omega\cap \Z^3} e^{iqu}} \ud u
\]
These kinds of integrals appear in estimates in \Cref{sec.abs.conv,sec.energy.in.box}. 
For an overview of such Lebesgue constants, see \cite{Ganzburg.Liflyand.2019,Liflyand.2006}. 
Of particular relevance for us is the fact that the Lebesgue constants are much smaller for polyhedral domains than for balls. 
Hence we introduce the polyhedron $P=P(N)$ as an approximation of the unit ball.
Then the scaled version $P_F = k_F P \cap \frac{2\pi}{L} \Z^3$ approximates the Fermi ball.
We will refer to $P_F$ as the \emph{Fermi polyhedron}. 
In \cite[Theorem 4.1]{Kolomoitsev.Lomako.2018} it is shown that for any fixed convex polyhedron $P'$ of $s$ vertices  
\begin{equation}\label{eqn.Lebesgue.constant.polyhedron}
	\mcL(RP')=\frac{1}{(2\pi)^3}\int_{[0,2\pi]^3} \abs{\sum_{q \in RP'\cap \Z^3} e^{iqu}} \ud u \leq C s (\log R)^{3} + C(s) (\log R)^2
\end{equation}
for any $R>2$, in particular for $R\sim N^{1/3}$,  where $C(s)$ is some unknown function of $s$.
We will improve on this bound for the specific polyhedron $P = P(N)$ to control the $s$-dependence of the subleading (in $R$) terms, i.e. of $C(s)$. 
For the specific polyhedron $P$ we have $C(s) \leq Cs$. This is the content of \Cref{lem.lebesgue.constant.fermi.polyhedron} below.
We first give an almost correct definition of the polyhedron $P$.

\begin{fakedef}[Simple definition]\label{defn.P_F}
The polyhedron $P$ is chosen to be the convex hull of $s=s(N)$ points $\kappa_1,\ldots,\kappa_s$ on a sphere of radius $1+\delta$,
where $\delta$ is chosen such that $\Vol(P) = 4\pi/3$. 
We moreover choose the set of points to have the following properties. 
\begin{itemize}
	\item 
	The points are \emph{evenly distributed}, meaning that the distance $d$ between any pair of points satisfies $d\gtrsim s^{-1/2}$, 
  and that for any $k$ on the sphere of radius $1+\delta$ 
  the distance from $k$ to the closest point is $\lesssim s^{-1/2}$.
	That is, for some constants $c,C > 0$ we have $d \geq cs^{-1/2}$ and $\inf_j |k - \kappa_j| \leq Cs^{-1/2}$.


	\item 
	$P$ is invariant under any map $(k^1, k^2, k^3) \mapsto (\pm k^a, \pm k^b, \pm k^c)$ for $\{a,b,c\} = \{1,2,3\}$, 
	i.e. reflection in or permutation of any of the axes.
\end{itemize}
The \emph{Fermi polyhedron} is the rescaled version defined as $P_F := k_FP \cap \frac{2\pi}{L}\Z^3$,
where $L$ is chosen large (depending on $k_F$) such that $k_FL$ is large.
\end{fakedef}

\begin{remark}
Note that the symmetry constraint adds a restriction on $s$. For instance, a generic point 
away from any plane of symmetry (i.e. $k^1,k^2,k^3$ all different and non-zero) 
has $48$ images (including itself) when reflected by the maps 
$(k^1, k^2, k^3) \mapsto (\pm k^a, \pm k^b, \pm k^c)$ for $\{a,b,c\} = \{1,2,3\}$.

For $s$ points on a sphere of radius $1+\delta$, the natural lengthscale is $(1+\delta)s^{-1/2} \sim s^{-1/2}$. 
The requirement that the points are evenly distributed then ensures that all pairs of close points 
(for any reasonable definition of ``close points'')
have a pairwise distance of this order. 
\end{remark}


\begin{remark}
For all purposes apart from the technical argument in \Cref{sec.derivative.lebesgue.constant} one may take this as the definition. 
In particular, the convergence criterion of the cluster expansion formulas of Gaudin, Gillespie and Ripka \cite{Gaudin.Gillespie.ea.1971}, 
given in \Cref{thm.gaudin.expansion}, holds also for this simpler definition of $P$.
We provide this simpler definition to better give an intuition of the construction.
\end{remark}

\noindent
We now give the actual definition of $P$. 
We first give the construction. Then in \Cref{rmk.comment.def.P_F} we give a few comments and in \Cref{rmk.motivation.P_F}
we give a short motivation.

\begin{defn}[Actual definition]\label{defn.P_F.true}
The polyhedron $P$ 
with $s$ corners 
and the ``centre'' $z$ is constructed as follows. 
\begin{itemize}
\item 
First, choose a big number $Q$, the ``size of the primes'' satisfying 
\[
	Q^{-1/4} \leq Cs^{-1}, 
	\qquad 
	N^{4/3} \ll Q \leq CN^C
\]
in the limit $N\to \infty$.

\item 
Pick three large distinct primes $Q_1, Q_2, Q_3$ with $Q_j \sim Q$.

\item 
Place $s$ \emph{evenly distributed} points $\kappa_1^\R,\ldots,\kappa_{s}^\R$ 
on the sphere of radius $Q^{-3/4}$ and
such that the set of points $\{\kappa_1^\R,\ldots,\kappa_{s}^\R\}$ is invariant under the symmetries 
$(k^1, k^2, k^3) \mapsto (\pm k^a, \pm k^b, \pm k^c)$ for $\{a,b,c\} = \{1,2,3\}$.

Here, \emph{evenly distributed} means that the distance between any pair of points is $d\gtrsim s^{-1/2}Q^{-3/4}$
and that for any $k$ on the sphere of radius $Q^{-3/4}$ the distance from $k$ to the nearest point is $\lesssim s^{-1/2}Q^{-3/4}$.
That is, $d \geq cs^{-1/2}Q^{-3/4}$ and $\inf_{j} \abs{k - \kappa_j^\R} \leq Cs^{-1/2}Q^{-3/4}$ for some constants $c,C > 0$.

\item 
Find points $\kappa_1,\ldots,\kappa_s$ of the form
\begin{equation}\label{eqn.define.kappa_j.polyhedron}
	\kappa_j = \left(\frac{p^1_j}{Q_1}, \frac{p^2_j}{Q_2}, \frac{p^3_j}{Q_3} \right), 
		\qquad p^\mu_j \in \Z, 
		\quad \mu=1,2,3,
		\quad j=1,\ldots,s,
\end{equation}
such that $\abs{\kappa_j - \kappa_j^\R} \lesssim Q^{-1}$ for all $j=1,\ldots,s$
and 
such that the set of points $\{\kappa_1,\ldots,\kappa_{s}\}$ is invariant under the symmetries 
$(k^1, k^2, k^3) \mapsto (\pm k^1, \pm k^2, \pm k^3)$.




\item 
Define $\tilde P$ as the convex hull of all the points $\kappa_1,\ldots,\kappa_s$.
That is, $\tilde P = \conv\{\kappa_1,\ldots,\kappa_s\}$.


\item 
Define $P$ as $\sigma \tilde P$, where $\sigma$ is chosen such that $\Vol(P) = 4\pi/3$.
We will refer to the scaled points $\sigma\kappa_j = \sigma(p^1_j / Q_1, p^2_j / Q_2, p^3_j/Q_3)$ for $j=1,\ldots,s$ 
as \emph{corners} of $P$.

\item 
Define $P^\R = \sigma \conv\{\kappa_1^\R,\ldots,\kappa_s^\R\}$ as the scaled convex hull of all the initial points $\kappa_1^\R,\ldots,\kappa_{s}^\R$.

\item 
Define the centre as $z = \sigma ( 1/Q_1, 1/Q_2, 1/Q_3)$.
\end{itemize}
The \emph{Fermi polyhedron} is the rescaled version defined as $P_F := k_FP \cap \frac{2\pi}{L}\Z^3$,
where $L$ is chosen large (depending on $k_F$) such that $\frac{k_FL}{2\pi}$ is rational and large.

We additionally define $P_F^\R := k_F P^\R \cap \frac{2\pi}{L}\Z^3$.
\end{defn}

\begin{remark}\label{rmk.choose.L.not.N}
We choose $N := \# P_F$, so that the Fermi polyhedron is filled. 
The dependence in $N$ of, for instance, $Q$ should therefore more precisely be given in terms 
of a dependence on $k_F L$. 
Note that $N = \rho L^3 \sim (k_FL)^3$ and $k_F = (6\pi^2\rho)^{1/3}(1 + O(N^{-1/3}))$.

We will choose also $s$ depending on $N$ (i.e. on $k_FL$) satisfying $s\to \infty$ as $N\to \infty$.
\end{remark}

\begin{remark}[Comments on and properties of the construction]\label{rmk.comment.def.P_F}
We collect here some properties of the Fermi polyhedron, some of which will only be needed in \Cref{sec.derivative.lebesgue.constant}.
\begin{itemize}



\item 
The points $\kappa_1,\ldots,\kappa_s$ are \emph{evenly distributed} on a thickened sphere of radius $Q^{-3/4}$ 
-- their radial coordinates are $|\kappa_j| = Q^{-3/4} + O(Q^{-1})$.
Indeed, the points $\kappa_1^\R,\ldots,\kappa_s^\R$ are \emph{evenly distributed} and $Q^{-1} \ll s^{-1/2}Q^{-3/4}$. 
For $s$ points on a thickened sphere of radius $Q^{-3/4}$, the natural lengthscale between points is $s^{-1/2}Q^{-3/4}$.

\item
There is some constraint on the number of points $s$. 
A generic point $\kappa$ (with $\kappa^1,\kappa^2,\kappa^3$ all different and non-zero) has $48$ images, including itself.
The constraint on $s$ is more or less the same as for the simpler \Cref{defn.P_F}. 

\item 
By choosing the points $\kappa_1,\ldots,\kappa_s$ as in \Cref{eqn.define.kappa_j.polyhedron} we break the symmetries 
of permuting the coordinates, i.e. 
$(k^1,k^2,k^3)\mapsto (k^a, k^b, k^c)$ if $(a,b,c) \ne (1,2,3)$.
These symmetries are however still almost satisfied, see \Cref{lem.preserve.symmetry}.

\item
We choose $s, Q$ such that $Q^{-1/4} \ll s^{-1/2}$ in the limit of large $N$.
Hence, for $N$ sufficiently large, all the chosen points $\{\kappa_1,\ldots,\kappa_s\}$ are extreme points of $\tilde P$,
i.e. all corners are extreme points of the polyhedron $P$.
That is, the name ``corner'' is well-chosen, and we do not have any superfluous points in the construction.


\item 
For any three points $(x_i, y_i, z_i)\in \R^3$, $i=1,2,3$ the plane through them is given by the equation 
\[
\begin{pmatrix}
	(y_2-y_1)(z_3 - z_1) - (y_3 - y_1)(z_2 - z_1)
	\\ (z_2 - z_1)(x_3 - x_1) - (z_3 - z_1)(x_2 - x_1)
	\\ (x_2 - x_1)(y_3 - y_1) - (x_3 - x_1)(y_2 - y_1)
\end{pmatrix}
\cdot 
\begin{pmatrix}
x \\ y \\ z
\end{pmatrix}
	= \textnormal{const.}
\]
Hence, for three points $K_1, K_2, K_3$ of the form $K_i = (p_{i}^1 / Q_1, p_{i}^2 / Q_2, p_{i}^3/Q_3), p_{i}^\mu \in \Z$, $i,\mu=1,2,3$ 
the plane through them is given by
\begin{equation}\label{eqn.plane.def.P}
	\frac{\alpha_1}{Q_2Q_3} k^1 + \frac{\alpha_2}{Q_1Q_3}k^2 + \frac{\alpha_3}{Q_1Q_2}k^3 = \gamma \in \Q, 	
\end{equation}
where 
\[
	\alpha_1 = (p_{2}^2-p_{1}^2)(p_{3}^3-p_{1}^3) - (p_{3}^2-p_{1}^2)(p_{2}^3-p_{1}^3) \in \Z
\] 
and similarly for $\alpha_2, \alpha_3$.
From these formulas it is immediate that for $K_i$'s corners of $P$ or the centre $z$ we have $|\alpha_j| \leq C\sqrt{Q}$ for $j=1,2,3$.
For some planes we may have $\alpha_j=0$ for some $j$.

\item
We claim that $\sigma = Q^{3/4}(1 + O(s^{-1}))$. In particular, that any point 
on the boundary $\partial P$ has radial coordinate $1 + O(s^{-1})$.
To see this, note that $Q^{3/4}\tilde P$ is a polyhedron whose corners are evenly spaced and have radial coordinates $r$ with 
$r = 1 + O(Q^{-1/4})$. 
Thus, by scaling $Q^{3/4}\tilde P$ by $1-CQ^{-1/4}$ we get that $(1-CQ^{-1/4})Q^{3/4}\tilde P \subset B_1(0)$ 
so that this has volume $\leq \frac{4\pi}{3}$.
It follows that $\sigma \geq Q^{3/4}(1-CQ^{-1/4})$.
On the other hand, scaling $Q^{3/4}\tilde P$ by $1+Cs^{-1}$ we have that $(1+Cs^{-1})Q^{3/4}\tilde P \supset B_1(0)$.
Indeed, since the distance from any point $k$ on the sphere of radius $1$ to any corner of $Q^{3/4}\tilde P$ 
is $\lesssim s^{-1/2}$, and the sphere is locally quadratic, the smallest radial coordinate $r$ of a point on the boundary 
$\partial (Q^{3/4}\tilde P)$ is $r\geq 1 - Cs^{-1}$. It follows that $\sigma \leq (1 + Cs^{-1})Q^{3/4}$.
Since $Q^{-1/4} \leq Cs^{-1}$ this shows the desired.


\item
Note moreover that $\sigma$ is irrational. Indeed, the volume of a polyhedron with rational corners is rational.
(This is easily seen for tetrahedra, of which any polyhedron is an essentially disjoint union.)
Thus $\sigma^3 = \pi r$ for a rational $r$.
Hence the equations of the planes defined by corners of $P$ (i.e. scaled points) are of the form \Cref{eqn.plane.def.P}
with an irrational constant $\sigma \gamma$ on the right-hand side. 
Indeed, the corners of $P$ (and the central point $z$) 
are all scaled by $\sigma$ compared to points of the form $(p^1 / Q_1, p^2 / Q_2, p^3/Q_3)$.
The equation of the plane through three scaled points only differ by scaling the constant term. 
Since $\sigma$ is irrational, and the constant term was rational for the unscaled points,
this shows the desired.

\item 
We now construct a triangulation of $\partial P$. 
For all ($2$-dimensional) triangular faces of $P$ simply consider these as part of the triangulation.
That is, we construct edges between any pair of the three corners of such a triangle.
Some of the ($2$-dimensional) faces of $P$ may be polygons of more than $3$ sides ($1$-dimensional faces). 
Construct edges between all pairs of corners sharing a side (i.e. a $1$-dimensional face)
and choose one corner and construct edges from this corner to all other corners of the polygon.

Doing this constructs a triangulation of $\partial P$ and we will refer to all pairs of corners with an edge between them as 
\emph{close} or \emph{neighbours}.
Since the points $\{\kappa_1,\ldots,\kappa_s\}$ are evenly distributed, that the distance 
between any pair of close corners is $d\sim s^{-1/2}$.

\item
Additionally, one may note that the corners of $P$ have $\leq C$ many neighbours 
since the points are evenly distributed.

\item
The reason we need $\frac{Lk_F}{2\pi}$ rational will only become apparent in \Cref{sec.derivative.lebesgue.constant}
and will be explained there.
\end{itemize}
\end{remark}

\begin{remark}[Motivation of construction]\label{rmk.motivation.P_F}
The purpose of the construction is twofold.
Firstly we avoid a casework argument as in the proof of \cite[Lemma 3.5]{Kolomoitsev.Lomako.2018} 
of whether the coefficients of the planes are rational or not. 
The argument in \Cref{lem.lemma.3.6,lem.lemma.3.9} is heavily inspired by \cite[Lemmas 3.6, 3.9]{Kolomoitsev.Lomako.2018},
where such casework is required.
Secondly we have good control over how many (and which) lattice points (i.e. points in $\frac{2\pi}{L}\Z^3$)  can lie on each plane 
(or, rather, a closely related plane, see \Cref{sec.derivative.lebesgue.constant} for the details). 

These are technical details only needed in \Cref{sec.derivative.lebesgue.constant}. 
We reiterate, that apart from the arguments in \Cref{sec.derivative.lebesgue.constant}, 
the reader may have the simpler \Cref{defn.P_F} in mind instead.
\end{remark}

\noindent
As mentioned in \Cref{rmk.comment.def.P_F} the Fermi polyhedron is almost symmetric under permutation of the axes.
This is formalized as follows.
\begin{lemma}\label{lem.preserve.symmetry}
For $\mu\ne \nu$ let $F_{\mu\nu}$ be the map that permutes $k^\mu$ and $k^\nu$
(i.e. $F_{12}(k^1,k^2,k^3) = (k^2,k^1,k^3)$, etc.).
Then for any function $t \geq 0$ we have 
\[
	\sum_{k\in \frac{2\pi}{L}\Z^3} \abs{\chi_{(k\in P_F)} - \chi_{(k\in F_{\mu\nu}(P_F))}} t(k)
	\lesssim Q^{-1/4} N \sup_{|k| \sim k_F} t(k)
	\lesssim N^{2/3} \sup_{|k| \sim k_F} t(k),
\]
where $Q$ is as in \Cref{defn.P_F.true} and $\chi$ denotes the indicator function.
\end{lemma}
\begin{proof}
Note that 
\begin{multline}\label{eqn.preserve.symmetry.triangle.ineq}
\sum_{k\in \frac{2\pi}{L}\Z^3} \abs{\chi_{(k\in P_F)} - \chi_{(k\in F_{\mu\nu}(P_F))}} t(k)
	\\
	\leq \sum_{k\in \frac{2\pi}{L}\Z^3} \abs{\chi_{(k\in P_F)} - \chi_{(k\in P_F^\R)}} t(k)
		+ \sum_{k \in \frac{2\pi}{L}\Z^3} \abs{\chi_{(k\in F_{\mu\nu}(P_F))} - \chi_{(k\in F_{\mu\nu}(P_F^\R))}} t(k)
\end{multline}
since $P_F^\R$ is invariant under permutation of the axes, i.e. $F_{\mu\nu}(P_F^\R) = P_F^\R$.
The points $\{\kappa_j\}_{j=1,\ldots,s}$ only differ from $\{\kappa_j^\R\}_{j=1,\ldots,s}$ by at most $\sim Q^{-1}$ thus 
the points 
$\{\sigma \kappa_j\}_{j=1,\ldots,s}$ (the corners of $P$) 
only differ from the points $\{\sigma\kappa_j^\R\}_{j=1,\ldots,s}$ by $\sim Q^{-1/4}$.
Hence, the support of $\chi_{(k\in P_F)} - \chi_{(k\in P_F^\R)}$ 
is contained in a shell of width $\sim k_FQ^{-1/4}$ around the surface $\partial(k_F P)$.
That is,
\[
	\supp \left(\chi_{(k\in P_F)} - \chi_{(k\in P_F^\R)}\right)
	\subset 
	\left\{k \in \frac{2\pi}{L}\Z^3 : \dist(k, \partial(k_FP)) \lesssim k_F Q^{-1/4}\right\}.
\]
The surface $\partial(k_F P)$ has area $\sim k_F^2$ so 
\[
	\Vol\left(\left\{k \in \R^3 : \dist(k, \partial(k_FP)) \lesssim k_F Q^{-1/4}\right\}\right)
	\lesssim k_F^3 Q^{-1/4}.
\] 
The spacing between the $k$'s in $\frac{2\pi}{L}\Z^3$ is $\sim L^{-1}$ and 
any $k$ with $\dist(k, \partial(k_FP)) \lesssim k_F Q^{-1/4}$ has $|k|\sim k_F$.
Thus
\[
	\sum_{k\in P_F} \abs{\chi_{(k\in P_F)} - \chi_{(k\in P_F^\R)}} t(k)
	\lesssim L^3 k_F^3 Q^{-1/4} \sup_{|k|\sim k_F} t(k)
	\sim Q^{-1/4} N \sup_{|k|\sim k_F} t(k).
\]
The same argument applies to the second summand in \Cref{eqn.preserve.symmetry.triangle.ineq}. 
We conclude the desired.
\end{proof}

\noindent 
We now improve on \Cref{eqn.Lebesgue.constant.polyhedron} for our polyhedron. 
\begin{lemma}\label{lem.lebesgue.constant.fermi.polyhedron}
The Lebesgue constant of the Fermi polyhedron satisfies
\[
\int_{\Lambda} \frac{1}{L^3} \abs{\sum_{k\in P_F} e^{ikx}} \ud x
= \frac{1}{(2\pi)^3} \int_{[0,2\pi]^3} \abs{\sum_{q \in \left(\frac{Lk_F}{2\pi}P\right)\cap \Z^3} e^{iqu}} \ud u \leq C s (\log N)^3.
\]
\end{lemma}
\noindent
The proof is (almost) the same as given in \cite[Theorem 4.1]{Kolomoitsev.Lomako.2018}. We need to be a bit more careful 
in the decomposition into tetrahedra. 
\begin{proof}
Define $R = \frac{Lk_F}{2\pi}$.
We decompose $RP$ into tetrahedra using the ``central'' point $z$ from the construction of $P$.
We triangulate the surface of $RP$ as in \Cref{rmk.comment.def.P_F}.
For each triangle in the triangulation add the point $Rz$ to form a tetrahedron. 
Note that $Rz \notin \Z^3$ since $|Rz| \leq C R Q^{-1/4} \ll 1$ and $z\ne 0$.
This gives $m=O(s)$ many (closed) tetrahedra $\{T_j\}$ such that $RP = \bigcup T_j$ and that $T_j\cap T_{j'}$ 
is a tetrahedron of lower dimension (i.e. the central point $Rz$, a line segment or a triangle).
Then, as in \cite[Theorem 4.1]{Kolomoitsev.Lomako.2018} by the inclusion--exclusion principle we have 
\[
\begin{aligned}
\mcL(RP) 
	& = \frac{1}{(2\pi)^3}\int \abs{\sum_j \sum_{q\in T_j \cap \Z^3} e^{iqu} - \sum_{j < j'} \sum_{q\in T_j \cap T_{j'}\cap \Z^3} e^{iqu} + \cdots} \ud u
	\\ & \leq 
		\sum_{\ell = 1}^m \sum_{j_1 < \ldots < j_\ell} \mcL(T_{j_1}\cap \ldots \cap T_{j_\ell}).
\end{aligned}
\]
In \cite[Theorem 4.1]{Kolomoitsev.Lomako.2018} it is shown that for a $d$-dimensional tetrahedron $T$ with $T \subset [0,n_1]\times \ldots \times [0,n_d]$
we have $\mcL(T)\leq C(d)\prod_{i = 1}^d\log (n_i + 1)$. 
All the 
tetrahedra in our construction are 
$d$-dimensional for $d\leq 3$ and
contained in boxes $[0, CR]^d$ (after translations by lattice vectors $\kappa \in \Z^3$). 
Hence for all tuples $T_{j_1},\ldots,T_{j_\ell}$ we have 
\[ 
	\mcL(T_{j_1}\cap \ldots \cap T_{j_\ell}) \leq C (\log R)^d \leq C (\log N)^3.
\]
We need to count how many summands we have. 
The $3$-dimensional tetrahedra each appear just once, and there are $m=O(s)$ many of them. 
The $2$-dimensional tetrahedra (triangles) appear just once, namely in the term $\mcL(T_j \cap T_{j'})$ where the triangle is the intersection $T_j \cap T_{j'}$.
Hence there are $O(s)$ many such terms. 
The $1$-dimensional tetrahedra (line segments) may appear more times, with $3,4,\ldots,C$ many $T_j$'s. 
Indeed an edge may be shared by more tetrahedra, but only a bounded number of them. (This follows from the points being well-distributed, 
so each corner of $P$ has a bounded number of neighbours.) 
Since there is also only $O(s)$ many $1$-dimensional line segments this gives also just a contribution $O(s)$. 
The central point appears many times, but all appearances contribute $0$, 
since $Rz \notin \Z^3$.
We conclude that $\mcL(RP) \leq C s(\log N)^3$ as desired.
\end{proof}

\noindent 
By replacing $B_F$ with $P_F$ we make an error in the kinetic energy. 
(The Fermi ball $B_F$ is the set of momenta of the Slater determinant with lowest kinetic energy.)
We now bound the error made with this approximation. That is, we consider 
\[
	\sum_{k \in P_F} |k|^2 - \sum_{k \in B_F} |k|^2. 
\]
Note that there might not be the same number of summands in both sums.
To compute this difference 
we interpret the sums as Riemann-sums and replace them with the corresponding integrals. 
It is a simple exercise to show that the error made in this replacement is $Ck_F^2N^{2/3}$. 
That is,
\[
	\sum_{k \in P_F} |k|^2 - \sum_{k \in B_F} |k|^2
	= \frac{L^3}{(2\pi)^3}\left(\int_{k_FP} |k|^2 \ud k - \int_{B(k_F)} |k|^2 \ud k\right) + O\left(k_F^2N^{2/3}\right).
\]
The integrals can be computed in spherical coordinates, 
\[
	\int_{k_FP} |k|^2 \ud k - \int_{B(k_F)} |k|^2 \ud k = \int_{\S^2} \left(\int_0^{k_FR(\omega)} r^4 \ud r - \int_0^{k_F} r^4 \ud r\right) \ud \omega
\]
For $k_FP$ the radial limit is $k_FR(\omega) = k_F(1 + \eps(\omega))$, 
where $\eps(\omega) = O(s^{-1})$ uniformly in $\omega$ by the argument in \Cref{rmk.comment.def.P_F}. 
Expanding the powers of $R$ we thus get 
\[
	\int_{k_FP} |k|^2 \ud k - \int_{B(k_F)} |k|^2 \ud k = k_F^5\int_{\S^2} \left(\eps(\omega) + O(s^{-2})\right) \ud \omega.
\]
By construction, $P$ has volume $4\pi/3$. That is, $k_FP$ and $B(k_F)$ have the same volume. This means that
\[
	0  = \int_{\S^2} \left(\int_0^{k_FR(\omega)} r^2 \ud r - \int_0^{k_F} r^2 \ud r\right) \ud \omega
		= k_F^3\int_{\S^2} \left(\eps(\omega) + O(s^{-2})\right) \ud \omega.
\]
We thus get that 
\[
	\sum_{k \in P_F} |k|^2 - \sum_{k \in B_F} |k|^2
	= O(k_F^2Ns^{-2}) + O\left(k_F^2N^{2/3}\right).
\]
We conclude the following.
\begin{lemma}\label{lem.KE.polyhedron}
The kinetic energy of the (Slater determinant with momenta in the) Fermi polyhedron satisfies 
\[
	\sum_{k\in P_F} |k|^2 = \sum_{k\in B_F} |k|^2 \left(1 + O(N^{-1/3}) + O(s^{-2})\right)
		= \frac{3}{5}(6\pi)^{2/3} \rho^{2/3} N\left(1 + O(N^{-1/3}) + O(s^{-2})\right).
\]
\end{lemma}
\begin{proof}
The computation above gives the first equality. The second follows by noting that $\sum_{k\in B_F} |k|^2$ is a Riemann sum for 
\[
\frac{L^3}{(2\pi)^3} \int_{|k|\leq k_F} |k|^2 \ud k = \frac{4\pi}{5(2\pi)^3}k_F^5 L^3 = \frac{3}{5}(6\pi^2)^{2/3}\rho^{2/3}N(1 + O(N^{-1/3})).
\qedhere 
\]
\end{proof}

\noindent
Completely analogously one can show that 
\begin{equation}\label{eqn.sumk4.P_F}
	\sum_{k\in P_F} |k|^4 
		= \frac{18\pi^2}{7}(6\pi)^{1/3} \rho^{4/3} N\left(1 + O(N^{-1/3}) + O(s^{-2})\right).
\end{equation}
We need this formula for \Cref{prop.2.density} below. 
Additionally we need a formula for 
$\sum_{k\in P_F} |k^1|^4$, 
where $k^1$ refers to the first coordinate of $k=(k^1,k^2,k^3)$.
Here we have 
\begin{equation}\label{eqn.sumk4.P_F.part2}
	\sum_{k\in P_F} |k^1|^4 
	= \frac{18\pi^2}{35}(6\pi)^{1/3} \rho^{4/3} N\left(1 + O(N^{-1/3}) + O(s^{-1})\right).
\end{equation}
To see this we compare it to $\sum_{k\in B_F} |k^1|^4$.
The only difference from above is when doing the spherical integral.
We have 
\[
\begin{aligned}
	& 
	\sum_{k\in P_F} |k^1|^4  - \sum_{k\in B_F} |k^1|^4 
	\\ & \quad 
	 	= \frac{L^3}{(2\pi)^3} 
		\int_{0}^{2\pi} \ud \theta 
		\int_0^\pi \ud \phi 
			\cos(\phi)^4 
		\left(\int_0^{k_F(1 + \eps(\phi,\theta))} r^6 \ud r - \int_0^{k_F} r^6 \ud r\right)
		+ O(k_F^4 N^{2/3})
	\\ & \quad 
		= O(k_F^4 N s^{-1}) + O(k_F^4 N^{2/3}),
\end{aligned}
\]
since we can't use the volume constraint that $\int_{\S^2} \eps(\omega) \ud \omega = O(s^{-2})$
but only that $\eps(\omega) = O(s^{-1})$. 
Thus we only get an error of (relative) size $s^{-1}$.
The sum over $B_F$ may be readily computed by computing the corresponding integral.

\subsection{Reduced densities of the Slater determinant}\label{sec.slater.densities}
We now consider the $2$-particle reduced density of the (normalized) Slater determinant.
We have the following. 
\begin{lemma}\label{prop.2.density}
The $2$-particle reduced density of the (normalized) Slater determinant $\frac{1}{\sqrt{N!}}D_N$ satisfies 
\[
\begin{aligned}
\rho^{(2)}(x_1,x_2) 
	& = \frac{(6\pi^2)^{2/3}}{5} \rho^{8/3} |x_1-x_2|^2
	\biggl(
	1 - \frac{3(6\pi^2)^{2/3}}{35}\rho^{2/3}|x_1-x_2|^2 
	\\ & \qquad 
	+ O(N^{-1/3}) 
	+ O(s^{-2}) 
	+ O(N^{-1/3}\rho^{2/3}|x_1-x_2|^2)
	+ O(\rho^{4/3}|x_1-x_2|^4)
	\biggr).
\end{aligned}
\]
\end{lemma}
\noindent
This follows from a Taylor expansion akin to the argument in \Cite[Lemma 11]{Agerskov.Reuvers.ea.2022}.
\begin{proof}
The Slater determinant is in particular a quasi-free state, hence we get by Wick's rule that
\begin{equation}\label{eqn.wick.rho2}
	\rho^{(2)}(x_1,x_2) = \rho^{(1)}(x_1)\rho^{(1)}(x_2) - \gamma^{(1)}_N(x_1;x_2) \gamma^{(1)}_N(x_2,x_1), 
\end{equation}
where $\gamma^{(1)}_N$ denotes the (kernel of the) reduced $1$-particle density matrix of the Slater determinant.
We have 
\[
	\gamma^{(1)}_N(x_1; x_2) = \sum_{k\in P_F} u_k(x_1) \overline{u_k(x_2)}
		= \frac{1}{L^3}\sum_{k\in P_F} e^{ik(x_1-x_2)},
	\qquad 
	\rho^{(1)}(x_1) = \rho.
\]
By translation invariance, $\gamma^{(1)}_N(x_1; x_2)$ is a function of $x_1-x_2$ only, and we shall  Taylor
expand $\gamma^{(1)}_N$ in $x_1 - x_2$. By construction $P_F$ is reflection symmetric in the axes, see \Cref{defn.P_F.true}.
This means that all odd orders vanish and that all off-diagonal second order terms vanish. 
Thus, by defining $x_{12} = (x_{12}^1, x_{12}^2, x_{12}^3) = x_1-x_2$ and expanding all the exponentials we get 
\[
\begin{aligned}
\gamma^{(1)}_N(x_1;x_2)
	& = \frac{1}{L^3}\sum_{k\in P_F} 
	\left(1 
		- \frac{1}{2}(k\cdot (x_1-x_2))^2
		+ \frac{1}{24}(k\cdot (x_1-x_2))^4
		+ O(|k|^6 |x_1-x_2|^6)
	\right)
	\\ 
	& = \rho - \frac{1}{2L^3}\sum_{k\in P_F} \left[|k^1|^2 |x_{12}^1|^2 + |k^2|^2|x_{12}^2|^2 + |k^3|^2|x_{12}^3|^2\right]
	\\ & \quad 
		+ \frac{1}{24L^3} 
		\left(
		\sum_{k\in P_F}  \left[|k^1|^4|x_{12}^1|^4+ |k^2|^4|x_{12}^2|^4 + |k^3|^4|x_{12}^3|^4\right]
		\right.
	\\
	& \qquad 
		\left. 
		+ 6\sum_{k\in P_F} \left[|k^1|^2 |k^2|^2|x_{12}^1|^2|x_{12}^2|^2
			+ |k^1|^2 |k^3|^2|x_{12}^1|^2|x_{12}^3|^2
			+ |k^2|^2 |k^3|^2|x_{12}^2|^2|x_{12}^3|^2\right]
		\right)
	\\ & \quad 
		+ O(\rho^{3}|x_1-x_2|^6).
\end{aligned}
\]
By \Cref{lem.preserve.symmetry} we may write 
\[
	\sum_{k\in P_F} |k^\mu|^2 
	= \frac{1}{3}\sum_{k\in P_F} |k|^2 + O\left(N^{2/3} k_F^2\right)
\]
and similar for the $\sum |k^\mu|^4$ and $\sum |k^\mu|^2|k^\nu|^2$-sums. 
Using this 
the second order term is given by 
\[
	-\frac{1}{6L^3} \sum_{k\in P_F} |k|^2 |x_{12}|^2 + O(\rho^{5/3} |x_{12}|^2 N^{-1/3}).
\]
Similarly by also rewriting everything in terms of $|x_{12}|^4$ and $\left[|x_{12}^1|^4+ |x_{12}^2|^4 + |x_{12}^3|^4\right]$ 
the fourth order term is given by 
\begin{multline*}
	\frac{1}{48L^3} \left[
	\left(\sum_{k\in P_F} |k|^4 - 3 \sum_{k\in P_F} |k^1|^4\right) |x_{12}|^4 
	+ \left(5 \sum_{k\in P_F} |k^1|^4- \sum_{k\in P_F} |k|^4\right) \left[|x_{12}^1|^4+ |x_{12}^2|^4 + |x_{12}^3|^4\right]
	\right]
	\\
	+ O(\rho^{7/3}|x_{12}|^4 N^{-1/3}).
\end{multline*}
Using \Cref{lem.KE.polyhedron,eqn.sumk4.P_F,eqn.sumk4.P_F.part2} we get that 
\[
\begin{aligned}
\gamma^{(1)}_N(x_1;x_2)
	& = \rho - \frac{(6\pi^2)^{2/3}}{10}\rho^{5/3}|x_1-x_2|^2 
		+ \frac{3\pi^2(6\pi^2)^{1/3}}{140} \rho^{7/3}|x_1-x_2|^4 
		+ O(\rho^{5/3}N^{-1/3} |x_1-x_2|^2)
	\\ & \qquad 
		+ O(\rho^{5/3} s^{-2} |x_1-x_2|^2)
		+ O(\rho^{7/3} N^{-1/3} |x_1-x_2|^4)
		+ O(\rho^{7/3} s^{-1} |x_1-x_2|^4)
	\\ & \qquad
		+ O(\rho^3 |x_1-x_2|^6).
\end{aligned}
\]
Plugging this into \Cref{eqn.wick.rho2} we conclude the desired.
\end{proof}

\noindent
Finally, we have the following bound on the $3$-particle reduced density
\begin{lemma}\label{lem.bound.rho3}
The $3$-particle reduced density of the (normalized) Slater determinant $\frac{1}{\sqrt{N!}} D_N$ satisfies
\[
  \rho^{(3)}(x_1,x_2,x_3) \leq C \rho^{3+4/3} |x_1-x_2|^2 |x_1-x_3|^2.
\]
\end{lemma}
\begin{proof}
Note that $\rho^{(3)}$ vanishes whenever any $2$ of the $3$ particles are incident
and moreover that $\rho^{(3)}$ is symmetric under exchange of the particles.
We may bound derivatives of $\rho^{(3)}$ as we did for $\rho^{(2)}$.
By Taylor's theorem we conclude the desired.
\end{proof}

\section{Gaudin-Gillespie-Ripka-expansion}\label{sec.gaudin.expansion}
We now present the cluster expansion of Gaudin, Gillespie and Ripka \cite{Gaudin.Gillespie.ea.1971}. 
The argument given here is essentially the same as in \cite{Gaudin.Gillespie.ea.1971}, 
only we give sufficient conditions for the formulas \cite[Equations (3.19), (4.9) and (8.4)]{Gaudin.Gillespie.ea.1971},
given in \Cref{thm.gaudin.expansion}, to hold, i.e., for absolute convergence of the expansion.

Recall the definition of the trial state $\psi_N$ in \Cref{eqn.define.trial.state}.
We calculate the the normalization constant $C_N$ and the reduced densities 
$\rho^{(1)}_{\textnormal{Jas}}, \rho^{(2)}_{\textnormal{Jas}}$ and $\rho^{(3)}_{\textnormal{Jas}}$
defined in \Cref{eqn.define.rho(n)}

We remark that the computation given in the following is not just valid for the function $f$
and (square of a) Slater determinant $|D_N|^2$ we choose, but these can be replaced 
by a more general function and determinant of a more general matrix.
We comment on this further in \Cref{rmk.general.f.gamma}.

\subsubsection{Calculation of the normalization constant}\label{sec.calc.C_N}
First we give the calculation of $C_N$. 
Rewriting the $f$'s in terms of $g = f^2-1$ we have
\[
C_N = \idotsint\prod_{i < j} f_{ij}^2 |D_N|^2 \ud x_1 \ldots \ud x_N =  \idotsint \prod_{i < j} (1 + g_{ij}) |D_N|^2 \ud x_1 \ldots \ud x_N
\]
We factor out the $g_{ij}$ and group terms with the same number of values $x_i$. 
(For instance $g_{12}g_{23}$ and $g_{45}g_{46}g_{56}$ both have $3$ values $x_i$ appearing, 
the values $x_1,x_2,x_3$ and $x_4,x_5,x_6$, respectively).
To state the result we define 
$\mcG_p$ as the set of all graphs on $\{1,\ldots,p\}$ such that each vertex has degree at least $1$
(i.e. is incident to at least one edge)
and define 
\[
	W_p(x_1,\ldots,x_p) := \sum_{G\in \mcG_p} \prod_{e\in G} g_e.
\]
(Note that for $p=0,1$ we have $\mcG_p = \varnothing$ and so $W_p = 0$.)
By the symmetry of permuting the coordinates we have
\[
\begin{aligned}
C_N
& = \idotsint \left[1 + \frac{N(N-1)}{2}W_2(x_1,x_2) + \frac{N(N-1)(N-2)}{3!} W_3(x_1,x_2,x_3) + \ldots\right]
\\ & \qquad \times  |D_N|^2 \ud x_1 \ldots \ud x_N
\\
&  = N!\left[1 + \sum_{p=2}^N \frac{1}{p!}\idotsint W_p(x_1,\ldots,x_p) 
\rho^{(p)}(x_1,\ldots,x_p)
\ud x_1 \ldots \ud x_p\right],
\end{aligned}
\]
where again the reduced densities are normalised as
\[
  \rho^{(p)}(x_1,\ldots,x_p) = N(N-1)\cdots(N-p+1) \idotsint \frac{1}{N!} |D_N(x_1,\ldots,x_N)|^2 \ud x_{p+1} \ldots \ud x_N.
\]
A simple calculation using the Wick rule shows that 
\[ 
  \rho^{(p)}
		= \det\left[\gamma_N^{(1)}(x_i;x_j)\right]_{1\leq i,j \leq p} 
		= \det\left[\sum_{k\in P_F} \overline{u_k(x_i)} u_k(x_j)\right]_{1\leq i,j \leq p} 
		= \det[S_p^*S_p],
\]
where $S_p$ is the $P_F\times p$ ``Slater''-matrix with entries $u_k(x_i)$.
This has rank $\min\{N, p\} = p$ and so by taking this determinant as the definition of 
$\rho^{(p)}$
for $p > N$ 
we have 
$\rho^{(p)} = 0$
for $p > N$. Thus we may extend the summation to $\infty$.
We now expand out the determinant and the $W_p$. That is
\[
\begin{aligned}
\frac{C_N}{N!}
& = 1 + \sum_{p=2}^\infty \frac{1}{p!} \sum_{\substack{G \in \mcG_p \\ \pi \in \mcS_p}} (-1)^\pi 
	\idotsint \prod_{e \in G} g_{e} \prod_{j=1}^p \gamma^{(1)}_N(x_j, x_{\pi(j)}) \ud x_1 \ldots \ud x_p,
\end{aligned}
\]
where $\mcS_p$ denotes the symmetric group on $p$ elements.
We will consider $\pi$ and $G$ together as a diagram $(\pi, G)$. 
We give a slightly more general definition for what a diagram is, as we will need such 
for the calculation of the reduced densities.
\begin{defn}\label{def.diagram}
Define the set $\mcG_p^q$ as the set of all graphs on $q$ ``external'' vertices $\{1,\ldots,q\}$ and 
$p$ ``internal'' vertices $\{q+1,\ldots,q+p\}$ such that all internal vertices have degree at least $1$,
i.e. each internal vertex has at least one incident edge,
and such that there are no edges between external vertices.
The external edges are allowed to have degree zero, i.e. have no incident edges.
For $q=0$ we recover $\mcG_p^0 = \mcG_p$.

A \emph{diagram} $(\pi, G)$ on $q$ ``external'' and $p$ ``internal'' vertices 
is a pair of a permutation $\pi\in \mcS_{q+p}$ (viewed as a directed graph on $\{1,\ldots,q+p\}$)
and a graph $G\in\mcG_p^q$.
We denote the set of all diagrams on $q$ ``external'' vertices and $p$ ``internal'' vertices by $\mcD_{p}^q$.

We will sometimes refer to edges in $G$ as $g$-edges, directed edges in $\pi$ as $\gamma^{(1)}_N$-edges and the 
graph $G$ as a $g$-graph.  
The value of a diagram $(\pi, G) \in \mcD_{p}^q$ is the function
\[
	\Gamma_{\pi, G}^q(x_1,\ldots,x_q) 
	:= 
		(-1)^\pi \idotsint \prod_{e \in G} g_e \prod_{j=1}^{q+p} \gamma^{(1)}_N(x_j, x_{\pi(j)}) \ud x_{q+1} \ldots \ud x_{q+p}.
\]
For $q=0$ we write $\Gamma_{\pi,G} = \Gamma_{\pi,G}^0$ and $\mcD_{p} = \mcD_p^0$.

A diagram $(\pi, G)$ is said to be \emph{linked} if the graph $\tilde G$ with edges the union of edges in $G$ and directed edges in $\pi$ is connected.
The set of all linked diagrams on $q$ ``external'' and $p$ ``internal'' vertices is denoted $\mcL_{p}^q$.
For $q=0$ we write $\mcL_p = \mcL_p^0$.
\end{defn}

\noindent
By the translation invariance we have that $\Gamma^1_{\pi,G}$ is a constant for any diagram $(\pi,G)$.

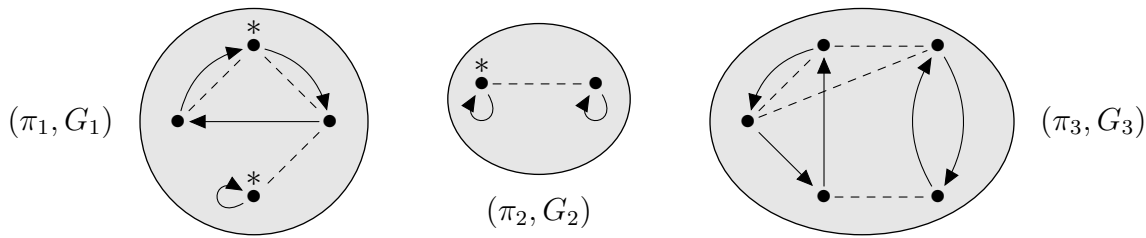
\begin{figure}[htb]
\centering
\begin{tikzpicture}[line cap=round,line join=round,>=triangle 45,x=1.0cm,y=1.0cm]
\draw [rotate around={0.:(2.,2.)},line width=0.5pt,fill=black,fill opacity=0.1] (2.,2.) ellipse (1.5cm and 1.5cm);
\draw [rotate around={0.:(9.25,2.)},line width=0.5pt,fill=black,fill opacity=0.1] (10,2.) ellipse (2.0cm and 1.5cm);
\draw [rotate around={0.:(5.5,3.)},line width=0.5pt,fill=black,fill opacity=0.1] (5.75,2.3) ellipse (1.2 cm and 1.0cm);
\node (1) at (1,2) {};
\node (2) at (2,3) {};
\node (3) at (3,2) {}; 
\node (4) at (2,1) {}; 
\node (5) at (5,2.5) {}; 
\node (6) at (6.5,2.5) {}; 
\node (7) at (8.5,2) {}; 
\node (8) at (9.5,1) {};
\node (9) at (9.5,3) {}; 
\node (10) at (11,1) {}; 
\node(11) at (11,3) {}; 
\foreach \x in {1,...,11} \node at (\x) {$\bullet$};
\foreach \x in {2,4,5} \node[anchor=south] at (\x) {$*$};
\draw[dashed] (1) -- (2) -- (3) -- (4);
\draw[dashed] (5) -- (6);
\draw[dashed] (7) -- (9) -- (11);
\draw[dashed] (8) -- (10);
\draw[dashed] (7) -- (11);
\draw[->] (1) to[bend left] (2);
\draw[->] (2) to[bend left] (3);
\draw[->] (3) to (1);
\draw[->] (4) to[out=-150,in=150,loop] ();
\draw[->] (5) to[out=-60,in=-120,loop] ();
\draw[->] (6) to[out=-60,in=-120,loop] ();
\draw[->] (7) to (8);
\draw[->] (8) to (9);
\draw[->] (9) to[bend right] (7);
\draw[->] (10) to[bend left] (11);
\draw[->] (11) to[bend left] (10);
\node[anchor = east] at (0.3,2) {$(\pi_1, G_1)$};
\node at (5.75, 0.8) {$(\pi_2, G_2)$};
\node[anchor = west] at (12.2,2) {$(\pi_3, G_3)$};
\end{tikzpicture}
\caption{A diagram $(\pi, G) \in \mcD_8^3$ decomposed into linked components. 
The dashed lines denote $g$-edges and the arrows $(i,j)$ denote that $\pi(i) = j$.
Vertices labelled with $*$ denote external vertices.}
\end{figure}

In terms of diagrams we thus have 
\[
	\frac{C_N}{N!} = 1 + \sum_{p=2}^\infty \frac{1}{p!} \sum_{(\pi, G)\in \mcD_p} \Gamma_{\pi, G}.
\]
If $(\pi, G)$ is not linked we may decompose it into its linked components.
Here the integration factorizes. We split the sum according to the number of linked components.
Each linked component has at least $2$ vertices, since each vertex must be connected 
to another vertex with an edge in the corresponding graph. We get 
\begin{equation}\label{eqn.normalization.linked.components}
\begin{aligned}
\frac{C_N}{N!}
& = 1 + \sum_{p=2}^\infty
	\underbrace{\vphantom{\sum_{(\pi_1,G_1) \textnormal{ linked}}} 
		\sum_{k=1}^\infty }_{\# \textnormal{ lnk. cps.}}
		\frac{1}{k!}
	\underbrace{\vphantom{\sum_{(\pi_1,G_1) \textnormal{ linked}}} 
		\sum_{p_1\geq2} \cdots \sum_{p_k \geq 2} }_{\textnormal{sizes linked cps.}}
		\chi_{(\sum p_\ell = p)}
	\underbrace{\sum_{(\pi_1,G_1)\in \mcL_{p_1}}\cdots \sum_{(\pi_k,G_k)\in \mcL_{p_k}}}_{\textnormal{linked components}}
	\frac{\Gamma_{\pi_1,G_1}}{p_1!} \cdots \frac{\Gamma_{\pi_k,G_k}}{p_k!}
\end{aligned}
\end{equation}
Here $\chi$ is the indicator function. 
The factor $1/(k!)$ comes from counting the possible labellings of the $k$ linked components.
The factors $1/(p_1!),\ldots,1/(p_k!)$ come from counting how to distribute the $p$ vertices $\{1,\ldots,p\}$ 
between the linked components of prescribed sizes $p_1,\ldots,p_k$.
This gives the factor $\binom{p}{p_1,\ldots,p_k} = \frac{p!}{p_1! \ldots p_k!}$, which together with the factor $1/p!$ already present gives the claimed formula.

We want to pull the $p$-summation inside the $p_1,\ldots,p_k$-summation. This is allowed once we check that 
$\sum_{p} \frac{1}{p!}\sum_{(\pi, G)\in\mcL_p} \Gamma_{\pi,G}$ is absolutely convergent. 
More precisely we need that the $p$-sum is absolutely convergent, i.e. 
$\sum_p \frac{1}{p!} \abs{\sum_{(\pi, G)\in\mcL_p} \Gamma_{\pi,G}} < \infty$.
This is the content of \Cref{lem.linked.sum.abs.conv} below.
We conclude that if the assumptions of \Cref{lem.linked.sum.abs.conv} are satisfied then
\begin{equation}\label{eqn.compute.final.C_N}
\begin{aligned}
\frac{C_N}{N!}	& = 1 + \sum_{k=1}^\infty \frac{1}{k!}
	\sum_{p_1\geq2} \cdots \sum_{p_k \geq 2}
	\sum_{(\pi_1,G_1)\in\mcL_{p_1}}\cdots \sum_{(\pi_k,G_k)\in\mcL_{p_k}}
	\frac{\Gamma_{\pi_1,G_1}}{p_1!} \cdots \frac{\Gamma_{\pi_k,G_k}}{p_k!}
\\
& = 1 + \sum_{k=1}^\infty \frac{1}{k!}\left(\sum_{p=2}^\infty \frac{1}{p!}\sum_{(\pi, G)\in\mcL_{p}} \Gamma_{\pi, G}\right)^k
= \exp\left(\sum_{p=2}^\infty \frac{1}{p!}\sum_{(\pi, G)\in\mcL_{p}} \Gamma_{\pi, G}\right).
\end{aligned}
\end{equation}

\subsubsection{Calculation of the \texorpdfstring{$1$}{1}-particle reduced density}\label{sec.calc.rho1}
We consider now the $1$-particle reduced density of the Jastrow trial state.
We have by the translation invariance that $\rho^{(1)}_{\textnormal{Jas}} = \rho^{(1)} = \rho$.
We nonetheless compute it here, as we need the formula in terms of (linked) diagrams.
We have similarly as before
\[
\begin{aligned}
\rho^{(1)}_{\textnormal{Jas}}(x_1)
& = \frac{N}{C_N} \idotsint \prod_{2\leq i \leq N} f_{1i}^2 \prod_{2\leq i< j \leq N } f_{ij}^2 |D_N|^2 \ud x_2 \ldots \ud x_N
\\ & = \frac{N!}{C_N}\left[\rho^{(1)}(x_1) + \sum_{p=1}^\infty \frac{1}{p!}\idotsint X_p^1 \rho^{(p+1)} \ud x_2 \ldots \ud x_{p+1}\right],
\end{aligned}
\]
where 
\[
	X_p^1 = \sum_{G\in\mcG_p^1} \prod_{e\in G} g_e,
\]
with $\mcG_p^1$ as in \Cref{def.diagram}.
Again this is what one gets by just expanding all the products in the first line and grouping them in terms of how many $x_i$'s appear.
The sum is again extended to $\infty$, since $\rho^{(p)} = 0$ for $p > N$.

We again expand out the determinant and the $X_p^1$'s. 
For each summand $\pi\in \mcS_{p+1}$ and $G\in\mcG_p^1$ we again think of them together as a diagram $(\pi, G)\in \mcD_p^{1}$. 
The formula for $\rho^{(1)}_{\textnormal{Jas}}$ in terms of diagrams is 
\[
	\rho^{(1)}_{\textnormal{Jas}}
 		= \frac{N!}{C_N}\sum_{p=0}^\infty \frac{1}{p!} \sum_{(\pi, G)\in \mcD_p^1} \Gamma_{\pi, G}^1
 		= \frac{N!}{C_N} \left[ \rho^{(1)} + \sum_{p=1}^\infty \frac{1}{p!} \sum_{(\pi, G)\in\mcD_p^1} \Gamma_{\pi, G}^1\right].
\]
As for the normalization we write out the diagrams in terms of their linked components.
There is a distinguished linked component, namely the one containing the vertex $\{1\}$. We will write its size as $p_*$.
It is convenient to take ``size'' to mean number of internal vertices, i.e. $p_* = 0$ if $\{1\}$ is not connected to any 
other vertex by either an edge in $G$ or an edge in $\pi$. 
Similarly ``number of linked components'' means disregarding the distinguished one.

Analogously to the computation in \Cref{eqn.normalization.linked.components} we thus get for any $p\geq 0$
\begin{equation*}
\begin{aligned}
\frac{1}{p!}\sum_{(\pi, G)\in\mcD_p^1}  \Gamma_{\pi, G}^1
& =  
	\sum_{(\pi_*, G_*)\in\mcL_{p}^1}
	\frac{\Gamma_{\pi_*, G_*}^1}{p!}
\\ & \quad 
	+ 
	\left[
	\sum_{k=1}^\infty \frac{1}{k!}
	\sum_{p_* \geq 0}
	\sum_{p_1\geq2} \cdots \sum_{p_k \geq 2} \chi_{(\sum_{\ell \in \{*, 1,\ldots,k\}} p_\ell = p)}
	\sum_{(\pi_*, G_*)\in\mcL_{p_*}^1}
	\frac{\Gamma_{\pi_*, G_*}^1}{p_*!}
	\right.
\\ & \qquad 
	\left.
	\times 
	\sum_{(\pi_1,G_1)\in\mcL_{p_1}}\cdots \sum_{(\pi_k,G_k)\in\mcL_{p_k}}
	\frac{\Gamma_{\pi_1,G_1}}{p_1!} \cdots \frac{\Gamma_{\pi_k,G_k}}{p_k!}
	\right],
\end{aligned}
\end{equation*}
where the superscript $1$ refers to the slightly modified structure as described in \Cref{def.diagram}, 
where there may be no $g$-edges connecting to $\{1\}$, and there is no integration over $x_1$.
Note that $(\pi_1, G_1)\in\mcL_{p_1},\ldots,(\pi_k, G_k)\in\mcL_{p_k}$ 
only deal with internal vertices. 

Again here we take the sum over $p$'s. We are allowed to permute the $p$-sum inside of the $p_*$- and $p_1,\ldots,p_k$-sums 
if the sums over linked diagrams are absolutely summable. That is, if 
\[
	\sum_{p\geq0} \frac{1}{p!} \abs{\sum_{(\pi, G)\in\mcL_{p}^1} \Gamma_{\pi, G}^1} < \infty,
	\qquad
	\sum_{p \geq 2} \frac{1}{p!} \abs{\sum_{(\pi, G)\in\mcL_{p}} \Gamma_{\pi, G}} < \infty,
\]
then we have, as for the normalization in \Cref{eqn.compute.final.C_N}, that 
\[
\begin{aligned}
\rho^{(1)}_{\textnormal{Jas}}
& = \frac{N!}{C_N}\sum_{p=0}^\infty \frac{1}{p!} \sum_{(\pi, G)\in\mcD_{p}^1} \Gamma_{\pi, G}^1
\\ & = \frac{N!}{C_N} 
	\left[\sum_{p_* \geq 0}\sum_{(\pi_*, G_*)\in\mcL_{p_*}} \frac{\Gamma_{\pi_*, G_*}^1}{p_*!}\right]
	\times 
	\left[\sum_{k=0}^\infty \frac{1}{k!} \left(\sum_{p\geq 2} \sum_{(\pi, G)\in\mcL_{p}} \frac{\Gamma_{\pi, G}}{p!}\right)^k\right]
\\
& = \sum_{p \geq 0}\sum_{(\pi, G)\in\mcL_{p}^1} \frac{\Gamma_{\pi, G}^1}{p!}
= \rho^{(1)} + \sum_{p \geq 1}\frac{1}{p!}\sum_{(\pi, G)\in\mcL_{p}^1} \Gamma_{\pi, G}^1,
\end{aligned}
\]
where we used \Cref{eqn.compute.final.C_N} 
and that the $p=0$ term just gives the $1$-particle density of the Slater determinant.
Thus, by translation invariance, we have 
\[ 
\rho = \rho^{(1)} 
	= \rho^{(1)}_{\textnormal{Jas}} 
	= \sum_{p \geq 0}\frac{1}{p!}\sum_{(\pi, G)\in\mcL_{p}^1} \Gamma_{\pi, G}^1.
\]

\subsubsection{Calculation of the \texorpdfstring{$2$}{2}-particle reduced density}\label{sec.calc.rho2}
Let us now compute the $2$-particle reduced density.
As before by expanding all the $f^2=1+g$ factors apart from the factor $f_{12}$ we get
\[
\rho^{(2)}_{\textnormal{Jas}}
= \frac{N!}{C_N} f_{12}^2 \sum_{p=0}^\infty \int X_p^{2} \rho^{(p+2)} \ud x_3\ldots\ud x_{p+2},
\qquad 
X_p^{2} = \sum_{G\in \mcG_p^{2}}\prod_{e\in G} g_e,
\]
where $\mcG_p^{2}$ is as in \Cref{def.diagram}. 

We again decompose the diagrams into linked components.
However, we need to distinguish between the cases where $\{1\}$ and $\{2\}$ are both in the same component or in two different components.
The computation is analogous to the computation above. We get
\[
\begin{aligned}
\rho^{(2)}_{\textnormal{Jas}}(x_1,x_2)
& = f_{12}^2 
	\vast(
	\underbrace{
	\sum_{p_1, p_2 \geq 0} \sum_{\substack{(\pi_1, G_1)\in\mcL_{p_1}^1 \\ (\pi_2, G_2)\in\mcL_{p_2}^1}} 
	\frac{\Gamma_{\pi_1, G_1}^1(x_1) \Gamma_{\pi_2,G_2}^1(x_2)}{p_1! p_2!}
	}_{\{1\} \textnormal{ and } \{2\} \textnormal{ in different linked components}}
	\, + \,  
	\underbrace{
	\vphantom{
	\sum_{p_1, p_2 \geq 0} \sum_{\substack{(\pi_1, G_1)\in\mcL_{p_1}^1 \\ (\pi_2, G_2)\in\mcL_{p_2}^1}} 
	\frac{\Gamma_{\pi_1, G_1}^1(x_1) \Gamma_{\pi_2,G_2}^1(x_2)}{p_1! p_2!}
	}
	\sum_{p_{12}\geq 0} \sum_{(\pi_{12}, G_{12})\in\mcL_{p_{12}}^2} 
		\frac{\Gamma_{\pi_{12}, G_{12}}^{2}(x_1,x_2)}{p_{12}!}}_{\{1\} \textnormal{ and } \{2\} \textnormal{ in same linked component}}
	\vast)
\\
& = f_{12}^2
	\left(
		\rho^{(1)}_{\textnormal{Jas}}(x_1) \rho^{(1)}_{\textnormal{Jas}}(x_2) 
		+ \sum_{p_{12}\geq 0} \frac{1}{p_{12}!}\sum_{(\pi_{12}, G_{12})\in\mcL_{p_{12}}^2} \Gamma_{\pi_{12}, G_{12}}^{2}(x_1,x_2)
	\right)
\end{aligned}
\]
after pulling in the sum over $p\geq 0$.
The $p_{12} = 0$ term together with the term 
$\rho^{(1)}_{\textnormal{Jas}}(x_1)\rho^{(1)}_{\textnormal{Jas}}(x_2) = \rho^{(1)}(x_1)\rho^{(1)}(x_2)=\rho^2$ gives $\rho^{(2)}$ by Wick's rule. 
The condition of absolute convergence is
\[	
	\sum_{p \geq 2} \frac{1}{p!} \abs{\sum_{(\pi, G)\in\mcL_{p}} \Gamma_{\pi, G}} < \infty,
	\quad
	\sum_{p\geq0} \frac{1}{p!} \abs{\sum_{(\pi, G)\in\mcL_{p}^1} \Gamma_{\pi, G}^1} < \infty,
	\quad
	\sum_{p\geq 0} \frac{1}{p!} \abs{\sum_{(\pi, G)\in\mcL_{p}^2} \Gamma_{\pi, G}^{2}} < \infty.
\]

\subsubsection{Calculation of the \texorpdfstring{$3$}{3}-particle reduced density}\label{sec.calc.rho3}
The calculation of the $3$-particle reduced density follows along the same arguments as for the $2$-particle reduced density.
We introduce the relevant diagrams and decompose these according to their linked components.
As for the $2$-particle reduced density we distinguish between the cases according to whether the external vertices $\{1,2,3\}$
are in the same or different linked components. 
They are either in $1, 2$ or $3$ different components. 
Thus, schematically
\[
	\rho^{(3)}_{\textnormal{Jas}}
	= f_{12}^2 f_{13}^2 f_{23}^2 
		\left[
			\sum_{\textnormal{all in different} } \Gamma^1 \Gamma^1 \Gamma^1
			+
			\left(\sum_{\textnormal{$2$ in one} } \Gamma^1(x_1) \Gamma^2(x_2,x_3)
									+ 
									\textnormal{permutations}\right)
			+
			\sum_{\textnormal{all in same}} \Gamma^3
		\right].
\]
Any case where one external vertex is in its own linked component, the contribution for 
such a linked component is $\rho^{(1)}_{\textnormal{Jas}} = \rho^{(1)}= \rho$ (assuming absolute convergence).
Thus, 
\[
\begin{aligned}	
	\rho^{(3)}_{\textnormal{Jas}}(x_1,x_2,x_3)
	& = f_{12}^2 f_{13}^2 f_{23}^2 
		\Biggl[
			\rho^3
			+
			\rho
			\sum_{p\geq 0} \frac{1}{p!}\sum_{(\pi, G)\in\mcL_{p}^2} 
			\left(\Gamma_{\pi, G}^{2}(x_1,x_2)
									+\Gamma_{\pi,G}^2(x_1,x_3)
									+\Gamma_{\pi,G}^2(x_2,x_3)\right)
	\\ & \qquad 
			+
			\sum_{p\geq 0} \frac{1}{p!}\sum_{(\pi, G)\in\mcL_{p}^3} 
			\Gamma_{\pi, G}^{3}(x_1,x_2,x_3)
		\Biggr].
\end{aligned}	
\]
All the $p=0$-terms together give $\rho^{(3)}$ by Wick's rule.
The condition for absolute convergence is 
that 
for any $q\leq 3$ we have
$\sum_{p\geq 0} \frac{1}{p!} \abs{\sum_{(\pi, G)\in\mcL_{p}^q} \Gamma_{\pi, G}^{q}} < \infty$.

\subsubsection{Summarising the results}
For the absolute convergence we have
\begin{lemma}\label{lem.linked.sum.abs.conv}
There exists a constant $c > 0$ such that if $s a^3\rho \log (b/a)(\log N)^3 < c$, then 
\[
	\sum_{p\geq 0} \frac{1}{p!} \abs{\sum_{(\pi, G)\in\mcL_{p}^{q}} \Gamma_{\pi, G}^{q}} < \infty
\]
for any $0\leq q\leq 3$. 
\end{lemma}

\begin{remark}\label{rmk.general.f.gamma}
As mentioned in the beginning of the section, the calculation just given is still valid 
if we replace $f$ by some general function $h\geq 0$ 
and replace $\abs{D_N}^2$ by some more general determinant 
$\det [\gamma(x_i-x_j)]_{1\leq i,j \leq N}$,
where $\gamma(x-y)$ is the kernel of some rank $N$ projection
(for instance the one-particle density matrix of a Slater determinant of $N$ particles).
The criterion for absolute convergence reads 
\[
  \sup_{x_1,\ldots,x_n}\prod_{\substack{1\leq i < j \leq n}} h(x_i-x_j) \leq C^{n} \textnormal{ for all } n\in \N,
  \qquad 
  \frac{1}{L^3}\sum_{k\in \frac{2\pi}{L}\Z^3} \abs{\hat \gamma(k)}
  \int_\Lambda \abs{h^2 - 1} \ud x \int_\Lambda \abs{\gamma} \ud y < c
\]
for some constants $c, C>0$, where the first condition is the ``stability condition'' of the tree-graph bound \cites[Proposition 6.1]{Poghosyan.Ueltschi.2009}{Ueltschi.2018}
and $\hat \gamma(k) := \int_\Lambda \gamma(x) e^{-ikx} \ud x$.
\end{remark}

\noindent
We give the proof of \Cref{lem.linked.sum.abs.conv} in \Cref{sec.abs.conv}. 
Thus, we have the following.
\begin{thm}\label{thm.gaudin.expansion}
There exists a constant $c > 0$ such that if $s a^3\rho \log (b/a)(\log N)^3 < c$, then
\begingroup
\allowdisplaybreaks
\begin{equation}\label{eqn.thm.gaudin}
\begin{aligned}
\frac{C_N}{N!} & = \exp \left(\sum_{p=2}^\infty \frac{1}{p!}\sum_{(\pi, G)\in\mcL_p} \Gamma_{\pi, G}\right),
\\
\rho^{(1)}_{\textnormal{Jas}}
& = \rho^{(1)} + \sum_{p = 1}^\infty \frac{1}{p!}\sum_{(\pi, G)\in\mcL_p^1} \Gamma_{\pi, G}^1,
\\
\rho^{(2)}_{\textnormal{Jas}}
& = f_{12}^2 \left[
	\rho^{(2)} 
	+ \sum_{p= 1}^\infty \frac{1}{p!} \sum_{(\pi, G)\in\mcL_p^2} 
		\Gamma_{\pi,G}^{2}
	\right].
\\
	\rho^{(3)}_{\textnormal{Jas}}
	& = f_{12}^2 f_{13}^2 f_{23}^2 
		\Biggl[
			\rho^{(3)}
			+
			\rho
			\sum_{p\geq 1} \frac{1}{p!}\sum_{(\pi, G)\in\mcL_{p}^2} 
			\left(\Gamma_{\pi, G}^{2}(x_1,x_2)
									+\Gamma_{\pi,G}^2(x_1,x_3)
									+\Gamma_{\pi,G}^2(x_2,x_3)\right)
	\\ & \qquad 
			+
			\sum_{p\geq 1} \frac{1}{p!}\sum_{(\pi, G)\in\mcL_{p}^3} 
			\Gamma_{\pi, G}^{3}
		\Biggr].
\end{aligned}
\end{equation}
\endgroup
\end{thm}
\noindent
The first three formulas are the same as those of \cite[Equations (3.19), (4.9) and (8.4)]{Gaudin.Gillespie.ea.1971}. Our main  contribution is to give a 
 criterion for convergence, and hence for validity of the formulas.

 \begin{remark}\label{rmk.why.Fermi.polyhedron}
The factor $s(\log N)^3$ results from the bound in \Cref{lem.lebesgue.constant.fermi.polyhedron}. 
If we had not introduced the Fermi polyhedron, and instead used the Fermi ball, we would instead have a factor $N^{1/3}$
as mentioned in \Cref{sec.polyhedron}. That is, the condition for absolute convergence would be 
$N^{1/3}a^3\rho\log (b/a) < c$ for some constant $c>0$.

In either case, the $N$-dependence prevents us from taking a thermodynamic limit directly, and we instead use a box method 
of gluing together multiple smaller boxes, where we may put some finite number of particles in each box, see \Cref{sec.box.method}.
For the case of using a Slater determinant with momenta in the Fermi ball, 
there is no way to choose the number of particles in each smaller box so that 
both the absolute convergence holds ($N^{1/3}a^3\rho\log (b/a) < c$), 
and the finite-size error made in the kinetic energy ($\sim N^{2/3}\rho^{2/3}$, see \Cref{lem.KE.polyhedron}) is smaller than the claimed 
energy contribution from the interaction ($\sim Na^3\rho^{5/3}$, see \Cref{thm.main}). 
For this reason we need the polyhedron of \Cref{sec.polyhedron}.
\end{remark}
\begin{remark}
The formulas for $\rho^{(2)}_{\textnormal{Jas}}$ and $\rho^{(3)}_{\textnormal{Jas}}$ only hold for periodic boundary conditions,
since in this case $\rho^{(1)}_{\textnormal{Jas}} = \rho^{(1)} = \rho$.
For different boundary conditions, one has to take into account that this equality is not valid. 
In general one has for $\rho^{(2)}_{\textnormal{Jas}}$ that 
\[
\begin{aligned}
\rho^{(2)}_{\textnormal{Jas}}
& = f_{12}^2 
	\left[
	\rho^{(2)} 
	+ \sum_{p= 1}^\infty \frac{1}{p!} \sum_{(\pi, G)\in\mcL_p^2} 
		\Gamma_{\pi, G}^{2}
	+ \left(\rho^{(1)}_{\textnormal{Jas}}(x_1) \rho^{(1)}_{\textnormal{Jas}}(x_2) - \rho^{(1)}(x_1)\rho^{(1)}(x_2)\right)
	\right]
\\
& = f_{12}^{2}
	\left[
	\rho^{(2)}
	+ \sum_{p= 1}^\infty \frac{1}{p!} \sum_{(\pi, G)\in\mcL_p^2} 
		\Gamma_{\pi, G}^{2}
	+ \sum_{\substack{p_1,p_2\geq 0 \\ p_1+p_2 \geq 1}} \frac{1}{p_1!p_2!} 
		\sum_{\substack{(\pi_1,G_1)\in\mcL_{p_1}^1 \\ (\pi_2,G_2)\in\mcL_{p_2}^1}}
		\Gamma_{\pi_1,G_1}^1(x_1) \Gamma_{\pi_2,G_2}^{1}(x_2)
	\right].
\end{aligned}
\]
One of the reasons we work with periodic boundary conditions is that by doing so, we don't have the complication of dealing with the additional term.
\end{remark}
\begin{remark}
By following the same procedure as in the previous sections, 
one can equally well get formulas for the higher order reduced particle densities.
Similarly one can extend the absolute convergence, \Cref{lem.linked.sum.abs.conv}, 
to any $q$, only one may have to change the constant $c>0$ to depend on $q$.
\end{remark}


\subsection{Absolute convergence}\label{sec.abs.conv}
We now prove \Cref{lem.linked.sum.abs.conv}, 
i.e. that the appropriate sums are absolutely convergent.

\begin{proof}[Proof of \Cref{lem.linked.sum.abs.conv}]
We consider the four sums 
$\sum_{p\geq 0} \frac{1}{p!} \abs{\sum_{(\pi, G)\in\mcL_{p}^{q}} \Gamma_{\pi, G}^{q}}$, $q=0,1,2,3$
one by one.

\subsubsection{Absolute convergence of the \texorpdfstring{$\Gamma$}{Gamma}-sum}
Consider first $\frac{1}{p!}\sum_{(\pi, G)\in\mcL_{p}} \Gamma_{\pi, G}$. 
Split the sum according to the number of connected components of $G$, labelled as $(G_1,\ldots,G_k)$ of sizes $n_1, \ldots, n_k$.
We call these \emph{clusters}. (Note that ``connected'' only refers to the graph $G$, and is independent of the permutation $\pi$.)
Name the vertices in $G_1$ as $\{1,\ldots,n_1\}$, in $G_2$ as $\{n_1 + 1,\ldots,n_1 + n_2\}$ and so on.
Then we have (for $p\geq 2$)
\[
\begin{aligned}
\frac{1}{p!}\sum_{(\pi, G)\in\mcL_p} \Gamma_{\pi, G}
	& = \sum_{k=1}^\infty \frac{1}{k!} 
		\sum_{n_1,\ldots,n_k \geq 2} \frac{1}{n_1! \cdots n_k!} 
		\chi_{\left({\sum n_\ell = p}\right)} 
		\sum_{\substack{G_1,\ldots,G_k \\ G_\ell \in \mcC_{n_\ell}}} 
		\sum_{\pi \in \mcS_p}
		(-1)^\pi 
		\chi_{\left((\pi, \cup G_\ell)\in\mcL_p\right)}
	\\
	& \qquad \times
		\idotsint \prod_{\ell=1}^k \prod_{e\in G_\ell} g_e \prod_{j=1}^p \gamma^{(1)}_N(x_j; x_{\pi(j)}) \ud x_1 \ldots\ud x_p,
\end{aligned}
\]
where $\mcC_n$ denotes the set of connected graphs on $n$ (labelled) vertices.
The factorial factors are similar to those of \Cref{eqn.normalization.linked.components}.
Indeed, the factor $1/(k!)$ comes from counting the possible labelling of the clusters, 
and the factors $1/(n_1!),\ldots,1/(n_k!)$ come from counting the number of ways to distribute 
the $p=\sum n_\ell$ vertices into the clusters and using the factor $1/(p!)$ already present.
\begin{figure}[htb]
\centering
\begin{tikzpicture}[line cap=round,line join=round,>=triangle 45,x=1.0cm,y=1.0cm]
\draw [rotate around={0.:(2.,2.)},line width=0.5pt,fill=black,fill opacity=0.1] (2.,2.) ellipse (1.7cm and 1.3cm);
\draw [rotate around={0.:(9.25,2.)},line width=0.5pt,fill=black,fill opacity=0.1] (8,2.) ellipse (2.0cm and 1.5cm);
\draw [rotate around={0.:(5.5,3.)},line width=0.5pt,fill=black,fill opacity=0.1] (4.75,3.5) ellipse (1.2 cm and 0.5cm);
\node (1) at (1,2) {};
\node (2) at (2,3) {};
\node (3) at (3,2) {};
\node (4) at (2,1) {};
\node (5) at (4,3.5) {};
\node (6) at (5.5,3.5) {};
\node (7) at (6.5,2) {};
\node (8) at (7.5,1) {};
\node (9) at (7.5,3) {};
\node (10) at (9,1) {};
\node (11) at (9,3) {};
\foreach \x in {1,...,11} \node at (\x) {$\bullet$};
\draw[dashed] (1) -- (2) -- (3) -- (4);
\draw[dashed] (5) -- (6);
\draw[dashed] (7) -- (9) -- (11);
\draw[dashed] (8) -- (10);
\draw[dashed] (7) -- (11);
\draw[dashed] (7) -- (8) -- (9);
\draw[->] (1) to[bend left] (2);
\draw[->] (2) to (5);
\draw[->] (3) to (1);
\draw[->] (4) to (8);
\draw[->] (5) to[bend left] (6);
\draw[->] (6) to[bend left] (11);
\draw[->] (7) to (3);
\draw[->] (8) to[bend left] (4);
\draw[->] (9) to[bend right] (7);
\draw[->] (10) to[out=120,in=60,loop] ();
\draw[->] (11) to[bend right] (9);
\node[anchor = east] at (0.2,2) {$G_1$};
\node at (4.75, 2.5) {$G_2$};
\node[anchor = west] at (10.3,2) {$G_3$};
\end{tikzpicture}
\caption{A linked diagram $(\pi, G)\in\mcL_{11}$ decomposed into clusters $G_1, G_2, G_3$. 
Dashed lines denote $g$-edges, and arrows $(i,j)$ denote that $\pi(i) = j$.}
\end{figure}
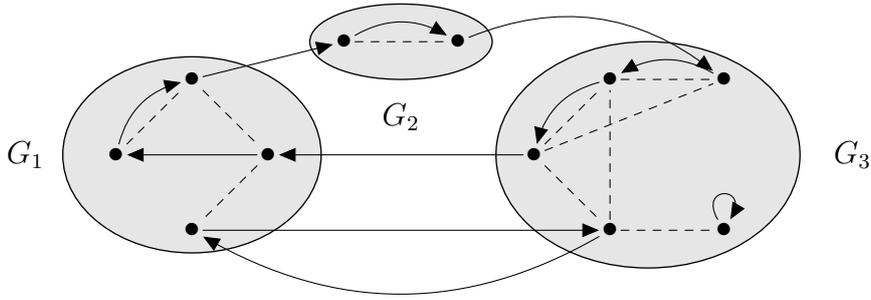

For the analysis we will need the following.
\begin{defn}
Let $A_1,\ldots,A_k$ denote disjoint non-empty sets.
The \emph{truncated} correlation function is 
\begin{equation}
\label{eqn.define.truncated.correlation}
\rho_{\rm t}^{(A_1,\ldots,A_k)}
:= \sum_{\pi \in \mcS_{\cup_\ell A_\ell}} (-1)^\pi \chi_{((\pi, \cup G_\ell) \textnormal{ linked})}
  \prod_{j\in \cup_\ell A_\ell} \gamma_{N}^{(1)}(x_j; x_{\pi(j)}).
\end{equation}
for some choice of connected graphs $G_\ell \in \mcC_{A_\ell}$.
The definition does not depend on the choice of graphs $G_\ell$.

If the underlying sets $A_1,\ldots,A_k$ are clear we will simply denote the truncated correlation by their sizes,
\[
  \rho_{\rm t}^{(|A_1|,\ldots,|A_k|)}
  =
  \rho_{\rm t}^{(A_1,\ldots,A_k)}.
\]
\end{defn}

\noindent
The truncated correlation functions 
are also sometimes referred to as \emph{connected} correlation functions \cite[Appendix D]{Giuliani.Mastropietro.ea.2021}.

\begin{remark}\label{rmk.form.truncated.correlation}
We write the characteristic function in \Cref{eqn.define.truncated.correlation} 
as $\chi_{\left((\pi, \cup G_\ell) \textnormal{ linked}\right)}$ 
for ease of generalizability to the cases in \Cref{sec.abs.conv.Gamma1,sec.abs.conv.Gamma2}
where we will need the notion of truncated correlations also for some of the vertices being external. 
For the truncated correlations it doesn't matter which (if any) vertices are external, only which vertices
are in which clusters.
\end{remark}

\noindent
Since $0\leq f\leq 1$ we have $-1\leq g\leq 0$. Thus, by the tree-graph bound \cites[Proposition 6.1]{Poghosyan.Ueltschi.2009}{Ueltschi.2018} we have
\[
	\abs{\sum_{G\in \mcC_{n}} \prod_{e\in G} g_e} \leq  \sum_{T \in \mcT_n} \prod_{e\in T} |g_e|,
\]
where $\mcT_n$ is the set of all trees on $n$ (labelled) vertices. Thus we get
\begin{equation}\label{eqn.bdd.sum.abs.linked}
\begin{aligned}
& \sum_{p=2}^\infty \frac{1}{p!}\abs{\sum_{(\pi, G)\in\mcL_p} \Gamma_{\pi, G}}
	\\ & \quad \leq 
		\sum_{k=1}^\infty \frac{1}{k!} \sum_{n_1,\ldots,n_k \geq 2} \frac{1}{n_1! \cdots n_k!} 
		\idotsint 
			\sum_{\substack{T_1,\ldots,T_k \\ T_\ell \in \mcT_{n_\ell}}} 
			\prod_{\ell=1}^k \prod_{e\in T_\ell} |g_e |
			\abs{\rho_{\rm t}^{(n_1,\ldots,n_k)}}
		\ud x_1 \ldots\ud x_{\sum n_\ell}.
\end{aligned}
\end{equation}
Here the vertices in $T_\ell$ are the same as in $G_\ell$, i.e.
$T_1$ has vertices $\{1,\ldots,n_1\}$, $T_2$ has vertices $\{n_1+1,\ldots,n_1+n_2\}$ and so on.

In \cite[Equation (D.53)]{Giuliani.Mastropietro.ea.2021} the following formula, known as the Brydges-Battle-Federbush (BBF) formula, is shown 
for the truncated correlation functions
\begin{equation}\label{eqn.truncated.correlation}
	\rho_{\rm t}^{(A_1,\ldots,A_k)}
	= \sum_{\tau\in \mcA^{(A_1,\ldots,A_k)}} \prod_{(i,j) \in \tau} \gamma^{(1)}_N(x_{i}; x_j) \int \ud\mu_\tau(r) \det \mathcal{N}(r),
\end{equation}
where $\mcA^{(A_1,\ldots,A_k)}$ is the set of all anchored trees on $k$ clusters with vertices $A_1,\ldots,A_k$.
(If the sets $A_1,\ldots,A_k$ are clear we will write $\mcA^{(|A_1|,\ldots,|A_k|)} = \mcA^{(A_1,\ldots,A_k)}$ as for the truncated correlations.)
An \emph{anchored tree} is a directed graph on all the $\cup_\ell A_\ell$ vertices, 
such that each vertex has at most one incoming and at most one outgoing edge 
(note that these are all $\gamma^{(1)}_N$-edges, and that the $g$-edges don't matter for this construction) 
and such that upon identifying all vertices in each cluster, the resulting graph is a (directed) tree.
The measure $\mu_\tau$ is a probability measure on 
$\{(r_{\ell\ell'})_{1\leq\ell\leq\ell'\leq k} : 0\leq r_{\ell\ell'}\leq 1\} = [0,1]^{k(k-1)/2}$
and depends only on $\tau$ but not on the factors $\gamma_N^{(1)}(x_i;x_j)$.
Finally, $\mathcal{N}$ is an $I\times J$ (square) matrix  
with entries $\mathcal{N}_{ij} = r_{c(i)c(j)}\gamma^{(1)}_N(x_i;x_j)$,
where $c(i)$ is the (label of the) cluster containing the vertex $\{i\}$
and $r_{\ell\ell'} := r_{\ell'\ell}$ if $\ell > \ell'$.
Here 
\[
	\textstyle
	I = \left\{i\in \bigcup_{\ell=1}^k A_\ell : \nexists j : (i,j)\in \tau\right\},
	\qquad 
	J = \left\{j\in \bigcup_{\ell = 1}^k A_\ell:  \nexists i : (i,j)\in \tau\right\},
\]
are the set of $i$'s (respectively $j$'s) not appearing as $i$'s (respectively $j$'s) 
in the anchored tree $\tau$.
\begin{figure}[htb]
\centering
\begin{tikzpicture}[line cap=round,line join=round,>=triangle 45,x=1.0cm,y=1.0cm]
\draw [line width=0.5pt,fill=black,fill opacity=0.1] (1.5,1.5) ellipse (0.9cm and 0.9cm);
\node (1) at (1,1) {};
\node (2) at (1,2) {};
\node (3) at (2,1) {};
\node (4) at (2,2) {};
\draw[dashed] (1) -- (2) -- (3) -- (4);
\node[anchor = east] at (0.7,1.5) {$T_1$};
\draw [line width=0.5pt,fill=black,fill opacity=0.1] (2,3.5) ellipse (0.2 cm and 0.9 cm);
\node (5) at (2,3) {};
\node (6) at (2,4) {};
\draw[dashed] (5) -- (6);
\node[anchor=east] at (1.8,3.5) {$T_2$};
\draw [line width=0.5pt,fill=black,fill opacity=0.1] (4,1.3) ellipse (1.3 cm and 1 cm);
\node (7) at (3.5,2) {};
\node (8) at (3,1) {};
\node (9) at (4.5,2) {};
\node (10) at (4,1) {};
\node (11) at (5,1) {};
\draw[dashed] (7) -- (9);
\draw[dashed] (8) -- (10);
\draw[dashed] (7) -- (11);
\draw[dashed] (7) -- (8);
\node[anchor = west] at (5,2) {$T_3$};
\draw [line width=0.5pt,fill=black,fill opacity=0.1] (4,3.6) ellipse (0.9cm and 0.8cm);
\node (12) at (3.5,4) {};
\node (13) at (4.5,4) {};
\node (14) at (4,3) {};
\draw[dashed] (12) -- (14) -- (13);
\node[anchor = west] at (4.8,4) {$T_4$};
\draw [line width=0.5pt,fill=black,fill opacity=0.1] (7,1.4) ellipse (0.9cm and 0.8cm);
\node (15) at (6.5,1) {};
\node (16) at (7.5,1) {};
\node (17) at (7,2) {};
\draw[dashed] (15) -- (17) -- (16);
\node[anchor = west] at (7.9,1.5) {$T_5$};
\draw [line width=0.5pt,fill=black,fill opacity=0.1] (7,3.5) ellipse (1.5cm and 1cm);
\node (18) at (6,3) {};
\node (19) at (6,4) {};
\node (20) at (7,3) {};
\node (21) at (7,4) {};
\node (22) at (8,3) {};
\node (23) at (8,4) {};
\draw[dashed] (18) -- (19);
\draw[dashed] (18) -- (20) -- (21);
\draw[dashed] (20) -- (22) -- (23);
\node[anchor=west] at (8.4,3) {$T_6$};
\draw[->] (2) to (5);
\draw[->] (5) to (7);
\draw[->] (6) to (12);
\draw[->] (11) to (15);
\draw[->] (18) to (17);
\foreach \i in {1,...,23} \draw[fill] (\i) circle [radius=1.5pt];
\end{tikzpicture}
\caption{An anchored tree $\tau$ (arrows) and trees $T_1,\ldots,T_6$ (dashed lines).}
\end{figure}
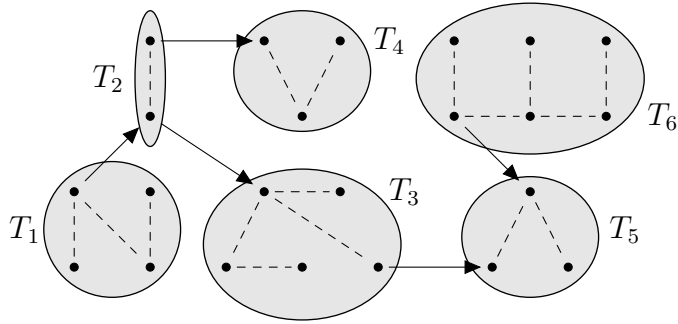

From \cite[Equation (D.57)]{Giuliani.Mastropietro.ea.2021} it follows that 
$\abs{\det \mathcal{N}} \leq \rho^{\sum n_\ell - (k-1)}$.
To see this, one has to adapt the argument in \cite[Lemma D.2]{Giuliani.Mastropietro.ea.2021} slightly. 
We sketch the argument here.
\begin{lemma}[{\cite[Lemmas D.2 and D.6]{Giuliani.Mastropietro.ea.2021}}]\label{lem.bound.det.N}
The matrix $\mcN(r)$ satisfies $\abs{\det \mathcal{N}(r)} \leq \rho^{\sum n_\ell - (k-1)}$
for all $r\in [0,1]^{k(k-1)/2}$.
\end{lemma}
\begin{proof}
First we bound $\rho^{(p)} = \det[\gamma_{N}^{(1)}(x_i; x_j)]_{1\leq i,j\leq p}$ following the strategy of \cite[Lemma D.2]{Giuliani.Mastropietro.ea.2021}.
This is done by writing (as in \cite[Equations (D.8), (D.9)]{Giuliani.Mastropietro.ea.2021})
\[	
	\gamma^{(1)}_N(x_i ; x_j) = \ip{\alpha_i}{\beta_j}_{\ell^2((2\pi/L)\Z^3)},
\] 
where for $k\in \frac{2\pi}{L}\Z^3$
\[
	\alpha_i(k) 
	= L^{-3/2} e^{-ikx_i} \chi_{(k\in P_F)}
	\qquad
	\beta_j(k) 
	= L^{-3/2} e^{-ikx_j} \chi_{(k\in P_F)} 
  = \alpha_j(k).
\]
By the the Gram-Hadamard inequality \cite[Lemma D.1]{Giuliani.Mastropietro.ea.2021} we have 
\[
	\abs{\rho^{(p)}} 
	= \abs{\det[\gamma_{N}^{(1)}(x_i; x_j)]_{1\leq i,j\leq p}} 
	\leq \prod_{i=1}^p \norm{\alpha_i}_{\ell^2((2\pi/L)\Z^3)}\norm{\beta_i}_{\ell^2((2\pi/L)\Z^3)} 
	= \rho^p.
\]
By modifying this argument exactly as described in the proof of \cite[Lemma D.6]{Giuliani.Mastropietro.ea.2021} 
and noting that $r_{\ell\ell'}\leq 1$ one concludes the desired.
\end{proof}

\begin{remark}
We denote the functions as $\alpha_j$ and $\beta_j$ (even though they denote the same function) 
for ease of modifying the argument later in order to prove \Cref{eqn.bound.det.N.tilde}.
\end{remark}

\noindent
In particular one  concludes the bound 
\begin{equation}\label{eqn.bound.truncated.correlation}
	\abs{\rho_{\rm t}^{(n_1,\ldots,n_k)}}
	\leq 
	\rho^{\sum n_\ell - (k - 1)}
	\sum_{\tau\in \mcA^{(n_1,\ldots,n_k)}} \prod_{(i,j) \in \tau} \abs{\gamma^{(1)}_N(x_{i}; x_j)}.
\end{equation}

\noindent
Plugging this into \Cref{eqn.bdd.sum.abs.linked} above we get 
\begin{equation*}	
\begin{aligned}
\sum_{p=2}^\infty \frac{1}{p!}\abs{\sum_{(\pi, G)\in\mcL_p} \Gamma_{\pi, G}}
		& 
		\leq 
		\sum_{k=1}^\infty \frac{1}{k!} \sum_{n_1,\ldots,n_k \geq 2} \frac{1}{n_1! \cdots n_k!} 
		\rho^{\sum n_\ell - (k-1)}
			\sum_{\substack{T_1,\ldots,T_k \\ T_\ell \in \mcT_{n_\ell}}} 
			\sum_{\tau \in \mcA^{(n_1,\ldots,n_k)}} 
	\\ & \qquad 
		\times 
		\idotsint 
		\ud x_1 \ldots\ud x_{\sum n_\ell}
			\prod_{\ell=1}^k \prod_{e\in T_\ell} |g_e |
			\prod_{(i,j) \in \tau} \abs{\gamma^{(1)}_N(x_{i}; x_j)}.
\end{aligned}
\end{equation*}	
To compute these integrals we note that by \Cref{prop.bound.f0.hc}
\begin{multline*}
	\int |g(x)| \ud x = \int \left(1 - f(x)^2\right) \ud x 
		\\
    \leq \frac{4\pi}{(1 - a^3/b^3)^2} \int_{a}^b \left(\left(1 - \frac{a^3}{b^3}\right)^2 - \left(1 - \frac{a^3}{r^3}\right)^2\right) r^2 \ud r 
		\leq C a^3 \log \frac{b}{a}.
\end{multline*}
That is, each factor of $g_e$ gives a contribution $Ca^3\log (b/a)$ after integration.
The $\gamma^{(1)}_N$-factors we can bound by \Cref{lem.lebesgue.constant.fermi.polyhedron} as 
\[
	\int \abs{\gamma^{(1)}_N(x;y)} \ud y \leq C s (\log N)^3.
\]
This takes care of all but one integration, which gives the volume factor $L^3$. 
We shall compute the integrations in the following order: 
\begin{enumerate}[(1.)]
\item Pick any leaf $\{j_0\}$ of the anchored tree $\tau$ lying in some cluster $\ell$,
meaning that there is exactly one edge of $\tau$ incident in $\ell$.

\item Consider $\{j_0\}$ as the root of $T_\ell$ and pick any leaf $\{j\}$ of $T_\ell$ and integrate over $x_j$. 
Since $\{j\}$ is a leaf of $T_\ell$ and $\{j_0\}$ is a leaf of $\tau$ we have that the only 
place $x_j$ appears in the integrand is in some factor $g_{ij}$ for $\{i\}$ the unique vertex connected to $\{j\}$
by a $g$-edge. 
Hence the $x_j$-integral contributes $\int |g|$ by the translation invariance. 

Remove $\{j\}$ and its incident edge from $T_\ell$.

Repeat for all vertices in the cluster until only $\{j_0\}$ remain. (At this point the entirety of $T_\ell$ has been removed.)

\item Integrate over $x_{j_0}$. Since $\{j_0\}$ is a leaf of $\tau$ the only 
place $x_{j_0}$ appears (in the remaining integrand) is in the $\gamma^{(1)}_N$-factor from $\tau$. 
Thus, the $x_{j_0}$-integral gives a contribution $\int |\gamma^{(1)}_N|$ by the translation invariance. 

Remove $\{j_0\}$ and its incident edge from $\tau$.

\item Repeat steps (1.)-(3.) until all integrals have been computed. The final integral gives the volume factor $L^3$.
\end{enumerate}
Steps (1.)-(3.) compute all integrations in one cluster. Repeating this process we integrate over the clusters one by one 
and thus compute all the integrals.
Note that each integration is always over a coordinate associated to a leaf of the relevant graphs. 
This is a key point, since then by translation invariance each integration contributes exactly $\int |g|$ 
or $\int |\gamma_N^{(1)}|$ whichever is appropriate.
In total we thus have the bound 
\[
\begin{aligned}
	\idotsint 
		\ud x_1 \ldots\ud x_{\sum n_\ell}
			\prod_{\ell=1}^k \prod_{e\in T_\ell} |g_e |
			\prod_{(i,j) \in \tau} \abs{\gamma^{(1)}_N(x_{i}; x_j)}
	\leq 
	\left(Ca^3\log(b/a)\right)^{\sum n_\ell - k} \left(Cs(\log N)^3\right)^{k-1} L^3.
\end{aligned}
\]
This bound is for each summand $\tau, T_1,\ldots,T_k$. 
By Cayley's formula $\# \mcT_{n} = n^{n-2} \leq C^n n!$, 
and by \cite[Appendix D.5]{Giuliani.Mastropietro.ea.2021}
$\#\mcA^{(n_1,\ldots,n_k)} \leq k! 4^{\sum n_\ell}$.
Thus, we get 
\[
\begin{aligned}
	\sum_{p=2}^\infty \frac{1}{p!}\abs{\sum_{(\pi, G)\in\mcL_p} \Gamma_{\pi, G}}
	& \leq CN \sum_{k=1}^\infty 
		\left[Cs(\log N)^3\right]^{k-1} 
		\left[\sum_{n= 2}^\infty (Ca^3\rho\log(b/a))^{n-1}\right]^{k}
	\\ & 
	\leq 
		C N a^3 \rho\log(b/a)\sum_{k=1}^\infty \left[C  sa^3 \rho\log(b/a)(\log N)^3\right]^{k-1} 
	\\ & 
	\leq CN a^3\rho\log(b/a) < \infty,
\end{aligned}
\]
for $sa^3\rho\log(b/a)(\log N)^3 $ sufficiently small.
This shows that $\sum_p \frac{1}{p!} \sum_{(\pi, G)\in\mcL_p} \Gamma_{\pi, G}$ is absolutely convergent
under this condition.

\subsubsection{Absolute convergence of the \texorpdfstring{$\Gamma^1$}{Gamma1}-sum}\label{sec.abs.conv.Gamma1}
Consider now $\frac{1}{p!} \sum_{(\pi, G)\in\mcL_p^1} \Gamma_{\pi, G}^1$. 
The argument is almost identical to the argument above. 
We again split the sum according to the connected components of $G$.
Call these $G_*, G_1,\ldots,G_k$, where $G_*$ is the distinguished connected component (cluster) containing the distinguished vertex $\{1\}$. 
Exactly as for $\frac{1}{p!} \sum_{(\pi, G)\in\mcL_p} \Gamma_{\pi, G}$ we have that 
(for $p = 0$ one has to interpret the empty product of integrals as $1$,
so $\sum_{(\pi,G)\in\mcL_0^1}\Gamma^1_{\pi,G} = \rho$)
\[
\begin{aligned}
& 
\frac{1}{p!} \sum_{(\pi, G)\in\mcL_p^1} \Gamma_{\pi, G}^1(x_1)
	\\ & \quad
		= \sum_{k=0}^\infty \frac{1}{k!} 
		\sum_{n_*\geq 0} 
		\sum_{n_1,\ldots,n_k \geq 2} \frac{1}{n_*!n_1! \cdots n_k!} 
		\chi_{\left({ \sum_{\ell\{*,1,\ldots,k\}} n_\ell = p}\right)}
		\sum_{\substack{G_1 \in \mcC_{n_1}, \ldots, G_k \in \mcC_{n_k}
    \\ G_* \in \mcC_{n_* + 1}}} 
		\sum_{\pi \in \mcS_{p+1}}
		(-1)^\pi 
	\\
	& \qquad \times
		\chi_{\left((\pi, \cup_{\ell\in\{*,1,\ldots,k\}} G_\ell)\in\mcL_p^1\right)}
		\idotsint \prod_{\ell\in\{*,1,\ldots,k\}} \prod_{e\in G_\ell} g_e \prod_{j=1}^{p+1} \gamma^{(1)}_N(x_j; x_{\pi(j)}) \ud x_2 \ldots\ud x_{p+1}.
\end{aligned}	
\]
Here for the $k=0$ term one should think of the $n_1,\ldots,n_k$-sums as being an empty product before it is an empty sum, i.e. it should give a factor $1$. 
That is, the $k=0$ term reads 
\[
		\sum_{G_* \in \mcC_{p + 1}} 
		\sum_{\pi \in \mcS_{p+1}}
		(-1)^\pi 
		\idotsint \prod_{e\in G_*} g_e \prod_{j=1}^{p+1} \gamma^{(1)}_N(x_j; x_{\pi(j)}) \ud x_2 \ldots\ud x_{p+1},
\]
since $(\pi, G_*)$ is trivially linked, since $G_*$ is connected.
From here on, we won't write out the $k=0$ term separately to make the formulas more concise.
As before we use the tree-graph inequality and the truncated correlation function (see \Cref{rmk.form.truncated.correlation})
to get 
\[
\begin{aligned}
\sum_{p\geq 0} \frac{1}{p!} \abs{\sum_{(\pi, G)\in\mcL_p^1} \Gamma_{\pi, G}^1}
	& \leq \sum_{k=0}^\infty \frac{1}{k!} 
		\sum_{n_*\geq 0} 
		\sum_{n_1,\ldots,n_k \geq 2} \frac{1}{n_*!n_1! \cdots n_k!} 
		\sum_{\substack{T_1\in\mcT_{n_1},\ldots,T_k\in\mcT_{n_k}
    \\ T_* \in \mcT_{n_*+1}}}
	\\
	& \qquad \times
		\idotsint \prod_{\ell\in\{*,1,\ldots,k\}} \prod_{e\in T_\ell} |g_e| 
		\abs{\rho_{\rm t}^{(n_*+1,n_1,\ldots,n_k)}} \ud x_2 \ldots\ud x_{\sum_{\ell\in\{*,1,\ldots,k\}} n_\ell +1}.
\end{aligned}
\]
To bound this we use the same bound, \Cref{eqn.bound.truncated.correlation}, on the truncated correlations as before. 
It reads 
\[
	\abs{\rho_{\rm t}^{(n_*+1,n_1,\ldots,n_k)}}
	\leq 
	\rho^{\sum_{\ell \in \{*,1,\ldots,k\}} n_\ell + 1 - (k+1 - 1)}
	\sum_{\tau\in \mcA^{(n_*+1, n_1,\ldots,n_k)}} \prod_{(i,j) \in \tau} \abs{\gamma^{(1)}_N(x_{i}; x_j)}.
\]
Computing the integrals is  as before, 
with a few differences. 
During each repeat (apart from the last one) of step (1.) we pick not just any leaf $j_0$ but a leaf $j_0$ not in the cluster containing $\{1\}$.
(Since any tree has at least $2$ leaves, this is always possible.)
For each of these repeats, the argument is the same as before.
For the last repeat of step (1.) where only the cluster containing $\{1\}$ remains 
we follow step (2.) with the slight change, that the root is chosen to be $\{1\}$.
(There are no $\gamma^{(1)}_N$-factors left, so we are free to choose any vertex as the root.)
There is then no step (3.) since we do not integrate over $x_1$.

This has the following effect. 
First, the last variable $x_1$ is not integrated over, so there is no volume factor $L^3$.
And second, there are $k$ integrals $\int |\gamma_N^{(1)}|$ instead of $k-1$ 
(since there are $k+1$ many clusters including the distinguished one). 
For the bounds of the sum of all terms we use that 
$\# \mcA^{(n_*+1,n_1,\ldots,n_k)}\leq (k+1)! 4^{\sum_{\ell \in\{*,1,\ldots,k\}} n_\ell + 1}$.
Thus, uniformly in $x_1$
\[
\begin{aligned}
& \sum_{p\geq 0} \frac{1}{p!} \abs{\sum_{(\pi, G)\in\mcL_p^1} \Gamma_{\pi, G}^1}
	\\ & \quad 
	\leq C \rho 
		\left(\sum_{n_*=0}^\infty \left[C a^3\rho\log(b/a)\right]^{n_*}\right)
		\left(
			\sum_{k=0}^\infty 
			(k+1)
			\left[Cs(\log N)^3\right]^{k} 
			\left[\sum_{n= 2}^\infty (Ca^3\rho\log(b/a))^{n-1}\right]^{k}
		\right)
	\\
	& \quad \leq C\rho < \infty,
\end{aligned}	
\]
for $sa^3\rho\log(b/a)(\log N)^3$ sufficiently small.
This shows that $\sum_p \frac{1}{p!} \sum_{(\pi, G)\in\mcL_p^1} \Gamma_{\pi, G}^1$ is absolutely convergent
under this condition.


\subsubsection{Absolute convergence of the \texorpdfstring{$\Gamma^{2}$}{Gamma2}-sum}\label{sec.abs.conv.Gamma2}
For the third sum the argument is mostly analogous. There are a few changes needed for the argument.
First, one has to distinguish between the two cases of whether or not 
the two distinguished vertices $\{1,2\}$ are in the same connected component (cluster) of the graph or not. 
One computes
\begin{equation}\label{eqn.compute.Gamma12}
\frac{1}{p!} \sum_{(\pi, G)\in\mcL_p^2} \Gamma_{\pi, G}^{2}
= \Sigma_{\textnormal{different}} + \Sigma_{\textnormal{same}},
\end{equation}
where 
\begin{equation}
\label{eqn.compute.Gamma2.details}
\begin{aligned}
\Sigma_{\textnormal{different}} 
& = \sum_{k=0}^\infty \frac{1}{k!} 
		\sum_{n_*, n_{**}\geq 0} 
		\sum_{n_1,\ldots,n_k \geq 2} \frac{1}{\prod_\ell n_\ell!} 
		\chi_{\left({ \sum n_\ell} = p\right) }
		\sum_{\substack{
    G_1\in\mcC_{n_1},\ldots,G_k\in \mcC_{n_k}
    \\ G_* \in \mcC_{n_*+\{1\}} \\ G_{**}\in \mcC_{n_{**}+\{2\}}}}
	\\ & \quad \times 
	\idotsint \prod_{\ell}\prod_{e\in G_\ell} g_e \rho^{(n_*+\{1\}, n_{**}+\{2\}, n_1,\ldots,n_k)}_{\rm t} \ud x_3 \ldots\ud x_{p+2},
\\
\Sigma_{\textnormal{same}}
	& = \sum_{k=0}^\infty \frac{1}{k!} 
		\sum_{n_*\geq 1} 
		\sum_{n_1,\ldots,n_k \geq 2} \frac{1}{\prod_\ell n_\ell!} 
		\chi_{\left({ \sum n_\ell} = p\right) }
		\sum_{\substack{
    G_1\in \mcC_{n_1},\ldots,G_k\in \mcC_{n_k}
    \\ G_* \in \mcC_{n_* +\{1,2\}} \\ (1,2)\notin G_*}} 
	\\ & \quad \times 
	\idotsint \prod_{\ell}\prod_{e\in G_\ell} g_e \rho^{(n_*+\{1,2\},n_1,\ldots,n_k)}_{\rm t} \ud x_3\ldots\ud x_{p+2}.
\end{aligned}
\end{equation}
Here $\sum_\ell$ and $\prod_\ell$ are over $\ell \in \{*, **, 1,\ldots,k\}$ or $\ell \in \{*, 1,\ldots,k\}$, 
whichever is appropriate.
With a slight abuse of notation we write $n_* + \{1\}$ for the set of vertices in the cluster containing the external vertex $\{1\}$
(similarly for $n_{**} + \{2\}, n_* + \{1,2\}$).
This set has exactly $n_*$ internal vertices.
For $p=0$ one has to interpret the empty product of integrals as a factor $1$.

The first part is the contribution where $\{1\}$ and $\{2\}$ are in distinct clusters (labelled $*$ and $**$),
the second part is the contribution from where they are in the same (labelled $*$). 
Note that in the second contribution we have $n_* \geq 1$. Indeed, $\{1\}$ and $\{2\}$ are connected, but not 
by an edge. Hence they must be connected by a path of length $\geq 2$, which necessarily goes through at least 
one vertex $\{j\}, j\ne 1,2$.

We treat the two cases separately.
In the case where the two distinguished vertices are in different clusters 
we may readily apply both the tree-graph bound and the bound on the truncated correlation \Cref{eqn.bound.truncated.correlation}.
The latter reads
\[
	\abs{\rho_{\rm t}^{(n_*+\{1\},n_{**}+\{2\},n_1,\ldots,n_k)}} 
	\leq 
	\rho^{(\sum_{\ell \in \{*,**,1,\ldots,k\}} n_\ell + 2) - (k+2 - 1)}
	\sum_{\tau \in \mcA^{(n_*+\{1\}, n_{**}+\{2\}, n_1, \ldots, n_k)}} 
	\prod_{(i,j)\in \tau} \abs{\gamma^{(1)}_N(x_i; x_j)}.
\]
The integration procedure is slightly modified compared to that of \Cref{sec.abs.conv.Gamma1}. 
In the anchored tree there is a path between (the cluster containing) $\{1\}$
and (the cluster containing) $\{2\}$. 
For the edge incident to (the cluster containing) $\{1\}$ on this path, we bound $|\gamma^{(1)}_N|\leq \rho$.
This cuts the anchored tree into two anchored trees $\tau_1, \tau_2$ such that (with a slight abuse of notation) 
$1\in \tau_1$ and $2\in \tau_2$. 
We may  follow the integration procedure exactly as for the $\Gamma^1$-sum for each of the anchored trees $\tau_1$ and $\tau_2$.
Recall the bound 
\[
	\#\mcA^{(n_*+\{1\},n_{**}+\{2\},n_1,\ldots,n_k)} \leq (k+2)! 4^{\sum_{\ell \in \{*,**,1,\ldots,k\}} n_\ell + 2} \leq C (k^2+1) k! 4^{\sum_{\ell \in \{*,**,1,\ldots,k\}} n_\ell}.
\]
We thus get for the contribution of all terms where the two distinguished vertices are in different clusters 
(assuming that $sa^3\rho \log(b/a) (\log N)^3$ is sufficiently small)
\begin{equation}\label{eqn.sum.different.Gamma2}
\begin{aligned}
& |\Sigma_{\textnormal{different}}| 
	\\ & \quad  
		\leq C \rho^2 
		\left(\sum_{n_*=0}^\infty \left[C a^3\rho\log(b/a)\right]^{n_*}\right)^2
		\left(
			\sum_{k=0}^\infty 
			(k^2+1)
			\left[Cs(\log N)^3\right]^{k} 
			\left[\sum_{n= 2}^\infty (Ca^3\rho\log(b/a))^{n-1}\right]^{k}
		\right)
	\\ & \quad 
	\leq C\rho^2 < \infty.
\end{aligned}
\end{equation}

\noindent
Now we consider the case where $\{1\}$ and $\{2\}$ are in the same distinguished cluster.
Here we may readily apply the bound in  \Cref{eqn.bound.truncated.correlation} on the truncated correlation
but we need to be a bit careful in applying the tree-graph bound.
Indeed, then the sum over graphs in the cluster containing the two vertices is not $\sum_{G_* \in \mcC_{n_*+2}}$,
but instead $\sum_{G_* \in \mcC_{n_*+2}, (1,2)\notin G_{*}}$, since in the construction, no $g$-edges are allowed between $\{1\}$ and $\{2\}$. 
To still apply the tree-graph bound, we define 
\[
	\tilde g_e := \begin{cases}
	g_e & e\ne (1,2) \\ 0 & e = (1,2).
	\end{cases}
\]
Then $-1\leq \tilde g_e\leq 0$ so we can apply the tree-graph bound with these edge-weights to get 
\[
	\abs{\sum_{G_* \in \mcC_{n_*+2}, (1,2)\notin G_{*}} \prod_{e\in G_*} g_e}
	= \abs{\sum_{G_* \in \mcC_{n_*+2}} \prod_{e\in G_*} \tilde g_e} 
	\leq \sum_{T_*\in\mcT_{n_*+2}}\prod_{e\in T_*} |\tilde g_e|
	= \sum_{T_*\in\mcT_{n_*+2}, (1,2)\notin T_*}\prod_{e\in T_*} |g_e|.
\]
We again have to modify the integrations slightly. 
The integrations over all clusters apart from the distinguished one may be computed as for the $\Gamma$- and $\Gamma^1$-sums.
For the distinguished cluster with $\{1\}$ and $\{2\}$ there is some path of $g$-edges connecting them. 
Pick the unique edge on this path incident with $\{1\}$ and bound $|g|\leq 1$ for this factor. 
This splits the tree $T_*$ into two trees $T_*^1$ and $T_*^2$ with $1\in T_*^1$ and $2\in T_*^2$. 
We may  compute the integrations over all the variables with index in the distinguished 
cluster exactly as for the $\Gamma^1$-sum for each tree $T_*^1$ and $T_*^2$ separately. 
One  gets for the contribution 
(assuming that $sa^3\rho \log(b/a) (\log N)^3$ is sufficiently small)
\begin{equation}\label{eqn.sum.same.Gamma2}
\begin{aligned}
& |\Sigma_{\textnormal{same}}| 
	\\ & \quad  
		\leq C \rho^2 
		\left(\sum_{n_*=1}^\infty \left[C a^3\rho\log(b/a)\right]^{n_*}\right)
		\left(
			\sum_{k=0}^\infty 
			(k+1)
			\left[Cs(\log N)^3\right]^{k} 
			\left[\sum_{n= 2}^\infty (Ca^3\rho\log(b/a))^{n-1}\right]^{k}
		\right)
	\\ & \quad 
	\leq Ca^3\rho^3\log (b/a)  < \infty.
\end{aligned}
\end{equation}



\noindent
We conclude that 
\[
	\sum_{p\geq 0} \frac{1}{p!} \abs{\sum_{(\pi, G)\in\mcL_p^2} \Gamma_{\pi, G}^{2}} \leq C\rho^2 < \infty,
\]
uniformly in $x_1,x_2$
for sufficiently small $sa^3\rho\log(b/a)(\log N)^3$.


\subsubsection{Absolute convergence of the \texorpdfstring{$\Gamma^3$}{Gamma3}-sum}
The argument for the last sum is completely analogous to the argument for the $\Gamma^2$-sum.
We have to distinguish between different cases of the clusters containing the external vertices $\{1,2,3\}$.
Either there is one cluster containing all of them, one cluster containing two of them and one cluster containing the last vertex,
or they are all in distinct clusters. 
One then deals with the different cases exactly as we did for the $\Gamma^2$-sum. We skip the details.
This concludes the proof of \Cref{lem.linked.sum.abs.conv}.
\end{proof}

\section{Energy of the trial state}\label{sec.energy.in.box}
In this section we bound the energy of the trial state $\psi_N$ defined in \Cref{eqn.define.trial.state}. 
Recall \Cref{eqn.compute.energy.first}.  
By \Cref{thm.gaudin.expansion} we have (for $s a^3\rho \log (b/a)(\log N)^3$ sufficiently small)
\[
\begin{aligned}
	\rho^{(2)}_{\textnormal{Jas}}(x_1, x_2) 
	& 
		= f(x_1-x_2)^2 
		\left(\rho^{(2)}(x_1, x_2) 
		+ \sum_{p= 1}^\infty \frac{1}{p!} \sum_{(\pi, G)\in\mcL_p^2} 
		\Gamma_{\pi, G}^2
		\right).
\end{aligned}
\]
We can expand $\rho^{(2)}$ in $x_1 - x_2$ using \Cref{prop.2.density}. 
The second term is an error term we have to control.
Additionally, also the three-body term is an error we have to control.
We claim that
\begin{lemma}\label{lem.derivative.sum.linked.diagrams}
There exist constants $c,C > 0$ such that if $s a^3\rho \log (b/a) (\log N)^3 < c$ and
$N=\# P_F > C$, then 
\[
\begin{aligned}
	& \abs{\sum_{p= 1}^\infty \frac{1}{p!} \sum_{(\pi, G)\in\mcL_p^2} \Gamma_{\pi, G}^{2}}
	\\ & \quad 
		\leq 
	Ca^6 \rho^4 (\log (b/a))^2 
		\Bigl[ 
			s^3 a^6 \rho^2 (\log b/a)^2 (\log N)^9 
			+ 1
		\Bigr]
	\\ & \qquad 
	+ Ca^3 \rho^{3+2/3} |x_1-x_2|^2 
		\Bigl[
			s^5 a^{12}\rho^4 (\log(b/a))^5  (\log N)^{16}
      + b^2 \rho^{2/3}
			+ \log(b/a)
		\Bigr].
\end{aligned}
\]
\end{lemma}
\begin{lemma}\label{lem.bound.3.body.term}
There exists a constant $c > 0$ such that if $sa^3\rho \log (b/a)(\log N)^3  < c$, then 
\[
	\rho^{(3)}_{\textnormal{Jas}} 
	= f_{12}^2f_{13}^2f_{23}^2\left[
		\rho^{(3)} 
		+ O\left(a^3\rho^{4} \log (b/a)
				\left[
				s^3  a^6 \rho^2 (\log (b/a))^2(\log N)^9 
				+ 1 \right]\right)
		\right]
\]
where the error is uniform in $x_1,x_2,x_3$.
\end{lemma}

\noindent
We give the proof of these lemmas in \Cref{sec.proof.subleading.2.body.diagrams,sec.proof.subleading.3.body.diagrams} below.
For the three-body term, we additionally have the bound
$\rho^{(3)}\leq C\rho^{3+4/3} |x_1-x_2|^2|x_2-x_3|^2$
by \Cref{lem.bound.rho3}.
Combining now \Cref{lem.KE.polyhedron,lem.derivative.sum.linked.diagrams,lem.bound.3.body.term,prop.2.density},
\Cref{thm.gaudin.expansion} and \Cref{eqn.compute.energy.first} 
we thus get (for $s a^3\rho \log (b/a)(\log N)^3$ sufficiently small and $N$ sufficiently large)
\begin{equation}\label{eqn.energy.w.all.errors}
\begin{aligned}
\longip{\psi_N}{H_N}{\psi_N}
	& = \frac{3}{5}(6\pi)^{2/3} \rho^{2/3} N\left(1 + O(N^{-1/3}) + O(s^{-2})\right) + L^3\int \ud x
		\left(
			|\nabla f(x)|^2 + \frac{1}{2}v(x) f(x)^2
		\right) 
	\\ & \qquad \times 
		\Bigg[
		 \frac{(6\pi^2)^{2/3}}{5} \rho^{8/3} |x|^2
	\biggl(
	1 - \frac{3(6\pi^2)^{2/3}}{35}\rho^{2/3}|x|^2 
	\\ & \qquad \qquad 
	+ O(N^{-1/3}) 
	+ O(s^{-2}) 
	+ O(N^{-1/3}\rho^{2/3}|x|^2)
	+ O(\rho^{4/3}|x|^4)
	\biggr)
	\\ & \qquad \quad 
			+ O\left( 
			a^6 \rho^4 (\log (b/a))^2 
		\Bigl[ 
			s^3 a^6 \rho^2 (\log b/a)^2 (\log N)^9 
			+ 1
		\Bigr]\right)
	\\ & \qquad \quad 
	+ O \left(a^3 \rho^{3+2/3} |x|^2 
				\Bigl[
					s^5 a^{12}\rho^4 (\log(b/a))^5  (\log N)^{16}
          + b^2\rho^{2/3}
					+ \log(b/a)
				\Bigr]\right)
		\Bigg]
	\\ & \quad 
		+ \iiint \ud x_1 \ud x_2 \ud x_3 f_{12}\nabla f_{12} f_{23}\nabla f_{23} f_{13}^2 
		\Bigl[
    O(\rho^{3+4/3} |x_1-x_2|^2|x_2-x_3|^2)
	\\ & \qquad
		+ O\left(a^3\rho^{4} \log (b/a)
				\left[
				s^3  a^6 \rho^2 (\log (b/a))^2(\log N)^9 
				+ 1 \right]
			\right)
		\Bigr].
\end{aligned}
\end{equation}
We will choose $N$ (really $L$, see \Cref{rmk.choose.L.not.N}) 
some large negative power of $a^3\rho$, so errors with $N^{-1/3}$ are subleading. 
We may  compute 
\begin{equation}\label{eqn.integral.v.scat.len}
\begin{aligned}
	& \int \ud x \left(|\nabla f(x)|^2 + \frac{1}{2}v(x) f(x)^2 \right) |x|^2 
	\\ & \quad 
		= \frac{1}{(1 - a^3/b^3)^2} \int_{|x|\leq b} \ud x \left(|\nabla f_0(x)|^2 + \frac{1}{2}v(x) f_0(x)^2 \right) |x|^2 
	\\ & \quad 
		\leq 12\pi a^3 \left(1 + O(a^3/b^3)\right),
\end{aligned}
\end{equation}
by \Cref{defn.scattering.length} since $f = \frac{1}{1- a^3/b^3} f_0$ for $|x|\leq b$ and $b>R_0$, the range of $v$.
For the higher moments we recall that $|\nabla f_0|\leq |\nabla f_{\textnormal{hc}}| = \frac{3a^3}{|x|^4}$ for $|x|\geq a$ by \Cref{prop.bound.f0.hc}. 
Then we have (for $n=4,6$)
\begin{equation}\label{eqn.integral.v.moments}
\begin{aligned}
	& \int \ud x \left(|\nabla f(x)|^2 + \frac{1}{2}v(x) f(x)^2 \right) |x|^n
	\\ & \quad 
		= \frac{1}{(1 - a^3/b^3)^2} \int_{|x|\leq b} \ud x \left(|\nabla f_0(x)|^2 + \frac{1}{2}v(x) f_0(x)^2 \right) |x|^n 
	\\ & \quad 
		\leq
		\frac{1}{2} R_0^{n-2} \int v |f_0|^2 |x|^2 \ud x 
		+ \int_{|x|\geq a} \left( \frac{3a^3}{|x|^4}\right)^2 |x|^n \ud x
		+ a^{n-2}\int_{|x|\leq a} |\nabla f_0|^2 |x|^{2} \ud x 
	\\ & \quad 
		\lesssim C R_0^{n-2}a^3.
\end{aligned}
\end{equation}
For $n=4$ we have more precisely
\[
\begin{aligned}
	\int \ud x \left(|\nabla f(x)|^2 + \frac{1}{2}v(x) f(x)^2 \right) |x|^4
	& \leq \int \ud x \left(|\nabla f_0(x)|^2 + \frac{1}{2}v(x) f_0(x)^2 \right) |x|^4
  \left(1 + O(a^3b^{-3})\right)
	\\
	& = 36\pi a^3a_0^2 + O(a^6 a_0^2  b^{-3}).
\end{aligned}
\]
For the lower moment, we have by \Cref{eqn.f.scatt.radial}
\begin{equation}\label{eqn.integral.v.zero.moment}
\begin{aligned}
\int \ud x \left(|\nabla f(x)|^2 + \frac{1}{2}v(x) f(x)^2 \right)
	& = 4\pi \int_{0}^b \left(|\partial_r f|^2 r^2 + r^2 f \partial_r^2 f + 4r f \partial_r f\right) \ud r
	\\ & 
		= \frac{12\pi a^3/b^2}{1 - a^3/b^3} + 8\pi \int_{0}^b r f \partial_r f \ud r
\end{aligned}
\end{equation}
where $\partial_r$ denotes the radial derivative, and we integrated by parts using that $f(r) = \frac{1 - a^3/r^3}{1 - a^3/b^3}$ outside the support of $v$.
By \Cref{prop.bound.f0.hc} we have 
\begin{equation}\label{eqn.integral.fdf}
	2\int_{0}^b r f \partial_r f \ud r = b - \int_0^b f^2 \ud r 
		\leq b - \frac{1}{(1-a^3/b^3)^2} \int_a^b \left(1 - \frac{a^3}{r^3}\right)^2 \ud r 
		\leq Ca.
\end{equation}
Hence 
\[
	\int \ud x \left(|\nabla f(x)|^2 + \frac{1}{2}v(x) f(x)^2 \right) \leq C a.
\]
This concludes the bounds on all the terms in \Cref{eqn.energy.w.all.errors} arising from the $2$-body term.
To bound those arising from the $3$-body term we bound 
$f_{13}\leq 1$. 
By the translation invariance one integration gives a volume factor $L^3$. 
The remaining two integrals then both give the same contribution.
That is,
\[
\begin{aligned}
	& \iiint \ud x_1 \ud x_2 \ud x_3 f_{12}\nabla f_{12} f_{23}\nabla f_{23} f_{13}^2 
		\Bigl[
    O(\rho^{3+4/3} |x_1-x_2|^2|x_2-x_3|^2)
	\\ & \qquad 
		+ O\left(a^3\rho^{4} \log (b/a)
				\left[
				s^3  a^6 \rho^2 (\log (b/a))^2(\log N)^9 
				+ 1 \right]
			\right)
		\Bigr]
	\\ & \quad 
		\leq C N\rho^{2+4/3} \left(\int |x|^2f\partial_r f  \ud x \right)^2
	\\ & \qquad
		+ CNa^3\rho^3 \log(b/a) \left[
				s^3  a^6 \rho^2 (\log (b/a))^2(\log N)^9 
				+ 1 \right]
				\left(\int f\partial_r f \ud x \right)^2.
\end{aligned}
\]
Using integration by parts and \Cref{prop.bound.f0.hc}, we have that 
\begin{multline}\label{eqn.integral.fdf.moments}
	\frac{1}{4\pi}\int |x|^n f\partial_r f \ud x
	= \int_0^b r^{n+2} f \partial_r f \ud r
	= \frac{b^{n+2}}{2} - \frac{n+2}{2}\int_0^b r^{n+1} f^2 \ud r
	\\ 
	\leq \frac{b^{n+2}}{2} - \frac{n+2}{2}\int_a^b r^{n+1} \left(\frac{1 - a^3/r^3}{1 - a^3/b^3}\right)^2 \ud r
	\leq 
	\begin{cases}
	C a^2 & n=0, \\ C a^3 b & n=2.
	\end{cases}
\end{multline}
Plugging all this into \Cref{eqn.energy.w.all.errors} we thus get for the energy density 
\begin{equation}
\label{eqn.energy.density.w.all.errors}
\begin{aligned}
\frac{\longip{\psi_N}{H_N}{\psi_N}}{L^3}
	& = \frac{3}{5}(6\pi)^{2/3} \rho^{5/3} + \frac{12\pi}{5}(6\pi^2)^{2/3} a^3\rho^{8/3} 
		- \frac{108\pi (6\pi^2)^{4/3}}{175}a^3a_0^2 \rho^{10/3} 
	\\ & \quad 
		+ O\left(s^{-2}\rho^{5/3}\right) 
		+ O\left(N^{-1/3}\rho^{5/3}\right) 
	\\ & \quad
		+ O\left(a^6b^{-3}\rho^{8/3}\right)
		+ O\left(a^6a_0^2b^{-3} \rho^{10/3}\right)
		+ O\left(R_0^4 a^3 \rho^4\right)
	\\ & \quad
		+ O\left( 
			a^7 \rho^4 (\log (b/a))^2 
		\Bigl[ 
			s^3 a^6 \rho^2 (\log b/a)^2 (\log N)^9 
			+ 1
		\Bigr]\right)
	\\ & \quad  
	+ O \left(a^6 \rho^{3+2/3} 
				\Bigl[
					s^5 a^{12}\rho^4 (\log(b/a))^5  (\log N)^{16}
          + b^2\rho^{2/3}
					+ \log(b/a)
				\Bigr]\right)
	\\ & \quad 
  + O\left(a^6b^2\rho^{3+4/3}\right)
	+ O\left(a^7\rho^4 \log(b/a)\left[s^3 a^6\rho^2(\log (b/a))^3(\log N)^9 + 1\right]\right).
\end{aligned}
\end{equation}
We  choose $L\sim a(a^3\rho)^{-10}$ 
still ensuring that $\frac{Lk_F}{2\pi}$ is rational. (More precisely one chooses $L\sim a (k_F a)^{-30}$, since $\rho$ is defined in terms of $L$.)
Then $N \sim (a^3\rho)^{-29}$. Choose moreover
\[
	b = a(a^3\rho)^{-\beta}, \qquad s \sim (a^3\rho)^{-\alpha}|\log (a^3\rho )|^{-\delta}.
\]
Note that we need 
\[	
  \frac{2}{9} < \beta < \frac{1}{2}, \qquad 
	\frac{5}{6} < \alpha < \frac{13}{15}
\]
for the error terms to be smaller than the desired accuracy of order $a^3a_0^2 \rho^{10/3}$.
We  get 
\begin{equation*}
\begin{aligned}
\frac{\longip{\psi_N}{H_N}{\psi_N}}{L^3}
	& = \frac{3}{5}(6\pi)^{2/3} \rho^{5/3} + \frac{12\pi}{5}(6\pi^2)^{2/3} a^3\rho^{8/3} 
		- \frac{108\pi (6\pi^2)^{4/3}}{175}a^3a_0^2 \rho^{10/3} 
	\\ & \quad
		+ O(\rho^{5/3} (a^3\rho)^{\gamma_1}|\log(a^3\rho)|^{\gamma_2}).
\end{aligned}
\end{equation*}
where 
\[
	\gamma_1 = \min 
	\left\{
		2\alpha, 1 + 3\beta, \frac{13}{3} - 3\alpha, 6 - 5\alpha, \frac{8}{3} - 2\beta
	\right\},
\]
and $\gamma_2$ is given by the power of the logarithmic factors of the largest error term.
Optimising in  $\alpha,\beta, \delta$ we see that for the choice
\[
	\beta = \frac{1}{3}, \qquad \alpha = \frac{6}{7}, \qquad \delta = 3
\]
we have $\gamma_1 = \frac{12}{7}$ and $\gamma_2 = 6$, i.e.
\begin{equation}
\label{eqn.energy.density.upper.bound.single.state}
\begin{aligned}
\frac{\longip{\psi_N}{H_N}{\psi_N}}{L^3}
	& = \frac{3}{5}(6\pi)^{2/3} \rho^{5/3} 
	\\ & \qquad 
		+ \frac{12\pi}{5}(6\pi^2)^{2/3} a^3\rho^{8/3}
		\biggl[
		1
		- \frac{9}{35}(6\pi^2)^{2/3}a_0^2\rho^{2/3}
		+ O((a^3\rho)^{2/3 + 1/21} |\log(a^3\rho)|^6)
		\biggr] 
\end{aligned}
\end{equation}
for $a^3\rho$ small enough.
Note that for this choice of $s,N$ we have $s\sim N^{6/203} (\log N)^{3}$. 
Thus any $Q$ with  $N^{4/3} \ll Q \leq CN^C$ satisfies the condition $Q^{-1/4}\leq C s^{-1}$
of \Cref{defn.P_F.true}.

\subsection{Thermodynamic limit via a box method}\label{sec.box.method}
In this section we construct a trial state in the thermodynamic limit using a box method 
of gluing trial states for finite $n$ together. 
Such a method has been used for many studies of dilute Bose and Fermi gases, see for instance 
\cite{Yau.Yin.2009,Lieb.Seiringer.ea.2005,Fournais.Girardot.ea.2022,Basti.Cenatiempo.ea.2021}.
First we show that we may choose periodic boundary conditions in the small boxes instead of using Dirichlet boundary conditions.
The setting and argument is due to Robinson \cite[Lemmas 2.1.12 and 2.1.13]{Robinson.1971}. 
We present a slightly modified version  in \cite[Section C]{Mayer.Seiringer.2020}.

\begin{lemma}[{\cite{Robinson.1971,Mayer.Seiringer.2020}}]\label{lem.change.b.c.}
Let $0 < d < L/2$ be a cut-off, let $H^{D}_{N,L+2d} = \sum_{j=1}^N -\Delta_{j, L+2d}^D + \sum_{i < j} v(x_i - x_j)$ denote the 
$N$-particle Hamiltonian with Dirichlet boundary conditions on a box of sides $L+2d$, and let 
$H^{\textnormal{per}}_{N,L} = \sum_{j=1}^N -\Delta_{j, L}^{\textnormal{per}} + \sum_{i < j} v_{\textnormal{per}}(x_i - x_j)$
denote the $N$-particle Hamiltonian with periodic boundary conditions on a box of sides $L$, with the interaction 
$v_{\textnormal{per}}(x) = \sum_{n\in \Z^3} v(x + nL)$, the periodized interaction.

Then, there exists an isometry $V: L^2_a(\Lambda_L^N) \to L^2_a(\Lambda_{L+2d}^N)$ such that 
for all $\psi$ in the form-domain of $H^{\textnormal{per}}_{N,L}$ we have $V\psi$ in the form-domain of $H^{D}_{N,L+2d}$ and 
\[
	\longip{V\psi}{H^D_{N,L+2d}}{V\psi} \leq \longip{\psi}{H^{\textnormal{per}}_{N,L}}{\psi} + \frac{6N}{d^2}\norm{\psi}^{2}
\]
\end{lemma}
\begin{proof}
This is a trivial modification of \cite[Lemma 4]{Mayer.Seiringer.2020}, noting that the explicitly constructed $V$
respects the anti-symmetry.
\end{proof}

\noindent
We now glue together trial states.
For any (sufficiently small) density $\rho$ we have above found that we may construct a (normalized) trial state $\psi_n$ 
on the torus $\Lambda_\ell = [-\ell/2, \ell/2]^3$ 
satisfying \Cref{eqn.energy.density.upper.bound.single.state}
with $\ell \sim a(a^3\rho)^{-10}$, i.e. $n \sim (a^3\rho)^{-29}$.
We now use the isometry $V$ from \Cref{lem.change.b.c.} to find a trial state $V\psi_n$ with Dirichlet boundary conditions on $\Lambda_{\ell + 2d}$.
Our trial state $\Psi_N$ for $N=M^3n$ is then obtained by gluing together $M^3$ copies of $V\psi_n$ arranged in boxes, with a distance $b$ between them, 
so that there is no interaction between the boxes. 
We choose the same $b$ as before.
More precisely, for configurations where the first $n$ particles are in box $1$ and so on,  
\[
	\Psi_N(x_1,\ldots,x_N) = \prod_{j = 1}^{M^3} V\psi_n(x_{n(j-1) + 1} - \tau_j, \ldots, x_{nj} - \tau_j),
\]
where $\tau_j \in \R^{3}$ denotes the centre of box number $j$.
The state $\Psi_N$ is then the antisymmetrization of this. 
Its energy is  
\[
	\longip{\Psi_N}{H_{N,M(\ell +  2d + b)}^D}{\Psi_N} 
	=  M^3 \longip{V\psi_n}{H_{n,\ell+2d}^{D}}{V\psi_N} \leq M^3 \left(\longip{\psi_n}{H_{n,\ell}^{\textnormal{per}}}{\psi_n} + \frac{6n}{d^2}\right).
\]
The particle density of the state $\Psi_N$ is $\tilde\rho = \frac{n}{(\ell + 2d + b)^3} = \rho(1 + O(d/\ell) + O(b/\ell))$.
The energy density is 
\begin{equation}\label{eqn.energy.density.infinite.box.trial.state}
\begin{aligned}
	e(\tilde\rho) 
	& \leq 
		\liminf_{M\to \infty}\frac{\longip{\Psi_N}{H_{N,M(\ell + 2d+b)}^{D}}{\Psi_N}}{M^3(\ell + 2d + b)^3}
		= 
		\frac{\longip{V\psi_n}{H_{n,\ell+2d}^{D}}{V\psi_N}}{(\ell + 2d + b)^3}
	\\ & \leq 
		\frac{\longip{\psi_n}{H_{n,\ell}^{\textnormal{per}}}{\psi_n}}{\ell^3} \left[1 + O(d/\ell) + O(b/\ell)\right] + O(\rho d^{-2}).
\end{aligned}
\end{equation}
Choosing  $d = a(a^3\rho)^{-5}$ and using \Cref{eqn.energy.density.upper.bound.single.state} 
we conclude that for $a^3\rho$ sufficiently small 
\[
\begin{aligned}
	e(\tilde\rho) 
	& \leq 
		\frac{3}{5}(6\pi)^{2/3} \rho^{5/3} 
	\\ & \qquad 
		+ \frac{12\pi}{5}(6\pi^2)^{2/3} a^3\rho^{8/3}
		\biggl[
		1
		- \frac{9}{35}(6\pi^2)^{2/3}a_0^2\rho^{2/3}
		+ O\left((a^3\rho)^{2/3 + 1/21} |\log(a^3\rho)|^6\right)
		\biggr] 
	\\ & 
		= \frac{3}{5}(6\pi)^{2/3} \tilde\rho^{5/3} 
	\\ & \qquad 
		+ \frac{12\pi}{5}(6\pi^2)^{2/3} a^3\tilde\rho^{8/3}
		\biggl[
		1
		- \frac{9}{35}(6\pi^2)^{2/3}a_0^2\tilde\rho^{2/3}
		+ O\left((a^3\tilde\rho)^{2/3 + 1/21} |\log(a^3\tilde\rho)|^6\right)
		\biggr] 
\end{aligned}
\]
since $\tilde \rho = \rho(1 + O( (a^3\rho)^5 ))$, so $\rho = \tilde \rho(1 + O( (a^3\tilde\rho)^5))$.
This concludes the proof of \Cref{thm.main}. 

It remains to give the proofs of \Cref{lem.derivative.sum.linked.diagrams,lem.bound.3.body.term}.

\subsection{Subleading \texorpdfstring{$2$}{2}-particle diagrams 
(proof of \texorpdfstring{\Cref{lem.derivative.sum.linked.diagrams}}{Lemma \ref*{lem.derivative.sum.linked.diagrams}})}
\label{sec.proof.subleading.2.body.diagrams}
In this section we give the proof of \Cref{lem.derivative.sum.linked.diagrams}.
Before doing this, we first discuss why we don't just use the bounds of these terms from the proof of \Cref{lem.linked.sum.abs.conv}.
\begin{remark}[{Why not use bounds of \Cref{lem.linked.sum.abs.conv}?}]\label{rmk.why.not.Taylor}
Inspecting the proof of \Cref{lem.linked.sum.abs.conv} 
(more precisely \Cref{eqn.sum.same.Gamma2,eqn.sum.different.Gamma2} of \Cref{sec.abs.conv.Gamma2}) we can extract 
the following bound 
\[
	\abs{\sum_{p= 1}^\infty \frac{1}{p!} \sum_{(\pi, G)\in\mcL_p^2} \Gamma_{\pi, G}^{2}}
	\leq C s a^3 \rho^3  \log (b/a)(\log N)^3.
\]
This is immediate by considering the bounds \Cref{eqn.sum.different.Gamma2,eqn.sum.same.Gamma2} and noting that since we have 
$p\geq 1$ in the sum $\sum_{p= 1}^\infty \frac{1}{p!} \sum_{(\pi, G)\in\mcL_p^2} \Gamma_{\pi, G}^{2}$, 
the summands either have $k\geq 1$ or $n_*\geq 1$ or $n_{**}\geq 1$.
Using this bound  
we would thus get  for the error in the ground state energy density 
 the bound $\sim s a^4 \rho^3$ (ignoring the $\log$-factors). However, as we saw in \Cref{sec.energy.in.box}, using \Cref{lem.KE.polyhedron}, 
the $s$-dependent error of the kinetic energy density is $\sim s^{-2}\rho^{5/3}$.
There is no way to choose $s$, such that both of these errors are smaller that $a^3\rho^{8/3}$, 
which is the precision we need in order to prove the leading correction to the kinetic energy in \Cref{thm.main}.

Similarly, by following the proof of \Cref{lem.linked.sum.abs.conv} one could get the bound 
\[
	\rho^{(3)}_{\textnormal{Jas}} \leq C\rho^3 f_{12}^2 f_{13}^2 f_{23}^2.
\]
This bound is not problematic in terms of getting all the error terms in the energy density smaller than $a^3\rho^{8/3}$.
However, we improve on this bound in \Cref{lem.bound.3.body.term} in order to get a better error bound in \Cref{thm.main}.
\end{remark}
\begin{proof}[Proof of \Cref{lem.derivative.sum.linked.diagrams}]
Note first that by  translation invariance  
\[
	\sum_{p= 1}^\infty \frac{1}{p!} \sum_{(\pi, G)\in\mcL_p^2} \Gamma_{\pi, G}^{2}
\] 
is a function of $x_1- x_2$ only.
Recall  \Cref{eqn.compute.Gamma12,eqn.compute.Gamma2.details}. 
We split the diagrams in $\mcL_p^2$ into three groups.
To define these three groups we first define for any diagram $(\pi, G)\in \mcL_p^2$ 
the number $k = k(\pi,G)$ as the number of clusters entirely containing internal vertices.
(This $k$ is exactly the same $k$ as in the proof of \Cref{lem.linked.sum.abs.conv}.)
Then we define
\[
	\textstyle
	\nu = \nu(\pi,G) = \sum_{\ell=1}^k n_\ell - 2k,
	\qquad 
	\nu^* = \nu^*(\pi,G) = n_* + n_{**},
\]
with the understanding that $n_{**}=0$  if $\{1\}$ and $\{2\}$ are in the same cluster.  
We think of $\nu + \nu^*$ as the ``number of added vertices''.
Indeed, given a $k$ the smallest number of vertices 
in a diagram $(\pi,G)\in \mcL_p^2$ with $k$ clusters is $2k+2$ and in this case we have $p=2k$.
For such a diagram, there are $p=2k$ internal vertices and $2$ external vertices.
The graph $G$ of such a diagram looks like 
\[
\begin{tikzpicture}[line cap=round,line join=round,>=triangle 45,x=1.0cm,y=1.0cm]
		\node (1) at (-1,-1.5) {};
		\node (2) at (2,-1.5) {};
		\node (3) at (0,-1) {};
		\node (4) at (1,-1) {};
		\node (5) at (0,-2) {};
		\node (6) at (1,-2) {};
		\draw[dashed] (3) -- (4);
		\draw[dashed] (5) -- (6);
		\foreach \i in {1,...,6} \draw[fill] (\i) circle [radius=1.5pt];
		\foreach \i in {1,...,3} \draw[fill] (0.5,-1 - \i/4) circle [radius=0.5pt];
		\draw[decoration={brace,raise=5pt,amplitude=5pt},decorate] (1,-0.8) -- node[right=6pt] {$\, k$} (1,-2.2);
		\foreach \i in {1,2} \node[anchor=north] at (\i) {$\i$};
    \foreach \i in {1,2} \node[anchor=south] at (\i) {$*$};
		\node at (-2,-1.5) {$G=$};
		\end{tikzpicture}
		.
\]
Then $\nu + \nu^*$ is the number of (internal) vertices  a diagram has more than this lowest number.

By following the bound in \Cref{eqn.sum.same.Gamma2,eqn.sum.different.Gamma2} we see that 
for $p = \nu_0 +2k_0$
\begin{equation}\label{eqn.bound.linked.sum.k.ng.cluster}
	\abs{\frac{1}{p!}\sum_{\substack{(\pi,G)\in \mcL_p^2 \\ \nu(\pi,G) + \nu^*(\pi,G) = \nu_{0} \\ k(\pi,G) = k_0}} \Gamma_{\pi,G}^2}
	\leq C\rho^2 (C s (\log N )^3)^{k_0} (C a^3 \rho \log (b/a))^{k_0 + \nu_{0}}.
\end{equation}
We split diagrams  into different groups depending on whether or not they are ``large'' 
and whether or not we will do a Taylor expansion of their values.
We first give some motivation for what ``large'' means.

\begin{remark}
Here ``large'' and ``small'' should be interpreted in the sense of how many vertices appear 
in the diagram. \Cref{eqn.bound.linked.sum.k.ng.cluster} describe how diagrams with more vertices (larger values of $k, \nu, \nu^*$) 
have a smaller value. 
More precisely, ``large'' should be thought of in terms of the bound in \Cref{eqn.bound.linked.sum.k.ng.cluster} in the following sense.

Recall that the ($s$-dependent) error in the kinetic energy density is $s^{-2}\rho^{5/3}$. 
For this error to be smaller than 
the desired accuracy of order $a^3a_0^2\rho^{10/3}$ we need $s \gg (a^3\rho)^{-5/6}$.
If we think of, say $s\sim (a^3\rho)^{-6/7}$, then (ignoring $\log$-factors)
\Cref{eqn.bound.linked.sum.k.ng.cluster} reads 
$ \lesssim \rho^2 (a^3\rho)^{k_0/7 + \nu_0}$.
The large diagrams (with $\nu^*\geq 1$) are those for which this bound 
gives a contribution to the energy density $\ll a^5\rho^{10/3}$, i.e. with $\nu + \nu^* + k/7 \geq 4/3$ by counting powers of $\rho$.
For diagrams with $\nu^*=0$ we obtain a differentiated version of the bound in \Cref{eqn.bound.linked.sum.k.ng.cluster}, 
where one effectively gains a power $a^2\rho^{2/3}$, see the details of the proof. 
The large diagrams (with $\nu^*=0$) are those for which the differentiated version 
gives a contribution to the energy density $\ll a^5\rho^{10/3}$, i.e. with $\nu + k/7 \geq 2/3$.
\end{remark}

\noindent
We split the diagrams into three (exhaustive) groups:
\begin{enumerate}
\item Small diagrams with 
\begin{enumerate}[$(A)$]
\item $\{1\}$ and $\{2\}$ in different clusters and $1\leq k \leq 4, \nu = 0, \nu^* = 0$,
\item $\{1\}$ and $\{2\}$ in different clusters and $0\leq k \leq 2, \nu = 0, \nu^* = 1$,
\item $\{1\}$ and $\{2\}$ in the same cluster and $0\leq k\leq 2, \nu = 0, \nu^* = 1$.
\end{enumerate} 
\item Large diagrams with $\nu^* = 0$ (in particular $\{1\}$ and $\{2\}$ are in different clusters) and
\begin{enumerate}[$(A)$]
\item $k\geq 5, \nu = 0$ or
\item $k\geq 1, \nu \geq 1$. 
\end{enumerate} 
\item Large diagrams with $\nu^* \geq 1$ and
\begin{enumerate}[$(A)$]
\item $k\geq 3, \nu = 0$ or
\item $\nu + \nu^* \geq 2$. 
\end{enumerate} 
\end{enumerate}
Note that we have $p\geq 1$, so the diagrams with $k=0, \nu = 0, \nu^* = 0$ are not present.
Moreover, if $k=0$ then clearly also $\nu = 0$.
For drawings of the small diagrams see \Cref{fig.small.diagrams} in \Cref{sec.small.diagrams.2.particle}.
We then write 
\begin{equation}\label{eqn.2.body.linked.diagrams.decompose}
	\sum_{p=1}^\infty \frac{1}{p!} \sum_{(\pi,G)\in \mcL_p^2} \Gamma_{\pi,G}^2
	= 
	\xi_{\textnormal{small}, 0} + \xi_{\textnormal{small}, \geq 1} +  \xi_{0} + \xi_{\geq 1},
\end{equation}
where $\xi_{\textnormal{small},0}$ is the contribution of all small diagrams of types $(A)$ and $(B)$,
$\xi_{\textnormal{small},\geq 1}$ is the contribution of all small diagrams of type $(C)$, 
$\xi_0$ is the contribution of all large diagrams with $\nu^* = 0$, and 
$\xi_{\geq 1}$ is the contribution of all large diagrams with $\nu^* \geq 1$.

The notation is motivated by that of the large diagrams, which were split into two groups depending
on whether $\nu^* = 0$ or $\nu^* \geq 1$.
We will treat the small diagrams of types $(A)$ and $(B)$ somewhat similar to the large diagrams 
in $\xi_0$ (hence the notation $\xi_{\textnormal{small},0}$)
and the small diagrams of type $(C)$ somewhat similar to the large diagrams in $\xi_{\geq 1}$
(hence the notation $\xi_{\textnormal{small},\geq 1}$).
Indeed, we will do a Taylor expansion of $\xi_{\textnormal{small},0}$ and $\xi_0$ 
but not of $\xi_{\textnormal{small},\geq 1}$ or $\xi_{\geq 1}$.

Using the bound in \Cref{eqn.bound.linked.sum.k.ng.cluster} and the absolute convergence (\Cref{lem.linked.sum.abs.conv}) we get 
\begin{equation}\label{eqn.bound.xi.geq1}
\begin{aligned}
  |\xi_{\geq1}(x_1,x_2)| 
  	& 
  		\leq \underbrace{C \rho^2 (s(\log N)^3)^3(a^3\rho\log(b/a))^4}_{\textnormal{type $(A)$ diagrams}}
  		+ \underbrace{C\rho^2 (a^3\rho\log(b/a))^2}_{\textnormal{type $(B)$ diagrams}}
  	\\ & 
  		\leq 
        C a^6\rho^{4} (\log (b/a))^2 
        \left[
        s^3 (\log N)^9 a^6 \rho^2 (\log (b/a))^2
        + 1 \right]
\end{aligned}
\end{equation}
uniformly in $x_1,x_2$.
For $\xi_{\textnormal{small},\geq 1}$ we have 
\begin{lemma}\label{lem.small.diagrams.type.C}
For the small diagrams of type $(C)$ we have the bound
\[
  \abs{\xi_{\textnormal{small},\geq 1}} \leq C a^3 b^2\rho^{3+4/3}  |x_1-x_2|^2 + C a^6\rho^4 (\log(b/a))^2.
\]
\end{lemma}
\noindent
The proof of this lemma is a (not very insightful) computation. We give it 
in \Cref{sec.small.diagrams.2.particle}.
Slightly more insightful however, is why we split off the small diagrams from the large diagrams.
\begin{remark}[Why one gets better bounds by computing small diagrams]
We could treat all the small diagrams exactly as we treat $\xi_0$ and $\xi_{\geq 1}$.
We do however gain better error bounds by treating them more directly, 
i.e. computing more precisely what the values of these small diagrams are. 
In exact calculations we can  make use the fact that $\int \gamma^{(1)}_N = 1$, instead of  
bounding the absolute value as $\int |\gamma^{(1)}_N| \leq Cs (\log N)^3$.
\end{remark}

\noindent
We  Taylor expand $\xi_{\textnormal{small},0}$ and $\xi_0$ to second order around the diagonal. 
We first claim that 
\begin{equation}\label{eqn.claimed.xi.x1=x2}
	\xi_{\textnormal{small},0}(x_2, x_2) + \xi_0(x_2,x_2) + \xi_{\textnormal{small},\geq 1}(x_2, x_2) + \xi_{\geq 1}(x_2,x_2) = 0
\end{equation}
Indeed by \Cref{thm.gaudin.expansion} we have 
\[
\begin{aligned}
  \xi_{\textnormal{small},0} + \xi_0 + \xi_{\textnormal{small},\geq1} + \xi_{\geq 1}
  & = 
    \frac{\rho^{(2)}_{\textnormal{Jas}}}{f_{12}^{2}} - \rho^{(2)}
    =
    \frac{N(N-1)}{C_N} \idotsint \prod_{\substack{i < j \\ (i,j) \ne (1,2)}} f_{ij}^2 D_N \ud x_3 \ldots \ud x_{N} - \rho^{(2)}.
\end{aligned}
\]
(Formally to do the division by $f$ in the first equality in case $f=0$ somewhere one uses 
\Cref{thm.gaudin.expansion} with all instances of $f$ replaced by $\tilde f^{(n)}$
for some sequence $\tilde f^{(n)} > 0$ with $\tilde f^{(n)}\searrow f$.
Then one readily applies the Lebesgue dominated convergence theorem to exchange the limit $\tilde f^{(n)}\to f$ with the relevant sums and integrals.)
Taking $x_1=x_2$ in this we have $D_N = 0$ and $\rho^{(2)} = 0$. 
This shows \Cref{eqn.claimed.xi.x1=x2}.
We may thus bound the zeroth order term of $\xi_{\textnormal{small},0}$ and $\xi_0$ by 
\begin{equation}\label{eqn.bound.zeroth.order.xi0}
\begin{aligned}
	& |\xi_{\textnormal{small},0}(x_2,x_2) + \xi_{0}(x_2,x_2)| 
	\\ & \quad 
	\leq 
	|\xi_{\textnormal{small},\geq 1}(x_2, x_2)| + |\xi_{\geq 1}(x_2,x_2)|
	\\ & \quad 
	\leq
	C a^6 \rho^4 (\log (b/a))^2
		\left[
		s^3 (\log N)^9 a^6 \rho^2 (\log (b/a))^2
		+
		1
		\right].
\end{aligned}
\end{equation}
Since both $\xi_{\textnormal{small},0}$ and $\xi_0$ are symmetric in $x_1$ and $x_2$ all first order terms vanish.
We are left with bounding the second derivatives.
For $\xi_{\textnormal{small},0}$ we have 
\begin{lemma}\label{lem.Taylor.small.diagrams}
For any $\mu,\nu=1,2,3$ we have 
\[
	\abs{\partial^\mu_{x_1}\partial^\nu_{x_1} \xi_{\textnormal{small},0}}
	\leq C a^3 \rho^{3+2/3}\log(b/a)
\]	
uniformly in $x_1,x_2$. Here $\partial^\mu_{x_1}$ denotes the derivative in the $x_1^\mu$-direction.
\end{lemma}
\noindent
The proof of this lemma is a (not very insightful) computation. We give it in \Cref{sec.small.diagrams.2.particle}.

Next we consider $\partial^\mu_{x_1}\partial^\nu_{x_1}\xi_0$.
We write $\xi_0$ in terms of truncated densities as in \Cref{eqn.compute.Gamma2.details}, i.e.
\[
\begin{aligned}
\xi_0 
& = 
    \sum_{k=0}^\infty \frac{1}{k!} 
    \sum_{\substack{n_1,\ldots,n_k \geq 2}} 
    \chi_{(k\geq 5, n_\ell = 2 \textnormal{ or } k\geq 1, \sum n_\ell = 2k+1)}
    \frac{1}{\prod_\ell n_\ell!} 
    \sum_{\substack{G_1,\ldots,G_k \\ G_\ell \in \mcC_{n_\ell}}}
  \\ & \quad \times 
  \idotsint \prod_{\ell}\prod_{e\in G_\ell} g_e \rho^{(\{1\}, \{2\}, n_1,\ldots,n_k)}_{\rm t} \ud x_3 \ldots\ud x_{p+2}.
\end{aligned}
\]
Since we consider terms with $n_* = n_{**} = 0$, there are no $g$-factors that depend on $x_1$ and thus all derivatives are of $\rho_{\rm t}^{(\{1\}, \{2\}, n_1,\ldots,n_k)}$.
We thus need to calculate $\partial^\mu_{x_1}\partial^\nu_{x_1}\rho_{\rm t}^{(\{1\}, \{2\}, n_1,\ldots,n_k)}$. 
For this we use the definition in \Cref{eqn.define.truncated.correlation} rather than the formula in \Cref{eqn.truncated.correlation}.
In \Cref{eqn.define.truncated.correlation} the variable $x_1$ appears exactly twice:
Once in an outgoing $\gamma^{(1)}_N$-edge from $\{1\}$ and once in an incoming $\gamma^{(1)}_N$-edge to $\{1\}$.
Taking the derivatives then amounts to replacing either one of these edges 
by its second derivative or both of them by their first derivatives. 
Thus, using that $\gamma^{(1)}_N(x_i;x_j) = \gamma^{(1)}_N(x_j;x_i)$ since $\gamma^{(1)}_N$ is real, 
and that for $(\pi,\cup G_\ell)$ to be linked necessarily $\pi(1)\ne 1$, 
we have (for $p = 2 + \sum_\ell n_\ell$)
\begin{equation}\label{eqn.partial.rhot.initial}
\begin{aligned}
	& \partial^\mu_{x_1}\partial^\nu_{x_1} \rho^{(\{1\},\{2\},n_1,\ldots,n_k)}_{\rm t}
 	\\ & \quad 
 	= 
 	\partial^\mu_{x_1}\partial^\nu_{x_1}
 	\sum_{\pi \in \mcS_p}
			(-1)^\pi \chi_{\left((\pi, \cup G_\ell)\in\mcL_p\right)}
			\prod_{j=1}^p \gamma^{(1)}_N(x_j; x_{\pi(j)})
	\\
	& \quad 
	= 
	 	\sum_{\pi \in \mcS_p}
			(-1)^\pi \chi_{\left((\pi, \cup G_\ell)\in\mcL_p\right)}
			\prod_{j\ne 1, j\ne\pi^{-1}(1)}
			\gamma^{(1)}_N(x_j; x_{\pi(j)})
			\left[
				2 \partial^\mu_{x_1}\partial^\nu_{x_1}\gamma^{(1)}_N(x_1;x_{\pi(1)}) \gamma^{(1)}_N(x_{\pi^{-1}(1)};x_1)
			\right.
	\\
	& \qquad 
			\left.
				+2 \partial^\mu_{x_1}\gamma^{(1)}_N(x_1;x_{\pi(1)})\partial^\nu_{x_1}\gamma^{(1)}_N(x_{\pi^{-1}(1)};x_1)  
			\right].
\end{aligned}
\end{equation}
With this formula we may then redo the computation of \cite[Equation (D.53)]{Giuliani.Mastropietro.ea.2021}
only now some of the $\gamma^{(1)}_N$-factors (precisely $1$ or $2$ of them) carry derivatives.
The $\gamma^{(1)}_N$-factors with derivatives may end up in the anchored tree, 
or they may end up in the matrix $\mcN(r)$.
If they end up in $\mcN(r)$ it is explained around \cite[Equation (D.9)]{Giuliani.Mastropietro.ea.2021}
how to modify \Cref{lem.bound.det.N}. 
One simply includes factors $ik^\mu$ in the definition of (some of) the functions $\alpha_i$ (and not of $\beta_j$) in the proof of \Cref{lem.bound.det.N}.
Since we may bound $|k|\leq C\rho^{1/3}$ for $k\in P_F$ we get
\begin{equation}\label{eqn.bound.det.N.tilde}
	\abs{\det \tilde\mcN(r)} \leq \begin{cases}
	\rho^{\sum n_\ell + 2 - (k +2 -1)} & \textnormal{if no derivatives end up in $\mcN$},
	\\
	C\rho^{\sum n_\ell +2 - (k + 2-1) + 1/3} & \textnormal{if one derivative ends up in $\mcN$},
	\\
	C\rho^{\sum n_\ell +2 - (k + 2-1) + 2/3} & \textnormal{if two derivatives end up in $\mcN$},
	\end{cases}
\end{equation}
where $\tilde\mcN(r)$ is the appropriate modification of $\mcN(r)$.
To get the formula for $\rho_{\rm t}$ we need also to consider two cases for how the anchored tree looks.
There could be both an incoming and an outgoing edge to/from the vertex $\{1\}$.
And if there is just one edge to/from $\{1\}$ it could be either an incoming or an outgoing edge.
Since $\gamma^{(1)}_N$ is real, incoming and outgoing edges gives the same factor $\gamma^{(1)}_N(x_1;x_j)$.
A simple calculation (essentially just undoing the product rule) then shows that
\begin{equation}\label{eqn.partial.rhot}
\begin{aligned}
	\partial^\mu_{x_1}\partial^\nu_{x_1} \rho^{(\{1\},\{2\},n_1,\ldots,n_k)}_{\rm t}
	& 
	= 
		\sum_{\partial \in \{1, \partial^\mu_{x_1}, \partial^\nu_{x_1}, \partial^\mu_{x_1}\partial^\nu_{x_1}\}}
		\left[
		\sum_{\substack{\tau \in \mcA^{(\{1\},\{2\},n_1,\ldots,n_k)} \\ \textnormal{two edges to/from $\{1\}$}}} 
		\partial \left[\gamma^{(1)}_N(x_{j_2}, x_1)\gamma^{(1)}_N(x_1; x_{j_1}) \right]
		\right.
	\\ & 
	\qquad 
		\left.
		+
		\sum_{\substack{\tau \in \mcA^{(\{1\},\{2\},n_1,\ldots,n_k)} \\ \textnormal{one edge to/from $\{1\}$}}} 
		\partial\gamma^{(1)}_N(x_1; x_{j_1})
		\right]
		\prod_{\substack{(i,j) \in \tau \\ i,j\ne 1}} \gamma^{(1)}_N(x_{i}; x_j) \int \ud\mu_\tau(r) \det \tilde\mcN_\partial(r),
\end{aligned}
\end{equation}
where $j_1$ and $j_2$ denote the vertices connected to $\{1\}$ by the relevant edges in $\tau$
and $\tilde\mcN_\partial$ is the appropriately modified version of $\mcN$, 
where the derivatives not in $\partial$ end up in $\mcN$, i.e. \Cref{eqn.bound.det.N.tilde} reads 
\[
	\abs{\det \tilde\mcN_\partial(r)} \leq C \rho^{\sum n_\ell + 2 - (k+2-1) + (2 - \#\partial)/3},
\]
where $\#\partial$ denotes the number of derivatives in $\partial$, 
i.e. $\# 1 = 0, \#\partial^\mu_{x_1} = 1$ and $\#\partial^\mu_{x_1}\partial^\nu_{x_1} = 2$.

We denote the contribution of the two terms in \Cref{eqn.partial.rhot} to $\partial^\mu_{x_1}\partial_{x_1}^\nu\xi_0$  by 
$(\partial^\mu_{x_1}\partial_{x_1}^\nu\xi_0)^{\rightarrow \bullet \rightarrow}$
and 
$(\partial^\mu_{x_1}\partial_{x_1}^\nu\xi_0)^{\bullet \rightarrow}$ respectively.

We first deal with the second term of \Cref{eqn.partial.rhot} 
where there is just one edge to/from $\{1\}$ in the anchored tree. 
We may bound the contribution of this term almost exactly as 
in the proof of \Cref{lem.linked.sum.abs.conv}. 
We give a sketch here.
Using  \Cref{eqn.bound.det.N.tilde} we get the bound 
\begin{equation}\label{eqn.bound.partial.rhot}
\begin{aligned}
  & \leq 
  C
  \sum_{\partial \in \{1, \partial^\mu_{x_1}, \partial^\nu_{x_1}, \partial^\mu_{x_1}\partial^\nu_{x_1}\}}
  \rho^{(2- \#\partial)/3}
  \sum_{\tau \in \mcA^{(\{1\},\{2\},n_1,\ldots,n_k)} } 
  \abs{\partial\gamma^{(1)}_N(x_1; x_{j_1})}
  \\ & \qquad \times 
  \prod_{\substack{(i,j) \in \tau \\ i,j\ne 1}} \abs{\gamma^{(1)}_N(x_{i}; x_j)} \rho^{(\sum_{\ell} n_\ell + 2) - (k + 2 - 1)},
\end{aligned}
\end{equation}
where again $\#\partial$ denotes the number of derivatives in $\partial$.

To bound the integrations we again follow the strategy of the proof of the $\Gamma^{2}$-sum of \Cref{lem.linked.sum.abs.conv}, 
\Cref{sec.abs.conv.Gamma2}.
The only difference is that the $\gamma_N^{(1)}$-edge on the path in the anchored tree between $\{1\}$ and $\{2\}$
incident to $\{1\}$ is the edge with derivatives, $\partial\gamma^{(1)}_N(x_1; x_{j_1})$. 
This we bound by 
$|\partial\gamma^{(1)}_N| \leq C \rho^{1 + \#\partial/3}$.
The integrations can then be performed exactly as in \Cref{sec.abs.conv.Gamma2}.
We conclude the bound 
\[
\begin{aligned}
& 
  \idotsint 
      \prod_{\ell=1}^k \prod_{e\in T_\ell} |g_e |
      \abs{\partial\gamma^{(1)}_N(x_1; x_{j_1})}
      \prod_{\substack{(i,j) \in \tau \\ i,j\ne 1}} \abs{\gamma^{(1)}_N(x_{i}; x_j)}
  \ud x_3 \ldots\ud x_{\sum n_\ell+2}
\\ & \quad 
  \leq 
  \rho^{1+\#\partial/3}\left(Ca^3\log(b/a)\right)^{\sum n_\ell - k} \left(Cs(\log N)^3\right)^{k}.
\end{aligned}
\]
Again, as in the proof of \Cref{lem.linked.sum.abs.conv} we have by Cayley's formula that $\# \mcT_n = n^{n-2} \leq C^n n!$ and 
by \cite[Appendix D.5]{Giuliani.Mastropietro.ea.2021} that 
$\#\mcA^{(\{1\},\{2\},n_1,\ldots,n_k)} \leq (k+2)! 4^{\sum n_\ell+2} \leq C(k^2 + 1) k! 4^{\sum n_\ell}$.
Following the same arguments as for \Cref{eqn.bound.linked.sum.k.ng.cluster} and recalling that the diagrams in $\xi_0$
have either $k\geq 5$ or $\nu\geq 1$ 
we get the contribution to $\partial^\mu_{x_1}\partial_{x_1}^\nu\xi_0$ of 
\begin{equation}\label{eqn.bound.xi0.one.edge}
\begin{aligned}
  \abs{(\partial^\mu_{x_1} \partial^\nu_{x_1} \xi_0)^{\bullet \rightarrow}
  }
  &
  	\leq C\rho^{2+2/3} 
  	\Bigl[\,
  		\underbrace{(s (\log N )^3)^{5} (a^3 \rho \log (b/a))^{5}}_{\textnormal{type $(A)$ diagrams}}
  		+
  		\underbrace{s (\log N )^3(a^3 \rho \log (b/a))^2}_{\textnormal{type $(B)$ diagrams}}
  	\,\Bigr]
  	\\
  	&
  	\leq 
  	Ca^6\rho^{4+2/3} (\log(b/a))^2 s(\log N)^3 \left[s^4 (\log N)^{12} a^{9} \rho^3 (\log (b/a))^3 + 1\right]
\end{aligned}
\end{equation}
uniformly in $x_1,x_2$.

Next consider the first term of \Cref{eqn.partial.rhot}. 
The argument is almost the same, only we have to distinguish between which $\gamma^{(1)}_N$-factor(s) the derivatives in $\partial$ hits. 
We consider the case $\partial = \partial^\mu_{x_1}\partial^{\nu}_{x_1}$. The other cases are similar.
Suppose that the $\gamma^{(1)}_N$-edge on the path (in the anchored tree) from $\{1\}$ to $\{2\}$ 
is $\gamma^{(1)}_N(x_1;x_{j_1})$ and the $\gamma^{(1)}_N$-factor not on the path is $\gamma^{(1)}_N(x_{j_2};x_1)$.
We distinguish between three cases:
\begin{enumerate}
\item 
If both derivatives are on $\gamma^{(1)}_N(x_1;x_{j_1})$ we may bound this exactly as above.

\item
If one derivative is on $\gamma^{(1)}_N(x_1;x_{j_1})$ (say $\partial^\nu_{x_1}$) and one derivative (say $\partial^\mu_{x_1}$)
is on $\gamma^{(1)}_N(x_{j_2};x_1)$ 
we bound $|\partial^\nu_{x_1} \gamma^{(1)}_N(x_1;x_{j_1})|\leq C\rho^{4/3}$. 
Then the argument is similar, only now one of the $\gamma^{(1)}_N$-integrations is with 
$\partial^\mu \gamma^{(1)}_N$ instead. 
Thus, in the computation leading to \Cref{eqn.bound.xi0.one.edge} we should replace one factor $Cs\rho^{1/3}(\log N)^3$ with 
$\int |\partial^\mu \gamma^{(1)}_N| \ud x$.

\item 
If both derivatives are on $\gamma^{(1)}_N(x_{j_2};x_1)$
then analogously we bound $|\gamma^{(1)}_N(x_1;x_{j_1})|\leq C\rho$
and in the computation leading to \Cref{eqn.bound.xi0.one.edge} we should replace one factor $Cs\rho^{2/3}(\log N)^3$ with 
$\int |\partial^\mu\partial^{\nu} \gamma^{(1)}_N| \ud x$.
\end{enumerate}




\noindent
In total we have the contribution to $\partial^\mu_{x_1}\partial_{x_1}^\nu\xi_0$ of 
\[
\begin{aligned}
  &
  \leq
  	C\rho^2
  	\left(\rho^{2/3} s(\log N)^3
  	+ \rho^{1/3} 
    \int_{\Lambda} \abs{\partial^{\mu} \gamma^{(1)}_N} \ud x
    +  
    \rho^{1/3}
    \int_{\Lambda} \abs{\partial^{\nu} \gamma^{(1)}_N} \ud x
    +
    \int_{\Lambda} \abs{\partial^{\mu}\partial^\nu \gamma^{(1)}_N} \ud x
  	\right)
\\ & \quad 
	\times 
  	\left[
  		(s (\log N )^3)^{4} (a^3 \rho \log (b/a))^{5}
  		+
  		(a^3 \rho \log (b/a))^2  	
  	\right]
\end{aligned}
\]
uniformly in $x_1,x_2$. 
One may do a similar computation for the other cases of $\partial$ and conclude that 
\begin{equation}\label{eqn.bound.xi0.two.edges}
\begin{aligned}
  	& \abs{(\partial^\mu_{x_1} \partial^\nu_{x_1} \xi_0)^{\rightarrow \bullet \rightarrow}}
  	\\ & \quad 
      \leq
    C\rho^2
    \left(\rho^{2/3} s(\log N)^3
    + \rho^{1/3} 
    \int_{\Lambda} \abs{\partial^{\mu} \gamma^{(1)}_N} \ud x
    +  
    \rho^{1/3}
    \int_{\Lambda} \abs{\partial^{\nu} \gamma^{(1)}_N} \ud x
    +
    \int_{\Lambda} \abs{\partial^{\mu}\partial^\nu \gamma^{(1)}_N} \ud x
    \right)
\\ & \qquad 
  \times 
    \left[
      (s (\log N )^3)^{4} (a^3 \rho \log (b/a))^{5}
      +
      (a^3 \rho \log (b/a))^2   
    \right]
\end{aligned}
\end{equation}
uniformly in $x_1,x_2$.
Thus, we need to bound the integrals
\begin{align*}
	\int_{\Lambda} \abs{\partial^{\mu} \gamma^{(1)}_N} \ud x
	& 
	= 
	\int_{\Lambda} \frac{1}{L^3} \abs{\sum_{k\in P_F} k^\mu e^{ikx}} \ud x 
	= \frac{1}{(2\pi)^2 L} \int_{[0,2\pi]^3} \abs{\sum_{q\in \left(\frac{Lk_F}{2\pi} P\right)\cap \Z^3} q^\mu e^{iqu}} \ud u,
\\
\intertext{and}
	\int_{\Lambda} \abs{\partial^{\mu}\partial^\nu \gamma^{(1)}_N} \ud x
	& 
	= 
	\int_{\Lambda} \frac{1}{L^3} \abs{\sum_{k\in P_F} k^\mu k^\nu e^{ikx}} \ud x 
	= \frac{1}{2\pi L^2} \int_{[0,2\pi]^3} \abs{\sum_{q\in \left(\frac{Lk_F}{2\pi} P\right)\cap \Z^3} q^\mu q^\nu e^{iqu}} \ud u.
\end{align*}
Here we have
\begin{lemma}\label{lem.derivative.lebesgue.constant}
The polyhedron $P$ from \Cref{defn.P_F.true} satisfies for any $\mu,\nu=1,2,3$ that 
\[
\begin{aligned}
	\int_{[0,2\pi]^3} \abs{\sum_{q\in \left(\frac{Lk_F}{2\pi} P\right)\cap \Z^3} q^\mu e^{iqu}} \ud u
	& \leq C s N^{1/3} (\log N)^3, 
	\\ 
	\int_{[0,2\pi]^3} \abs{\sum_{q\in \left(\frac{Lk_F}{2\pi} P\right)\cap \Z^3} q^\mu q^\nu e^{iqu}} \ud u
	& \leq C s N^{2/3} (\log N)^4
\end{aligned}
\]
for sufficiently large $N$.
\end{lemma}

\noindent
The proof of \Cref{lem.derivative.lebesgue.constant} is a long and technical computation, which we give in \Cref{sec.derivative.lebesgue.constant}.
Applying the lemma we conclude that 
\[
	\int_{\Lambda} \abs{\partial^{\mu} \gamma^{(1)}_N} \ud x \leq C s \rho^{1/3} (\log N)^3,
	\qquad 
	\int_{\Lambda} \abs{\partial^{\mu}\partial^\nu \gamma^{(1)}_N} \ud x \leq C s \rho^{2/3} (\log N)^4.
\]
By combining this with \Cref{eqn.bound.xi0.two.edges,eqn.bound.xi0.one.edge} we get 
\begin{equation}\label{eqn.bound.2nd.order.xi0}
\begin{aligned}
  \abs{\partial^\mu_{x_1} \partial^\nu_{x_1} \xi_0}
  	&
  	\leq 
  	Ca^6\rho^{4+2/3} (\log(b/a))^2 s(\log N)^4 \left[s^4 (\log N)^{12} a^{9} \rho^3 (\log (b/a))^3 + 1\right]
\end{aligned}
\end{equation}
uniformly in $x_1,x_2$. 
Combining  \Cref{lem.small.diagrams.type.C,lem.Taylor.small.diagrams,eqn.bound.xi.geq1,eqn.bound.zeroth.order.xi0,eqn.bound.2nd.order.xi0}
and using that for any real number $t>0$ and integer $n\geq1$ we may bound $t \lesssim t^n+1$
this shows the desired.
\end{proof}

\begin{remark}[Treating more diagrams as small]\label{rmk.more.small.diagrams}
One can improve the error bound in \Cref{thm.main} slightly by treating more diagrams as small, i.e. calculating their values more precisely. 
This is similar to what is done in \cite{Basti.Cenatiempo.ea.2022a} for the dilute Bose gas.
(In \cite{Basti.Cenatiempo.ea.2022a} the Bose gas is treated with a method very similar to a cluster expansion. 
Their expansion is performed to some arbitrarily high order [denoted by $M$ in \cite{Basti.Cenatiempo.ea.2022a}], 
which if chosen sufficiently large yields the bounds of \cite{Basti.Cenatiempo.ea.2022a}.)
We sketch the overall idea.

If we choose $s\sim (a^3\rho)^{-1+\eps/2}$ then the error from the $s$-dependent term in the kinetic energy 
to the energy density is $\rho^{5/3}(a^3\rho)^{2-\eps}$.
Then choose as ``large'' the diagrams 
for which the bound in \Cref{eqn.bound.linked.sum.k.ng.cluster} gives contributions to the energy density much smaller than $\rho^{5/3}(a^3\rho)^{2-\eps}$.
This happens for $k_0 > K$ for some large $K \sim \eps^{-1}$.
We can then evaluate all small diagrams as in \Cref{sec.small.diagrams} and conclude that their contributions 
are as given in \Cref{sec.small.diagrams} only with some $K$-dependent constants, since 
there is some $K$-dependent number of small diagrams.
In total we would then get 
an error of size $O_\eps( \rho^{5/3}(a^3\rho)^{2-\eps} )$ in \Cref{thm.main}.


\end{remark}

\subsection{Subleading \texorpdfstring{$3$}{3}-particle diagrams 
(proof of \texorpdfstring{\Cref{lem.bound.3.body.term}}{Lemma \ref*{lem.bound.3.body.term}})}
\label{sec.proof.subleading.3.body.diagrams}
\begin{proof}[Proof of \Cref{lem.bound.3.body.term}]
Recall the formula for $\rho^{(3)}_{\textnormal{Jas}}$ of \Cref{thm.gaudin.expansion}.
In this formula there are terms like $\rho\sum_{p\geq 1} \frac{1}{p!} \sum_{(\pi,G)\in\mcL_p^2} \Gamma_{\pi,G}^2(x_2,x_3)$.
We have $\rho = \rho^{(1)}_{\textnormal{Jas}}(x_1) = \sum_{p'\geq 1} \frac{1}{p'!} \sum_{(\pi',G')\in\mcL_{p'}^1} \Gamma_{\pi',G'}^1(x_1)$
by translation invariance.
Joining the two diagrams $(\pi,G)\in \mcL_p^2$ and $(\pi',G')\in\mcL_{p'}^1$ we get a new (no longer linked) diagram
$(\pi^{\prime\prime},G^{\prime\prime})\in\mcD_{p+p'}^3$
with two linked components, one of which contains the vertices $\{2\}$ and $\{3\}$ and one of which contains the vertex $\{1\}$.
Doing this for all three terms of this type, we are led to define the set 
\[
	\tilde\mcL_p^3 := \mcL_p^3 \cup \bigcup_{q+q' = p}(\mcL_q^2 \oplus \mcL_{q'}^1),
\]
where $\oplus$ refers to the operation of joining two diagrams as above.
The set 
$\tilde\mcL_p^3$ is then the set of diagrams on $3$ external and $p$ internal vertices 
such that there is at most two linked components, and that each linked component contains at least one external vertex.
With this, the formula for $\rho^{(3)}_{\textnormal{Jas}}$ of \Cref{thm.gaudin.expansion} reads
(assuming that $s a^3\rho \log (b/a) (\log N)^3$ is sufficiently small)
\[
\begin{aligned}
	\rho^{(3)}_{\textnormal{Jas}}
	& = f_{12}^2 f_{13}^2 f_{23}^2 
		\Biggl[
			\rho^{(3)}
			+
			\sum_{p\geq 1} \frac{1}{p!}\sum_{(\pi, G)\in\tilde \mcL_{p}^3} 
			\Gamma_{\pi, G}^{3}
		\Biggr].
\end{aligned}
\]
We  split the diagrams in $\tilde\mcL_p^3$ 
into two groups, large and small similarly to the proof of \Cref{lem.derivative.sum.linked.diagrams}.
To do this, we similarly define for a diagram $(\pi, G)\in \tilde\mcL_p^3$ 
the number $k = k(\pi,G)$ as the number of clusters entirely containing internal vertices.
(This $k$ is exactly the same $k$ as in the proof of \Cref{lem.linked.sum.abs.conv}.)
Then we define
\[
	\textstyle
	\nu = \nu(\pi,G) = \sum_{\ell=1}^k n_\ell - 2k,
	\qquad 
	\nu^* = \nu^*(\pi,G) = n_* + n_{**} + n_{***},
\]
where we understand $n_{**}=0$ and/or $n_{***}=0$ if $\{1,2,3\}$ are not all in different clusters.
(One defines $n_{***}$ as the number of internal vertices in the cluster containing $\{3\}$ 
if all $\{1,2,3\}$ are in different clusters, exactly as for the $n_{**}$ and $n_*$ 
of \Cref{sec.abs.conv.Gamma2,sec.abs.conv.Gamma1}.)
We may still think of $\nu + \nu^*$ as the ``number of added vertices''.
As for \Cref{eqn.bound.linked.sum.k.ng.cluster} we have (for $p = 2k_0 + \nu_0$)
\begin{equation}\label{eqn.bound.linked.sum.k.ng.cluster.3.body}
	\abs{\frac{1}{p!}\sum_{\substack{(\pi,G)\in \tilde\mcL_p^3 \\ \nu(\pi,G) + \nu^*(\pi,G) = \nu_0 \\ k(\pi,G) = k_0}} \Gamma_{\pi,G}^3}
	\leq C\rho^3 (C s (\log N )^3)^{k_0} (C a^3 \rho \log (b/a))^{k_0 + \nu_0}.
\end{equation}
The main difference compared to \Cref{eqn.bound.linked.sum.k.ng.cluster} is that we here allow diagrams that are not linked. 
This doesn't matter, since when we compute the integrals (as in \Cref{sec.abs.conv.Gamma2}) 
we anyway have to cut the diagram up into $3$ parts
(either by bounding $g$-edges or $\gamma^{(1)}_N$-edges) as described in \Cref{sec.abs.conv.Gamma2}.
We split the diagrams into two (exhaustive) groups:
\begin{enumerate}
\item Small diagrams with $1\leq k\leq 2, \nu = 0, \nu^* = 0$, and 
\item Large diagrams as the rest, i.e. with 
\begin{enumerate}[$(A)$]
\item $k\geq 3$, or 
\item $\nu + \nu^* \geq 1$.
\end{enumerate}
\end{enumerate}
As in \Cref{sec.proof.subleading.2.body.diagrams}, the splitting is motivated by counting powers in \Cref{eqn.bound.linked.sum.k.ng.cluster.3.body}. 
Note that for $p\geq 1$  the diagrams with $k=0, \nu = 0, \nu^* = 0$ are not present.
We  write 
\[
	\sum_{p\geq 1} \frac{1}{p!}\sum_{(\pi, G)\in\tilde \mcL_{p}^3} 
		\Gamma_{\pi, G}^{3}
		= \xi_{\textnormal{small}}^3 + \xi_{\textnormal{large}}^3,
\]
where $\xi_{\textnormal{small}}^3$ and $\xi_{\textnormal{large}}^3$ 
are the contributions of small and large diagrams respectively.
Exactly as in \Cref{eqn.bound.xi.geq1} we may bound, using \Cref{eqn.bound.linked.sum.k.ng.cluster.3.body}
\[
\begin{aligned}
  |\xi_{\textnormal{large}}^3| 
  	& 
  		\leq \underbrace{C \rho^3 (s(\log N)^3)^3(a^3\rho\log(b/a))^3}_{\textnormal{type $(A)$ diagrams}}
  		+ \underbrace{C\rho^3a^3\rho\log(b/a)}_{\textnormal{type $(B)$ diagrams}}
  	\\ & 
  		\leq 
  		C a^3\rho^{4} \log (b/a)
  			\left[
  			s^3  a^6 \rho^2 (\log (b/a))^2(\log N)^9 
  			+ 1 \right].
\end{aligned}
\]
For the small diagrams we have 
\begin{lemma}\label{lem.bound.3.body.small.diagrams}
We have 
\[
	|\xi_{\textnormal{small}}^3|
	\leq C a^3 \rho^4 \log(b/a)
\]
uniformly in $x_1,x_2,x_3$.
\end{lemma}
\noindent
As with \Cref{lem.Taylor.small.diagrams}, the proof is simply a computation, which we give in \Cref{sec.small.diagrams.3.particle}.
We conclude the desired.
\end{proof}

\section{One and two dimensions}\label{sec.lower.dimensions}
In this section we sketch the necessary changes one needs to make for the argument to apply in dimensions $d=1$ and $d=2$.
We will abuse notation slightly and denote by the same symbols as in \Cref{sec.gaudin.expansion,sec.preliminary.computations,sec.energy.in.box}
the relevant $1$- and $2$-dimensional analogues.

\subsection{Two dimensions}\label{sec.two.dimensions}
Similarly to the $3$-dimensional setting, the $p$-wave scattering function $f_0$ in $2$ dimensions 
is radial and 
solves the equation 
\begin{equation}\label{eqn.scat.d=2}
-\partial_r^2 f_0 - \frac{3}{r}\partial_r f_0 + \frac{1}{2}vf_0 = 0,
\end{equation}
see \Cref{sec.scattering.function} and recall \Cref{defn.scattering.length.d=2}.
Thus, it is the same as the $s$-wave scattering function in $4$ dimensions.
In particular it satisfies the bound
\begin{lemma}[{\cite[Lemma A.1]{Lieb.Yngvason.2001}, \Cref{prop.bound.f0.hc}}]
\label{prop.bound.f0.hc.d=2}
The scattering function 
satisfies $\left[1 - \frac{a^2}{|x|^2}\right]_+ \leq f_0(x) \leq 1$ for all $x$ and 
$|\nabla f_0(x)|\leq \frac{2a^2}{|x|^3}$ for $|x| > a$.
\end{lemma}

\noindent
As for the $3$-dimensional setting we consider the trial state 
\[
	\psi_N(x_1,\ldots,x_N) = \frac{1}{\sqrt{C_N}} \prod_{i<j} f(x_i-x_j) D_N(x_1,\ldots,x_N),
\]
where $f$ is a rescaled scattering function 
\[
	f(x) = \begin{cases}
	\frac{1}{1 - a^2/b^2} f_0(|x|) & |x|\leq b, \\ 1 & |x|\geq b
	\end{cases}
\]
and
\[
	D_N(x_1,\ldots,x_N) = \det[u_k(x_i)]_{\substack{1\leq i \leq N \\ k\in P_F}},
	\qquad 
	u_k(x) = \frac{1}{L} e^{ikx},
	\qquad 
	N = \# P_F.
\]
Here $P_F$ denotes the ``Fermi polygon'', 
the $2$-dimensional analogue of the ``Fermi polyhedron'' $P_F$.
It is defined as follows. (Compare to \Cref{defn.P_F.true}.)
\begin{defn}\label{defn.P_F.d=2}
The polygon $P$ is defined as follows. 
\begin{itemize}
\item
First pick $Q,$ satisfying 
\[
	Q^{-1/2} \leq Cs^{-2},
	\qquad 
	N^{3/2} \ll Q \leq CN^C
\]
in the limit $N\to \infty$.
(The exponents arise as $\frac{-1}{2} = \frac{-1}{2(d-1)}$, $-2 = \frac{-2}{d-1}$ and $\frac{3}{2} = \frac{d+1}{d}$ for $d=2$.)

\item 
Pick two distinct primes $Q_1,Q_2 \sim Q$.

\item 
Place $s$ \emph{evenly distributed} points $\kappa_1^\R,\ldots,\kappa_s^\R$ on the circle of radius $Q^{-1/2}$
such that the points are invariant under the symmetries $(k^1,k^2)\mapsto(\pm k^a,\pm k^b)$
for $\{a,b\} = \{1,2\}$.
(Here the exponent arises as $\frac{-1}{2} = \frac{1}{2(d-1)} - 1$ for $d=2$.)

\emph{Evenly distributed} means that the distance between any pair of points is $d\gtrsim s^{-1}Q^{-1/2}$ 
and that for any $k$ on the sphere of radius $Q^{-1/2}$ the distance from $k$ to the nearest point is $\lesssim s^{-1}Q^{-1/2}$.
(On the circle we can naturally order the points. Then the condition for being \emph{evenly distributed} reads that 
consecutive points are separated by a distance $\sim s^{-1}Q^{-1/2}$.)

\item 
Find now points $\kappa_1,\ldots,\kappa_s$
of the form 
\[
	\kappa_j = \left(\frac{p_j^1}{Q_1}, \frac{p_j^2}{Q_2}\right),
	\qquad p_j^\mu \in \Z, \quad \mu=1,2, \quad j=1,\ldots,s
\]
such that the points are invariant under the symmetries $(k^1,k^2)\mapsto(\pm k^1,\pm k^2)$
and such that for any $j=1,\ldots,s$ we have $\abs{\kappa_j - \kappa_j^{\R}}\lesssim Q^{-1}$.


\item 
Define $\tilde P$ as the convex hull of the points $\kappa_1,\ldots,\kappa_s$ and $P = \sigma \tilde P$ where $\sigma$ is such that $\Vol(P) = \pi$.

\item 
Define the centre $z = \sigma (1/Q_1, 1/Q_2)$.
\end{itemize}
\noindent
The ``Fermi polygon'' is the rescaled version defined as $P_F = k_F P \cap \frac{2\pi}{L}\Z^2$, 
where $L$ is chosen large (depending on $k_F$) such that $\frac{k_FL}{2\pi}$ is rational and large.
\end{defn}


\noindent
Similarly as in \Cref{rmk.comment.def.P_F} we have that $\sigma$ is irrational and $\sigma = Q^{1/2}(1 + O(s^{-2}))$.
In particular, any point on the boundary $\partial P$ has radial coordinate $1+O(s^{-2})$.
(The power of $s$ here comes from the circle being locally quadratic and the distance between close points being $\sim s^{-1}$.
Compare to \Cref{rmk.comment.def.P_F}.)
Moreover, $P_F$ is almost symmetric under the map $(k^1,k^2)\mapsto (k^2,k^1)$ similarly to \Cref{lem.preserve.symmetry}.
\begin{lemma}\label{lem.preserve.symmetry.d=2}
Let $F_{12}$ be the map $(k^1,k^2)\mapsto (k^2,k^1)$.
For any function $t\geq 0$ we have 
\[
	\sum_{k\in \frac{2\pi}{L}\Z^2} \abs{\chi_{(k\in P_F)} - \chi_{(k\in F_{12}(P_F))}} t(k)
	\lesssim Q^{-1/2} N \sup_{|k| \sim k_F} t(k)
	\lesssim N^{1/4} \sup_{|k| \sim k_F} t(k),
\]
where $Q$ is as in \Cref{defn.P_F.d=2}.
\end{lemma}



\noindent
The analogue of \Cref{lem.lebesgue.constant.fermi.polyhedron} is then 
\begin{lemma}\label{lem.lebesgue.constant.fermi.d=2}
The Lebesgue constant of the Fermi polygon satisfies
\[
\int_{\Lambda} \frac{1}{L^2} \abs{\sum_{k\in P_F} e^{ikx}} \ud x
= \frac{1}{(2\pi)^2} \int_{[0,2\pi]^2} \abs{\sum_{q \in \left(\frac{Lk_F}{2\pi}P\right)\cap \Z^2} e^{iqu}} \ud u \leq C s (\log N)^2.
\]
\end{lemma}

\noindent
We can again compute the kinetic energy of the Slater determinant analogously to \Cref{lem.KE.polyhedron} and its $2$-particle reduced density analogously to \Cref{prop.2.density}.
\begin{lemma}\label{lem.KE.polygon.d=2}
The kinetic energy of the (Slater determinant with momenta in the) Fermi polygon satisfies 
\[
	\sum_{k\in P_F} |k|^2 = \sum_{k\in B_F} |k|^2 \left(1 + O(N^{-1/2}) + O(s^{-4})\right)
    = 2\pi \rho N\left(1 + O(N^{-1/2}) + O(s^{-4})\right).
\]
\end{lemma}
\begin{lemma}\label{prop.2.density.d=2}
The $2$-particle reduced density of the (normalized) Slater determinant satisfies 
\[
\rho^{(2)}(x_1,x_2) 
  = \pi \rho^{3} |x_1-x_2|^2\left(1 + O(N^{-1/2}) + O(s^{-4}) + O( \rho|x_1 - x_2|^2)\right).
\]
\end{lemma}


\noindent
The computations in \Cref{sec.gaudin.expansion} make no reference to the dimension and are thus also valid in dimension $d=2$.
For the absolute convergence in \Cref{sec.abs.conv} and \Cref{lem.linked.sum.abs.conv} 
one should simply replace occurrences of $g$ and $\gamma^{(1)}_N$
with their $2$-dimensional analogues. 
Here we have the bounds (using \Cref{prop.bound.f0.hc.d=2,lem.lebesgue.constant.fermi.d=2})
\begin{equation}\label{eqn.2d.bound.int.g}
\begin{aligned}
	\int_\Lambda |g| 
		& \lesssim a^2 + \frac{1}{(1 - a^2/b^2)^2}\int_a^b \left[\left(1 - \frac{a^2}{b^2}\right)^2 - \left(1 - \frac{a^2}{x^2}\right)^2\right] x \ud x 
	\lesssim a^2 \log (b/a),
	\\
	\int_\Lambda |\gamma^{(1)}_N| 
	& \lesssim s(\log N)^2.
\end{aligned}
\end{equation}
Thus the absolute convergence holds as long as $s a^2\rho \log(b/a) (\log N)^2$ is sufficiently small.
That is, the analogue of \Cref{thm.gaudin.expansion} reads
\begin{thm}\label{thm.gaudin.expansion.d=2}
There exists a constant $c > 0$ such that if $s a^2\rho \log (b/a)(\log N)^2 < c$, then
the formulas in \Cref{eqn.thm.gaudin} hold 
(with $\rho_{\textnormal{Jas}}^{(n)}$ and $\Gamma_{\pi,G}^n$ interpreted as appropriate in the two-dimensional setting).
\end{thm}

\noindent
The analogues of \Cref{lem.derivative.sum.linked.diagrams,lem.bound.3.body.term} read
\begin{lemma}\label{lem.derivative.sum.linked.diagrams.d=2}
There exist constants $c,C > 0$ such that if $sa^2\rho \log (b/a)(\log N)^2 < c$ and $N = \#P_F > C$, then 
\begingroup
\allowdisplaybreaks
\begin{align*}
	\abs{\sum_{p= 1}^\infty \frac{1}{p!} \sum_{(\pi, G)\in\mcL_p^2} \Gamma_{\pi, G}^{2}}
	&
		\leq 
	Ca^4 \rho^4 (\log (b/a))^2 
		\Bigl[ 
			s^3 a^4 \rho^2 (\log b/a)^2 (\log N)^6 
			+ 1
		\Bigr]
	\\ & \quad 
	+ Ca^2 \rho^{4} |x_1-x_2|^2 
		\Bigl[
			s^5 a^{8}\rho^4 (\log(b/a))^5  (\log N)^{11}
      +b^2\rho
			+ \log(b/a)
		\Bigr],
\intertext{and}
\rho^{(3)}_{\textnormal{Jas}} 
	& \leq C f_{12}^2f_{13}^2f_{23}^2
		\Bigl[
    \rho^{5}|x_1-x_2|^2|x_2-x_3|^2
	\\ & \qquad 
		+ a^2\rho^{4} \log (b/a)
				\left[
				s^3  a^4 \rho^2 (\log (b/a))^2(\log N)^6 
				+ 1 
				\right]
		\Bigr].
\end{align*}
\endgroup
\end{lemma}


\noindent
The proof is again similar to the $3$-dimensional case 
replacing the bounds on $\int |g|$ and $\int |\gamma^{(1)}_N|$ as in \Cref{eqn.2d.bound.int.g} above. 
Apart from this, there are two main changes.
The first is in the proof of the analogue of \Cref{lem.small.diagrams.type.C}, namely \Cref{eqn.bound.int.g.|x|2},
where one bounds $\int |x|^2 |1-f^2|$. 
In two dimensions this bound is, using \Cref{prop.bound.f0.hc.d=2},
\[
	\int_{\R^2}\left(1 - f(x)^2\right)|x|^2 \ud x
	\leq Ca^4 + \frac{C}{(1-a^2/b^2)^2} \int_a^b \left[\left(1 - \frac{a^2}{b^2}\right)^2 - \left(1 - \frac{a^2}{r^2}\right)^2 \right] r^3 \ud r
	\leq C a^2 b^2.
\]
The other main difference is for the analogue of \Cref{lem.derivative.lebesgue.constant}.
Here the $2$-dimensional analogue reads 
\begin{lemma}\label{lem.derivative.lebesgue.constant.d=2}
The polygon $P$ from \Cref{defn.P_F.d=2} satisfies for any $\mu,\nu=1,2$ that 
\[
\begin{aligned}
	\int_{[0,2\pi]^2} \abs{\sum_{q\in \left(\frac{Lk_F}{2\pi} P\right)\cap \Z^2} q^\mu e^{iqu}} \ud u
	& \leq C s N^{1/2} (\log N)^2, 
	\\ 
	\int_{[0,2\pi]^2} \abs{\sum_{q\in \left(\frac{Lk_F}{2\pi} P\right)\cap \Z^2} q^\mu q^\nu e^{iqu}} \ud u
	& \leq C s N (\log N)^3
\end{aligned}
\]
for sufficiently large $N$.
\end{lemma}
\noindent
The proof is similar to that of \Cref{lem.derivative.lebesgue.constant} given in \Cref{sec.derivative.lebesgue.constant}
only one skips \Cref{sec.induction.d=3} and notes that $R = \frac{Lk_F}{2\pi}\sim N^{1/2}$.

Putting together the formulas in
\Cref{thm.gaudin.expansion.d=2}
with the bounds in
\Cref{prop.2.density.d=2,lem.KE.polygon.d=2,lem.derivative.sum.linked.diagrams.d=2} 
we  easily find the analogue of \Cref{eqn.energy.w.all.errors}.
We then need to bound a few terms. 
Following the type of arguments of \Cref{sec.energy.in.box}, 
namely \Cref{eqn.integral.v.scat.len,eqn.integral.v.zero.moment,eqn.integral.v.moments,eqn.integral.fdf,eqn.integral.fdf.moments} 
and using \Cref{prop.bound.f0.hc.d=2}
we get the bounds 
\[
\begin{aligned}
\int \left(|\nabla f(x)|^2 + \frac{1}{2}v(x) f(x)^2\right) |x|^n \ud x 
	& \leq \begin{cases}
	C, & n=0, 
	\\
	4\pi a^2 + O(a^4b^{-2}), & n=2, 
	\\
	Ca^4 \log(b/a) + CR_0^2a^2, & n=4,
	\end{cases}
\\
\int |x|^n f \partial_r f \ud x 
	& \leq \begin{cases}
	Ca, & n=0, 
	\\
	Ca^2b, & n= 2.
	\end{cases}
\end{aligned}
\]
Plugging this into the analogue of \Cref{eqn.energy.w.all.errors} we get the analogue of \Cref{eqn.energy.density.w.all.errors},
\begin{equation}
\label{eqn.energy.density.w.all.errors.d=2}
\begin{aligned}
\frac{\longip{\psi_N}{H_N}{\psi_N}}{L^2}
  & = 2\pi\rho^2 + 4\pi^2a^2\rho^{3} 
	\\ & \quad 
		+ O\left(s^{-4}\rho^{2}\right) 
		+ O\left(N^{-1/2}\rho^{2}\right) 
	\\ & \quad
		+ O\left(a^4b^{-2}\rho^{3}\right)
		+ O\left(a^4\rho^4 \log(b/a)\right)
		+ O\left(R_0^2a^2 \rho^{4}\right)
	\\ & \quad
		+ O\left( 
			a^4 \rho^4 (\log (b/a))^2 
		\Bigl[ 
			s^3 a^4 \rho^2 (\log b/a)^2 (\log N)^6 
			+ 1
		\Bigr]\right)
	\\ & \quad  
	+ O \left(a^4 \rho^{4} 
				\Bigl[
					s^5 a^{8}\rho^4 (\log(b/a))^5  (\log N)^{11}
          + b^2\rho
					+ \log(b/a)
				\Bigr]\right)
	\\ & \quad 
  + O\left(a^4b^2\rho^5\right)
	+ O\left(a^4\rho^4 \log(b/a)\left[s^3 a^4\rho^2(\log (b/a))^3(\log N)^6 + 1\right]\right).
\end{aligned}
\end{equation}
As above, we can choose $L\sim a(a^2\rho)^{-10}$
still ensuring that $\frac{Lk_F}{2\pi}$ is rational. 
(More precisely one chooses $L\sim a (k_F a)^{-20}$, since $\rho$ is defined in terms of $L$.)
Then $N \sim (a^2\rho)^{-19}$. 
Choose moreover
\[
	b = a(a^2\rho)^{-\beta}, \qquad s \sim (a^2\rho)^{-\alpha}|\log (a^2\rho )|^{-\gamma}.
\]
Optimising in  $\alpha,\beta, \gamma$ we see that for the choice
\[
	\beta = \frac{1}{2}, \qquad \alpha = \frac{4}{7}, \qquad \gamma = \frac{10}{7}
\]
we have
\begin{equation}
\label{eqn.energy.density.upper.bound.single.state.d=2}
\begin{aligned}
\frac{\longip{\psi_N}{H_N}{\psi_N}}{L^3}
  & = 2\pi\rho^2 + 4\pi^2a^2\rho^{3} 
		\biggl[
		1
		+ O\left(a^2\rho |\log(a^2\rho)|^2\right)
		\biggr]
\end{aligned}
\end{equation}
for $a^2\rho$ small enough.
Note that for this choice of $s,N$ we have $s\sim N^{4/133} (\log N)^{10/7}$. 
Thus any $Q$ with $N^{3/2} \ll Q \leq CN^C$ 
satisfies the condition $Q^{-1/2}\leq C s^{-2}$ of \Cref{defn.P_F.d=2}.

The extension to the thermodynamic limit of \Cref{sec.box.method} is readily generalized. We thus conclude the proof of \Cref{thm.main.d=2}.

\subsection{One dimension}\label{sec.one.dimension}
Similarly to the $2$- and $3$-dimensional settings, the $p$-wave scattering function $f_0$ in $1$ dimension 
is even 
and solves the equation (here $\partial^2$ denotes the second derivative)
\begin{equation}\label{eqn.scat.d=1}
-\partial^2 f_0 - \frac{2}{r}\partial f_0 + \frac{1}{2}vf_0 = 0,
\end{equation}
see \Cref{sec.scattering.function} and recall \Cref{defn.scattering.length.d=1}.
Thus, it is the same as the $s$-wave scattering function in $3$ dimensions.
In particular it satisfies the bound
\begin{lemma}[{\cite[Lemma A.1]{Lieb.Yngvason.2001}, \Cref{prop.bound.f0.hc}}]
\label{prop.bound.f0.hc.d=1}
The scattering function 
satisfies $\left[1 - \frac{a}{|x|}\right]_+ \leq f_0(x) \leq 1$ for all $x$ and 
$|\partial f_0(x)|\leq \frac{a}{|x|^2}$ for $|x| > a$.
\end{lemma}

\noindent
Before giving the proof of \Cref{thm.main.d=1} we first compare our definition of the scattering length to that of \cite{Agerskov.Reuvers.ea.2022}.
In \cite{Agerskov.Reuvers.ea.2022} the following definition is given.
\begin{defn}[{\cite[Section 1.3]{Agerskov.Reuvers.ea.2022}}]\label{def.odd.wave.scat.len}
The \emph{odd-wave scattering length} $a_{\textnormal{odd}}$ is given by 
\[
	\frac{4}{R-a_{\textnormal{odd}}} 
		= \inf  \left\{\int_{-R}^R \left(2 |\partial h|^2 + v |h|^2\right) \ud x : h(R) = - h(-R) = 1\right\}
\]
for any $R > R_0$, the range of $v$. 
\end{defn}
\noindent
The value of $a_{\textnormal{odd}}$ is independent of $R > R_0$ so $a_{\textnormal{odd}}$ is well-defined.
We claim that 
\begin{prop}\label{prop.odd.wave.=.p.wave}
The \emph{$p$-wave scattering length} $a$ defined in \Cref{defn.scattering.length.d=1} and the 
\emph{odd-wave scattering length} $a_{\textnormal{odd}}$ defined in \Cref{def.odd.wave.scat.len} agree, i.e. $a = a_{\textnormal{odd}}$.
\end{prop}
\begin{proof}
Note first that $h\mapsto \mcE(h)= \int_{-R}^R \left(2 |\partial h|^2 + v |h|^2\right) \ud x$ is convex, so 
by replacing $h$ by $(h(x)-h(-x))/2$ we can only lower its value. Thus, we have
\[
	\frac{4}{R-a_{\textnormal{odd}}} 
		= \inf  \left\{\int_{-R}^R \left(2|\partial h|^2 + v |h|^2\right) \ud x : h(x) = - h(-x), \, h(R) = 1\right\}.
\]
Any $h$ we write as $h(x) = \frac{xf(x)}{R}$. Using this and integration by parts we get 
\[
\begin{aligned}
	& \frac{4}{R-a_{\textnormal{odd}}} 
	\\ & \quad 
   = \frac{1}{R^2}
		\inf  \left\{\int_{-R}^R \left(2 |f|^2 + 4 x f\partial f + 2 |x|^2 |\partial f|^2 +  v |f|^2|x|^2\right) \ud x : 
		f(x) = f(-x), \, f(R) = 1\right\}
	\\
	& \quad = 
	\frac{4}{R} + \frac{2}{R^2} \inf \left\{\int_{-R}^R \left(|\partial f|^2 +  \frac{1}{2} v |f|^2\right)|x|^2 \ud x : f(x) = f(-x), \, f(R) = 1\right\}.
\end{aligned}
\]
That is, 
\[
	2R\left(\frac{1}{1-a_{\textnormal{odd}}/R}  - 1\right)
	 = 
	\inf \left\{\int_{-R}^R \left(|\partial f|^2 +  \frac{1}{2} v |f|^2\right)|x|^2 \ud x : f(x) = f(-x), \, f(R) = 1\right\}.	
\]
Taking $R\to \infty$ in this we recover the definition of $a$.
We conclude that $a = a_\textnormal{odd}$.
\end{proof}

\noindent
Concerning the assumption on $v$ that $\int \left(\frac{1}{2}vf_0^2 + |\partial f_0|^2\right) \ud x < \infty$ we have the following two propositions.
\begin{prop}\label{prop.h.c.implies.int.vf2}
Suppose that $v\geq 0$ is even and compactly supported and that for some interval $[x_1,x_2]$, $0\leq x_1 < x_2$ we have 
$v(x) = \infty$ for $x_1\leq x \leq x_2$.
Then $\int \left(\frac{1}{2}vf_0^2 + |\partial f_0|^2\right) \ud x < \infty$, 
where $f_0$ denotes the $p$-wave scattering function.
\end{prop}
\begin{proof}
Let $[x_1,x_2]$ be an interval where $v(x) = \infty$ for $x_1\leq x \leq x_2$
and note that $f_0(x)=0$ for all $|x|\leq x_2$.
Then we have 
\[
\int \left(\frac{1}{2}v f_0^2 + |\partial f_0|^2\right) \ud x 
\leq \frac{1}{x_2^2} \int_{|x|\geq x_2} \left(\frac{1}{2}v f_0^2 + |\partial f_0|^2\right) |x|^2 \ud x
 = 2 a x_2^{-2} < \infty.
 \qedhere
\]
\end{proof}
\begin{prop}\label{prop.smooth.implies.int.vf2}
Suppose that $v\geq 0$ is even, compactly supported and smooth. 
Then $\int \left(\frac{1}{2}vf_0^2 + |\partial f_0|^2\right) \ud x < \infty$, 
where $f_0$ denotes the $p$-wave scattering function.
\end{prop}
\begin{proof}
For smooth $v$ also the scattering function $f_0$ is smooth. 
Recall the scattering equation \eqref{eqn.scat.d=1}.
Then a simple calculation using integration by parts shows that
\[
\int \left(\frac{1}{2}vf_0^2 + |\partial f_0|^2\right) \ud x
= 2\int_0^\infty \left( f_0 \partial^2 f_0 + \frac{2f_0\partial f_0}{x} + (\partial f_0)^2 \right) \ud x
= 2\int_0^\infty \frac{f_0(x)^2 - f_0(0)^2}{x^2} \ud x.
\]
The function $f_0$ is smooth and even. Thus for small $x$ we have $f(x) = f(0) + O(|x|^2)$, hence the integral converges around $0$. 
By the decay of $\frac{1}{x^2}$ the integral converges at $\infty$.
We conclude the desired.
\end{proof}

\noindent
We now give the proof of \Cref{thm.main.d=1}. We consider the trial state given in \Cref{eqn.define.trial.state}
where $f$ is a rescaled scattering function 
\[
	f(x) = \begin{cases}
	\frac{1}{1 - a/b} f_0(|x|) & |x|\leq b, \\ 1 & |x|\geq b
	\end{cases}
\]
and
\[
	D_N(x_1,\ldots,x_N) = \det[u_k(x_i)]_{\substack{1\leq i \leq N \\ k\in B_F}},
	\qquad 
	u_k(x) = \frac{1}{L^{1/2}} e^{ikx},
	\qquad 
	N = \# B_F.
\]
In $1$ dimension, there is no difference between a ball and a polyhedron, 
so we may use the Fermi ball $B_F = \{k\in \frac{2\pi}{L}\Z : |k|\leq k_F\}$ for the momenta in the Slater determinant.
In this case  we have (see \cite[Lemma 3.2]{Kolomoitsev.Lomako.2018} or \Cref{lem.eval.1d.integrals})
\begin{lemma}\label{lem.lebesgue.constant.fermi.d=1}
The Lebesgue constant of the Fermi ball satisfies
\[
\int_{-L/2}^{L/2} \frac{1}{L} \abs{\sum_{k\in B_F} e^{ikx}} \ud x
= \frac{1}{2\pi} \int_{0}^{2\pi} \abs{\sum_{q \in \left(B\left(\frac{Lk_F}{2\pi}\right)\right)\cap \Z^2} e^{iqu}} \ud u \leq C \log N.
\]
\end{lemma}

\noindent
As for the $2$-dimensional setting one easily generalizes the computation of the kinetic energy in \Cref{lem.KE.polyhedron} 
and the calculation of the $2$-particle reduced density for a Slater determinant in  \Cref{prop.2.density}.
That is, 
\begin{lemma}\label{lem.KE.polygon.d=1}
The kinetic energy of the (Slater determinant with momenta in the) Fermi ball satisfies 
\[
	\sum_{k\in B_F} |k|^2 =  \frac{\pi^2}{3} \rho^2 N\left(1 + O(N^{-1})\right).
\]
\end{lemma}
\begin{lemma}\label{prop.2.density.d=1}
The $2$-particle reduced density of the (normalized) Slater determinant satisfies 
\[
\rho^{(2)}(x_1,x_2) = \frac{\pi^2}{3} \rho^{4} |x_1-x_2|^2\left(1 + O(N^{-1}) + O( \rho^2|x_1 - x_2|^2)\right).
\]
\end{lemma}

\noindent
For the Gaudin-Gillespie-Ripka-expansion 
we 
replace occurrences of $g$ and $\gamma^{(1)}_N$
with their $1$-dimensional analogues as for the $2$-dimensional setting. 
Here we have the bounds (using \Cref{prop.bound.f0.hc.d=1,lem.lebesgue.constant.fermi.d=1})
\begin{equation}\label{eqn.1d.bound.int.g}
	\int_\Lambda |g| 
		\lesssim a \log (b/a),
		\qquad 
	\int_\Lambda |\gamma^{(1)}_N| 
	\lesssim \log N.
\end{equation}
Then, the $1$-dimensional analogue of \Cref{thm.gaudin.expansion} reads
\begin{thm}\label{thm.gaudin.expansion.d=1}
There exists a constant $c > 0$ such that if $a\rho \log (b/a)\log N < c$, then
the formulas in \Cref{eqn.thm.gaudin} hold 
(with $\rho_{\textnormal{Jas}}^{(n)}$ and $\Gamma_{\pi,G}^{n}$ interpreted as appropriate for the $1$-dimensional setting.)
\end{thm}

\noindent
For the analogues of \Cref{lem.derivative.sum.linked.diagrams,lem.bound.3.body.term} we have to a bit more careful. 
In order to get errors smaller than the desired accuracy of the leading interaction term (of order $a\rho^4$ for the energy density)
we need to also do a Taylor expansion of (some of) the $3$-particle diagrams. 
(Pointwise we only have the bound 
$\abs{\Gamma^3_{\pi,G}} \leq C a \rho^4 \log(b/a) \log N$ (see \Cref{sec.proof.subleading.3.body.diagrams,sec.small.diagrams.3.particle})
for any subleading diagram $(\pi,G)$, i.e. for $(\pi,G)\in \tilde\mcL_p^3$ with $p\geq 1$.)

\begin{remark}[{Why this was not a problem for dimensions $d=2,3$}]
In dimensions $d=1,2,3$ the analoguous bound reads $\abs{\Gamma^3_{\pi,G}} \leq C s a^d \rho^4 \log(b/a) (\log N)^d$
(if $d=1$ then there is no $s$)
for any subleading diagram, see \Cref{eqn.bound.linked.sum.k.ng.cluster.3.body}.
This bound should be compared to the energy density of the leading interaction term of order $a^d \rho^{2 + 2/d}$.
Considering just the power of $\rho$, we see that 
such terms are subleading compared to the interaction term for $d\ne 1$.
\end{remark}



\noindent
Similarly the argument for $\Gamma^2$ is also slightly different compared to that of \Cref{lem.derivative.sum.linked.diagrams}.
We have the bounds 
\begin{lemma}\label{lem.derivative.sum.linked.diagrams.d=1}
There exists a constant $c > 0$ such that if $a\rho \log (b/a)\log N < c$, then 
\begingroup
\allowdisplaybreaks
\[
\begin{aligned}
  \abs{\sum_{p= 1}^\infty \frac{1}{p!} \sum_{(\pi, G)\in\mcL_p^2} \Gamma_{\pi, G}^{2}}
	& 
		\leq 
    C a^2 \rho^4 \left[ a\rho (\log(b/a))^3 (\log N)^3 + b^4\rho^4\right]
	\\ & \quad 
		+ 
    C a \rho^5 |x_1-x_2|^2 \left[b^2\rho^2 + Nab^4\rho^5 + \log(b/a)\right]
\end{aligned}
\]
and
\[
\begin{aligned}
\rho^{(3)}_{\textnormal{Jas}} 
	& 
		\leq C f_{12}^2f_{13}^2f_{23}^2
		\Bigl[
    \rho^{7}|x_1-x_2|^2 |x_2-x_3|^2
		+ a^2 \rho^5 (\log(b/a))^2 (\log N)^2
	\\ & \qquad 
    + a \rho^6 \left((b^2\rho^2 + \log(b/a))\right)\left[|x_1-x_2|^2 + |x_1-x_3|^2 + |x_2-x_3|^2\right]
		\Bigr].
\end{aligned}
\]
\endgroup
\end{lemma}
\noindent
The proof is similar to that of \Cref{lem.derivative.sum.linked.diagrams,lem.bound.3.body.term}.
We postpone it to the end of this section.
Note here that the $N$-dependence is not just via logarithmic factors.
Thus, we need to be more careful in choosing the size of the smaller boxes when applying the box method arguments of \Cref{sec.box.method}.
With this we get  the analogue of \Cref{eqn.energy.w.all.errors} in $1$ dimension,
\begin{equation}\label{eqn.energy.w.all.errors.d=1}
\begin{aligned}
& \longip{\psi_N}{H_N}{\psi_N}
	\\ & \quad = \frac{\pi^2}{3} \rho^{2} N\left(1 + O(N^{-1})\right) 
		+ L\int \ud x
		\left(
			|\partial f(x)|^2 + \frac{1}{2}v(x) f(x)^2
		\right) 
	\\ & \qquad \quad \times 
		\Bigg[
		 \frac{\pi^2}{3} \rho^{4} |x|^2
	\left(
	1 + O(N^{-1}) + O(\rho^{2}|x|^2)
	\right)
	+ O \left(a \rho^{5} |x|^2 
				\Bigl[
          b^2\rho^2 
          + N a b^4 \rho^5
					+ \log(b/a)
				\Bigr]\right)
	\\ & \qquad \qquad 
			+ O\left( 
			a^2 \rho^4
		\Bigl[ 
			a\rho  (\log b/a)^3 (\log N)^3 
      + b^4\rho^4
		\Bigr]\right)
		\Bigg]
	\\ & \qquad 
		+ \iiint \ud x_1 \ud x_2 \ud x_3 \, 
		f_{12}\partial f_{12} f_{23}\partial f_{23} f_{13}^2 
	\\ & \qquad \quad \times
		\Biggl[
    O(\rho^7 |x_1-x_2|^2|x_2-x_3|^2)
		+ O\left(a^2\rho^{5} (\log (b/a))^2 (\log N)^2
			\right)
	\\ & \qquad  \qquad
		+ O\left(
      a\rho^6 \left[b^2\rho^2 + \log(b/a)\right]\left[|x_1-x_2|^2 + |x_1-x_3|^2 + |x_2-x_3|^2\right]\right)
		\Biggr].
\end{aligned}
\end{equation}
For the $2$-body error terms we may follow the type of arguments of \Cref{sec.energy.in.box}, 
namely \Cref{eqn.integral.v.scat.len,eqn.integral.v.zero.moment,eqn.integral.v.moments,eqn.integral.fdf,eqn.integral.fdf.moments} 
exactly as for the $2$-dimensional case.
By using \Cref{prop.bound.f0.hc.d=1} we  get the bounds 
\[
\begin{aligned}
\int \left(|\partial f(x)|^2 + \frac{1}{2}v(x) f(x)^2\right) |x|^n \ud x 
	& \leq \begin{cases}
	2a, & n=2, 
	\\
	Ca^2b, & n=4,
	\end{cases}
\\
\int |x|^n f \partial f \ud x 
	& \leq \begin{cases}
	C, & n=0, 
	\\
	Ca\log(b/a), & n=1, 
	\\
	Cab, & n= 2.
	\end{cases}
\end{aligned}
\]
Define $a_0$ by 
\[	
	\frac{1}{2a_0} = \int \left(|\partial f(x)|^2 + \frac{1}{2}v(x) f(x)^2\right) \ud x
\]
and recall by assumption on $v$ that $a_0 > 0$, i.e. that $1/a_0 < \infty$. 
For the $3$-body terms we may do as for the $3$-dimensional case, \Cref{sec.energy.in.box}. 
For the first term we bound 
$f_{13} \leq 1$.
By the translation invariance one integration gives a volume (i.e. length) factor $L$. 
That is,
\[
\begin{aligned}
& \iiint \ud x_1 \ud x_2 \ud x_3 \, 
		\abs{f_{12}\partial f_{12} f_{23}\partial f_{23} f_{13}^2}
	\\ &  \qquad \times
		\Biggl[
    O(\rho^7 |x_1-x_2|^2|x_2-x_3|^2)
		+ O\left(a^2\rho^{5} (\log (b/a))^2 (\log N)^2
			\right)
	\\ & \qquad \quad 
		+ O\left(
      a\rho^6 \left[b^2\rho^2 + \log(b/a)\right]\left[|x_1-x_2|^2 + |x_1-x_3|^2 + |x_2-x_3|^2\right]\right)
		\Biggr]
	\\ & \quad 
		\leq 
    C N \rho^{6} \left(\int_0^b |x|^2 f \partial f \ud x\right)^2
		+ CNa^2 \rho^4 (\log(b/a))^2 (\log N)^2 \left(\int_0^b f \partial f\right)^2 
	\\ & \qquad 
    + C N a \rho^5 \left[b^2 \rho^2 + \log(b/a)\right] 
    \left[
    \left(\int_0^b |x|^2 f \partial f \ud x\right) \left(\int_0^b f \partial f \ud x\right)
    + \left(\int_0^b |x| f \partial f \ud x\right)^2
    \right].
  \\ & \quad \leq 
  CN a^2 \rho^4 \left[b^2 \rho^2 + (\log(b/a))^2 (\log N)^2 + b\rho \left[b^2 \rho^2 + \log(b/a)\right]\right].
\end{aligned}
\]
We conclude the analogue of \Cref{eqn.energy.density.w.all.errors} in dimension $1$
\begin{equation}
\label{eqn.energy.density.w.all.errors.d=1}
\begin{aligned}
\frac{\longip{\psi_N}{H_N}{\psi_N}}{L}
	& 
	= \frac{\pi^2}{3}\rho^{3} + \frac{2\pi^2}{3}a\rho^{4} 
		+ O\left(N^{-1}\rho^{3}\right) 
		+ O\left(a^2b^{-1}\rho^{4}\right)
		+ O\left(a^2 b \rho^6\right)
	\\ & \quad
		+ O\left( 
      a^2 \rho^4 a_0^{-1} \left[a \rho (\log (b/a))^3 (\log N)^3 + b^4\rho^4\right]
			\right)
	\\ & \quad  
	+ O \left(a^2\rho^5
				\left[
          b^2\rho^{2}
          + N a b^4\rho^5
					+ \log(b/a)
				\right]\right)
	\\ & \quad 
	+ O\left(
  a^2 \rho^5 \left[b^2 \rho^2 + (\log(b/a))^2 (\log N)^2 + b\rho \left[b^2 \rho^2 + \log(b/a)\right]\right]
	\right).
\end{aligned}
\end{equation}
We need to be careful how we choose $N$ (i.e. how we choose $L$), 
since the error depends on $N$ not just via logarithmic terms.
We choose 
\[
	N = (a\rho)^{-\alpha}, \quad \alpha > 1
	\qquad 
	b = a(a\rho)^{-\beta}, \quad 0 < \beta < 1
\]
where the bounds on $\alpha, \beta$ are immediate for all the error-terms to be smaller than the desired accuracy 
(there is similarly also an upper limit for $\alpha$, which we do not write). 
Keeping then only the leading error terms we get 
\begin{equation}
\label{eqn.energy.density.w.all.errors.d=1.reduced}
\begin{aligned}
\frac{\longip{\psi_N}{H_N}{\psi_N}}{L}
	&  
	= \frac{\pi^2}{3}\rho^{3} + \frac{2\pi^2}{3}a\rho^{4} 
		+ O\left(N^{-1}\rho^{3}\right) 
		+ O\left(a^2b^{-1}\rho^{4}\right)
		+ O\left( 
      a^2 a_0^{-1} b^2 \rho^6
			\right) 
  + O \left(N a^3 b^4 \rho^{10} \right).
\end{aligned}
\end{equation}
Using the box method similarly as in \Cref{sec.box.method} we also have to be careful with how we choose the parameter $d$. 
As in \Cref{eqn.energy.density.infinite.box.trial.state} we get 
\[
\begin{aligned}
	e(\tilde\rho) 
	& \leq \frac{\longip{\psi_n}{H_n}{\psi_n}}{\ell} \left[1 + O(d/\ell) + O(b/\ell)\right] + O\left(\rho d^{-2}\right)
	\\
	& \leq \frac{\pi^2}{3}\rho^{3} + \frac{2\pi^2}{3}a\rho^{4} 
		+ O\left(n^{-1}\rho^{3}\right) 
		+ O\left(a^2b^{-1}\rho^{4}\right)
		+ O\left( 
      a^2 a_0^{-1} b^2 \rho^6
			\right) 
  + O \left(n a^3 b^4 \rho^{10} \right)
	\\ & \quad
	+ O\left(d \ell^{-1} \rho^3\right)
	+ O\left(b \ell^{-1} \rho^3\right)
	+ O\left(\rho d^{-2}\right).
\end{aligned}
\] 
Here we change notation from $N$ to $n$ and choose $d = a(a\rho)^{-\delta}$.
To get the error smaller than desired, we see that we need to choose $\delta > 3/2$. 
In particular then 
the error is $O(\rho^3 (a\rho)^{\gamma})$, where 
\[
  \gamma = \min\{1+\beta, 5-4\beta, 7-\alpha-4\beta, \alpha + 1 - \delta, 2\delta - 2\}.
\]
Then, also $\tilde\rho = \rho\left(1 + O( (a\rho)^\gamma)\right)$ so $\rho = \tilde \rho(1 + O( (a\tilde\rho)^\gamma))$.
Optimising in $\alpha,\beta,\delta$ we see that for 
\begin{equation}\label{eqn.optimal.choices.d=1}
  \alpha = \frac{33}{13},
  \qquad 
  \beta = \frac{9}{13},
  \qquad 
  \delta = \frac{24}{13}
\end{equation}
we get $\gamma = 22/13$, i.e.
\[
\begin{aligned}
	e(\tilde\rho) 
	& \leq \frac{\pi^2}{3}\tilde\rho^{3} + \frac{2\pi^2}{3}a\tilde\rho^{4}\left(1 + O\left( (a\tilde\rho)^{9/13}\right)\right).
\end{aligned}
\] 
This concludes the proof of \Cref{thm.main.d=1}.

It remains to give the
\begin{proof}[{Proof of \Cref{lem.derivative.sum.linked.diagrams.d=1}}]
Note first that, 
completely analogously to \Cref{eqn.bound.linked.sum.k.ng.cluster,eqn.bound.linked.sum.k.ng.cluster.3.body}, we have 
\begin{equation}\label{eqn.bound.linked.sum.k.ng.cluster.d=1}
\begin{aligned}
  \frac{1}{p!}
		\abs{\sum_{\substack{(\pi,G)\in \mcL_p^2 \\ \nu(\pi,G) + \nu^*(\pi,G) = \nu_0 \\ k(\pi,G) = k_0}} \Gamma_{\pi,G}^2}
	& \leq C\rho^2 (C \log N )^{k_0} (C a \rho \log (b/a))^{k_0 + \nu_0},
  && p = 2k_0 + \nu_0
	\\
  \frac{1}{p!}
		\abs{\sum_{\substack{(\pi,G)\in \tilde\mcL_p^3 \\ \nu(\pi,G) + \nu^*(\pi,G) = \nu_0 \\ k(\pi,G) = k_0}} \Gamma_{\pi,G}^3}
	& \leq C\rho^3 (C \log N )^{k_0} (C a \rho \log (b/a))^{k_0 + \nu_0},
  && p = 2k_0 + \nu_0.
\end{aligned}
\end{equation}
We will  use this to split the diagrams of $\mcL_p^2$ and $\tilde\mcL_p^3$ into groups.
We split diagrams in $\mcL_p^2$ into three (exhaustive) groups:
\begin{enumerate}
\item Small diagrams with $1\leq k + \nu + \nu^* \leq 2$, $\{1\}$ and $\{2\}$ in different clusters 
\begin{enumerate}[$(A)$]
\item and $k\geq 1$,
\item and $k=0, \nu^* = 1$.
\end{enumerate}
\item Small diagrams with $1\leq k + \nu + \nu^* \leq 2$ and
\begin{enumerate}[$(A)$]
\item $\{1\}$ and $\{2\}$ in different clusters and $k=0, \nu^* = 2$,
\item $\{1\}$ and $\{2\}$ in the same cluster, 
\end{enumerate} 
\item Large diagrams with $k + \nu + \nu^* \geq  3$. 
\end{enumerate}
We then split 
\[
	\sum_{p=1}^\infty \frac{1}{p!} \sum_{(\pi,G)\in \mcL_p^2} \Gamma_{\pi,G}^2
	= \xi_{\textnormal{small},0} + \xi_{\textnormal{small},\geq 1} + \xi_{\geq 1},
\]
where $\xi_{\textnormal{small},0}$ is the contribution of all small diagram in the first group, 
$\xi_{\textnormal{small},\geq1}$ is the contribution of all small diagrams in the second group 
and $\xi_{\geq1}$ is the contribution of all large diagrams.
We will then do a Taylor expansion of $\xi_{\textnormal{small},0}$ but not of the other terms.

We split diagrams in $\tilde\mcL_p^3$ into three (exhaustive) groups:
\begin{enumerate}
\item Small diagrams with $k + \nu + \nu^* = 1$ and $\{1\}$, $\{2\}$ and $\{3\}$ in $3$ different clusters. (Then $\nu = 0$.)
\item Small diagrams with $k + \nu + \nu^* = 1$ and $\{1\}$, $\{2\}$ and $\{3\}$ in $<3$ different clusters. (Then $k=\nu = 0$.)
\item Large diagrams with $k + \nu + \nu^* \geq 2$. 
\end{enumerate}
We then split 
\[
	\sum_{p=1}^\infty \frac{1}{p!} \sum_{(\pi,G)\in \tilde\mcL_p^3} \Gamma_{\pi,G}^3
	= \xi_{\textnormal{small},0}^3 + \xi_{\textnormal{small},\geq 1}^3 + \xi_{\geq 1}^3,
\]
where $\xi_{\textnormal{small},0}^3$ is the contribution of all small diagram in the first group, 
$\xi_{\textnormal{small},\geq1}^3$ is the contribution of all small diagrams in the second group 
and $\xi_{\geq1}^3$ is the contribution of all large diagrams.
Again,  we do a Taylor expansion of $\xi_{\textnormal{small},0}^3$ but not of the other terms.
For simplicity 
we will only compute the derivatives $\partial_{x_1}^2$.
With this bound the error term for the energy density is $O(a^2b \rho^6 \log(b/a))$ 
and so it is even smaller than the accuracy $a^2\rho^5$ with $b$ chosen as in \Cref{eqn.optimal.choices.d=1}.
(By the symmetry, we could bound $\xi_{\textnormal{small},0}$ by bounding its $6$th derivative 
$\partial_{x_1}^{2}\partial_{x_2}^2\partial_{x_3}^2$ instead.)
To keep the result symmetric in $x_1,x_2,x_3$ we will symmetrize the result afterwards.

We have immediately by \Cref{eqn.bound.linked.sum.k.ng.cluster.d=1} that 
\begin{equation}\label{eqn.bound.xi.geq1.d=1}
	\abs{\xi_{\geq 1}} \leq C a^3 \rho^5 (\log(b/a))^3 (\log N)^3,
	\qquad 
	\abs{\xi_{\geq 1}^3} \leq C a^2 \rho^5 (\log(b/a))^2 (\log N)^2.
\end{equation}
Similarly as in the proof of \Cref{lem.derivative.sum.linked.diagrams}
we have for $x_1=x_2$
\[
\begin{aligned}
	\xi_{\textnormal{small},0}(x_2,x_2) + \xi_{\textnormal{small},\geq 1}(x_2,x_2) + \xi_{\geq 1}(x_2,x_2) 
	& = 0,
	\\
	\xi_{\textnormal{small},0}^3(x_2,x_2,x_3) + \xi_{\textnormal{small},\geq 1}^3(x_2,x_2,x_3) + \xi_{\geq 1}^3(x_2,x_2,x_3) 
	& = 0.
\end{aligned}
\]
Hence we may bound the zeroth order by 
\[
\begin{aligned}
	\abs{\xi_{\textnormal{small},0}(x_2,x_2)}
	& \leq \abs{\xi_{\textnormal{small},\geq 1}(x_2,x_2)} + \abs{\xi_{\geq 1}(x_2,x_2)},
	\\
	\abs{\xi_{\textnormal{small},0}^3(x_2,x_2,x_3)}
	& \leq \abs{\xi_{\textnormal{small},\geq 1}^3(x_2,x_2,x_3)} + \abs{\xi_{\geq 1}^3(x_2,x_2,x_3)}.
\end{aligned}	
\]
For the diagrams in $\xi_{\textnormal{small},0}$ and $\xi_{\textnormal{small},0}^3$
we have similarly to \Cref{lem.Taylor.small.diagrams} that 
\begin{equation}\label{eqn.bound.partial.xi0.d=1}
	\abs{\partial_{x_1}^2\xi_{\textnormal{small},0}} 
	\leq C a \rho^{5} \log(b/a),
	\qquad 
	\abs{\partial_{x_1}^2\xi_{\textnormal{small},0}^3} 
	\leq C a \rho^{6} \log(b/a)	
\end{equation}
uniformly in $x_1,x_2,x_3$.
For the diagrams in $\xi_{\textnormal{small},\geq 1}$ and $\xi_{\textnormal{small},\geq 1}^3$ 
the analysis is somewhat similar to the proof of \Cref{lem.small.diagrams.type.C}.
We have 
\begin{lemma}\label{lem.small.diagrams.d=1}
For the small diagrams in $\xi_{\textnormal{small},\geq 1}$ and $\xi_{\textnormal{small},\geq 1}^3$ we have the bounds 
\begin{align}
	\abs{\xi_{\textnormal{small},\geq 1}} 
	& \leq 
  C a^2 b^4 \rho^8 + C a b^2 \rho^7 |x_1-x_2|^2\left[1 + N a b^2 \rho^3\right],
	\label{eqn.2body.small.diagrams.d=1}
	\\ 
	\abs{\xi_{\textnormal{small},\geq 1}^3}
  & \leq C a b^2 \rho^{8} \left(|x_1-x_2|^2 + |x_1-x_3|^2 + |x_2-x_3|^2\right)
	\label{eqn.3body.small.diagrams.d=1}
\end{align}
uniformly in $x_1,x_2,x_3$. 
\end{lemma}
\noindent
We give the proof of \Cref{lem.small.diagrams.d=1} in \Cref{sec.small.diagrams.d=1}.
Combining \Cref{lem.small.diagrams.d=1} and \Cref{eqn.bound.xi.geq1.d=1,eqn.bound.partial.xi0.d=1}
concludes the proof of \Cref{lem.derivative.sum.linked.diagrams.d=1}.
\end{proof}

\subsubsection*{Acknowledgements}
A.B.L. would like to thank Johannes Agerskov and Jan Philip Solovej for valuable discussions.
We thank Alessandro Giuliani for helpful discussions and for pointing out the reference \cite{Giuliani.Mastropietro.ea.2021}. 
Funding from the European Union's Horizon 2020 research and innovation programme under
the ERC grant agreement No~694227 is acknowledged. 
Financial support by the Austrian Science Fund (FWF) through project number~I~6427-N (as part of the SFB/TRR~352) is gratefully acknowledged.


\appendix

\section{Small diagrams}\label{sec.small.diagrams}
In this appendix we compute the contributions of all the small diagrams of 
\Cref{lem.Taylor.small.diagrams,lem.bound.3.body.small.diagrams,lem.small.diagrams.type.C,lem.small.diagrams.d=1}. 
We first consider those of \Cref{lem.Taylor.small.diagrams,lem.small.diagrams.type.C}.
\subsection{Small \texorpdfstring{$2$}{2}-particle diagrams 
(proof of \texorpdfstring{\Cref{lem.Taylor.small.diagrams,lem.small.diagrams.type.C}}{Lemmas \ref*{lem.small.diagrams.type.C} and \ref*{lem.Taylor.small.diagrams}})}
\label{sec.small.diagrams.2.particle}
Recall from the proof of \Cref{lem.derivative.sum.linked.diagrams}, \Cref{sec.proof.subleading.2.body.diagrams} that 
\[
	\xi_{\textnormal{small},0} + \xi_{\textnormal{small},\geq 1}
	=
	\sum_{p=1}^\infty \frac{1}{p!} \sum_{\substack{(\pi,G)\in \mcL_p^2 \\ (\pi,G) \textnormal{ small}}} \Gamma_{\pi,G}^2.
\]
The criterion for being small is defined in the proof of \Cref{lem.derivative.sum.linked.diagrams} around \Cref{eqn.2.body.linked.diagrams.decompose},
and will be recalled below.
The diagrams are split into types $(A), (B)$ and $(C)$  according their underlying graphs $G$ as in the proof of \Cref{lem.derivative.sum.linked.diagrams}.
We further split the type $(B)$ into two types $(B_1)$ and $(B_2)$.
The diagrams of type $(B_1)$ are those diagrams for which the extra vertex $\{3\}$ 
in the distinguished clusters is in the cluster containing $\{1\}$, 
i.e. connected to $\{1\}$.
The diagrams of type $(B_2)$ are those diagrams for which the extra vertex $\{3\}$ 
is in the cluster containing $\{2\}$, i.e. connected to $\{2\}$. 
That is, the different types are as follows.
See also \Cref{fig.small.diagrams}.
\begin{enumerate}[$(A)$]
\item $\{1\}$ and $\{2\}$ in different clusters and $1\leq k \leq 4, \nu = 0, \nu^* = 0$,
\item $\{1\}$ and $\{2\}$ in different clusters and $0\leq k \leq 2, \nu = 0, \nu^* = 1$,
\begin{enumerate}
\item[$(B_1)$] and $n_* = 1, n_{**}=0$,
\item[$(B_2)$] and $n_* = 0, n_{**} = 1$,
\end{enumerate}
\item $\{1\}$ and $\{2\}$ in the same cluster and $0\leq k\leq 2, \nu = 0, \nu^* = 1$.
\end{enumerate} 
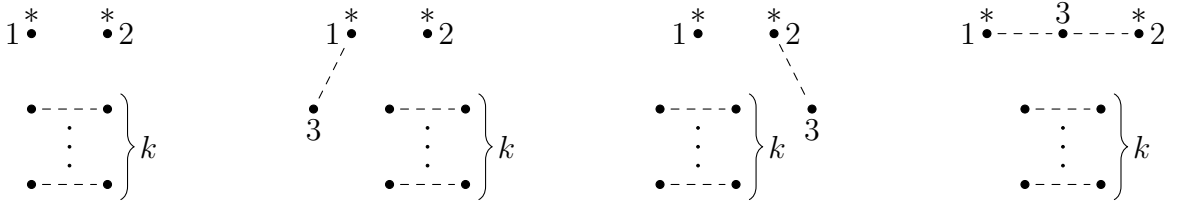
\begin{figure}[htb]
\centering
	\begin{subfigure}[t]{0.24\textwidth} 
		\centering
		\begin{tikzpicture}[line cap=round,line join=round,>=triangle 45,x=1.0cm,y=1.0cm]
		\node (1) at (0,0) {};
		\node (2) at (1,0) {};
		\node (3) at (0,-1) {};
		\node (4) at (1,-1) {};
		\node (5) at (0,-2) {};
		\node (6) at (1,-2) {};
		\draw[dashed] (3) -- (4);
		\draw[dashed] (5) -- (6);
		\foreach \i in {1,...,6} \draw[fill] (\i) circle [radius=1.5pt];
		\foreach \i in {1,...,3} \draw[fill] (0.5,-1 - \i/4) circle [radius=0.5pt];
		\draw[decoration={brace,raise=5pt,amplitude=5pt},decorate] (1,-0.8) -- node[right=6pt] {$\, k$} (1,-2.2);
		\foreach \i in {2} \node[anchor=west] at (\i) {$\i$};
    \foreach \i in {1} \node[anchor=east] at (\i) {$\i$};
    \foreach \i in {1,2} \node[anchor=south] at (\i) {$*$};
		\end{tikzpicture}
		\caption{Type $(A), 1\leq k\leq 4$}
	\end{subfigure}\hfill
	\begin{subfigure}[t]{0.25\textwidth} 
		\centering
		\begin{tikzpicture}[line cap=round,line join=round,>=triangle 45,x=1.0cm,y=1.0cm]
		\node (3) at (0.5,-1) {};
		\node (1) at (1,0) {};
		\node (2) at (2,0) {};
		\node (4) at (1.5,-1) {};
		\node (5) at (2.5,-1) {};
		\node (6) at (1.5,-2) {};
		\node (7) at (2.5,-2) {};
		\draw[dashed] (3) -- (1);
		\draw[dashed] (4) -- (5);
		\draw[dashed] (6) -- (7);
		\foreach \i in {1,...,7} \draw[fill] (\i) circle [radius=1.5pt];
		\foreach \i in {1,...,3} \draw[fill] (2,-1 - \i/4) circle [radius=0.5pt];
		\draw[decoration={brace,raise=5pt,amplitude=5pt},decorate] (2.5,-0.8) -- node[right=6pt] {$\, k$} (2.5,-2.2);
		\foreach \i in {2} \node[anchor=west] at (\i) {$\i$};
    \foreach \i in {1} \node[anchor=east] at (\i) {$\i$};
    \foreach \i in {1,2} \node[anchor=south] at (\i) {$*$};
		\node[anchor = north] at (3) {$3$};
		\end{tikzpicture}
		\caption{Type $(B_1), 0\leq k\leq 2$}
	\end{subfigure}\hfill
	\begin{subfigure}[t]{0.24\textwidth}
		\centering	
		\begin{tikzpicture}[line cap=round,line join=round,>=triangle 45,x=1.0cm,y=1.0cm]
		\node (1) at (0,0) {};
		\node (2) at (1,0) {};
		\node (3) at (1.5,-1) {};
		\node (4) at (-0.5,-1) {};
		\node (5) at (0.5,-1) {};
		\node (6) at (-0.5,-2) {};
		\node (7) at (0.5,-2) {};
		\draw[dashed] (2) -- (3);
		\draw[dashed] (4) -- (5);
		\draw[dashed] (6) -- (7);
		\foreach \i in {1,...,7} \draw[fill] (\i) circle [radius=1.5pt];
		\foreach \i in {1,...,3} \draw[fill] (0,-1 - \i/4) circle [radius=0.5pt];
		\draw[decoration={brace,raise=5pt,amplitude=5pt},decorate] (0.5,-0.8) -- node[right=6pt] {$\, k$} (0.5,-2.2);
		\foreach \i in {2} \node[anchor=west] at (\i) {$\i$};
    \foreach \i in {1} \node[anchor=east] at (\i) {$\i$};
    \foreach \i in {1,2} \node[anchor=south] at (\i) {$*$};
		\node[anchor = north] at (3) {$3$};
		\end{tikzpicture}
		\caption{Type $(B_2), 0\leq k\leq 2$}
	\end{subfigure}\hfill
	\begin{subfigure}[t]{0.24\textwidth}
		\centering	
		\begin{tikzpicture}[line cap=round,line join=round,>=triangle 45,x=1.0cm,y=1.0cm]
		\node (1) at (0,0) {};
		\node (2) at (2,0) {};
		\node (3) at (1,0) {};
		\node (4) at (0.5,-1) {};
		\node (5) at (1.5,-1) {};
		\node (6) at (0.5,-2) {};
		\node (7) at (1.5,-2) {};
		\draw[dashed] (4) -- (5);
		\draw[dashed] (6) -- (7);
		\draw[dashed] (1) -- (3) -- (2);
		\foreach \i in {1,...,7} \draw[fill] (\i) circle [radius=1.5pt];
		\foreach \i in {1,...,3} \draw[fill] (1,-1 - \i/4) circle [radius=0.5pt];
		\draw[decoration={brace,raise=5pt,amplitude=5pt},decorate] (1.5,-0.8) -- node[right=6pt] {$\, k$} (1.5,-2.2);
		\foreach \i in {2} \node[anchor=west] at (\i) {$\i$};
    \foreach \i in {1} \node[anchor=east] at (\i) {$\i$};
    \foreach \i in {3} \node[anchor=south] at (\i) {$\i$};
    \foreach \i in {1,2} \node[anchor=south] at (\i) {$*$};
		\end{tikzpicture}
		\caption{Type $(C), 0\leq k \leq 2$}
	\end{subfigure}
\caption{$g$-graphs of small diagrams of different types. 
For each diagram only the graph $G$ is drawn. 
The relevant diagrams come with permutations $\pi$ such that the diagrams are linked.}
\label{fig.small.diagrams}
\end{figure}

\noindent
We first give the
\begin{proof}[{Proof of \Cref{lem.small.diagrams.type.C}}]
Consider first all diagrams of type $(C)$ of smallest size, i.e. with $g$-graph 
\[
	\begin{tikzpicture}[line cap=round,line join=round,>=triangle 45,x=1.0cm,y=1.0cm]
		\node[anchor = east] at (-0.5,0) {$G_0=$};
		\node (1) at (0,0) {};
		\node (2) at (2,0) {};
		\node (3) at (1,0) {};
		\draw[dashed] (1) -- (3) -- (2);
		\foreach \i in {1,...,3} \draw[fill] (\i) circle [radius=1.5pt];
		\foreach \i in {2} \node[anchor=west] at (\i) {$\i$};
    \foreach \i in {1} \node[anchor=east] at (\i) {$\i$};
    \foreach \i in {1,2} \node[anchor=south] at (\i) {$*$};
		\end{tikzpicture}
\]
Since this graph is connected, all $\pi\in \mcS_3$ give rise to a linked diagram $(\pi,G_0)$.
By Wick's rule, the $\pi$-sum then gives the factor $\rho^{(3)}$.
That is,
\[
	\sum_{\pi \in \mcS_3 : (\pi, G_0)\in\mcL_1^2} \Gamma_{\pi,G_0}^2 
	= 
		\int 
		g_{13}g_{23}
		\sum_{\pi \in \mcS_3} 
		(-1)^\pi 
		\prod_{j=1}^{3} \gamma^{(1)}_N(x_j, x_{\pi(j)}) \ud x_{3}
	= 
	\int g_{13}g_{23} \rho^{(3)} \ud x_3.
\]
Recall the bound $\rho^{(3)}\leq C\rho^{3+4/3} |x_1-x_2|^2 |x_1-x_3|^2$ from \Cref{lem.bound.rho3}.
Now we bound $|g_{23}|\leq 1$.
Thus 
\[
	\abs{\sum_{\pi \in \mcS_3 : (\pi, G_0)\in\mcL_1^2} \Gamma_{\pi,G_0}^2}
	\leq 
  C\rho^{3+4/3} |x_1-x_2|^2 \int\left(1 - f(x)^2\right)|x|^2 \ud x.
\]
Recalling \Cref{prop.bound.f0.hc} we may bound 
\begin{equation}\label{eqn.bound.int.g.|x|2}
	\int\left(1 - f(x)^2\right)|x|^2 \ud x
	\leq Ca^5 + \frac{C}{(1-a^3/b^3)^2} \int_a^b \left[\left(1 - \frac{a^3}{b^3}\right)^2 - \left(1 - \frac{a^3}{r^3}\right)^2 \right] r^4 \ud r
	\leq C a^3 b^2.
\end{equation}
We conclude that all diagrams of smallest size contribute $\leq Ca^3 b^2 \rho^{3+4/3} |x_1-x_2|^2$.

For the larger diagrams, we consider an example diagram
\[
\begin{tikzpicture}[line cap=round,line join=round,>=triangle 45,x=1.0cm,y=1.0cm]
		\node[anchor = east] at (-0.5,-0.5) {$(\pi, G) = $};
		\node (1) at (0,0) {};
		\node (3) at (1,0) {};
		\node (2) at (2,0) {};
		\node (4) at (0.5,-1) {};
		\node (5) at (1.5,-1) {};
		\draw[dashed] (4) -- (5);
		\draw[dashed] (1) -- (3) -- (2);
		\draw[->] (1) to (4);
		\draw[->] (4) to (3);
		\draw[->] (3) to[bend right] (1);
		\draw[->] (2) to[bend right] (5);
		\draw[->] (5) to[bend right] (2);
		\foreach \i in {1,...,5} \draw[fill] (\i) circle [radius=1.5pt];
		\foreach \i in {2} \node[anchor=west] at (\i) {$\i$};
    \foreach \i in {1} \node[anchor=east] at (\i) {$\i$};
    \foreach \i in {3} \node[anchor=south] at (\i) {$\i$};
    \foreach \i in {1,2} \node[anchor=south] at (\i) {$*$};
		\end{tikzpicture}
\]
For this diagram we have 
\[
\begin{aligned}
	\Gamma_{\pi,G}^2
	& 
	= (-1)^{\pi} \iiint \gamma^{(1)}_N(x_1;x_4)\gamma^{(1)}_N(x_4;x_3)\gamma^{(1)}_N(x_3;x_1)\gamma^{(1)}_N(x_2;x_5)\gamma^{(1)}_N(x_5;x_2) 
		g_{13}g_{23}g_{45} \ud x_3 \ud x_4 \ud x_5
	\\ & 
	= \frac{-1}{L^{15}} \sum_{k_1,\ldots,k_5\in P_F}
		\iiint e^{ik_1(x_1-x_4)}e^{ik_2(x_4-x_3)}e^{ik_3(x_3-x_1)}e^{ik_4(x_2-x_5)}e^{ik_5(x_5-x_2)} g_{13}g_{23}g_{45} \ud x_3 \ud x_4 \ud x_5
	\\ & 
	= \frac{-1}{L^{15}} \sum_{k_1,\ldots,k_5\in P_F} e^{i(k_1-k_3)x_1}e^{i(k_4-k_5)x_2} 
		\int \ud x_3\, e^{i(k_3-k_2) x_3} g(x_1 - x_3) g(x_2-x_3)
	\\ & \quad 
		\times 
		\int \ud x_4\, \left[e^{i(k_2 - k_1 + k_5 - k_4)x_4} 
			\int \ud x_5\, e^{-i(k_5-k_4)(x_4-x_5)}g(x_4-x_5)\right]
	\\ & 
	= \frac{-1}{L^{12}} \sum_{k_1,\ldots,k_5\in P_F} e^{i(k_1-k_3)x_1}e^{i(k_4-k_5)x_2} 
		\int \ud x_3\, e^{i(k_3-k_2) x_3} g(x_1 - x_3) g(x_2-x_3)
	\\ & \quad 
		\times 
		\chi_{(k_2-k_1 = k_4-k_5)} \hat g(k_5-k_4),
\end{aligned}
\]
where $\hat g(k) := \int_\Lambda g(x) e^{-ikx} \ud x$.
Bounding $|g_{23}|\leq 1$ and $|\hat g(k)| \leq \int |g| \leq a^3 \log(b/a)$ we get that 
\[
	\abs{\Gamma_{\pi,G}^2} \leq C a^6 \rho^4 (\log(b/a))^2.
\]
One may do a similar computation for all the remaining diagrams.
By computing the integrations of the vertices in the internal clusters first, these give some factor $\hat g(k_i-k_j)$
and a factor $L^3\chi_{(k_i - k_j = k_{i'} - k_{j'})}$. 
By bounding as above we conclude that the contribution of small diagrams of type $(C)$ is bounded as desired.
\end{proof}

\begin{proof}[{Proof of \Cref{lem.Taylor.small.diagrams}}]
As with the larger diagrams of type $(C)$ we only give calculations for a few example diagrams and explain how
the calculation for the remaining diagrams are similar. We consider the examples in \Cref{fig.diagrams.type.A}.
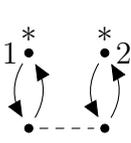
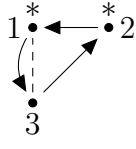
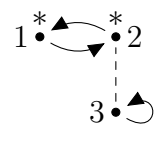
\begin{figure}[htb]
\centering
	\begin{subfigure}[t]{0.3\textwidth}
		\centering
		\begin{tikzpicture}[line cap=round,line join=round,>=triangle 45,x=1.0cm,y=1.0cm]
		\node (1) at (0,0) {};
		\node (2) at (1,0) {};
		\node (3) at (0,-1) {};
		\node (4) at (1,-1) {};
		\draw[dashed] (3) -- (4);
		\draw[->] (1) to[bend right] (3);
		\draw[->] (3) to[bend right] (1);
		\draw[->] (2) to[bend right] (4);
		\draw[->] (4) to[bend right] (2);
		\foreach \i in {1,...,4} \draw[fill] (\i) circle [radius=1.5pt];
		\foreach \i in {2} \node[anchor=west] at (\i) {$\i$};
    \foreach \i in {1} \node[anchor=east] at (\i) {$\i$};
    \foreach \i in {1,2} \node[anchor=south] at (\i) {$*$};
		\end{tikzpicture}
		\caption{Example of a type $(A)$ diagram of smallest size}\label{fig.diagram.A.smallest}
	\end{subfigure}\hfill
  \begin{subfigure}[t]{0.3\textwidth}
    \centering
    \begin{tikzpicture}[line cap=round,line join=round,>=triangle 45,x=1.0cm,y=1.0cm]
    \node (3) at (1,-1) {};
    \node (1) at (1,0) {};
    \node (2) at (2,0) {};
    \draw[dashed] (3) -- (1);
    \foreach \i in {1,...,3} \draw[fill] (\i) circle [radius=1.5pt];
    \foreach \i in {2} \node[anchor=west] at (\i) {$\i$};
    \foreach \i in {1} \node[anchor=east] at (\i) {$\i$};
    \foreach \i in {1,2} \node[anchor=south] at (\i) {$*$};
    \node[anchor = north] at (3) {$3$};
    \draw[->] (1) to[bend right] (3);
    \draw[->] (3) to (2);
    \draw[->] (2) to (1);
    \end{tikzpicture}
    \caption{Example of a type $(B_1)$ diagram of smallest size}\label{fig.diagram.B1.smallest}
  \end{subfigure}
  \hfill
	\begin{subfigure}[t]{0.3\textwidth}
		\centering
		\begin{tikzpicture}[line cap=round,line join=round,>=triangle 45,x=1.0cm,y=1.0cm]
		\node (3) at (2,-1) {};
		\node (1) at (1,0) {};
		\node (2) at (2,0) {};
		\draw[dashed] (3) -- (2);
		\draw[->] (1) to[bend right] (2);
		\draw[->] (2) to[bend right] (1);
		\draw[->] (3) to[out=-30,in=30,loop] (3);
		\foreach \i in {1,...,3} \draw[fill] (\i) circle [radius=1.5pt];
		\foreach \i in {2} \node[anchor=west] at (\i) {$\i$};
    \foreach \i in {1} \node[anchor=east] at (\i) {$\i$};
    \foreach \i in {1,2} \node[anchor=south] at (\i) {$*$};
		\node[anchor = east] at (3) {$3$};
	\end{tikzpicture}
	\caption{Example of a type $(B_2)$ diagram of smallest size}\label{fig.diagram.B2.smallest}
	\end{subfigure}
\caption{Exemplary diagrams of types $(A),(B_1)$ and $(B_2)$. 
The dashed lines denote $g$-edges, and the arrows denote (directed) edges of the permutation.}
\label{fig.diagrams.type.A}
\end{figure}

The contribution of the diagram in \Cref{fig.diagram.A.smallest} to 
$\partial^\mu_{x_1}\partial^\nu_{x_1} \xi_{\textnormal{small}}$ is
\[
\begin{aligned}
	& \frac{1}{2}\partial^\mu_{x_1}\partial^\nu_{x_1} \Gamma_{\pi,G}^2
	\\
	& = 
		\frac{-1}{2L^{12}} \sum_{k_1,\ldots,k_4\in P_F} (k_1^\mu - k_2^\mu)(k_1^\nu - k_2^\nu) \iint 
		e^{ik_1(x_1-x_3)}e^{ik_2(x_3-x_1)}e^{ik_3(x_2-x_4)}e^{ik_4(x_4-x_2)} g_{34} \ud x_3 \ud x_4 
	\\
	& = 
		\frac{-1}{2L^{12}} \sum_{k_1,\ldots,k_4\in P_F} (k_1^\mu - k_2^\mu)(k_1^\nu - k_2^\nu)
		e^{i(k_1-k_2)x_1}e^{i(k_3-k_4)x_2}
	\\
	& \quad \times \iint e^{-i(k_4-k_3)(x_3-x_4)} e^{i(k_2-k_1+k_4-k_3)x_3} g(x_3-x_4) \ud x_3 \ud x_4
	\\
	& = 
		\frac{-1}{2L^{9}} \sum_{k_1,\ldots,k_4\in P_F} (k_1^\mu - k_2^\mu)(k_1^\nu - k_2^\nu)
		e^{i(k_1-k_2)x_1}
		e^{i(k_3-k_4)x_2}
		\chi_{(k_2-k_1=k_3-k_4)} \hat g(k_4-k_3)
	\\
	& = 
		O(\rho^{3+2/3} a^3 \log (b/a))
\end{aligned}
\]
using that $\hat g(k) = \int_\Lambda g(x) e^{-ikx} \ud x$ satisfies $|\hat g(k)|\leq \int |g| \leq C a^3 \log (b/a)$.
The same type of computation is valid for all other diagrams of type $(A)$.

Consider now the diagram in \Cref{fig.diagram.B2.smallest} of type $(B_2)$.
This contributes 
\[
\begin{aligned}
&
	\partial^\mu_{x_1}\partial^{\nu}_{x_1}
	\frac{-1}{L^9} \sum_{k_1,k_2,k_3\in P_F}
	\int e^{ik_1(x_1-x_2)}e^{ik_2(x_2-x_1)}g(x_2-x_3) \ud x_3
\\
&	
	= \frac{1}{L^9} \sum_{k_1,k_2,k_3\in P_F} 
		(k_1^\mu - k_2^\mu)(k_1^\nu - k_2^\nu)
		e^{i(k_1-k_2)x_1}
		e^{i(k_2-k_1)x_2} 
		\int 
		g(x_2-x_3)
		\ud x_3
\\
& 
	= O(\rho^{3+2/3} a^3 \log(b/a))
\end{aligned}
\]
exactly as for type $(A)$.
Similarly, all other diagrams of type $(B_2)$ may be bounded using the same method as for types $(A)$.

Finally, we consider the diagram in \Cref{fig.diagram.B1.smallest} of type $(B_1)$. 
Here we have 
\[
\begin{aligned}
	\partial^\mu_{x_1}\partial^{\nu}_{x_1}\Gamma_{\pi,G}^2
& =
	\partial^\mu_{x_1}\partial^{\nu}_{x_1}
	\frac{1}{L^9} \sum_{k_1,k_2,k_3\in P_F}
	\int e^{ik_1(x_1-x_3)}e^{ik_2(x_3-x_2)}e^{ik_3(x_2-x_1)}g(x_1-x_3) \ud x_3
\\
&	
  = \partial^\mu_{x_1}\partial^{\nu}_{x_1}
  \frac{1}{L^9} \sum_{k_1,k_2,k_3\in P_F}
  e^{i(k_2-k_3)x_1} e^{i(k_3-k_2)x_2}
  \int e^{-i(k_2 - k_1)(x_1-x_3)}g(x_1-x_3) \ud x_3
\\ & 
  = 
  \frac{-1}{L^9} \sum_{k_1,k_2,k_3\in P_F}
  (k_2^\mu - k_3^\mu)(k_2^\nu - k_3^\nu)
  e^{i(k_2-k_3)x_1} e^{i(k_3-k_2)x_2}
  \hat g(k_2-k_1)
\\ & = O(\rho^{3+2/3} a^3\log(b/a)).
\end{aligned}
\]
All larger diagrams of type $(B_1)$ 
may be bounded similarly.
We conclude the desired.
\end{proof}


\subsection{Small \texorpdfstring{$3$}{3}-particle diagrams 
(proof of \texorpdfstring{\Cref{lem.bound.3.body.small.diagrams}}{Lemma \ref*{lem.bound.3.body.small.diagrams}})}
\label{sec.small.diagrams.3.particle}
We now give the 
\begin{proof}[{Proof of \Cref{lem.bound.3.body.small.diagrams}}]
Recall that 
\[
	\xi_{\textnormal{small}}^3
	=
	\sum_{p=1}^\infty \frac{1}{p!} \sum_{\substack{(\pi,G)\in \tilde\mcL_p^3 \\ (\pi,G) \textnormal{ small}}} \Gamma_{\pi,G}^3,
\]
where \lq\lq small\rq\rq\  refers to diagrams with $G$-graph 
\[
\begin{tikzpicture}[line cap=round,line join=round,>=triangle 45,x=1.0cm,y=1.0cm]
		\node[anchor=east] at (-0.5,-1) {$G = $};
		\node[anchor=west] at (3,-1) {$k=1,2$};
		\node (1) at (0,0) {};
		\node (2) at (1,0) {};
		\node (3) at (2,0) {};
		\node (4) at (0.5,-1) {};
		\node (5) at (1.5,-1) {};
		\node (6) at (0.5,-2) {};
		\node (7) at (1.5,-2) {};
		\draw[dashed] (4) -- (5);
		\draw[dashed] (6) -- (7);
		\foreach \i in {1,...,7} \draw[fill] (\i) circle [radius=1.5pt];
		\foreach \i in {1,...,3} \draw[fill] (1,-1 - \i/4) circle [radius=0.5pt];
		\draw[decoration={brace,raise=5pt,amplitude=5pt},decorate] (1.5,-0.8) -- node[right=6pt] {$\, k$} (1.5,-2.2);
		\foreach \i in {1,2,3} \node[anchor=north] at (\i) {$\i$};
    \foreach \i in {1,2,3} \node[anchor=south] at (\i) {$*$};
		\end{tikzpicture}	
\]
and permutation $\pi$ such that $(\pi,G)$ has at most two linked components,
both of which contain at least one external vertex.
As in the proof of \Cref{lem.Taylor.small.diagrams,lem.small.diagrams.type.C} in \Cref{sec.small.diagrams.2.particle}
we compute the value of a few examples and explain how to compute the value of the remaining diagrams.
We consider the examples of \Cref{fig.3.body.small.diagrams}
\begin{figure}[htb]
\centering
	\begin{subfigure}[t]{0.4\textwidth} 
		\centering
		\begin{tikzpicture}[line cap=round,line join=round,>=triangle 45,x=1.0cm,y=1.0cm]
		\node (1) at (0,0) {};
		\node (2) at (1,0) {};
		\node (3) at (2,0) {};
		\node (4) at (0.5,-1) {};
		\node (5) at (1.5,-1) {};
		\draw[dashed] (4) -- (5);
		\draw[->] (1) to (4);
		\draw[->] (4) to (2);
		\draw[->] (2) to (1);
		\draw[->] (3) to[bend right] (5);
		\draw[->] (5) to[bend right] (3);
		\foreach \i in {1,...,5} \draw[fill] (\i) circle [radius=1.5pt];
		\foreach \i in {1} \node[anchor=east] at (\i) {$\i$};
    \foreach \i in {2,3} \node[anchor=west] at (\i) {$\i$};
    \foreach \i in {1,2,3} \node[anchor=south] at (\i) {$*$};
		\end{tikzpicture}
		\caption{Example of a diagram of smallest size with one linked component}\label{fig.3.body.small.diagram.linked}
	\end{subfigure}\hspace{2em}
	\begin{subfigure}[t]{0.4\textwidth} 
		\centering
		\begin{tikzpicture}[line cap=round,line join=round,>=triangle 45,x=1.0cm,y=1.0cm]
		\node (1) at (0,0) {};
		\node (2) at (1,0) {};
		\node (3) at (2,0) {};
		\node (4) at (0.5,-1) {};
		\node (5) at (1.5,-1) {};
		\draw[dashed] (4) -- (5);
		\draw[->] (1) to (4);
		\draw[->] (4) to[bend right] (5);
		\draw[->] (5) to (1);
		\draw[->] (2) to[bend right] (3);
		\draw[->] (3) to[bend right] (2);
		\foreach \i in {1,...,5} \draw[fill] (\i) circle [radius=1.5pt];
		\foreach \i in {1,2} \node[anchor=east] at (\i) {$\i$};
    \foreach \i in {3} \node[anchor=west] at (\i) {$\i$};
    \foreach \i in {1,2,3} \node[anchor=south] at (\i) {$*$};
		\end{tikzpicture}
		\caption{Example of a diagram of smallest size with two linked components}\label{fig.3.body.small.diagram.not.linked}
	\end{subfigure}
\caption{Exemplary small diagrams in $\tilde\mcL_2^3$}
\label{fig.3.body.small.diagrams}
\end{figure}
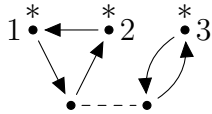
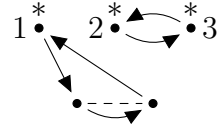

The contribution of the diagram in \Cref{fig.3.body.small.diagram.linked} is 
\[
\begin{aligned}
\Gamma_{\pi, G}^3
	& = \frac{-1}{L^{15}} \sum_{k_1,\ldots,k_5\in P_F} \iint 
		e^{ik_1(x_1-x_4)}e^{ik_2(x_4-x_2)}e^{ik_3(x_2-x_1)} e^{ik_4(x_3-x_5)}e^{ik_5(x_5-x_3)} g_{45} \ud x_4 \ud x_5 
	\\
	& = \frac{-1}{L^{12}}\sum_{k_1,\ldots,k_5\in P_F} e^{i(k_1-k_3)x_1} e^{i(k_3-k_2)x_2} e^{i(k_4-k_5)x_3}
		\chi_{(k_2-k_1 = k_4-k_5)} \hat g(k_5-k_4)
	\\
	& = O(a^3\rho^4 \log(b/a)).
\end{aligned}
\]
Similarly, the contribution of the diagram in \Cref{fig.3.body.small.diagram.not.linked} is 
\[	
\begin{aligned}
\Gamma_{\pi,G}^3 
	& = \frac{-1}{L^{15}} \sum_{k_1,\ldots,k_5\in P_F} \iint 
		e^{ik_1(x_1-x_4)}e^{ik_2(x_4-x_5)}e^{ik_3(x_5-x_1)} e^{ik_4(x_2-x_3)}e^{ik_5(x_3-x_2)} g_{45} \ud x_4 \ud x_5 
	\\
	& = \frac{-1}{L^{12}}\sum_{k_1,\ldots,k_5\in P_F} e^{i(k_1-k_3)x_1} e^{i(k_4-k_5)x_2} e^{i(k_5-k_4)x_3}
		\chi_{(k_2-k_1 = k_2-k_3)} \hat g(k_3-k_2)
	\\
	& = O(a^3\rho^4 \log(b/a)).
\end{aligned}
\]
One may follow this kind of computation for any diagram. 
The central property we used is that the internal vertices are all 
in the same linked component as some external vertex. 
This means that the integrals 
over internal vertices either gives a factor of $\hat g(k_i - k_j)$ or a factor of $L^3\chi_{(k_i - k_j = k_{i'} - k_{j'})}$. 
We conclude the desired.
\end{proof}

\subsection{Small diagrams in \texorpdfstring{$1$}{1} dimension 
	(proof of \texorpdfstring{\Cref{lem.small.diagrams.d=1}}{Lemma \ref*{lem.small.diagrams.d=1}})}
\label{sec.small.diagrams.d=1}
We now give the 
\begin{proof}[{Proof of \Cref{lem.small.diagrams.d=1}}]
We first give the proof of \Cref{eqn.2body.small.diagrams.d=1}.
We split the two cases $(A)$ and $(B)$ of small diagrams further.
They are given as follows. 
\begin{enumerate}[$(A)$]
\item $\{1\}$ and $\{2\}$ in different clusters and $k=0, \nu^* = 2$,
\begin{enumerate}
\item[$(A_1)$] $n_* = 2, n_{**} = 0$ (or $n_* = 0, n_{**} = 2$),
\item[$(A_2)$] $n_* = 1, n_{**} = 1$.
\end{enumerate}
\item $\{1\}$ and $\{2\}$ in the same cluster and $1\leq k + \nu + \nu^* \leq 2$,
\begin{enumerate}
\item[$(B_1)$] $k=0$,
\item[$(B_2)$] $k=1$.
\end{enumerate}
\end{enumerate} 
See also \Cref{fig.small.2body.diagrams.d=1}.
\begin{figure}[htb]
\centering
	\begin{subfigure}[t]{0.23\textwidth} 
		\centering
		\begin{tikzpicture}[line cap=round,line join=round,>=triangle 45,x=1.0cm,y=1.0cm]
		\node (1) at (0,0) {};
		\node (2) at (1,0) {};
		\node (3) at (-0.5,-1) {};
		\node (4) at (0.5,-1) {};
		\draw[dashed] (1) -- (3) -- (4) -- (1);
		\foreach \i in {1,...,4} \draw[fill] (\i) circle [radius=1.5pt];
		\foreach \i in {1} \node[anchor=east] at (\i) {$\i$};
    \foreach \i in {2} \node[anchor=west] at (\i) {$\i$};
    \foreach \i in {1,2} \node[anchor=south] at (\i) {$*$};
		\end{tikzpicture}
		\caption{Type $(A_1)$}
		\label{fig.small.2body.diagrams.d=1.type.A1}
	\end{subfigure}\hfill
	\begin{subfigure}[t]{0.23\textwidth} 
		\centering
		\begin{tikzpicture}[line cap=round,line join=round,>=triangle 45,x=1.0cm,y=1.0cm]
		\node (1) at (0,0) {};
		\node (2) at (1,0) {};
		\node (3) at (0,-1) {};
		\node (4) at (1,-1) {};
		\draw[dashed] (1) -- (3);
		\draw[dashed] (2) -- (4);
		\foreach \i in {1,...,4} \draw[fill] (\i) circle [radius=1.5pt];
		\foreach \i in {1} \node[anchor=east] at (\i) {$\i$};
    \foreach \i in {2} \node[anchor=west] at (\i) {$\i$};
    \foreach \i in {1,2} \node[anchor=south] at (\i) {$*$};
		\end{tikzpicture}
		\caption{Type $(A_2)$}
		\label{fig.small.2body.diagrams.d=1.type.A2}
	\end{subfigure}\hfill
	\begin{subfigure}[t]{0.23\textwidth}
		\centering	
		\begin{tikzpicture}[line cap=round,line join=round,>=triangle 45,x=1.0cm,y=1.0cm]
		\node (1) at (0,0) {};
		\node (2) at (2,0) {};
		\node (3) at (1,0.5) {};
		\node (4) at (1,-0.5) {};
		\draw[dashed] (1) -- (3) -- (2) -- (4) -- (1);
		\draw[dashed] (3) -- (4);
		\foreach \i in {1,...,3} \draw[fill] (\i) circle [radius=1.5pt];
		\draw (4) circle [radius=1.5pt];
		\foreach \i in {1} \node[anchor=east] at (\i) {$\i$};
    \foreach \i in {2} \node[anchor=west] at (\i) {$\i$};
    \foreach \i in {1,2} \node[anchor=south] at (\i) {$*$};
		\end{tikzpicture}
		\caption{Type $(B_1)$}
		\label{fig.small.2body.diagrams.d=1.type.B1}
	\end{subfigure}\hfill
	\begin{subfigure}[t]{0.23\textwidth}
		\centering	
		\begin{tikzpicture}[line cap=round,line join=round,>=triangle 45,x=1.0cm,y=1.0cm]
		\node (1) at (0,0) {};
		\node (2) at (2,0) {};
		\node (3) at (1,0) {};
		\node (4) at (0.5,-1) {};
		\node (5) at (1.5,-1) {};
		\draw[dashed] (4) -- (5);
		\draw[dashed] (1) -- (3) -- (2);
		\foreach \i in {1,...,5} \draw[fill] (\i) circle [radius=1.5pt];
		\foreach \i in {1} \node[anchor=east] at (\i) {$\i$};
    \foreach \i in {2} \node[anchor=west] at (\i) {$\i$};
    \foreach \i in {1,2} \node[anchor=south] at (\i) {$*$};
		\end{tikzpicture}
		\caption{Type $(B_2)$}
		\label{fig.small.2body.diagrams.d=1.type.B2}
	\end{subfigure}
\caption{$g$-graphs of small diagrams of different types. For each diagram only the graph $G$ is drawn. 
The relevant diagrams come with permutations $\pi$ such the the diagrams are linked.
The diagrams of type $(A_1)$ and $(B_1)$ may have some of the drawn $g$-edges not present, but the same connected components.
Moreover, the diagrams of type $(B_1)$ may have one of the internal vertices drawn not present
(indicated by a $\circ$).
With the modification of the drawings described here these are all small diagrams.}
\label{fig.small.2body.diagrams.d=1}
\end{figure}
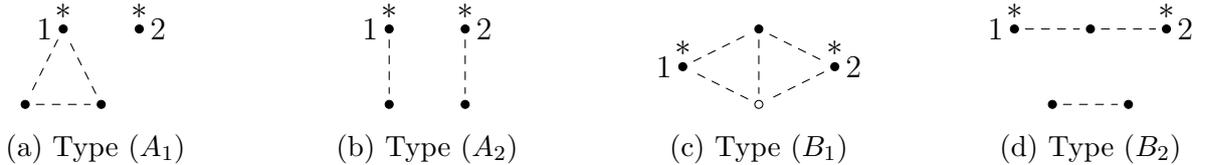

We will consider some examples of diagrams. 
Namely those drawn in \Cref{fig.small.2body.diagrams.d=1} (but not modified as described in the caption), 
except for the diagram of type $(B_1)$, where 
we will consider diagrams of smallest size, with $g$-graph 
\begin{equation}\label{eqn.fig.small.2.body.diagram.type.B1}
\begin{tikzpicture}[line cap=round,line join=round,>=triangle 45,x=1.0cm,y=1.0cm]
		\node (1) at (0,0) {};
		\node (2) at (2,0) {};
		\node (3) at (1,0) {};
		\draw[dashed] (1) -- (3) -- (2);
		\foreach \i in {1,...,3} \draw[fill] (\i) circle [radius=1.5pt];
		\foreach \i in {1} \node[anchor=east] at (\i) {$\i$};
    \foreach \i in {2} \node[anchor=west] at (\i) {$\i$};
    \foreach \i in {1,2} \node[anchor=south] at (\i) {$*$};
		\node[anchor = east] at (-0.5,0) {$G_0 = $};
		\end{tikzpicture}
\end{equation}
All other diagrams can be treated in a similar fashion.
For the argument we will need a different formula for $\rho_{\rm t}$.
Recall the definition in \Cref{eqn.define.truncated.correlation}.
We may write the characteristic function as 
\[
	\chi( (\pi,\cup G_\ell)\textnormal{ linked}) = 1 - \chi( (\pi,\cup G_\ell) \textnormal{ not linked}).
\]
That is, 
\begin{equation*}
\rho_{\rm t}^{(A_1,\ldots,A_k)} = \sum_{\pi \in \mcS_{\cup A_\ell}}
			(-1)^\pi 
			\prod_{j\in \cup A_\ell} \gamma^{(1)}_N(x_j; x_{\pi(j)})
			- \sum_{\pi \in \mcS_{\cup A_\ell}}
			(-1)^\pi 
			\chi_{\left((\pi, \cup G_\ell) \textnormal{ not linked}\right)}
			\prod_{j\in \cup A_\ell} \gamma^{(1)}_N(x_j; x_{\pi(j)}).
\end{equation*}
For our case we only need to consider cases where there are at most two clusters. 
If there is just one cluster then $\rho_{\rm t}^{(A)} = \rho^{(|A|)}( (x_j)_{j\in A})$. 
So suppose we have two clusters $A_1, A_2$.
Here, all the  $\pi$'s for which $(\pi, \cup G_\ell)$ is not linked are exactly those arising as products $\pi = \pi_1\pi_2$,
where $\pi_1 \in \mcS_{A_1}$ and $\pi_2 \in \mcS_{A_1}$ are permutations of the vertices in the 2 clusters. 
Thus,
\begin{equation}\label{eqn.truncated.2.clusters}
\begin{aligned}
\rho_{\rm t}^{(A_1,A_2)}( (x_j)_{j\in A_1\cup A_2})
	& = \sum_{\pi \in \mcS_{A_1\cup A_2}}
			(-1)^\pi 
			\prod_{j\in A_1\cup A_2} \gamma^{(1)}_N(x_{j}; x_{\pi(j)})
	\\ & \quad 
			- 
			\sum_{\pi_1 \in \mcS_{A_1}}
			(-1)^{\pi_1}
			\prod_{j \in A_1} \gamma^{(1)}_N(x_{j}; x_{\pi_1(j)})
			\sum_{\pi_2 \in \mcS_{A_2}}
			(-1)^{\pi_2}			
			\prod_{j\in A_2} \gamma^{(1)}_N(x_{j}; x_{\pi_2(j)})
	\\
	& = \rho^{(|A_1| + |A_2|)}( (x_j)_{j\in A_1\cup A_2})
		- \rho^{(|A_1|)}( (x_j)_{j\in A_1})
			\rho^{(|A_2|)}((x_j)_{j\in A_2}).
\end{aligned}
\end{equation}

\noindent
We now consider the diagrams in 
\Cref{fig.small.2body.diagrams.d=1.type.A1,fig.small.2body.diagrams.d=1.type.A2,fig.small.2body.diagrams.d=1.type.B2} 
and \eqref{eqn.fig.small.2.body.diagram.type.B1}. 
We get 
\begingroup
\allowdisplaybreaks
\begin{equation}\label{eqn.eval.small.1d.diagrams}
\begin{aligned}
\textnormal{Type $(A_1):$} && 
	\sum_{\substack{\pi \in \mcS_4 : (\pi, G_0)\in\mcL_2^2}} \Gamma_{\pi,G_0}^2 
	& = 
		\iint 
		g_{13}g_{14}g_{34}
		\rho_{\rm t}^{(\{1,3,4\},\{2\})}
		\ud x_{3} \ud x_4,
\\
\textnormal{Type $(A_2):$} &&
	\sum_{\substack{\pi \in \mcS_4 : (\pi, G_0)\in\mcL_2^2}} \Gamma_{\pi,G_0}^2 
	& = 
		\iint 
		g_{13}g_{24}
		\rho_{\rm t}^{(\{1,3\}, \{2,4\})}
		\ud x_{3} \ud x_4,
\\
\textnormal{Type $(B_1):$} &&
	\sum_{\substack{\pi \in \mcS_3 : (\pi, G_0)\in\mcL_1^2}} \Gamma_{\pi,G_0}^2 
	& = 
		\int 
		g_{13}g_{23}
		\rho^{(3)}
		\ud x_{3},
\\
\textnormal{Type $(B_2):$} &&
	\sum_{\substack{\pi \in \mcS_5 : (\pi, G_0)\in\mcL_3^2}} \Gamma_{\pi,G_0}^2 
	& = 
		\iiint 
		g_{13}g_{23} g_{45}
		\rho^{(\{1,2,3\},\{4,5\})}_{\rm t}
		\ud x_{3} \ud x_4 \ud x_5.
\end{aligned}
\end{equation}
\endgroup
Using \Cref{eqn.truncated.2.clusters} and (the $1$-dimensional versions of) \Cref{lem.bound.rho3,prop.2.density} 
and similar bounds for the $4$- and $5$-particle reduced densities 
we get the  bounds on the truncated correlations
\begin{align*}
(A_1)
&& \abs{\rho_{\rm t}^{(\{1,3,4\},\{2\})}
} 
& \leq 
\rho^{(4)}(x_1,\ldots,x_4) + \rho^{(3)}(x_1,x_3,x_4) \rho^{(1)}(x_2)
\nonumber
\\ 
&&& \leq 
C \rho^{8} |x_1-x_3|^2 |x_1-x_4|^2,
\\
(A_2)
&&
\abs{\rho_{\rm t}^{(\{1,3\},\{2,4\})}
}
& \leq \rho^{(4)}(x_1,\ldots,x_4) + \rho^{(2)}(x_1,x_3)\rho^{(2)}(x_2,x_4)
\nonumber
\\ &&& 
\leq C \rho^8 |x_1-x_3|^2 |x_2-x_4|^2,
\\
(B_1)
&&
\rho^{(3)}
& \leq C\rho^7 |x_1-x_2|^2 |x_1-x_3|^2,
\\
(B_2)
&&
\abs{\rho^{(\{1,2,3\},\{4,5\})}_{\rm t}
}
& \leq \rho^{(5)}(x_1,\ldots,x_5) + \rho^{(3)}(x_1,x_2,x_3) \rho^{(2)}(x_4,x_5)
\\ &&&
\leq C\rho^{11} |x_1-x_2|^2 |x_1-x_3|^2 |x_4-x_5|^2.
\end{align*}
Bounding moreover, $g_{34}\leq 1$ for the diagram of type $(A_1)$ 
we thus get by the translation invariance
\[
\abs{\sum_{\substack{\pi \in \mcS_4 : (\pi, G_0)\in\mcL_2^2 \\ \textnormal{type $(A_1)$}}} \Gamma_{\pi,G_0}^2 }
  \leq C \rho^{8} \left(\int |g(x)| |x|^2 \ud x\right)^2.
\]
For the diagram of type $(A_2)$ we get 
\[
\abs{\sum_{\substack{\pi \in \mcS_4 : (\pi, G_0)\in\mcL_2^2 \\ \textnormal{type $(A_2)$}}} \Gamma_{\pi,G_0}^2 }
	\leq C \rho^{8} \left(\int |g(x)| |x|^2 \ud x\right)^2.
\]
For the diagram of type $(B_1)$ we get by bounding
$g_{23}\leq 1$ 
(as in the proof of \Cref{lem.small.diagrams.type.C})
\[
\abs{\sum_{\substack{\pi \in \mcS_3 : (\pi, G_0)\in\mcL_1^2 \\ \textnormal{type $(B_1)$}}} \Gamma_{\pi,G_0}^2 }
  \leq C \rho^{7} |x_1-x_2|^2 \int |g(x)| |x|^2 \ud x.
\]
Finally, for the diagram of type $(B_2)$ we get in the same way 
\[
\abs{\sum_{\substack{\pi \in \mcS_5 : (\pi, G_0)\in\mcL_3^2 \\ \textnormal{type $(B_2)$}}} \Gamma_{\pi,G_0}^2 }
  \leq C N \rho^{10} |x_1-x_2|^2 \left(\int |g(x)| |x|^2 \ud x\right)^2.
\]
We may bound $\int |x|^2 |g| \ud x$ similarly as in $3$ and $2$ dimensions, 
\[
	\int_{\R}\left(1 - f(x)^2\right)|x|^2 \ud x
	\leq Ca^3 + \frac{C}{(1-a/b)^2} \int_a^b \left[\left(1 - \frac{a}{b}\right)^2 - \left(1 - \frac{a}{r}\right)^2 \right] r^2 \ud r
	\leq C a b^2.
\]
The other diagrams of types $(A_1)$ and $(B_1)$ (there are no other diagrams of type $(A_2)$ or $(B_2)$) we may treat  similarly
by bounding some of the $g$-edges by $|g|\leq 1$.
Combining these bounds we conclude the proof of \Cref{eqn.2body.small.diagrams.d=1}.

To  prove \Cref{eqn.3body.small.diagrams.d=1} we recall that we consider all diagrams with $g$-graph 
\[
	\begin{tikzpicture}[line cap=round,line join=round,>=triangle 45,x=1.0cm,y=1.0cm]
	\node[anchor = east] at (-0.5,0) {$G_0=$};
	\node (1) at (0,0) {};
	\node (2) at (2,0) {};
	\node (4) at (1,0) {};
	\node (3) at (3,0) {};
	\draw[dashed] (1) -- (4) -- (2);
	\foreach \i in {1,...,4} \draw[fill] (\i) circle [radius=1.5pt];
	\foreach \i in {1,2,3} \node[anchor=north] at (\i) {$\i$};
  \foreach \i in {1,2,3} \node[anchor=south] at (\i) {$*$};
  \node at (5,0) {or};
  \node[anchor = east] at (7.5,0) {$G_1=$};
  \node (n1) at (8,-0.5) {};
  \node (n2) at (8,0.5) {};
  \node (n4) at (8.85,0) {};
  \node (n3) at (9.85,0) {};
  \draw[dashed] (n1) -- (n4) -- (n2);
  \draw[dashed] (n3) -- (n4);
  \foreach \i in {1,...,4} \draw[fill] (n\i) circle [radius=1.5pt];
  \foreach \i in {1,2,3} \node[anchor=north] at (n\i) {$\i$};
  \foreach \i in {1,2,3} \node[anchor=south] at (n\i) {$*$};
	\end{tikzpicture}
\]
(and graphs that look like $G_0$ where $\{1,2,3\}$ are permuted).
One may treat this similarly as the diagrams above, with the result that 
\[
\begin{aligned}
\abs{\sum_{\substack{\pi \in \mcS_4 : (\pi, G_0)\in\mcL_1^3}} \Gamma_{\pi,G_0}^3 }
	& \leq 
		\int |g_{14}||g_{24}| \abs{\rho_{\rm t}^{(\{1,2,4\},\{3\})}
    } \ud x_4
  \leq C a b^2 \rho^{8} |x_1-x_2|^2
\\
\abs{\sum_{\substack{\pi \in \mcS_4 : (\pi, G_1)\in\mcL_1^3}} \Gamma_{\pi,G_1}^3 }
  & \leq 
    \int |g_{14}||g_{24}||g_{34}| \rho^{(4)} \ud x_4
  \leq C a b^2 \rho^{8} |x_1-x_2|^2.
\end{aligned}
\]
Summing this over all the permutations of $\{1,2,3\}$ we conclude the proof of \Cref{eqn.3body.small.diagrams.d=1}.
\end{proof}

\section{Derivative Lebesgue constants 
(proof of Lemma~\texorpdfstring{\ref{lem.derivative.lebesgue.constant}}{\ref*{lem.derivative.lebesgue.constant}})}
\label{sec.derivative.lebesgue.constant}
In this appendix we give the proof of \Cref{lem.derivative.lebesgue.constant}. 
We recall the statement in slightly different notation for convenience.
\begin{replemma}{lem.derivative.lebesgue.constant}
The polyhedron $P$ from \Cref{defn.P_F.true} satisfies for any $\mu,\nu=1,2,3$ that 
\[
  	\int_{[0,2\pi]^3} \abs{\sum_{k \in RP\cap \Z^3} k^\mu e^{ikx}} \ud x \leq C s R (\log R)^3,
  	\quad 
	\int_{[0,2\pi]^3} \abs{\sum_{k \in RP\cap \Z^3} k^\mu k^\nu e^{ikx}} \ud x \leq C s R^{2} (\log R)^4
\]
for sufficiently large $R= \frac{Lk_F}{2\pi}$.
\end{replemma}


\noindent
Recall that by construction $R\sim N^{1/3}$ is rational.

The proof follows quite closely the argument in \cite{Kolomoitsev.Lomako.2018}. 
In particular the structure is that  of induction. The $3$-dimensional integral is bounded one dimension at a time.
We start by introducing some notation from \cite{Kolomoitsev.Lomako.2018}.
\begin{notation}\label{notation.[]}
For any real number $x$ we will write $[x]$ for either $\lfloor x \rfloor$ or $\lceil x \rceil$. 
Similarly we will write $\expect{x} = x - [x]$, i.e. $\expect{x}$ is either the fractional part $\{ x\} = x - \lfloor x \rfloor$ or $x - \lceil x \rceil$.
For any computation we do below, the definition of $[x]$ is fixed, but the computations hold with either choice.

Additionally for a $d$-dimensional vector $x = (x^1, \ldots, x^d)$ we write 
$x^{(\tilde d)} = (x^1, \ldots, x^{\tilde d})$ for the first $\tilde d \leq d$ components.

We emphasize that expressions like $k^2,x^3,\ldots$ do not denote squares or cubes of numbers $k,x$, but instead refer to coordinates 
of vectors $k,x$. 
The instances where we do want to denote a square, cube or higher power should be clear.
\end{notation}

\noindent
By potentially relabelling the coordinates 
it suffices to consider the cases $\mu=1$,  $\mu = \nu = 1$ and $\mu = 1, \nu = 2$.
(Alternatively, by appealing to \Cref{lem.preserve.symmetry} and choosing $Q \gtrsim N^{4}$ in \Cref{defn.P_F.true}
we have a symmetry of coordinates up to error-terms which are subleading compared to \Cref{lem.derivative.lebesgue.constant}.)
Hence define 
\[
	t_1(k) = k^1,
	\qquad 
	t_2(k) = k^1k^1 = (k^1)^2, 
	\qquad 
	t_3(k) = k^1k^2. 
\]
We want to show that
\[
\int_{[0,2\pi]^3} \abs{\sum_{k \in RP\cap \Z^3} t_j(k) e^{ikx}} \ud x 
\leq 
\begin{cases}
C s R (\log R)^{3} & j=1, \\ 
C s R^2 (\log R)^{4} & j=2,3.
\end{cases}
\]
As in the proof of \Cref{lem.lebesgue.constant.fermi.polyhedron} 
we write $RP$ as a union of $O(s)$ closed tetrahedra. 
We also recall that $Rz\notin \Z^3$.
As in the proof of \Cref{lem.lebesgue.constant.fermi.polyhedron}
we get by the inclusion exclusion principle  $O(s)$ terms with tetrahedra of lower dimension (triangles or line segments).
All the $3$-dimensional (closed) tetrahedra are convex and hence of the form 
\[
  T = \left\{
    k\in \Z^3 : \lambda_1\leq k^1 \leq \Lambda_1, 
    \lambda_2(k^1)\leq k^2\leq \Lambda_2(k^1),
    \lambda_3(k^1,k^2)\leq k^3 \leq \Lambda_3(k^1,k^2)
    \right\},
\]
for some piecewise affine functions $\lambda_i,\Lambda_i, i=1,2,3$.
They are the equations of the planes bounding the tetrahedron $T$.
Since any $k\in T$ has integer coordinates we can replace $\Lambda_j$ by $\lfloor \Lambda_j \rfloor$
and $\lambda_j$ by $\lceil \lambda_j \rceil$.
It will be convenient to not distinguish between $\lfloor \cdot \rfloor$ and $\lceil \cdot \rceil$
and use instead the notation $[\cdot]$ introduced in \Cref{notation.[]}.
Then the tetrahedra are of the form
\begin{equation}\label{eqn.generic.tetrahedron}
	T = \left\{
		k\in \Z^3 : [\lambda_1]\leq k^1 \leq [\Lambda_1], 
		[\lambda_2(k^1)]\leq k^2\leq [\Lambda_2(k^1)],
		[\lambda_3(k^1,k^2)]\leq k^3 \leq [\Lambda_3(k^1,k^2)]
		\right\},
\end{equation}
where we allow $[\cdot]$ to be different in any of the $6$ instances it appears. 

Sums over lower-dimensional tetrahedra can be written as differences of sums over $3$-dimensional tetrahedra 
(with potentially different meanings of $[\cdot]$).
We will thus only consider $3$-dimensional tetrahedra. 
That is, for a tetrahedron $T$ of the form \Cref{eqn.generic.tetrahedron}, we need to bound 
\begin{equation}\label{eqn.derivative.lebesgue.constant.general.tetrahedron}
\int_{[0,2\pi]^3} \abs{\sum_{k \in T\cap \Z^3} t_j(k) e^{ikx}} \ud x 
\leq 
\begin{cases}
C R (\log R)^{3} & j=1, \\ 
C R^2 (\log R)^{4} & j=2,3.
\end{cases}
\end{equation}
Gluing together tetrahedra  as in \Cref{lem.lebesgue.constant.fermi.polyhedron} we conclude the desired 
bound, \Cref{lem.derivative.lebesgue.constant}. 
The remainder of this section gives the proof of \Cref{eqn.derivative.lebesgue.constant.general.tetrahedron}.

\subsection{Reduction to simpler tetrahedron}\label{sec.reduction.tetrahedron}
We first reduce to the case of a simpler tetrahedron $T$.
Consider  what happens by shifting all $k$'s by some fixed lattice vector $\kappa \in \Z^3$ with $|\kappa|\leq CR$.
For $t_2$ we have 
\[
\begin{aligned}
	\sum_{k \in T\cap \Z^3} (k^1)^2 e^{ikx}
	& = \sum_{k \in (T - \kappa)\cap \Z^3} (k^1 + \kappa^1)^2 e^{ikx}
	\\ & = \sum_{k \in (T - \kappa)\cap \Z^3} (k^1)^2 e^{ikx} 
		+ 2\kappa^1 \sum_{k \in (T - \kappa)\cap \Z^3} k^1 e^{ikx} 
		+ (\kappa^1)^2 \sum_{k \in (T - \kappa)\cap \Z^3} e^{ikx}.
\end{aligned}
\]
A similar computation holds for $t_1,t_3$.
We may bound $|\kappa| \leq CR$ and thus
we may assume that $T\subset [0, CR]^3$.
(Recall that $\int_{[0,2\pi]^3} \abs{\sum_{k \in T\cap \Z^3} e^{ikx}} \ud x\leq C (\log R)^3$ by \cite[Theorem 4.1]{Kolomoitsev.Lomako.2018}, 
see the proof of \Cref{lem.lebesgue.constant.fermi.polyhedron}.)

For any tetrahedron of the form \eqref{eqn.generic.tetrahedron} we may write the $k$-sum as three $1$-dimensional sums 
\[
	\sum_{k \in T\cap \Z^3} 
	= \sum_{k^1 = [\lambda_1]}^{[\Lambda_1]} 
		\sum_{k^2 = [\lambda_2(k^1)]}^{[\Lambda_2(k^1)]} 
		\sum_{k^3 = [\lambda_3(k^1, k^2)]}^{[\Lambda_3(k^1, k^2)]}
	= \sum_{k^1 = [\lambda_1]}^{[\Lambda_1]} 
		\sum_{k^2 = [\lambda_2(k^1)]}^{[\Lambda_2(k^1)]} 
		\left(\sum_{k^3 = 0}^{[\Lambda_3(k^1, k^2)]} - \sum_{k^3 = 0}^{[\lambda_3(k^1, k^2)-1]}\right),
\]
where the $\lambda_j$'s and $\Lambda_j$'s are the equations of the planes bounding the tetrahedron $T$, 
i.e. piecewise affine functions. 
As in \Cref{eqn.generic.tetrahedron} each instance of $[\cdot]$ may be either of the definitions of \Cref{notation.[]}.
By splitting the $k^1, k^2$ sums into at most 4 parts, we may ensure that both $\Lambda_3$ and $\lambda_3 - 1$ are only from one bounding plane,
i.e. they are affine functions. 
When we do this splitting, we have to choose (in each new tetrahedron) which definition of $[\cdot]$ to use for the new bounding plane. 
This may give rise to some ``boundary term'', if we choose definitions of $[\cdot]$ in the new tetrahedra 
such that the $k$'s on the splitting face are either in both or in neither of the two tetrahedra sharing this face. 
These boundary terms are sums over lower-dimensional tetrahedra, and may thus be bounded by sums over $3$-dimensional ones as above.

\begin{remark}\label{rmk.problem.k2.split}
One may similarly let the $k^1$- and $k^2$-sums go from $0$ by writing e.g. 
\[
	\sum_{k^2 = [\lambda_2(k^1)]}^{[\Lambda_2(k^1)]} = \sum_{k^2 = 0}^{[\Lambda_2(k^1)]} - \sum_{k^2 = 0}^{[\lambda_2(k^1)-1]}.
\]
However, the upper limits $\Lambda_3(k^1, k^2)$ and $\lambda_3(k^1,k^2)-1$ for the $k^3$-sum 
may become much larger than $R$ 
for $k^2 \leq \lambda_2(k^1)$.
This is why we don't do this.
\end{remark}

\noindent
The terms with $\Lambda_3$ and $\lambda_3 - 1$ may be treated the same way, so we just look at the one with $\Lambda_3$.
We thus want to bound 
\[
\int_{[0,2\pi]^3} \abs{
		\sum_{k^1 = [\lambda_1]}^{[\Lambda_1]} 
		\sum_{k^2 = [\lambda_2(k^1)]}^{[\Lambda_2(k^1)]} t_j(k^1, k^2)
		\sum_{k^3 = 0}^{[\Lambda_3(k^1, k^2)]}  e^{ikx}
			} \ud x
	\lesssim \begin{cases}
  R (\log R)^3 & j=1,
  \\
	R^2(\log R)^4 & j=2,3. 
	\end{cases}
\]

\subsection{Reduction from \texorpdfstring{$d=3$ to $d=2$}{d=3 to d=2}}\label{sec.induction.d=3}
We show that we may bound the three-dimensional integrals by analogous two-dimensional integrals 
up to a factor of $(\log R + \log Q) \sim \log N$.

First, before shifting by a constant $\kappa\in\Z^3$, $\Lambda_3$ is given by either the plane through 3 close corners of $RP$ 
(points $R\sigma(p^1/Q_1, p^2/Q_2, p^3/Q_3)$)
or of two close corners and the centre $Rz$. 
This follows from the construction of $P$ in \Cref{defn.P_F.true}, since forming the edges between pairs of close points constructs a triangulation of $P$.

The equation for a plane through the three points $R\sigma(p^1_{j}/Q_1, p^2_{j}/Q_2, p^3_{j}/Q_3)$, $j=1,2,3$ is given by 
\[
	\frac{\alpha_1}{Q_2Q_3} k^1 + \frac{\alpha_2}{Q_1Q_3}k^2 + \frac{\alpha_3}{Q_1Q_2}k^3 = R\sigma \gamma
\]
where by construction of $P$, see \Cref{defn.P_F.true}, we have
\[
	\sigma \notin \Q, \qquad \gamma \in \Q, 
	\qquad \alpha_j \in \Z, 
	\qquad |\alpha_j| \leq C \sqrt{Q},
	\quad j=1,2,3.
\]
We might have that $\alpha_j = 0$. 
If  $\alpha_3=0$ then this plane is parallel to the $k^3$-axis 
and so does not give rise to a bound on the $k^3$-sum. Hence $\alpha_3 \ne 0$.
By choice of $L$, we have that $R$ is rational, and so $R\sigma \gamma \notin \Q$.
(The choice of $L$ such that $R$ is rational, is exactly so that $R\sigma \gamma \notin \Q$.)
The equation for $\Lambda_3$ is an integer shift of this plane, hence it is of the form
\begin{equation}
\label{eqn.Lambda3}
\begin{aligned}
	\Lambda_3(k^1, k^2) 
 	= n_3 - m^1 k^1 - m^2 k^2
	= n_3 - \frac{Q_1\alpha_1}{Q_3 \alpha_3} k^1 - \frac{Q_2\alpha_2}{Q_3 \alpha_3} k^2,
  \\ 
  n_3\notin \Q, 
	\quad |\alpha_j| \leq C\sqrt{Q}, \, j=1,2,3.
\end{aligned}
\end{equation}


\noindent
Define  for $j=1,2,3$ the quantities
\begingroup
\allowdisplaybreaks
\begin{equation}\label{eqn.define.D_d^j}
\begin{aligned}
	D_3^j(x) & := \sum_{k^1 = [\lambda_1]}^{[\Lambda_1]} 
		\sum_{k^2 = [\lambda_2(k^1)]}^{[\Lambda_2(k^1)]} t_j(k^1, k^2)
		\sum_{k^3 = 0}^{[\Lambda_3(k^1, k^2)]}  e^{ik^{(3)}x^{(3)}},
		\\	
	\tilde D_2^j(x) & := \sum_{k^1 = [\lambda_1]}^{[\Lambda_1]} 
		\sum_{k^2 = [\lambda_2(k^1)]}^{[\Lambda_2(k^1)]} t_j(k^1, k^2) e^{ik^{(2)}x^{(2)}},
		\\
	G_3^j(x) & := \frac{1}{e^{ix^3} - 1} \left(e^{i (n_3 + 1) x^3} \tilde D_{2}^j(x^{(2)} - m^{(2)} x^3 ) - \tilde D_{2}^j(x^{(2)})\right),
		\\
	F_3^j(x) & := \frac{e^{i(n_3 + 1)x^3}}{e^{ix^3}-1} 
		\sum_{k^1 = [\lambda_1]}^{[\Lambda_1]} 
		\sum_{k^2 = [\lambda_2(k^1)]}^{[\Lambda_2(k^1)]} t_j(k^1,k^2) e^{ik^{(2)}(x^{(2)} - m^{(2)}x^3)} 
		\left(e^{-i\expect{\Lambda_3(k^{(2)})} x^3} - 1\right),
\end{aligned}
\end{equation}%
\endgroup
where $m^{(2)} = (m^1, m^{2})$ is defined in \Cref{eqn.Lambda3}. 
We shall prove the following bound.
\begin{lemma}\label{lemma.derivative.lebesgue.induction}
We have for some  $k_0^{(2)}\in \Z^2$, some (non-zero)
$\kappa=\kappa^{(2)} \in \Z^2$ and a $h\in \Z$, $h\geq 0$ 
with  $|k^{(2)}_0| \leq CR$ and $h|\kappa^{(2)}| \leq CR$ that 
for any $j=1,2,3$
\[
\begin{aligned}
	& \int_{[0,2\pi]^3} \abs{ D_3^j(x^{(3)})} \ud x^{(3)} 
	\\ & \quad 
	\lesssim (\log R + \log Q )\int_{[0,2\pi]^2} \abs{\tilde D_2^j(x^{(2)})} \ud x^{(2)} 
	+1
	+ \int_{0}^{2\pi} \abs{\sum_{\tau = 0}^h t_j \left(k_0^{(2)} + \tau \kappa^{(2)}\right) e^{i \tau |\kappa^{(2)}| x^{1}}} \ud x^{1}.
\end{aligned}
\]
\end{lemma}

\noindent
As a first step, consider the case where both $\alpha_1=\alpha_2=0$ in \Cref{eqn.Lambda3}.
Then the $k^3$-sum and $x^3$-integral in \Cref{lemma.derivative.lebesgue.induction} factors out.
Using \cite[Lemma 3.2]{Kolomoitsev.Lomako.2018}
to evaluate the $k^3$-sum and $x^3$-integral we conclude the desired.
Hence we can assume that at most one of $\alpha_1, \alpha_2$ is $0$.
(This will be relevant for \Cref{lem.extra.term.integers}, but only then.)

A simple calculation shows that \cite[Lemma 3.1]{Kolomoitsev.Lomako.2018}
\begin{equation}
\label{eqn.lemma.3.1}
D_3^j(x) = G_3^j(x) + F_3^j(x), \quad j=1,2,3.
\end{equation}
By a straightforward modification of the argument in \cite[Lemma 3.3]{Kolomoitsev.Lomako.2018} (including the factor $t_j$) we have
\begin{lemma}[{\cite[Lemma 3.3]{Kolomoitsev.Lomako.2018}}]\label{lem.lemma.3.3}
For any $j=1,2,3$ we have
\[
\int_{[0,2\pi]^{3}} \abs{G_3^j(x)} \ud x \lesssim 	 \log R \int_{[0,2\pi]^{2}} \abs{\tilde D_{2}^j(x^{(2)})} \ud x^{(2)}.
\]
\end{lemma}

\noindent
We thus want to bound the integral of $F_3^j$. 
Again, by a straightforward modification of the argument in \cite[Lemma 3.7]{Kolomoitsev.Lomako.2018} (including the factor $t_j$) we have 
\begin{lemma}[{\cite[Lemma 3.7]{Kolomoitsev.Lomako.2018}}]\label{lem.lemma.3.7}
For any $j=1,2,3$ we have 
\[
	\int_{[0,2\pi]^{3}} \abs{F_3^j(x)} \ud x 
		\lesssim \sum_{r = 1}^\infty \frac{(2\pi)^r}{r!} 
			\int_{[0,2\pi]^{2}} 
			\abs{\sum_{k^1 = [\lambda_1]}^{[\Lambda_1]} 
		\sum_{k^2 = [\lambda_2(k^1)]}^{[\Lambda_2(k^1)]} t_j(k^1,k^2) 
		e^{i k^{(2)}x^{(2)}} \expect{\Lambda_{3}(k^{(2)})}^r } \ud x^{(2)}.
\]
\end{lemma}

\noindent
To bound the right hand side of \Cref{lem.lemma.3.7} we bound either definition of $\expect{\cdot}$ by the fractional part $\{\cdot\}$.
This follows the strategy in \cite{Kolomoitsev.Lomako.2018}.
In analogy with \cite[Lemma 3.6]{Kolomoitsev.Lomako.2018} we have
\begin{lemma}[{\cite[Lemma 3.6]{Kolomoitsev.Lomako.2018}}]\label{lem.lemma.3.6}
For either definition of $\expect{\cdot}$ we have the bound 
\[
\begin{aligned} 	
	& 
	\int_{[0,2\pi]^{2}} 
	\abs{\sum_{k^1 = [\lambda_1]}^{[\Lambda_1]} 
		\sum_{k^2 = [\lambda_2(k^1)]}^{[\Lambda_2(k^1)]} t_j(k^1,k^2) 
		e^{i k^{(2)}x^{(2)}} \expect{\Lambda_{3}(k^{(2)})}^r } 
	\ud x^{(2)}
	\\ & \quad \leq 
		\int_{[0,2\pi]^{2}} 
		\abs{\tilde D_2^j(x^{(2)})}
		\ud x^{(2)}
	+ 
	 	\sum_{\nu = 1}^r \binom{r}{\nu} 
		\int_{[0,2\pi]^{2}} 
		\abs{
			\sum_{k^1 = [\lambda_1]}^{[\Lambda_1]} 
			\sum_{k^2 = [\lambda_2(k^1)]}^{[\Lambda_2(k^1)]} t_j(k^1,k^2) 
			e^{i k^{(2)}x^{(2)}} \{\Lambda_{3}(k^{(2)})\}^\nu
			}
		\ud x^{(2)}
\end{aligned}
\]
uniformly in (integer) $r\geq 1$.
\end{lemma}
\begin{proof}
If $\expect{\cdot} = \{\cdot\}$ this is clear. Hence suppose that $\expect{x} = x - \lceil x \rceil$.
Then
\[
	\expect{x} = \{ x \} - 1 + \begin{cases}
	1 & \textnormal{if } x \in \Z \\ 0 & \textnormal{otherwise}.
	\end{cases}
\]
By construction $\Lambda_3(k^1,k^2)\notin \Z$ for $k^{1},k^2\in\Z$. Thus, $\expect{\Lambda_3(k^1,k^2)} = \{\Lambda_3(k^1,k^2)\}-1$.
Then
\begin{align*}
& 	\sum_{k^1 = [\lambda_1]}^{[\Lambda_1]} 
	\sum_{k^2 = [\lambda_2(k^1)]}^{[\Lambda_2(k^1)]} 
	t_j(k^1,k^2) e^{i k^{(2)}x^{(2)}} \expect{\Lambda_{3}(k^{(2)})}^r
\\ & \quad = 
	\sum_{\nu = 0}^r (-1)^{r - \nu} \binom{r}{\nu} 
	\sum_{k^1 = [\lambda_1]}^{[\Lambda_1]} 
	\sum_{k^2 = [\lambda_2(k^1)]}^{[\Lambda_2(k^1)]} 
	t_j(k^{1},k^2) e^{i k^{(2)}x^{(2)}} \{\Lambda_{3}(k^{(2)})\}^\nu
\\ & \quad = 
	(-1)^r \tilde D_2^j(x^{(2)})
	+ \sum_{\nu = 1}^r (-1)^{r - \nu} \binom{r}{\nu} 
	\sum_{k^1 = [\lambda_1]}^{[\Lambda_1]} 
	\sum_{k^2 = [\lambda_2(k^1)]}^{[\Lambda_2(k^1)]} 
	t_j(k^{1},k^2) e^{i k^{(2)}x^{(2)}} \{\Lambda_{3}(k^{(2)})\}^\nu.
\qedhere
\end{align*}
\end{proof}

\noindent
We now bound the second summand of \Cref{lem.lemma.3.6} 
similarly to \cite[Lemmas 3.8 and 3.9]{Kolomoitsev.Lomako.2018}. 
We first define $\tilde\Lambda_3$, a rational approximation of $\Lambda_3$.
Recall the definition of $\Lambda_3$ in \Cref{eqn.Lambda3}. 
By Dirichlet's approximation theorem we may for any $Q_\infty$ find integers $p,q$ with $1\leq q\leq Q_\infty$ such that 
\[
	\gamma_3 := n_3 - \frac{p}{q} \qquad \textnormal{satisfies} \qquad  \abs{\gamma_3} < \frac{1}{qQ_\infty}.
\]
We will choose $Q_\infty = Q_3 \alpha_3$.
Define then 
\begin{equation}\label{eqn.define.Lambda3tilde}
	\tilde \Lambda_3(k^{(2)})
		= \Lambda_3(k^{(2)}) - \gamma_3 
		= \frac{p}{q} - \frac{Q_1\alpha_1}{Q_3 \alpha_3} k^1 - \frac{Q_2\alpha_2}{Q_3 \alpha_3} k^2.
\end{equation}
Note that this takes values in $\frac{1}{qQ_\infty}\Z$ for integers $k^1,k^2$.
In particular 
(for integers $k^1, k^2$)
$\{\tilde \Lambda_3(k^{(2)})\} \in \{0, \frac{1}{qQ_\infty}, \ldots, \frac{qQ_\infty - 1}{qQ_\infty}\}$.
Thus, since $|\gamma_3| < \frac{1}{qQ_\infty}$ we have
\begin{equation}\label{eqn.Lambda3.Lambda3tilde.fractional.part}
	\{\Lambda_3(k^{(2)})\} = \gamma_3 + \{ \tilde{\Lambda}_3(k^{(2)}) \} + \begin{cases}
	1 & \textnormal{if }  \gamma_3 < 0 \textnormal{ and } \tilde \Lambda_3(k^{(2)}) \in \Z, 
	\\ 0 & \textnormal{otherwise}.
	\end{cases}
\end{equation}
We claim that 
\begin{lemma}\label{lem.lemma.3.9}
For $N$ sufficiently large, we have uniformly in (integer) $r\geq 1$ that 
\[
\begin{aligned}
	& \int_{[0,2\pi]^2} 
	\abs{ \sum_{k^1 = [\lambda_1]}^{[\Lambda_1]} 
	\sum_{k^2 = [\lambda_2(k^1)]}^{[\Lambda_2(k^1)]} 
	t_j(k^{1},k^2) e^{i k^{(2)}x^{(2)}} \{\Lambda_{3}(k^{(2)})\}^r 
	} \ud x^{(2)} 
	\\ & \quad 
	\lesssim \log(rQ) \int_{[0,2\pi]^2} \abs{ \tilde D_2^j(x^{(2)})} \ud x^{(2)} 
	+ 
	\int_{[0,2\pi]^2}  
		\abs{\sum_{k^1 = [\lambda_1]}^{[\Lambda_1]} 
			\sum_{\substack{k^2 = [\lambda_2(k^1)] \\ \tilde \Lambda_3(k^1,k^2) \in \Z}}^{[\Lambda_2(k^1)]}
 			t_j(k^{1},k^2) e^{i k^{(2)}x^{(2)}} }
 		\ud x^{(2)}
	+ 2^r 
\end{aligned}
\]
\end{lemma}

\noindent
The proof differs from that of \cite[Lemmas 3.8 and 3.9]{Kolomoitsev.Lomako.2018} in a few key location, 
so we give it here.
\begin{proof}
Using \Cref{eqn.Lambda3.Lambda3tilde.fractional.part} we have 
\[
\begin{aligned}
	& \sum_{k^1 = [\lambda_1]}^{[\Lambda_1]} 
	\sum_{k^2 = [\lambda_2(k^1)]}^{[\Lambda_2(k^1)]} 
	t_j(k^{1},k^2) e^{i k^{(2)}x^{(2)}} \{\Lambda_{3}(k^{(2)})\}^r
	\\ & \quad 
		= \sum_{k^1 = [\lambda_1]}^{[\Lambda_1]} 
			\sum_{k^2 = [\lambda_2(k^1)]}^{[\Lambda_2(k^1)]} 
			t_j(k^{1},k^2) e^{i k^{(2)}x^{(2)}} \left(\gamma_3 + \{ \tilde{\Lambda}_3(k^{(2)}) \}\right)^r 
	\\ & \qquad
		+ \chi_{(\gamma_3 < 0)}
			\sum_{k^1 = [\lambda_1]}^{[\Lambda_1]} 
			\sum_{\substack{k^2 = [\lambda_2(k^1)] \\ \tilde \Lambda_3(k^1,k^2) \in \Z}}^{[\Lambda_2(k^1)]}
 			t_j(k^{1},k^2) e^{i k^{(2)}x^{(2)}} 
		+ \textnormal{mixed terms}.
\end{aligned}
\]
All the mixed terms have at least one power of $\gamma_3 + \{ \tilde{\Lambda}_3(k^{(2)})\} = \gamma_3$. 
(Indeed, in the mixed terms we have $\tilde\Lambda_3(k^{(2)})\in \Z$ so $\{ \tilde{\Lambda}_3(k^{(2)})\} = 0$.)
Since $|\gamma_3| < 1/(qQ_\infty) \leq 1/Q$
the sum of all mixed terms may be bounded by $2^r R^4 Q^{-1} \lesssim 2^r$ for $N$ sufficiently large 
(independent of $r$) by our choice of $Q$, see \Cref{defn.P_F.true}.
Similarly expanding the first summand, all the terms with at least one power of $\gamma_3$ may be bounded the same way. 
We thus have 
\begin{equation}
\label{eqn.decompose.lemma.3.9}
\begin{aligned}
	& \sum_{k^1 = [\lambda_1]}^{[\Lambda_1]} 
	\sum_{k^2 = [\lambda_2(k^1)]}^{[\Lambda_2(k^1)]}  
	t_j(k^{1},k^2) e^{i k^{(2)}x^{(2)}} \{\Lambda_{3}(k^{(2)})\}^r
	\\ & \quad 
		= \sum_{k^1 = [\lambda_1]}^{[\Lambda_1]} 
		\sum_{k^2 = [\lambda_2(k^1)]}^{[\Lambda_2(k^1)]}  
		t_j(k^{1},k^2) e^{i k^{(2)}x^{(2)}} \{ \tilde{\Lambda}_3(k^{(2)}) \}^r 
	\\ & \qquad 
		+ \chi_{(\gamma_3 < 0)}\sum_{k^1 = [\lambda_1]}^{[\Lambda_1]} 
			\sum_{\substack{k^2 = [\lambda_2(k^1)] \\ \tilde \Lambda_3(k^1,k^2) \in \Z}}^{[\Lambda_2(k^1)]}
 			t_j(k^{1},k^2) e^{i k^{(2)}x^{(2)}}  
		+ O(2^r),
\end{aligned}
\end{equation}
where the error is $O(2^r)$ uniform in $x^{(2)}$.
For the first summand we have by a simple modification of \cite[Lemma 3.8]{Kolomoitsev.Lomako.2018} (including the factor $t_j$) that 
\[
	\int_{[0,2\pi]^2} 
	\abs{\sum_{k^1 = [\lambda_1]}^{[\Lambda_1]} 
		\sum_{k^2 = [\lambda_2(k^1)]}^{[\Lambda_2(k^1)]}  
		t_j(k^{1},k^2) e^{i k^{(2)}x^{(2)}} \{ \tilde{\Lambda}_3(k^{(2)}) \}^r} \ud x^{(2)}
	\lesssim \log(rqQ_\infty) \int_{[0,2\pi]^2} \abs{ D_2^j(x^{(2)}) } \ud x^{(2)}.
\]
This importantly uses that 
$\{\tilde \Lambda_3(k^{(2)})\} \in \{0, \frac{1}{qQ_\infty}, \ldots, \frac{qQ_\infty - 1}{qQ_\infty}\}$ 
for integers $k^1, k^2$, so that on can find some smartly chosen function $h(u)\approx u^r$ 
on $[0,1]$ but with a smooth cut-off at $1$ and $h(\{\tilde \Lambda_3(k^{(2)})\}) = \{\tilde \Lambda_3(k^{(2)})\}^r$ 
for which one can bound Fourier coefficients, see \cite[Lemma 3.8]{Kolomoitsev.Lomako.2018}.

We have $q\leq Q_\infty = Q_3 \alpha_3 \leq C Q^{3/2}$. 
We conclude the desired.
\end{proof}

\noindent
Next we bound the second term in \Cref{lem.lemma.3.9}, where $\tilde \Lambda_3$ is integer.
If there are no valid choices of $k^1,k^2$ for which $\tilde \Lambda_3(k^1,k^2)$ is an integer, then this term is clearly zero.
Otherwise we have the following.
\begin{lemma}\label{lem.extra.term.integers}
Let $N$ be sufficiently large and suppose that the set 
\[
I_0 = \left\{ (k^1,k^2)\in \Z^2 \, : \,  [\lambda_1] \leq k^1 \leq [\Lambda_1], \,  [\lambda_2(k^1)]\leq k^2 \leq [\Lambda_2(k^1)], \, \tilde \Lambda_3(k^1,k^2) \in \Z\right\}
\]
is non-empty. 
Then we may find a point $k_0^{(2)} \in I_0$, 
a (non-zero) lattice vector $\kappa = \kappa^{(2)}\in \Z^2$ and an integer $h \geq 0$ with $k^{(2)}_0 + h\kappa \in I_0$ 
(in particular $h|\kappa|\lesssim R$) such that 
$I_0 = \{ k^{(2)}_0 + \tau \kappa^{(2)} : \tau \in \{0, \ldots, h\}\}$. In particular
\begin{equation}\label{eqn.lem.number.theory}
\int_{[0,2\pi]^2}  
		\abs{\sum_{k^1 = [\lambda_1]}^{[\Lambda_1]} 
			\sum_{\substack{k^2 = [\lambda_2(k^1)] \\ \tilde \Lambda_3(k^1,k^2) \in \Z}}^{[\Lambda_2(k^1)]}
 			t_j(k^{1},k^2) e^{i k^{(2)}x^{(2)}} }
 		\ud x^{(2)}
 	\lesssim
 	\int_{0}^{2\pi} \abs{\sum_{\tau = 0}^h t_j \left(k_0^{(2)} + \tau \kappa^{(2)}\right) e^{i \tau |\kappa^{(2)}| x}} \ud x.
\end{equation}
\end{lemma}
\noindent
The proof is an exercise in elementary number theory analysing the set $I_0$.
\begin{proof}
Define $k^{(2)}_0$ to be any point in the (non-empty) set $I_0$. 
Recall \Cref{eqn.define.Lambda3tilde}, and that 
$\abs{\tilde \Lambda_3(k^1,k^2)}\leq CR$ for any $k^{(2)} \in I_0$.
(This follows since the relevant tetrahedron is contained in $[0,CR]^3$.)
By redefining $\alpha_j$ as $\alpha_j/\gcd(\alpha_1,\alpha_2,\alpha_3)$ we may assume that $\alpha_1,\alpha_2,\alpha_3$ have no shared prime factors.
(This only decreases their values, so that still $|\alpha_j| \leq C \sqrt{Q}$.)
In case one of the $\alpha_j$'s is zero we will use the convention that $\gcd(\alpha, \beta,0) = \gcd(\alpha,\beta)$ and $\gcd(\alpha,0) = \alpha$
for $\alpha, \beta> 0$.

\paragraph*{Solving the general problem.}
We first consider the general problem of finding all $k^1,k^2 \in \Z$ for which $\tilde \Lambda_3(k^1,k^2)$ is an integer.
This set has the form $k^{(2)}_0 + \Gamma$ for some two-dimensional lattice $\Gamma$.
We now find spanning lattice vectors of $\Gamma$.

Define $\alpha_{ij} = \gcd(\alpha_i,\alpha_j)$ for $i\ne j$. 
(Note that the $\alpha_j$'s are not necessarily pairwise coprime, only all $3$ $\alpha_j$'s have no shared factor by the reduction above.
Also, since $\alpha_1$ and $\alpha_2$ are not both $0$, we have $\alpha_{12} \ne 0$ is well-defined.)
Shifting $k^{(2)}_0$ by $\kappa_0 := (Q_2 \frac{\alpha_2}{\alpha_{12}}, -Q_1 \frac{\alpha_1}{\alpha_{12}})$
we have 
\[
	\tilde \Lambda_3 ( k^{(2)}_0 + b \kappa_0) = \tilde \Lambda_3 ( k^{(2)}_0 ) \in \Z, \qquad b\in \Z
\]
and $\kappa_0$ is the shortest lattice vector with this property. 
One should note here that $\kappa_0$ is not ``short''. Indeed $|\kappa_0| \gtrsim Q$ since both $Q_1,Q_2\gtrsim Q$, see \Cref{defn.P_F.true}, 
and $\alpha_1,\alpha_2$ are not both $0$.
We now look for  
the lattice vector in $\Gamma$ giving the smallest possible (integer) increase of $\tilde \Lambda_3$.
This lattice vector together with $\kappa_0$ spans $\Gamma$.
Note that 
\begin{equation}\label{eqn.def.delta.Lambda}
	\delta \tilde \Lambda_3(\kappa)
	:= 
	\tilde \Lambda_3(k^{(2)}_0 + \kappa) - \tilde \Lambda_3(k^{(2)}_0)
		=  \frac{-Q_1\alpha_1\kappa^1 - Q_2\alpha_2\kappa^2}{Q_3 \alpha_3}.
\end{equation}
Suppose first that either $\alpha_1=0$ or $\alpha_2=0$, say $\alpha_2=0$. 
Now, $Q_1 \ne Q_3$ and $|\alpha_j|\leq C\sqrt{Q}$ so $Q_j$ is not a factor of $\alpha_i$ for any $i=1,3, j=1,2,3$.
Thus, $\gcd(Q_3\alpha_3,Q_1\alpha_1) = \gcd(\alpha_1,\alpha_3) = 1$ since $\alpha_2 = 0$.
For the ratio $\delta \tilde \Lambda_3(\kappa)$ to be an integer we need that the numerator is some multiple of $Q_3\alpha_3$,
and thus that $|\kappa| \gtrsim Q_3 \gg R$.
Thus there is at most one $k_0^{(2)}\in I_0$ and the lemma is clear.

Suppose then that $\alpha_1\ne 0$, $\alpha_2\ne 0$.
Varying $\kappa\in \Z^2$ we have by Bézout's lemma that the numerator in \Cref{eqn.def.delta.Lambda} assumes as values 
all multiples of $\gcd(Q_1\alpha_1, Q_2\alpha_2)$.
We have  $\gcd(Q_1\alpha_1, Q_2\alpha_2) = \gcd(\alpha_1, \alpha_2) = \alpha_{12}$.
For the ratio $\delta \tilde \Lambda_3(\kappa)$ to be an integer we need that the numerator is some multiple of $Q_3\alpha_3$. 
Since by assumption there are no prime factors shared by all $\alpha_j$'s and $Q_3$ is not a factor of $\alpha_{12}$ 
we have $\gcd(\alpha_{12}, Q_3\alpha_3) = 1$.
Thus, the smallest integer increase of $\tilde \Lambda_3$ is $\alpha_{12}\geq 1$
and this happens along some lattice vector $\kappa_1$.
Immediately then $\Gamma \supset \{a\kappa_1 + b\kappa_0 : a,b\in \Z\}$.
To see that $\Gamma \subset \{a\kappa_1 + b\kappa_0 : a,b\in \Z\}$ note that by Bézout's lemma the (integer) solutions  to the equation 
\[
  -Q_1\alpha_1\kappa^1 - Q_2\alpha_2\kappa^2 = Q_3\alpha_3 A,
\]
for some integer $A\in \Z$, is exactly $(\kappa^1,\kappa^2)\in \left\{\frac{A}{\alpha_{12}}\kappa_1 + b\kappa_0 : b\in \Z\right\}$
if $\alpha_{12}$ divides $A$ and there are no solutions otherwise.
In summary then
\begin{equation}\label{eqn.delta.Lambda.integer}
\Gamma = \{ a \kappa_1 + b \kappa_0 : a,b \in \Z\},
\qquad 
\tilde \Lambda_3( k^{(2)}_0 + a\kappa_1 + b \kappa_0) = \tilde \Lambda_3(k^{(2)}_0) + a \alpha_{12},
\quad a,b\in \Z.
\end{equation}
Moreover
\[
	I_0 = \left(k^{(2)}_0 + \Gamma\right) \cap \left\{ (k^1,k^2)\in \Z^2 \,:\, [\lambda_1] \leq k^1 \leq [\Lambda_1],\,  [\lambda_2(k^1)]\leq k^2 \leq [\Lambda_2(k^1)]\right\}.
\]
\paragraph*{Finding the candidate for $\kappa$.}
We now find the candidate for the $\kappa$ in the lemma. 
Either $I_0 = \{k^{(2)}_0\}$, in which case the lemma is clear (take $h=0$), 
or there exists some (non-zero) $\kappa = a\kappa_1 + b\kappa_0\in \Gamma$ such that $k^{(2)}_0 + \kappa \in I_0$. 
For such $\kappa$ we have (for sufficiently large $N$) that $a\ne 0$ as $|\kappa_0| \gtrsim Q \gg R$ and any such $\kappa$ has $|\kappa| \leq CR$.
Let $\kappa_2 = a_2 \kappa_1 + b_2 \kappa_0$ be the $\kappa$ such that $k^{(2)}_0 + \kappa \in I_0$
with minimal value of $|a_2|$. 
($\kappa_2$ is unique up to potentially a sign if both $k^{(2)}_0 - \kappa_2\in I_0$ and $k^{(2)}_0 + \kappa_2 \in I_0$.)
It follows from \Cref{eqn.delta.Lambda.integer} that 
$|a_2|\leq CR/\alpha_{12} \leq CR$ since $|\delta\tilde\Lambda_3(\kappa_2)|\leq CR$ as the tetrahedron is contained in $[0,CR]^3$.

If $b_2 = 0$ then $a_2=\pm 1$, else if $b_2\ne 0$ then $\gcd(a_2, b_2) = 1$. 
Indeed, if $a_2$ and $b_2$ shared some common factor, 
we could factor this out to find a $\kappa$ with smaller value $|a|$ contradicting the minimality of $|a_2|$.

\paragraph*{Characterizing all allowed $\kappa$'s.}
We  claim that by potentially redefining $k^{(2)}_0$ to $k^{(2)}_0 - a\kappa_2$ with $a\in \Z$ largest such that still $k^{(2)}_0 - a\kappa_2 \in I_0$
we have that 
\begin{equation}\label{eqn.form.Ik.claimed}
\begin{aligned}
	I_0 
	& = \{ k^{(2)}_0 + \tau \kappa_2 : \tau\in \{0,\ldots,h\} \},
	\qquad \textnormal{for some }  h \in \Z, h\geq 0.
\end{aligned}
\end{equation}
(The intuition for the remainder of the argument is as follows.
Essentially, if some $\kappa$ had $k^{(2)}_0 + \kappa \in I_0$ but was not a multiple of $\kappa_2$, 
it would have to differ from some multiple of $\kappa_2$ by at least $\kappa_0$ or $\kappa_1$. 
Since $|\kappa_0| \gg R$ and either $\kappa_1 = \kappa_2$ or $|\kappa_1|\gg R$, this is impossible.)
To prove \Cref{eqn.form.Ik.claimed} we first introduce the following notation.
We view  a lattice vector $\kappa\in\Z^2$ as a vector $\kappa\in\R^2$ and write
$\kappa^\parallel$ for its component parallel to $\kappa_0$. 
Note that $\kappa^\parallel$ need not have integer coordinates.
Define the constant $A$ such that $\kappa_1^\parallel = A \kappa_0$. (Note that $A$ need not be an integer.)
Let  $0\ne \kappa = a\kappa_1 + b\kappa_0\in \Gamma$ with $k^{(2)}_0 + \kappa \in I_0$.
We have 
\[
  \kappa^{\parallel} = a\kappa_1^\parallel + b \kappa_0
    = (aA + b)\kappa_0.
\]
Thus, since $|\kappa_0| \gtrsim Q$, $|\kappa|\lesssim CR$ and $|a| \geq 1$ (since $\kappa \ne 0$) 
we have $\abs{\frac{b}{a} + A} \leq \frac{CR}{Q}$.


Using this also for $\kappa_2 = a_2\kappa_1 + b_2\kappa_0$ we get 
\[
	\abs{ba_2 - b_2 a}
	= 
	\abs{\frac{b}{a} - \frac{b_2}{a_2}}|a a_2| \leq |a a_2| \left(\abs{\frac{b}{a} + A} + \abs{-\frac{b_2}{a_2} - A}\right) \leq CR^2\frac{R}{Q}
	\ll 1.
\]
But $ba_2 - b_2 a$ is an integer. Hence (for $N$ sufficiently large) we have $ba_2 = b_2 a$.
Now, if $b_2=0$ then $b=0$ and so $a_2=\pm 1$ is a divisor of $a$ so $\kappa = \pm a\kappa_2$.
If $b_2\ne 0$ then $\gcd(a_2,b_2)=1$ and thus $a_2$ is again a divisor of $a$ and $a/a_2 = b/b_2$.
Then $\kappa = \frac{a}{a_2} \kappa_2$ is a multiple of $\kappa_2$. This shows the desired.

\paragraph*{Integral form.}
To prove \Cref{eqn.lem.number.theory} 
we do the following. 
Define $e_2 = \kappa_2/|\kappa_2|$ as the unit vector parallel to $\kappa_2$ and 
$e_2^\perp$ as the unit vector perpendicular to $\kappa_2$.
Then define the domain
\[
	S_0 := \left\{ x^{(2)}\in \R^2 : \abs{x^{(2)} \cdot e_2} \leq 4\pi, \abs{x^{(2)}\cdot e_2^\perp}\leq 4\pi \right\}
\]
and note that $[0,2\pi]^2 \subset S_0$. Thus, using \Cref{eqn.form.Ik.claimed}
\[
\begin{aligned}
\int_{[0,2\pi]^2}  
		\abs{\sum_{k\in I_0}
 			t_j(k^{1},k^2) e^{i k^{(2)}x^{(2)}} }
 		\ud x^{(2)}
 	& \leq
 	\int_{S_0} \abs{\sum_{\tau = 0}^h t_j \left(k_0^{(2)} + \tau \kappa^{(2)}\right) e^{i \tau \kappa^{(2)} x^{(2)}}} \ud x^{(2)}.
\end{aligned}
\]
The integrand is constant in the $e_2^\perp$-direction, and $2\pi$-periodic in the $e_2$-direction. 
Thus, computing the integral in these coordinates we have
\[
\begin{aligned}
 	\int_{S_0} \abs{\sum_{\tau = 0}^h t_j \left(k_0^{(2)} + \tau \kappa^{(2)}\right) e^{i \tau \kappa^{(2)} x^{(2)}}} \ud x^{(2)}
 	& =
 	32\pi  \int_{0}^{2\pi} \abs{\sum_{\tau = 0}^h t_j \left(k_0^{(2)} + \tau \kappa^{(2)}\right) e^{i \tau |\kappa^{(2)}| x}} \ud x.
\end{aligned}
\]
This concludes the proof.
\end{proof}

\noindent
Combining  \Cref{lem.lemma.3.7,lem.lemma.3.9,lem.lemma.3.6,lem.extra.term.integers} the $r$- and $\nu$-sums 
in \Cref{lem.lemma.3.7,lem.lemma.3.6}
are readily bounded because of the factor $1/r!$ from 
\Cref{lem.lemma.3.7}. We conclude that 
\[
	\int_{[0,2\pi]^3} \abs{ F_3^j(x)} \ud x 
	\lesssim \log Q \int_{[0,2\pi]^2} \abs{\tilde D_2^j(x)} \ud x 
	+ 1 
	+ 
	\int_{0}^{2\pi} \abs{\sum_{\tau = 0}^h t_j \left(k_0^{(2)} + \tau \kappa^{(2)}\right) e^{i \tau |\kappa^{(2)}| x^{1}}} \ud x^{1},
\]
where $k^{(2)}_0$ and $\kappa^{(2)}$ are as in \Cref{lem.extra.term.integers}.
If the set $I_0$ from \Cref{lem.extra.term.integers} is empty, then the bound is valid without the last term. 
In particular it is valid with any $k^{(2)}_0 \in [0,CR]^2$, (non-zero) $\kappa=\kappa^{(2)} \in \Z^2$ and $h=0$.
Thus, by \Cref{lem.lemma.3.3,eqn.lemma.3.1} we prove the desired bound, \Cref{lemma.derivative.lebesgue.induction}.

\subsection{Reduction from \texorpdfstring{$d=2$ to $d=1$}{d=2 to d=1}}\label{sec.induction.d=2}
For $j=1,2$ we will do one more step reducing the dimension. 
The argument is basically the same as for going from dimension $d=3$ to $d=2$ in \Cref{sec.induction.d=3}. 
We sketch the main differences.

As we did in \Cref{sec.reduction.tetrahedron} for $d=3$ by adding and subtracting the lower tail of the sum, 
we may assume that the $k^2$-sum is $\sum_{k^2 = 0}^{[\Lambda_2(k^1)]}$. 
\begin{remark}\label{rmk.k2.split}
It is valid here to make the $k^2$-sum go from $0$, 
since now the $k^2$-sum is the innermost sum and we do not risk 
values of $k^3$ much larger that $R$ by doing so (as in \Cref{rmk.problem.k2.split}).
Indeed, we already computed the sum over the relevant $k^3$.
We could at this point also do the same splitting of the $k^1$-sum, but we would have the same problems 
that $\Lambda_2(k^1)$ or $\lambda_2(k^1)$ might be much larger than $R$ for $k^1 \leq \lambda_1$
as in \Cref{rmk.problem.k2.split}.
\end{remark}

\noindent
Additionally, by splitting the $k^1$-sum into at most $2$ parts, 
we may assume that $\Lambda_2$ is just the equation for a line.
Here again one needs to be careful with what to do with the boundary terms. 
This gives some sums over $1$-dimensional tetrahedra (i.e. line segments),
which we can write as differences of sums over $2$-dimensional tetrahedra exactly as for the $3$-dimensional case.
We are led to define the quantities 
\begingroup
\allowdisplaybreaks
\begin{align*}
	D_2^j(x) & := \sum_{k^1 = [\lambda_1]}^{[\Lambda_1]} t_j(k^1)
		\sum_{k^2 = 0}^{[\Lambda_2(k^1)]} e^{ik^{(2)}x^{(2)}},
		\\	
	\tilde D_1^j(x) & := \sum_{k^1 = [\lambda_1]}^{[\Lambda_1]} t_j(k^1) e^{ik^{1}x^{1}},
		\\
	G_2^j(x) & := \frac{1}{e^{ix^2} - 1} \left(e^{i (n_2 + 1) x^2} \tilde D_{1}^j(x^{1} - m_1 x^2 ) - \tilde D_{1}^j(x^{1})\right),
		\\
	F_2^j(x) & := \frac{e^{i(n_2 + 1)x^2}}{e^{ix^2}-1} 
		\sum_{k^1 = [\lambda_1]}^{[\Lambda_1]} t_j(k^1) e^{ik^{1}(x^{1} - m_1x^2)} 
		\left(e^{-i\expect{\Lambda_2(k^{1})} x^2} - 1\right).
\end{align*}%
\endgroup
We  claim the following inductive bound. 
\begin{lemma}\label{lemma.derivative.lebesgue.induction.d=2}
For  $j=1,2$ we have for $N$ sufficiently large that 
\[
	\int_{[0,2\pi]^2} \abs{ D_2^j(x^{(2)})} \ud x^{(2)} \lesssim (\log R + \log Q) \int_{[0,2\pi]} \abs{\tilde D_{1}^j(x^{1})} \ud x^{1} 
	+ \begin{cases}
	R & j=1, \\ R^2 & j=2.
	\end{cases}	
\]
\end{lemma}
\begin{proof}
As for $\Lambda_3$, we have that  
the equation of a line between any two points $(p^1_i/Q_1, p^2_i/Q_2)$, $i=1,2$ is given by 
\[
	\frac{p_1^1 - p_2^1}{Q_2} k^1 + \frac{p^2_2 - p^2_1}{Q_1}k^2 = \textnormal{const}\,.
\]
If we choose the points to be either corners of $RP$ or the central point $Rz$
we get the equation 
\[
	\frac{\alpha_1}{Q_2} k^1 + \frac{\alpha_2}{Q_1} k^2 = R\sigma \gamma \notin \Q.
\]
Here we might have that $\alpha_1=0$ or $\alpha_2=0$. 

If $\alpha_2=0$ this line is parallel to the $k^2$-axis and so does not give rise to a bound for the $k^2$-sum.
Thus $\alpha_2\ne 0$.
If $\alpha_1=0$ the sum in $D_2^j(x)$ and integral thereof factorizes, and hence by \cite[Lemma 3.2]{Kolomoitsev.Lomako.2018}
we have that 
\[
	\int_{[0,2\pi]^2} \abs{ D_2^j(x^{(2)})} \ud x^{(2)} \leq C \log R \int_{0}^{2\pi} \abs{\tilde D_1^j(x^1)} \ud x^1.
\]
Hence, this case yields the desired inductive bound, \Cref{lemma.derivative.lebesgue.induction.d=2}.
Suppose then $\alpha_1,\alpha_2 \ne 0$.

Then 
\[
	\Lambda_2(k^1) = n_2 - m_1 k^1 = n_2 - \frac{Q_1\alpha_1}{Q_2 \alpha_2} k^1, \qquad n_2 \notin \Q, \qquad |\alpha_j|\leq C Q^{1/4}, \quad j=1,2.
\]
\Cref{lem.lemma.3.3,lem.lemma.3.7,lem.lemma.3.6} are readily adapted and proven as before. 
The adaptation of \Cref{lem.lemma.3.9} is then mostly analogous. One chooses $Q_\infty=Q_2\alpha_2$
and finds the rational approximation of $\Lambda_2$ as 
\[
	\tilde\Lambda_2(k^1) = \Lambda_2(k^1) - \gamma_2 = \frac{p}{q} - \frac{Q_1\alpha_1}{Q_2\alpha_2}k^1,
	\qquad 
	|\gamma_2| < \frac{1}{qQ_\infty}\leq Q^{-1}.	
\]
The rest of the argument follows exactly as for $d=3$ only that the extra term of the sum where $\tilde\Lambda_2(k^1)\in \Z$ may be bounded as follows.
\[
	\abs{\sum_{\substack{k^1 = [\lambda_1] \\ \tilde \Lambda_2(k^1) \in \Z}}^{[\Lambda_1]}
 			t_j(k^{1}) e^{i k^{1}x^{1}}  
 	}
 	\leq \begin{cases}
 	R & j=1 \\ R^2 & j=2
 	\end{cases}
\]
since there is at most one $k^1$ such that $\tilde \Lambda_2(k^1)$ is an integer. 
To see this note that $\gcd(\alpha_1,Q_2) = 1$ since $|\alpha_1| \leq CQ^{1/4} \ll Q_2$, hence the change in $k^1$ to change $\tilde \Lambda_2(k^1)$ by an integer is at least $Q_3 \gg R$.
We thus conclude the desired bound.
\end{proof}

\subsection{Bounding the one-dimensional integrals}\label{sec.1d.integrals}
Now we bound $\int |\tilde D_1^j|$ and $\int \abs{\sum_{\tau = 0}^h t_j \left(k_0^{(2)} + \tau \kappa\right) e^{i \tau |\kappa| x}} \ud x$
from the right-hand-sides of \Cref{lemma.derivative.lebesgue.induction.d=2,lemma.derivative.lebesgue.induction}.
For $\tilde D_1^j$ we may assume that the lower bound of the summations are at $0$ 
by the same procedure as in \Cref{sec.reduction.tetrahedron}.
Expanding $t_j \left(k_0^{(2)} + \tau \kappa\right)$ we see that $j=1$ gives an affine expression in $\tau$ 
and $j=2,3$ give quadratic expressions in $\tau$.
For instance,
\[
	t_2\left(k_0^{(2)} + \tau \kappa\right) = (k_0^1)^2 + 2 k_0^1 \kappa^1 \tau + (\kappa^1)^2 \tau^2.
\]
Thus, bounding both the integrals amounts to bounding 
the following:
\begin{lemma}\label{lem.eval.1d.integrals}
Let $M\geq 2$ be an integer. Then
\begin{enumerate}[(1)]
\item $\displaystyle \int_0^{2\pi} \abs{ \sum_{k=0}^{M} e^{ikx}} \ud x \leq C \log M$,
\item $\displaystyle \int_0^{2\pi} \abs{ \sum_{k=0}^{M} k e^{ikx}} \ud x \leq CM \log M$,
\item $\displaystyle \int_0^{2\pi} \abs{ \sum_{k=0}^{M} k^2 e^{ikx}} \ud x \leq CM^2 \log M$.
\end{enumerate}
\end{lemma}
\begin{proof}
The bound $(1)$ is elementary, see also \cite[Lemma 3.2]{Kolomoitsev.Lomako.2018}.
For any $M\in \N$ and $q \in \C\setminus\{1\}$ we have
\begingroup
\allowdisplaybreaks
\begin{equation}
\label{eqn.sum.kq^k}
\begin{aligned}
	\sum_{k=0}^M q^k 
	& = \frac{q^{M+1}-1}{q-1}
	\\
	\sum_{k=0}^M kq^k 
	& = \frac{q}{(q-1)^2} \left[q^M(Mq - M-1) + 1\right]
	\\
	\sum_{k=0}^M k^2 q^k 
	& = \frac{q}{(q-1)^3}\left[ q^M \left(M^2 (q-1)^2 - 2M (q-1) + q + 1\right) - q - 1\right].
\end{aligned}
\end{equation}
\endgroup
Consider now the integrals $(2)$ and $(3)$. 
By symmetry of complex conjugation $\int_0^{2\pi} = 2\int_0^\pi$. 
We split the integrals according according to whether $x\leq 1/M$ or $x\geq 1/M$. 
For $x \leq 1/M$ we have 
\[
\begin{aligned}
\int_0^{1/M} \abs{ \sum_{k=0}^{M} k e^{ikx}} \ud x 
	& \lesssim \int_0^{1/M} M^2 \ud x \lesssim M,
\\
\int_0^{1/M} \abs{ \sum_{k=0}^{M} k^2 e^{ikx}} \ud x 
	& \lesssim \int_0^{1/M} M^3 \ud x \lesssim M^2.
\end{aligned}
\]
For $x \geq 1/M$ we use \Cref{eqn.sum.kq^k} and note that $\abs{e^{ix} - 1} \geq c x$ for $x\leq \pi$.
Expanding the exponentials $e^{ix} = 1 + O(x)$ we thus have 
\begin{align*}
	\int_{1/N}^{\pi} 
	\abs{ \sum_{k=0}^{M} k e^{ikx}} \ud x 
	& \lesssim \int_{1/M}^\pi \frac{1}{x^2} \left[Me^{iMx}(e^{ix} - 1) + 1 - e^{iMx}\right] \ud x
	\\ &
		\lesssim \int_{1/M}^\pi \left(\frac{M}{x} + \frac{1}{x^2}\right) \ud x \lesssim M \log M
	\\ \intertext{and}
	\int_{1/N}^{\pi}
	\abs{ \sum_{k=0}^{M} k^2 e^{ikx}} \ud x  
	& \lesssim \int_{1/M}^\pi \frac{1}{x^3} \left[ e^{iMx}\left(M^2 (e^{ix}-1)^2 - 2M (e^{ix}-1) + e^{ix} + 1\right) - e^{ix} - 1\right] \ud x
	\\ & 
		\lesssim \int_{1/M}^\pi \left(\frac{M^2}{x} + \frac{M}{x^2} + \frac{1}{x^3}\right) \ud x \lesssim M^2 \log M.
\end{align*}
This concludes the proof.
\end{proof}

\noindent
With this we may thus bound for ($j=2$, say)
\[
\begin{aligned}
 & \int_{0}^{2\pi} \abs{\sum_{\tau = 0}^h t_2 \left(k_0^{(2)} + \tau \kappa\right) e^{i \tau |\kappa| x}} \ud x
 \\ & \quad 
 = \int_{0}^{2\pi} \abs{\sum_{\tau = 0}^h \left((k_0^1)^2 + 2 k_0^1 \kappa^1 \tau + (\kappa^1)^2 \tau^2\right) e^{i \tau |\kappa| x}} \ud x
 \\ & \quad 
 \leq CR^2 \int_{0}^{2\pi} \abs{\sum_{\tau = 0}^h e^{i \tau |\kappa| x}} \ud x
  + CR|\kappa| \int_{0}^{2\pi} \abs{\sum_{\tau = 0}^h \tau e^{i \tau |\kappa| x}} \ud x
  + C|\kappa|^2 \int_{0}^{2\pi} \abs{\sum_{\tau = 0}^h \tau^2 e^{i \tau |\kappa| x}} \ud x\,.
\end{aligned}
\]
Substituting $y = |\kappa|x$, using \Cref{lem.eval.1d.integrals} and recalling that $h|\kappa|\lesssim R$ and $|\kappa|\geq 1$ by \Cref{lem.extra.term.integers}
we may bound this by $R^2\log R$.
An analoguous bound holds for $j=1$.
This takes care of all the one-dimensional integrals.
In combination  with \Cref{lemma.derivative.lebesgue.induction,lemma.derivative.lebesgue.induction.d=2}
we get the bounds for $j=1,2$ of \Cref{eqn.derivative.lebesgue.constant.general.tetrahedron}.
It remains to consider the two-dimensional integral for $j=3$.

\subsection{Bounding the \texorpdfstring{$j=3$}{j=3} two-dimensional integral}
We are left with bounding the integral $\int |\tilde D_2^3|$ on the right-hand-side of \Cref{lemma.derivative.lebesgue.induction}.
We first reduce to the case of a simpler tetrahedron (triangle).
By shifting the sums by a fixed $\kappa = (\kappa^1,0)\in\Z^2$ 
and using the bounds in \Cref{lem.eval.1d.integrals} to evaluate the extra contributions of the shift, 
we may assume that the $k^1$-sum starts at $0$.
By splitting the $k^2$-sum as in \Cref{sec.reduction.tetrahedron}
we may assume that that $k^2$-sum also starts at $0$.
That is, we  need to evaluate the integral
\[	
	\iint_{[0,2\pi]^2} \abs{ \sum_{k=0}^{[\Lambda_1]} \sum_{\ell =0}^{[\Lambda_2(k)]} k \ell e^{ikx} e^{i\ell y}} \ud x \ud y, 
\]
where 
$\Lambda_2(k) = n_2 - \frac{Q_1\alpha_1}{Q_2 \alpha_2} k$ for an irrational $n_2$. 
Recall that $|\Lambda_1|\leq CR$ and for any $0\leq k\leq [\Lambda_1]$ we have $|\Lambda_2(k)|\leq CR$.

The analysis given here is in spirit the same as given in \Cref{sec.induction.d=3,sec.induction.d=2,sec.1d.integrals}.
It is sufficiently different
that we find it easier to do the arguments separately. We shall show the following. 
\begin{lemma}\label{lem.compute.kl.sum.integral}
We have the following bound
\[ 
	\displaystyle \iint_{[0,2\pi]^2} \abs{ \sum_{k=0}^{[\Lambda_1]} \sum_{\ell =0}^{[\Lambda_2(k)]} k \ell e^{ikx} e^{i\ell y}} \ud x \ud y 
	\leq CR^2 (\log R)^2 \log Q.
\]
\end{lemma}

\noindent
Combining then \Cref{lemma.derivative.lebesgue.induction,lem.eval.1d.integrals,lem.compute.kl.sum.integral,lemma.derivative.lebesgue.induction.d=2} 
and choosing $Q$ some sufficiently large power of $N$ as required in \Cref{defn.P_F.true}
we conclude the proof of \Cref{eqn.derivative.lebesgue.constant.general.tetrahedron} 
and thus of 
\Cref{lem.derivative.lebesgue.constant}.
It remains to give the proof of \Cref{lem.compute.kl.sum.integral}.
\begin{proof}
Denote $M = [\Lambda_1]$ and recall
$\Lambda_2(k) = n_2-m_1k = n_2 - \frac{Q_1\alpha_1}{Q_2 \alpha_2} k$.
First note that by mapping $\ell \mapsto [\Lambda_2(k)] - \ell$ we may assume that $m_1 \geq 0$ . 
If $m_1=0$ the sum factors, and so does the integral into two one-dimensional sums/integrals. These may be bounded 
using \Cref{lem.eval.1d.integrals}. 
In this case we get the bound $\leq CR^2 (\log R)^2$ as desired. Hence assume that $m_1 > 0$. 
Moreover, if $n_2 > m_1 M$ we may split the $(k,\ell)$-sum into two parts, 
\[
\sum_{k=0}^{M}\sum_{\ell = 0}^{[\Lambda_2(k)]} 
	= \sum_{k=0}^{M}\sum_{\ell = 0}^{[n_2-m_1M]} + \sum_{k=0}^{M}\sum_{\ell = [n_2-m_1M]+1}^{[\Lambda_2(k)]}.
\]
The first sum factors into one-dimensional integrals which we may bound using \Cref{lem.eval.1d.integrals} again. 
The second we may shift by a constant $\ell$ (again then using \Cref{lem.eval.1d.integrals} to evaluate the contribution of the shift) and assume that 
the lower limit of the $\ell$-sum is $0$. 
The upper limit then becomes $[\Lambda(k)]$, where 
\[
	\Lambda(k) = n_2 - ([n_2 - m_1M] + 1) - m_1 k := n - mk\,.
\]
Geometrically, this means that the domain of the $(k,\ell)$-sum is 
a triangle with two sides along the axes. We thus need to bound 
\[
\iint_{[0,2\pi]^2} \abs{ \sum_{k=0}^{M} \sum_{\ell =0}^{[\Lambda(k)]} k \ell e^{ikx} e^{i\ell y}} \ud x \ud y,
\]
where 
\[
\Lambda(k) = n - mk, \qquad 
M \leq R, \qquad mM = n + O(1), \qquad n \leq R. 
\]
By the symmetries of translation invariance and complex conjugation we may integrate over the domain $[-\pi, \pi]\times [0, \pi]$ instead.
We  evaluate the $\ell$-sum using \Cref{eqn.sum.kq^k}. 
Recall that $[\Lambda(k)] = \Lambda(k) - \expect{\Lambda(k)}$.
We thus have 
\[
\begin{aligned}
	\sum_{\ell=0}^{[\Lambda(k)]} \ell e^{i\ell y} 
	& = \frac{e^{iy}}{(e^{iy} - 1)^2} \left[e^{i[\Lambda(k)] y} ([\Lambda(k)] e^{iy} - [\Lambda(k)] - 1) + 1\right]
	\\
	& = \frac{e^{iy}}{(e^{iy} - 1)^2} \left[
		\left(e^{i\Lambda(k) y}(\Lambda(k) e^{iy} - \Lambda(k) - 1) + 1\right)
		\right.
	\\ & \qquad 
		+ \left(e^{-i\expect{\Lambda(k)}y} - 1\right)\left(e^{i\Lambda(k) y}(\Lambda(k) e^{iy} - \Lambda(k) - 1)+ 1\right)
	\\ & \qquad 
		\left.
		- \left(\expect{\Lambda(k)}e^{i(\Lambda(k) - \expect{\Lambda(k)})y}(e^{iy} - 1) + (e^{-i\expect{\Lambda(k)}y} - 1)\right)
		\right]
	\\ & =: (\textrm{I}) + (\textrm{II}) + (\textrm{III}).
\end{aligned}
\]
The third summand $(\textrm{III})$ may be calculated as
\[
	\frac{- e^{iy}}{(e^{iy} - 1)^2}\left[ \expect{\Lambda(k)} \left(e^{i[\Lambda(k)]y} - 1\right) iy + O(y^2)\right].
\]
The factor $\frac{- e^{iy}}{(e^{iy} - 1)^2}$ may be bounded by $1/y^2$.
For this term we split the $y$-integral according to whether $y \leq 1/n$ or $y \geq 1/n$.
For $y \leq 1/n$ we expand additionally $e^{i[\Lambda(k)]y} - 1 = O(ny)$. We get the contribution
\[
\int_{-\pi}^\pi \ud x \int_{0}^{1/n} \ud y \frac{1}{y^2}
	\abs{ \sum_{k=0}^{M} k e^{ikx}\left[ \expect{\Lambda(k)} \left(e^{i[\Lambda(k)]y} - 1\right)y + O(y^2)\right]}
	\lesssim \frac{1}{n} M^2 n + \frac{1}{n} M^2 \lesssim R^2.
\]
For $y \geq 1/n$ we bound $e^{i[\Lambda(k)]y} - 1 = O(1)$. We get 
\[
\int_{-\pi}^\pi \ud x \int_{1/n}^\pi \ud y \frac{1}{y^2}\abs{ \sum_{k=0}^{M} k e^{ikx}\left[ \expect{\Lambda(k)} \left(e^{i[\Lambda(k)]y} - 1\right)y + O(y^2)\right]}
	\lesssim (\log n) M^2 + M^2 \lesssim R^2 \log R.
\]
For the second summand $(\textrm{II})$ we again split the integral according to whether $y \leq 1/n$ or $y \geq 1/n$.
If $y \leq 1/n$ we have 
\[
\frac{e^{iy}}{(e^{iy} - 1)^2} \left(e^{-i\expect{\Lambda(k)}y} - 1\right)\left(e^{i\Lambda(k) y}(\Lambda(k) e^{iy} - \Lambda(k) - 1)+ 1\right) 
= O(\Lambda(k)^2 y) = O(n).
\]
Hence this contributes the term 
\[
	\int_{-\pi}^\pi \ud x \int_0^{1/n} \ud y \abs{ \sum_{k=0}^M k e^{ikx} O(\Lambda(k)^2 y)} 
	\lesssim \frac{1}{n} M^2 n \lesssim R^2.
\]
For $y \geq 1/n$ we write 
\[
\begin{aligned}
	& \frac{e^{iy}}{(e^{iy} - 1)^2} \left(e^{-i\expect{\Lambda(k)}y} - 1\right)\left(e^{i\Lambda(k) y}(\Lambda(k) e^{iy} - \Lambda(k) - 1)+ 1\right)
	\\ & \quad 
		= \frac{e^{iy}}{(e^{iy} - 1)^2} \sum_{\nu =1}^{\infty} \frac{(-iy)^\nu }{\nu !} 
			\expect{\Lambda(k)}^\nu  \left(e^{i\Lambda(k) y}(\Lambda(k) e^{iy} - \Lambda(k) - 1)+ 1\right).
\end{aligned}
\]
Again we bound the factor $\frac{e^{iy}}{(e^{iy} - 1)^2}$ as $1/y^2$.
We treat each summand 
similarly as in \Cref{lem.lemma.3.6,lem.lemma.3.9}
(or rather, the $2$-dimensional version of these as used in \Cref{sec.induction.d=2}.)
Completely analogously to \Cref{lem.lemma.3.6} we see that for any integer $r\geq 1$ we have 
\[
\begin{aligned}
& 
	\iint \abs{\sum_{k=0}^M k e^{ikx} \frac{e^{i\Lambda(k) y}(\Lambda(k) e^{iy} - \Lambda(k) - 1)+ 1}{y^2} \expect{\Lambda(k)}^r} \ud x \ud y
\\ & \quad 
	\leq \iint \abs{\sum_{k=0}^M k e^{ikx} \frac{e^{i\Lambda(k) y}(\Lambda(k) e^{iy} - \Lambda(k) - 1)+ 1}{y^2}} \ud x \ud y 
\\& \qquad 
	+ \sum_{\nu = 1}^{r}  \binom{r}{\nu} \iint \abs{\sum_{k=0}^M k e^{ikx} \frac{e^{i\Lambda(k) y}(\Lambda(k) e^{iy} - \Lambda(k) - 1)+ 1}{y^2} \{\Lambda(k)\}^\nu}\ud x \ud y,
\end{aligned}
\]
for either definition of $\expect{\cdot}$ (i.e. either $\expect{\cdot} = \{\cdot\}$ or $\expect{\cdot} = \cdot - \lceil\cdot\rceil$).
Also the application of \Cref{lem.lemma.3.9} is analogous to its use in \Cref{sec.induction.d=2}.
There is at most one $k$ such that $\tilde\Lambda (k) \in \Z$ for the appropriate rational approximation $\tilde\Lambda$ of $\Lambda$. 
Using that $e^{iy} = 1 + O(y)$ we obtain the 
bound
\[  
\abs{k e^{ikx}\frac{e^{i\Lambda(k) y}(\Lambda(k) e^{iy} - \Lambda(k) - 1)+ 1}{y^2}}
\lesssim M\frac{ny + 1}{y^2} \lesssim \frac{R^2}{y} + \frac{R}{y^2},
\]
valid for any $k$. 
Hence this error term contributes at most 
\[
	\int_{-\pi}^\pi \ud x \int_{1/n}^\pi \ud y \left(\frac{R^2}{y} + \frac{R}{y^2}\right)
	\lesssim R^2\log n + Rn \lesssim R^2 \log R.
\]
The rest of the argument in \Cref{lem.lemma.3.9} is the same.
We  conclude that we may bound the contribution of the term $(\textrm{II})$ by 
that of $({\textrm I})$ up to a factor of $\log Q$ and an error $R^2\log R$, i.e.
\[
	\iint \abs{\sum_k k e^{ikx} (\textrm{II})} \ud x \ud y \lesssim \log Q \iint \abs{\sum_k k e^{ikx} ({\textrm I})} \ud x \ud y + R^2 \log R.
\]
In particular  
\begin{equation}
\label{eqn.compute.kl.integral.reduce.bracket}
\begin{aligned}
& \iint_{[0,2\pi]^2} \abs{ \sum_{k=0}^{M} \sum_{\ell =0}^{[\Lambda(k)]} k \ell e^{ikx} e^{i\ell y}} \ud x \ud y
\\ & \quad \lesssim \log Q \int_{-\pi}^\pi \ud x \int_0^\pi \ud y 
	\abs{ \sum_{k=0}^{M} k e^{ikx} \frac{e^{i\Lambda(k) y}(\Lambda(k) e^{iy} - \Lambda(k) - 1) + 1}{y^2}} 
	+ R^2\log R. 
\end{aligned}
\end{equation}
In order to evaluate the integral on the right-hand side, we split the integration domain into $5$ regions, see \Cref{fig.split.integrals.sum.kl}.
\[
\begin{aligned}
	I_1 & = \{ |x| \leq 2/M, y \leq 2/n\},
	&&&
	I_2 & = \{|x| \leq 1/M, y \geq 2/n\},
\\
	I_3 & = \{ y \leq 1/n, |x| \geq 2/M\},
	&&&
	I_4 & = \{y\geq 1/n, |x| \geq 1/M, |x-my| \geq 1/M\}
\\
	I_5 & = \{|x-my| \leq 1/M, (x,y)\notin I_1\}.
\end{aligned}
\]
We will be a bit sloppy with notation and refer to both the domain of integration 
and the value of the integration over that domain by $I_j$.
\begin{figure}[htb]
\center 
\begin{tikzpicture}[line cap=round,line join=round,>=triangle 45,x=1cm,y=1cm]
\draw  (-6.,0.)-- (6.,0.);
\draw  (6.,0.)-- (6.,6.);
\draw  (6.,6.)-- (-6.,6.);
\draw (-6,0) -- (-6,6);
\draw  (2.,1.5)-- (6.,1.5);
\draw (-6,1.5)-- (-2,1.5);
\draw  (6.,1.5)-- (6.,0.);
\draw  (2.,0.)-- (2.,3.);
\draw (-2,0) -- (-2,3);
\draw  (2.,1.5)-- (5.,6.);
\draw (-2,3) -- (2,3);
\draw (1,3) -- (1,6);
\draw (-1,3) -- (-1,6);
\draw (1,3) -- (3,6);
\draw (6,0) node[anchor=north west] {$x$};
\draw (0,6) node[anchor=south east] {$y$};
\draw (2,0) node[anchor=north] {$2/M$};
\node[anchor=north] at (-2,0) {$-2/M$};
\draw (0,3) node[anchor=north east] {$2/n$};
\draw (4,6) node[anchor=south] {$x=my$};
\node (1) at (0,1.5) {$I_1$};
\node (2) at (0,4.5) {$I_2$};
\draw (4,0.75) node {$I_3$};
\node at (-4,0.75) {$I_3$};
\draw (5,3) node {$I_4$};
\node at (-4,4) {$I_4$};
\node at (1.5,5) {$I_4$};
\node (5) at (3,4.5) {$I_5$};
\draw[dashed] (0,0) -- (5) -- (4,6);
\draw[dashed] (0.,0.) -- (1) -- (2) -- (0.,6.);
\end{tikzpicture}
\caption{Decomposition of the domain $[-\pi,\pi]\times [0,\pi]$ into different regions.
}
\label{fig.split.integrals.sum.kl}
\end{figure}
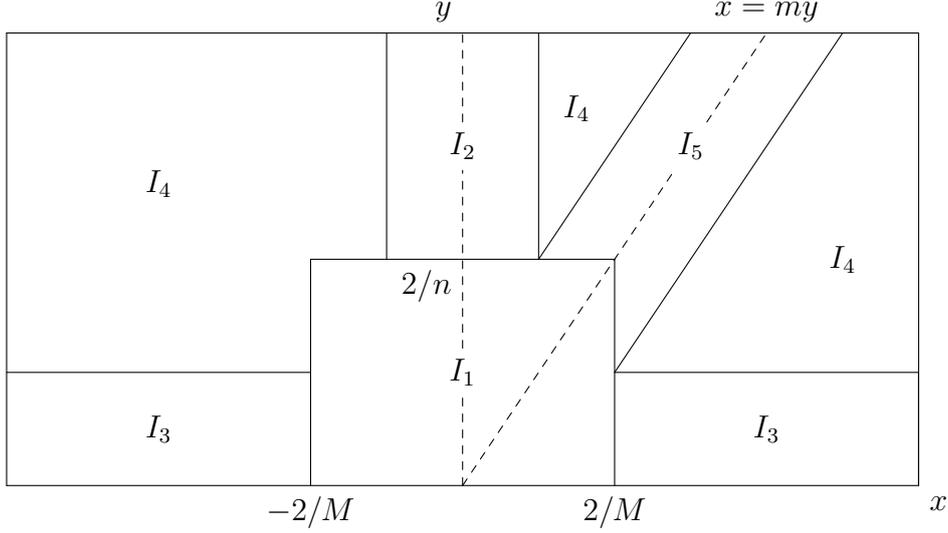

$(I_1)$.
We expand 
\[
(*) := \sum_{k=0}^{M} k e^{ikx} \frac{e^{i\Lambda(k) y}(\Lambda(k) e^{iy} - \Lambda(k) - 1) + 1}{y^2}
\] 
(or rather the numerator) to second order in $y$.
Using that $\Lambda(k) = O(n)$ we get that $(*)  \lesssim M^2n^2$.
Thus the integral gives 
\[
	I_1 \lesssim \int_{-2/M}^{2/M} \ud x \int_0^{2/n} \ud y M^2 n^2 \lesssim Mn \lesssim R^2.
\]

$(I_2)$.
We expand $e^{iy} = 1 + O(y)$ in $(*)$. Then 
$(*) \lesssim \frac{M^2n}{y} + \frac{M^2}{y^2}$. The integral is then 
$I_2 \lesssim R^2\log R$.

$(I_3,I_4,I_5)$.
For the remaining integrals we use the explicit formula for $\Lambda(k) = n - mk$. 
Then 
\begin{equation}
\label{eqn.explicit.evaluation.sum.kl}
\begin{aligned}
 (*)
 	& 
 	= 
 	\frac{1}{y^2}\sum_{k=0}^{M} k e^{ikx} (e^{i\Lambda(k) y}(\Lambda(k) e^{iy} - \Lambda(k) - 1) + 1)
 \\ &  
 	= 
 	\frac{1}{y^2}\sum_{k=0}^{M} \left(k e^{ikx} + k e^{ik(x-my)} e^{iny}(ne^{iy} - n - 1) + k^2 e^{ik(x-my)} e^{iny} (m-me^{iy})\right)
 \\ & 
 	= 
 	\frac{1}{y^2} \left( -i\partial D(x) -i e^{iny}(ne^{iy} - n - 1) \partial D(x-my) + m e^{iny} (e^{iy} - 1)\partial^2 D(x-my)\right),
\end{aligned}
\end{equation}
where we introduced $D(z) = \sum_{k=0}^M e^{ikz} = \frac{e^{i(M+1)z} - 1}{e^{iz} - 1}$. 
From \Cref{eqn.sum.kq^k} we conclude that we may bound derivatives of $D$ as 
\begin{equation}\label{eqn.bound.partialD}
	|\partial D(z)| \lesssim \frac{M}{z} + \frac{1}{z^2},
	\quad 
	|\partial^2 D(z)| \lesssim \frac{M^2}{z} + \frac{M}{z^2} + \frac{1}{z^3},
	\quad 
	|\partial^3 D(z)| \lesssim \frac{M^3}{z} + \ldots + \frac{1}{z^4}.
\end{equation}

$(I_3)$. 
We have $y \leq 1/n$ and $|x| \geq 2/M$. 
We expand \Cref{eqn.explicit.evaluation.sum.kl} to second order in $y$. 
Expanding first the exponentials and then derivatives of $D$ where needed we get 
\[
\begin{aligned}
	\eqref{eqn.explicit.evaluation.sum.kl} 
	& = \frac{1}{y^2}\Big(-i\partial D(x) + i\partial D(x-my) + imy \partial^2 D(x-my) 
	\\ & \qquad 
		+ O(n^2y^2\partial D(x-my)) + O(nmy^2\partial^2 D(x-my))\Big)	
	\\
	& \lesssim n^2 \sup_y|\partial D(x-my)| + nm \sup_y|\partial^2 D(x-my)| + m^2 \sup_y |\partial^3 D(x-my)|.
\end{aligned}
\]
Now we use the bounds \Cref{eqn.bound.partialD} and use that $z := x - my$ has $|z|\geq |x| - m/n = |x| - 1/M + O(1/(Mn))$ 
(recall that $mM = n + O(1)$) and $|x| \geq 2/M$.
Thus 
\[
	I_3 \lesssim \frac{1}{n} \int_{1/M}^\pi \ud z \left(n^2 |\partial D(z)| + nm |\partial^2 D(z)| + m^2  |\partial^3 D(z)|\right)
	\lesssim R^2 \log R.
\]

$(I_4)$. We expand the exponentials $e^{iy} = 1 + O(y)$. 
Then 
\[
	\abs{\eqref{eqn.explicit.evaluation.sum.kl}}
	\leq \frac{\abs{\partial D(x)}}{y^2} + \frac{n|\partial D(x-my)|}{y} + \frac{|\partial D(x-my)|}{y^2} + \frac{m |\partial^2 D(x-my)|}{y}.
\]
Using the bounds \Cref{eqn.bound.partialD} as before 
and noting that $|x|\geq 2/M$ and $z = x-my$ has $|z|\geq 1/M$
one easily sees that $I_4 \lesssim R^2 (\log R)^2$.


$(I_5)$. 
Again, expanding the exponentials $e^{iy} = 1 + O(y)$ we have as for $I_4$ that 
\[
	\abs{\eqref{eqn.explicit.evaluation.sum.kl}}
	\leq \frac{\abs{\partial D(x)}}{y^2} + \frac{n|\partial D(x-my)|}{y} + \frac{|\partial D(x-my)|}{y^2} + \frac{m |\partial^2 D(x-my)|}{y}.
\]
We use the bounds 
\[
	|\partial D(z)| = \abs{\sum_{k=0}^M k e^{ikz}} \leq M^2, 
	\qquad 
	|\partial^2 D(z)| = \abs{\sum_{k=0}^M k^2 e^{ikz}} \leq M^3.
\]
Thus 
\[
	I_5 \lesssim \frac{1}{M} \int_{1/n}^{\pi} \frac{M^2}{y^2} + \frac{nM^2 + mM^3}{y} \ud y \lesssim R^2 \log R.
\]
We conclude that 
\[
\int_{-\pi}^\pi\ud x \int_0^\pi \ud y \abs{ \sum_{k=0}^{M} k e^{ikx} \frac{e^{i\Lambda(k) y}(\Lambda(k) e^{iy} - \Lambda(k) - 1) + 1}{y^2}} 
\lesssim R^2 (\log R)^2.
\]
Together with \Cref{eqn.compute.kl.integral.reduce.bracket} this concludes the proof.
\end{proof}

\printbibliography
\end{document}